\documentclass{llncs}

\usepackage[utf8]{inputenc}
 \usepackage{thm-restate}
\usepackage{amssymb}
\usepackage{amsmath}

\usepackage{thm-restate}

 \usepackage{color}
 \definecolor{gris}{gray}{0.85}
\usepackage{msc}
\usepackage{xspace}
\usepackage{multirow}

\usepackage{centernot}
\usepackage{mathtools}
\usepackage{stmaryrd}

\makeatletter
\newcommand{\rightdoublearrow}{%
  \rightarrow\mkern-10mu\protect\joinrel\rightarrow}

\newcommand{\rightdoublearrowfill@}
  {\arrowfill@\relbar\relbar\rightdoublearrow}
\newcommand{\mapstodoublearrowfill@}
  {\arrowfill@{\mapstochar\relbar}\relbar\rightdoublearrow}
\newcommand{\xrightdoublearrow}[2][]
  {\ext@arrow 3{15}59\rightdoublearrowfill@{#1}{#2}}
\newcommand{\xmapstodoublearrow}[2][]
  {\ext@arrow 3{15}59\mapstodoublearrowfill@{#1}{#2}}
\makeatother

\title{Composing security protocols: from confidentiality to privacy
}

\author{Myrto Arapinis\inst{1} \and Vincent Cheval\inst{2} \and St\'ephanie Delaune\inst{3}}

\institute{
School of Informatics, University of Edinburgh, UK
\and
LORIA, CNRS, France
\and
LSV, CNRS \& ENS Cachan,  France
}

\newcommand{\seq}[1]{\overline{#1}}

\newcommand{\Channel}{\mathcal{C}{h}} 
\newcommand{\N}{\mathcal{N}}
\newcommand{\M}{\mu}
\newcommand{\X}{\mathcal{X}}
\newcommand{\T}{\mathcal{T}}

\newcommand{\E} {\mathsf{E}}
\newcommand{\Sigmazero}{\Sigma_0}
\newcommand{\Ezero}{\E_0}

\newcommand{\dom}{\operatorname{dom}}

\newcommand{\defi}{\mathsf{def}}
\newcommand{\st}{\mathit{st}}   

\newcommand{\mydownarrow}{\mathord{\downarrow}}

\newcommand{\pair}[2]{\langle#1,#2\rangle}
\newcommand{\proj}{\mathsf{proj}}
\newcommand{\senc}{\mathsf{senc}}
\newcommand{\aenc}{\mathsf{aenc}}
\newcommand{\sdec}{\mathsf{sdec}}
\newcommand{\adec}{\mathsf{adec}}
\newcommand{\sencDiffie}{\mathsf{senc}_{DH}}
\newcommand{\aencDiffie}{\mathsf{aenc}_{DH}}
\newcommand{\sdecDiffie}{\mathsf{sdec}_{DH}}
\newcommand{\adecDiffie}{\mathsf{adec}_{DH}}
\newcommand{\sign}{\mathsf{sign}}
\newcommand{\pk}{\mathsf{pk}}
\newcommand{\pkDiffie}{\mathsf{pk}_{DH}}
\newcommand{\checksign}{\mathsf{check}}
\newcommand{\h}{\mathsf{h}}

\newcommand{\mac}{\mathsf{mac}}
\newcommand{\vk}{\mathsf{vk}}

\newcommand{\Let}{\texttt{let}\xspace}

\newcommand{\ffun}{\mathsf{f}}
\newcommand{\gfun}{\mathsf{g}}

\newcommand{\Diffie}{\mathsf{DH}}

\newcommand{\fct}{\mathit{Fct}}
\newcommand{\racine}{\mathsf{root}}
\newcommand{\racinebis}{\mathsf{tagroot}}
\newcommand{\Or}{\mathcal{O}}

\newcommand{\unfolding}[2]{\mathsf{Unf}_{#2}(#1)}

\newcommand{\fv} {\mathit{fv}}
\newcommand{\bv} {\mathit{bv}}
\newcommand{\fn} {\mathit{fn}}
\newcommand{\bn} {\mathit{bn}}

\newcommand{\triple}[3]{(#1 ; #2 ; #3)}
\newcommand{\quadruple}[4]{(#1 ; #2; #3; #4)}

\newcommand{\Ec}{\mathcal{E}}
\newcommand{\p}{\mathcal{P}}
\newcommand{\q}{\mathcal{Q}}
\newcommand{\In} {\texttt{in}}
\newcommand{\Out}{\texttt{out}}
\newcommand{\myIf}{\texttt{if}\xspace}
\newcommand{\myElse}{\texttt{else}\xspace}
\newcommand{\myThen}{\texttt{then}\xspace}
\newcommand{\new}{\texttt{new}\xspace}

\newcommand{\tr}{\mathsf{tr}}
\newcommand{\bad}{\mathit{bad}}

\newcommand{\TAGG}[1]{[\! [ #1 ]\! ]}

\newcommand{\Tag}{\mathsf{tag}}
\newcommand{\unTag}{\mathsf{untag}}
\newcommand{\TAG}[2]{[#1]_{{#2}}}
\newcommand{\TestTag}[2]{\mathsf{test}_{#2}(#1)}
\newcommand{\Flawed}[1]{\mathsf{Flawed}(#1)}
\newcommand{\FlawedColor}[2]{\mathsf{Flawed}^{#1}(#2)}
\newcommand{\true}{\mathsf{true}}
\newcommand{\key}{\mathit{key}}

\newcommand{\id}{\mathit{id}}

\newcommand{\sk}{\mathit{sk}}

\newcommand{\diff}{\mathsf{diff}}
\newcommand{\fst}{\mathsf{fst}}
\newcommand{\snd}{\mathsf{snd}}

\newcommand{\bi}{\mathsf{bi}}

\makeatletter
\def\rightarrowfillstar@{\arrowfill@\relbar\relbar{\rightarrow\smash{^*}}}
\newcommand{\xrightarrowstar}[2][]{\ext@arrow
  0{13}{15}8\rightarrowfillstar@{#1}{#2}}
\newcommand{\lrstep}{\@ifstar{\xrightarrowstar}{\xrightarrow}}
\makeatother

\newcommand{\LRstep}[1]{{\xRightarrow{#1}}}
\newcommand{\statequiv}{\sim}
\newcommand{\trace}{\textsf{trace}}

\newcommand{\refer}{\vartriangleright}

\newcommand{\vcol}{\mathit{col}}

\newcommand{\KPrAA}{\mathit{sk_P}}
\newcommand{\KPrDS}{\mathit{sk_{DS}}}
\newcommand{\KPrSN}{\mathit{sk_{SN}}}
\newcommand{\DG}{\mathit{dg}}
\newcommand{\SOD}{\mathit{sod}}
\newcommand{\PAuth}{\textit{PA}}
\newcommand{\AAuth}{\textit{AA}}
\newcommand{\BAC}{\textit{BAC}}
\newcommand{\AKA}{\textit{AKA}}

\newcommand{\sSMS}{\textit{sSMS}}

\newcommand{\ksenc}{\mathit{xkenc}}
\newcommand{\ksmac}{\mathit{xkmac}}
\newcommand{\KE}{\mathit{ke}}
\newcommand{\KM}{\mathit{km}}
\newcommand{\KIMSI}{\mathit{k_{IMSI}}}
\newcommand{\getC}{\mathsf{get\_C}}
\newcommand{\New}{\mathsf{new}}

\newcommand{\If}{\texttt{if}}
\newcommand{\Else}{\texttt{else}}
\newcommand{\Then}{\texttt{then}}
\newcommand{\yes}{\mathsf{yes}}
\newcommand{\no}{\mathsf{no}}
\newcommand{\ok}{\mathsf{ok}}
\newcommand{\ko}{\mathsf{ko}}

\newcommand{\ndisjunction}[4]{\mathit{Ass}_{#3 := #2}^{#4}(#1)}
\newcommand{\ass}{\mathsf{ass}}

\newcommand{\imsi}{\mathit{IMSI}}

\newcommand{\sms}{\mathit{sms}}

\newcommand{\para}{\mathsf{par}}
\newcommand{\sequ}{\mathsf{seq}}

\begin{document}

\maketitle

\begin{abstract}
  Security protocols are 
  used in many of our daily-life applications, and our privacy largely
  depends on their design.  Formal verification techniques have proved
  their usefulness to analyse these protocols, but they become so
  complex that modular techniques have to be developed.
  We propose several results 
  to safely compose security protocols.  We consider arbitrary
  primitives modeled using an equational theory, and a rich process
  algebra close to the applied pi calculus.

  Relying on these composition results, we are able to derive some
  security properties on a protocol from the security analysis
  performed on each of its sub-protocols individually. We consider
  parallel composition and the case of key-exchange protocols. Our
  results apply to deal with confidentiality but also privacy-type
  properties (\emph{e.g.}  anonymity, unlinkability) expressed using a
  notion of equivalence.  We illustrate the usefulness of our
  composition results on protocols from the 3G phone application and electronic passport.
\end{abstract}

\section{Introduction}
\label{sec:intro}

Privacy means that one can control when, where, and how information
about oneself is used and by whom, and it is actually an important
issue in many modern applications. For instance, nowadays, it is
possible to wave an electronic ticket,
a building access card, a government-issued ID, or even a smartphone
in front of a reader to go through a gate, or to pay for some
purchase.  Unfortunately, as often reported by the media, this technology also makes it possible for
anyone to capture some of our personal information. 
%
To secure the applications mentioned above and to protect our
  privacy, some specific cryptographic protocols are deployed.  For
  instance, the 3G telecommunication application allows one to send
  SMS encrypted with a key that is established with the \AKA\xspace
  protocol~\cite{TS33102}.  The aim of this design is to provide some
  security guarantees: \emph{e.g.} the SMS exchanged between phones
  should remain confidential from third parties.



Because security protocols are notoriously difficult to design and
analyse, formal verification techniques are 
important. They
have become mature and have known several 
successes. 
For instance, a flaw has been discovered in the Single-Sign-On
protocol used 
by Google Apps~\cite{DBLP:conf/ccs/ArmandoCCCT08}, and several
verification tools are nowadays
available (\emph{e.g.} ProVerif~\cite{BlanchetAF08}, the AVANTSSAR
platform~\cite{sysdesc-TACAS2012}).  These tools perform well in
practice, at least for standard security properties (\emph{e.g.}
secrecy, authentication).
Regarding privacy properties, the techniques and tools are more
recent.  Most of the verification techniques are only able to analyse
a bounded number of sessions and consider a quite restrictive class of
protocols (\emph{e.g.} fixed set of primitives and/or no conditional
branching~\cite{Tiu-csf10}).  A slightly different approach
consists in 
analysing a stronger notion of equivalence, namely
\emph{diff-equivalence}.  In particular, ProVerif implements a
semi-decision procedure for checking
diff-equivalence~\cite{BlanchetAF08}.


Security protocols used in practice are more and more complex and it
is difficult to analyse them altogether.  For example, the UMTS
standard~\cite{TS33102} specifies tens of sub-protocols running
concurrently in 3G phone systems.
%
While one may hope to verify each protocol in isolation, it is however
unrealistic to expect that the whole application will be checked
relying on a unique automatic tool.
%
Existing tools have their own specificities that prevent them
to be used in some cases.
Furthermore, most of the techniques do not scale up well on large
systems, and sometimes the ultimate solution is to rely on a manual
proof. It is therefore important that the protocol under study is as
small as possible.


\paragraph*{Related work.}
There are many results studying the composition of security protocols
in the symbolic model~\cite{GuttmanT00,CD-fmsd08,CC-csf10}, as well as
in the computational
model~\cite{Canetti-FOCS04,KuestersTuengerthal-CCS-2011} in which the
so-called UC (universal composability) framework has been first
developed before being adapted in the symbolic
setting~\cite{symbolic-uc}.  This result belongs to the first
approach.
Most of the existing composition results are concerned with
trace-based security properties, and in most cases only with secrecy
(stated as a reachability property),
\emph{e.g.}~\cite{GuttmanT00,CD-fmsd08,CC-csf10,Gross-csf11}. They are
quite restricted in terms of the class of protocols that can be
composed, \emph{e.g.} a fixed set of cryptographic primitives and/or
no else branch. Lastly, they often only consider parallel
composition. Some notable exceptions are the results presented
in~\cite{MV-esorics09,Gross-csf11,CC-csf10}.
This paper is clearly inspired 
from the approach developed  in~\cite{CC-csf10}.

Regarding privacy-type properties, very few composition results
exist. In a previous work~\cite{ACD-csf12}, we consider parallel
composition only. More precisely, we identify sufficient conditions under which protocols can ``safely'' be
executed in parallel as long as they have been proved secure in
isolation. This composition theorem is quite general from the point of
view of the cryptographic primitives allowed. We consider 
arbitrary primitives that can be modelled by a set of equations, and
protocols may share some standard primitives provided  they are
tagged differently. 
We choose to reuse this quite general setting in this work. 


\paragraph*{Our contributions.}
Our main goal is to analyse privacy-type properties in a modular way.
These security properties  are usually expressed as equivalences
between processes. 
Roughly, two processes $P$ and $Q$ are equivalent ($P \approx Q$) if, 
however they behave, the messages observed by the attacker are
indistinguishable. 
Actually, it is well-known that: \\[1mm]
\null\hfill$\mbox{if }P_1 \approx P_2 \mbox{ and } Q_1 \approx Q_2 \mbox{ then }
P_1 \mid P_2 \approx Q_1 \mid Q_2.$\hfill\null

\smallskip{}

%
However, this parallel
composition result works because the processes that are composed are
disjoint (\emph{e.g.} they share no key).  Moreover, here we
want to go beyond parallel composition.  In particular, we want to
capture the case where a protocol uses a sub-protocol to establish
some keys.

To achieve our goal, we first enrich the applied pi calculus with an
assignment construction. This will allow us to share some data
(\emph{e.g.} session keys) between
sub-protocols. Our calculus is presented in Section~\ref{sec:model}.
In Section~\ref{sec:confidentiality-par}, we present a first composition result to deal
with confidentiality properties. The purpose of
this section is to review the difficulties that arise when composing
security protocols even in a simple setting.
In Section~\ref{sec:confidentiality-seq}, we go beyond parallel composition, and we consider
the case of key-exchange protocols. 
We present in
Section~\ref{sec:difficulties-equiv} some additional  difficulties that arise when we
want to consider privacy-type properties expressed using trace
equivalence. In Section~\ref{sec:privacy}, we present our composition results for
privacy-type properties. We consider parallel composition as well as
the case of key-exchange protocols.

Actually, all these composition results are derived from
a generic composition result which is quite technical and presented
only  in Appendix~\ref{sec:app-compo-main}.
This result allows one to map a
trace of the composed protocol into a trace of a disjoint case
(protocol where the sub-protocols do not share any data), {and
  conversely}. It can be seen as an extension of the result presented
in~\cite{CC-csf10} where only a mapping from the shared case to the
disjoint case is provided (and not the converse).  Moreover, we
consider a richer process algebra than the one used
in~\cite{CC-csf10}. In particular, we are able to deal with protocols
with else branches and to compose protocols that both rely on
asymmetric primitives (\emph{i.e.} asymmetric encryption and
signature).

In Section~\ref{sec:casestudies}, we illustrate the usefulness of our composition results 
on  protocols from the 3G phone application, as well as on 
protocols from the e-passport application.
 We show how to derive some
security guarantees from the analysis performed
on each sub-protocol in isolation.

 \section{Models for security protocols}
\label{sec:model}

Our calculus is close to the applied pi
calculus~\cite{AbadiFournet2001}.
We consider an assignment operation to make explicit the data that are
shared among different processes.

\subsection{Messages}
\label{subsec:messages}

As usual in this kind of models, messages are modelled using an
abstract term algebra.  We assume an infinite set of \emph{names}~$\N$
of \emph{base type} (used for representing keys, nonces,
\ldots) and a set~$\Channel$ of names of \emph{channel
  type}. 
We also 
consider a set of \emph{variables} $\X$, and a signature
$\Sigma$ consisting of a finite set of \emph{function symbols}.  We
rely on a sort system for terms. The details of the sort system are
unimportant, as long as the base type differs from the channel type,
and
we suppose that function symbols only operate on and return terms of
base type. 

\emph{Terms} are defined as names, variables, and function symbols
applied to other terms.
The set of terms built from $\mathsf{N}\subseteq \N \cup \Channel$,
and $\mathsf{X}\subseteq \X$ by applying function symbols in $\Sigma$
(respecting sorts and arities) is denoted by $\T (\Sigma, \mathsf{N}
\cup \mathsf{X})$.
We write $\fv(u)$ (resp. $\fn(u)$) for the set of variables
(resp. names) occurring in a term~$u$. A term~$u$ is \emph{ground} if it
does not contain any variable, \emph{i.e.} $\fv(u) = \emptyset$.

The algebraic properties of cryptographic primitives are specified by
the means of an \emph{equational theory} which is defined by a finite
set~$\E$ of equations $u = v$ with $u,v \in \T(\Sigma,\X)$,
\emph{i.e.} $u,v$ do not contain names.  We denote by~$=_\E$ the
smallest equivalence relation on terms, that contains~$\E$ and that is
closed under application of function symbols and substitutions of
terms for variables.

\begin{example} 
\label{ex:theory} 
Consider the
  signature~$\Sigma_{\Diffie} = \{\aenc, \adec, \pk, \gfun, \ffun,
  \langle \; \rangle, \proj_1, \proj_2\}$.
  The function symbols~$\adec$, $\aenc$ of arity~2 represent
  asymmetric decryption and encryption. We denote by $\pk(sk)$ the
  public key associated to the private key~$sk$. The two 
  function symbols~$\ffun$ of arity~2, and~$\gfun$ of arity~1 are used
  to model the Diffie-Hellman primitives, whereas the three remaining symbols
  are used to model pairs. The equational theory
  $\E_\Diffie$ is defined by: \\[1mm]
\null\hfill
$\E_\Diffie =
\left\{ \;\;\;
\begin{array}{ccc}
\proj_1(\langle x, y\rangle) = x & \;\;\; & \adec(\aenc(x,\pk(y)),y)=x\\
\proj_2(\langle x, y\rangle) =  y & \;\;\;\;&
  \ffun(\gfun(x),y) = \ffun(\gfun(y),x)
\end{array}
\right.$\hfill\null

\smallskip{}

\noindent  Let $u_0 = \aenc(\langle n_A, \gfun(r_A) \rangle,\pk(sk_B))$.
  We have that:\\[1mm]
\null\hfill
$\ffun(\proj_2(\adec(u_0,sk_B)),r_B) =_{\E_\Diffie} \ffun(\gfun(r_A), r_B)
  =_{\E_\Diffie} \ffun(\gfun(r_B),r_A).$\hfill\null
\end{example}

\subsection{Processes}
\label{subsec:processes}


As in the applied pi calculus, we consider \emph{plain processes} as
well as \emph{extended processes} that represent processes having
already evolved by \emph{e.g.}  disclosing some terms to the
environment.
\noindent \emph{Plain processes} are defined by the following grammar:\\[1mm]
\null\hfill
$\begin{array}{llcl@{\qquad\qquad}lcl}
  P,Q := & 0 &&\mbox{null}
  & P \mid Q &&\mbox{parallel}\\[0.3mm]
  & \new\; n. P &&\mbox{restriction}
  & ! P && \mbox{replication}\\[0.3mm]
  & [x := v]. P && \mbox{assignment}
  & \mbox{\texttt{if} {$\varphi$} \texttt{then} $P$ \texttt{else} $Q$}
  && \mbox{conditional}\\[0.3mm]
  & \In (c,x).P  &&\mbox{input}
  & \Out(c,v).Q && \mbox{output} 
\end{array}$\hfill\null

\smallskip{}

\noindent where $c$ is a name of channel type, {$\varphi$ is a
  conjunction of tests of the form $u_1 = u_2$} where $u_1, u_2$ are
terms of base type, $x$ is a variable of base type, $v$ is a term of
base type, and $n$ is a name of any type.  Note that the terms that
occur in $\varphi$ and $v$ may contain variables.
We consider an assignment operation that instantiates~$x$ with a
term~$v$. Note that we consider private channels but we do not allow
channel passing.
%
%
For the sake of clarity, we often omit the null
process, and 
when there is no ``\texttt{else}'',
it means ``\texttt{else}\, $0$''.

Names and variables have scopes, which are delimited by restrictions,
inputs, {and assignment operations.}  We write $\fv(P)$, $\bv(P)$,
$\fn(P)$ and $\bn(P)$ for the sets of \emph{free} and \emph{bound
  variables}, and \emph{free} and \emph{bound names} of a plain
process~$P$. 

\begin{example}
  \label{ex:process}
Let $P_{\Diffie} = \new\, sk_A. \new \, sk_B. (P_A \mid P_B)$ a
  process that models a Diffie-Hellman key exchange protocol:
  \begin{itemize}
  \item $P_A \stackrel{\mathsf{def}}{=} \new \,r_A . \new \, n_A. \Out(c,
    \aenc(\langle n_A, \gfun(r_A) \rangle,\pk(sk_B))). \In(c,y_A).$\\
\null\hfill $\texttt{if } \proj_1(\adec(y_A, sk_A)) = n_A \texttt{ then }[x_A :=
    \ffun(\proj_2(\adec(y_A,sk_A)),r_A)]. 0$
  \item $P_B \stackrel{\mathsf{def}}{=} \new \,r_B . \In(c,y_B). \Out(c,
    \aenc(\langle \proj_1(\adec(y_B,sk_B)),\gfun(r_B)\rangle,\pk(sk_A))).$\\
\null \hspace{6cm} $[x_B :=
    \ffun(\proj_2(\adec(y_B,sk_B)),r_B)].0$
  \end{itemize}
  The process $P_A$ generates two fresh random numbers~$r_A$ and
  $n_A$, sends a message
 on the channel~$c$, and waits for a
  message containing the nonce $n_A$ in order to compute his own view of the key that will be
  stored in~$x_A$.  The process~$P_B$ proceeds in a similar way and
  stores
the computed value in~$x_B$.
\end{example}

\emph{Extended processes} add a set of restricted names~$\Ec$ (the
names that are \emph{a priori} unknown to the attacker), a sequence of
messages $\Phi$ (corresponding to the messages that have been sent so
far on public channels) and a substitution~$\sigma$ which is used to
store the messages that have been received as well as those that have
been stored in assignment variables.

\begin{definition}
  An \emph{extended process} is a tuple
  $\quadruple{\Ec}{\p}{\Phi}{\sigma}$ with:
  \begin{itemize}
  \item $\Ec$ a set of names that represents the names that are
    restricted in~$\p$, $\Phi$ and~$\sigma$;

  \item $\p$ a multiset of \emph{plain processes} {where null processes
    are removed} and with $\fv(\p)
    \subseteq \dom(\sigma)$;

  \item $\Phi = \{w_1 \refer u_1, \ldots, w_n \refer u_n\}$ where
    $u_1, \ldots, u_n$ are ground terms, $w_1, \ldots, w_n$ are
    variables; 

  \item $\sigma =\{x_1 \mapsto v_1, \ldots, x_m \mapsto v_m\}$ where
    $v_1, \ldots, v_m$ are ground terms, $x_1, \ldots, x_m$ are
    variables.

  \end{itemize}
\end{definition}




For the sake of simplicity,
we assume that extended processes are \emph{name and variable
  distinct}, \emph{i.e.} a name (resp. variable) is either free or bound, and in the
latter case, it is at most bound once.
%
Moreover, we often write~$\triple{\Ec}{P}{\Phi}$ instead of
$\quadruple{\Ec}{P}{\Phi}{\emptyset}$.

 \begin{figure*}[h]
\noindent $
\begin{array}{rclr}
  \quadruple{\Ec}{\{\mbox{\texttt{if} $\varphi$ \texttt{then} $Q_1$ \texttt{else} $Q_2$}\}\uplus\p}{\Phi}{\sigma} 
  &\lrstep{\tau} & \quadruple{\Ec}{Q_1\uplus\p}{\Phi}{\sigma}  &  \mbox{(\sc Then)}\\
  \multicolumn{4}{r}{\mbox{if $u\sigma =_\E v\sigma$ for each $u = v \in \varphi$}}  \\[1mm]
  \quadruple{\Ec}{\{\mbox{\texttt{if} $\varphi$ \texttt{then} $Q_1$ \texttt{else} $Q_2$}\}\uplus\p}{\Phi}{\sigma}
  & \lrstep{\tau} & \quadruple{\Ec}{Q_2\uplus\p}{\Phi}{\sigma} & \mbox{(\sc Else)}\\
  \multicolumn{4}{r}{\mbox{if $u\sigma \neq_\E v\sigma$ for some $u = v \in \varphi$}}  \\[1mm]
  \quadruple{\Ec}{\{\Out(c,u) . Q_1; \In(c,x).Q_2\}\uplus\p}{\Phi}{\sigma}
  & \lrstep{\tau} &
  \multicolumn{2}{l}{\quadruple{\Ec}{Q_1 \uplus Q_2 
      \uplus\p}{\Phi}{\sigma \cup \{x \mapsto {u\sigma}\}}  
    \mbox{(\sc Comm)} }\\[1mm]
  \quadruple{\Ec}{\{[x := v] . Q\}\uplus\p}{\Phi}{\sigma}
  & \lrstep{\tau} &
  \quadruple{\Ec}{Q \uplus\p}{\Phi}{\sigma \cup \{x \mapsto {v\sigma}\}} 
  & \mbox{(\sc Assgn)}
\end{array}
$
\[
\begin{array}{rclr}
  \quadruple{\Ec}{\{\In(c,z).Q\}\uplus\p}{\Phi}{\sigma} 
  & \lrstep{\In(c,M)} & 
  \quadruple{\Ec}{Q \uplus\p}{\Phi}{\sigma \cup\{z \mapsto u\}} &
  \hspace{1.5cm} \mbox{(\sc In)}
  \\ \multicolumn{4}{r}{\mbox{if ${c} \not\in \Ec$, $M\Phi =u$, $\fv(M) \subseteq \dom(\Phi)$ and
      $\fn(M) \cap \Ec = \emptyset$}}
  \\[1mm]
  \quadruple{\Ec}{\{\Out(c,u).Q\}\uplus\p}{\Phi}{\sigma}
  & \lrstep{\nu w_i.\Out(c,w_i)} & \quadruple{\Ec}{Q\uplus\p}{\Phi\cup\{w_i\refer u\sigma\}}{\sigma}&\;\;\;\;
  \mbox{(\sc Out-T)}
  \\
  \multicolumn{4}{r}{\mbox{if ${c} \not\in \Ec$, $u$ is a term of base
      type, and $w_i$ is a variable such that $i = | \Phi| +
      1$}}
\end{array}
\]
%
%
\[
\begin{array}{rclr}
  \hspace{2.4cm}\quadruple{\Ec}{\{\new\; n.Q\} \uplus \p}{\Phi}{\sigma} & \lrstep{\tau} &
  \quadruple{\Ec \cup \{ n \}}{Q\uplus \p }{\Phi}{\sigma}&\hspace{1.2cm} \;\;\;
  \mbox{(\sc New)}
  \\[1mm]
  \quadruple{\Ec}{\{!Q\} \uplus \p}{\Phi}{\sigma} & \lrstep{\tau} &
  \quadruple{\Ec}{\{!Q; Q\rho\} \uplus \p }{\Phi}{\sigma}& \;\;\; \hfill \mbox{(\sc
    Repl)}\\
  \multicolumn{4}{r}{\mbox{$\rho$ is used to rename  $\bv(Q)/ \bn(Q)$ with fresh variables/names}}
  \\[1mm]
  \quadruple{\Ec}{\{P_1 \mid P_2\} \uplus \p}{\Phi}{\sigma} & \lrstep{\tau} &
  \quadruple{\Ec}{\{P_1,P_2\} \uplus \p}{\Phi}{\sigma} & \hfill \mbox{(\sc Par)}
\end{array}
\]



where $n$ is a name, $c$ is a name of channel type, $u$, $v$ are
terms of base type, and $x, z$ are variables of base type. 
\caption{Semantics of extended processes}
\label{fig:semantics}
\end{figure*}

\smallskip{}

The semantics is given by a set of labelled rules that allows one to
reason about processes that interact with their environment (see
Figure~\ref{fig:semantics}).
This defines the relation $\lrstep{\;\ell\;}$ where~$\ell$ is either
an input, an output, or a silent action~$\tau$. The
relation~$\lrstep{\; \tr \;}$ where $\tr$ denotes a sequence of labels
is defined in the usual way whereas the relation $\LRstep{\; \tr' \;}$
on processes is defined by: $A \,\LRstep{\;\tr' \;}\, B$ if, and only
if, there exists a sequence $\tr$ such that $A \lrstep{\; \tr \;} B$
and~$\tr'$ is obtained by erasing all occurrences of the silent
action~$\tau$ in $\tr$.

\begin{example}
 \label{ex:semantics}
Continuing Example~\ref{ex:theory} and Example~\ref{ex:process}, we
consider the sequence $\Phi_{\Diffie} \stackrel{\mathsf{def}}{=}
\{w_1\refer \pk(sk_A), w_2 \refer \pk(sk_B)\}$, and
$A_{\Diffie} \stackrel{\mathsf{def}}{=} \triple{\{sk_A,
  sk_B\}}{P_A \mid P_B}{\Phi_{\Diffie}}$.
We have that:
\[
\begin{array}{rcl} A_{\Diffie} &\LRstep{\nu
    w_3. \Out(c,w_3). \In(c,w_3).\nu w_4. \Out(c,w_4). \In(c,w_4)} & \quadruple{\Ec}{\emptyset}{\Phi_{\Diffie} \uplus \Phi}{\sigma\cup\sigma'} \\
 \end{array}
\]
\noindent where
${\Phi =_{\E_\Diffie} \{w_3 \refer u_0,w_4 \refer \aenc(\langle n_A,\gfun(r_B)\rangle,pk_A)\}}$, $\Ec = \{sk_A,sk_B,r_A,r_B,n_A\}$, ${\sigma =_{\E_\Diffie} \{y_A \mapsto \aenc(\langle n_A,\gfun(r_B)\rangle,pk_A), \;y_B \mapsto \aenc(\langle n_A,\gfun(r_A) \rangle,pk_B) \}}$, and lastly $\sigma' =_{\E_\Diffie} \{x_A \mapsto \ffun(\gfun(r_B),r_A), \;x_B \mapsto \ffun(\gfun(r_A),r_B)\}$. We used $pk_A$ (resp. $pk_B$) as a shorthand for $\pk(sk_A)$ (resp. $\pk(sk_B)$).
\end{example}

\subsection{Process equivalences}
\label{subsec:properties}

We are particularly interested in
security properties expressed using a notion of equivalence such
as those studied in \emph{e.g.}~\cite{AMRR-CCS12,kostas-csf10}.  For
instance, the notion of \emph{strong unlinkability} 
can be formalized using an
equivalence between two situations: one where each user can execute
the protocol 
multiple times, and one where each user can execute the protocol at
most once.



We consider here the notion of \emph{trace equivalence}.  Intuitively,
two protocols~$P$ and~$Q$ are in trace equivalence, denoted $P \approx
Q$, if whatever the messages they received (built upon previously sent
messages), the resulting sequences of messages sent on public
channels 
are indistinguishable from the point of view of an outsider.
%
%
Given an extended process~$A$, we define its set of
traces as follows: \\[1mm]
\null\hfill
$\trace(A) = \{(\tr,\new \; \Ec.\Phi) ~|~ A
\,\LRstep{\;\tr\;}\, \quadruple{\Ec}{\p}{\Phi}{\sigma} \mbox{ for some
  process $\quadruple{\Ec}{\p}{\Phi}{\sigma}$}\}.$\hfill\null

\smallskip{}

The sequence of messages $\Phi$ together with the set of restricted
names $\Ec$ (those unknown to the attacker) is called the \emph{frame}.

\begin{definition}
We say that a term $u$ is \emph{deducible} (modulo $\E$) from a frame
$\phi = \new \, \Ec.\Phi$, denoted $\new\, \Ec. \Phi \vdash u$, when
 there exists a term $M$ (called a \emph{recipe}) such that $\fn(M)
 \cap \Ec = \emptyset$, $\fv(M) \subseteq \dom(\Phi)$, and $M\Phi
 =_\E u$.
 \end{definition}

%
%

Two frames are indistinguishable when the attacker cannot detect the
difference between the two situations they represent.


\begin{definition}
  \label{def:statequiv}
  Two frames $\phi_1$ and $\phi_2$ with $\phi_i = \new \, \Ec. \Phi_i$
  ($i \in \{1,2\}$) are \emph{statically equivalent}, denoted by
  $\phi_1 \statequiv \phi_2$, when $\dom(\Phi_1) = \dom(\Phi_2)$, and
  for all terms $M, N$ with  $\fn(\{M,N\})
  \cap \Ec = \emptyset$ {and $\fv(\{M,N\}) \subseteq \dom(\Phi_1)$}, we have that:\\[1mm]
  \null\hfill $M\Phi_1=_\E N\Phi_1, \mbox{ if and only if, } M\Phi_2
  =_\E N\Phi_2$.  \hfill\null
\end{definition}




\begin{example}
  Consider $\Phi_1 = \{w_1 \refer \gfun(r_A), w_2 \refer \gfun(r_B),
  w_3 \refer \ffun(\gfun(r_A),r_B)\}$, and $\Phi_2 = \{w_1 \refer
  \gfun(r_A), w_2 \refer \gfun(r_B), w_3 \refer k\}$.  Let $\Ec =
  \{r_A, r_B, k\}$. We have that $\new \, \Ec. \Phi_1 \statequiv
  \new\, \Ec. \Phi_2$ (considering the equational
  theory~$\E_{\Diffie}$). This equivalence shows that the term
  $\ffun(\gfun(r_A),r_B)$ (the Diffie-Hellman key) is
  indistinguishable from a random key. This indistinguishability
  property holds even if the messages $\gfun(r_A)$ and $\gfun(r_B)$
  have been observed by the attacker.
\end{example}




Two processes are trace equivalent if, whatever the messages they
{sent and received}, their frames are in static equivalence.
\begin{definition}
  \label{def : trace eq}
  Let $A$ and $B$ be two extended processes, ${A \sqsubseteq B}$ if
  for every $(\tr,\phi) \in \trace(A)$, there exists $(\tr', \phi')
  \in \trace(B)$ such that ${\tr = \tr'}$ and~${\phi \statequiv
    \phi'}$.
  We say that~$A$ and~$B$ are \emph{trace equivalent}, denoted by $A
  \approx B$, if ${A \sqsubseteq B}$ and~${B \sqsubseteq A}$.
\end{definition}


This notion of equivalence allows us to express many interesting
privacy-type properties \emph{e.g.} vote-privacy, strong versions of
anonymity and/or unlinkability.

\section{Composition result: a simple setting}
\label{sec:confidentiality-par}

It is well-known that even if two protocols are secure in isolation,
it is not possible to compose them in arbitrary ways still preserving
their security.  This has already been observed for different kinds of
compositions (\emph{e.g.} parallel~\cite{GuttmanT00},
sequential~\cite{CC-csf10}) and when studying standard security
properties~\cite{CD-fmsd08} (\emph{e.g.}  secrecy, authentication) and even
privacy-type properties~\cite{ACD-csf12}.
 In this section, we introduce some well-known hypotheses that are needed to
 safely compose security protocols.


\subsection{Sharing primitives}
A protocol can be used as an oracle by another protocol to decrypt a
message, and then compromise the security of the whole application. To
avoid this kind of interactions, most of the composition results
assume that protocols do not share any primitive or allow a list of
standard primitives (\emph{e.g.} signature, encryption) to be shared
as long as they are tagged in different ways. In this paper, we
 adopt the latter hypothesis and
consider the fixed common signature:\\[1mm]
\null\hfill
$\Sigma_0 = \{\sdec,\allowbreak
\senc,\allowbreak \adec,\allowbreak \aenc,\allowbreak \pk,\allowbreak
\langle, \rangle,\allowbreak \proj_1,\allowbreak \proj_2,\allowbreak
\sign,\allowbreak \checksign,\allowbreak \vk,\allowbreak
\h\}$\hfill\null

\smallskip{}

\noindent equipped with the equational theory $\Ezero$, defined by the following
equations:\\[1mm]
\null\hfill
$\begin{array}{rclcrcl}
  \sdec(\senc(x,y),y) &=& x&\;\;\;\;\;&  \checksign(\sign(x,y), \vk(y))&=& x \\
  \adec(\aenc(x,\pk(y)),y) &=& x &&
  \proj_i(\langle x_1,x_2 \rangle ) &=&
  x_i \mbox{ with $i \in \{1,2\}$)}
\end{array}$
\hfill\null

\smallskip{}

\noindent This allows us to model symmetric/asymmetric encryption,
concatenation, signatures, and hash functions.  We consider a type
\emph{seed} which is a subsort of the base type that only contains names.
We denote by $\pk(\mathit{sk})$ (resp.~$\vk(\mathit{sk})$) the public
key (resp.  the verification key) associated to the private
key~$\mathit{sk}$ which has to be a name of type \emph{seed}.
%
We allow protocols to both rely on $\Sigma_0$ provided that each
application of $\aenc$, $\senc$, $\sign$, and $\h$ is tagged (using
disjoint sets of tags for the two protocols), and adequate tests are
performed when receiving a message to ensure that the tags are
correct. Actually, we consider the same tagging mechanism as the one
we have
introduced in~\cite{ACD-csf12} (see Appendix~\ref{sec:app-tagging} for
more details). Here, we simply illustrate this tagging mechanism on
our running example.
Note that this tagging mechanism has no effect when protocols do
not rely on $\Sigma_0$ 
(\emph{i.e.} when
 the protocols we want to compose do not share any primitive). 

\begin{example}
  \label{ex:tag example Diffe}
 In order to compose the protocol
introduced in Example~\ref{ex:process} with
  another one that also relies on the primitive~$\aenc$, we may want to
  consider a tagged version of this protocol. For this, we introduce
  two function symbols $\Tag_1/\unTag_1$, and the equation
  $\unTag_1(\Tag_1(x)) = x$ to model the interaction between these two
  symbols.
  The tagged version (using~$\Tag_1$) of~$P_B$ is given below (with $u = \unTag_1(\adec(y_B,\sk_B))$):\\[1mm]
\null\hfill
$\left\{
    \begin{array}{l}
      \new \,r_B .\In(c,y_B).\\
      \If\; \Tag_1(\unTag_1(\adec(y_B,sk_B))) = \adec(y_B,sk_B)\; \Then\\
      \If\; u = \langle \proj_1(u), \proj_2(u) \rangle\; \Then\; \\
      \Out(c,\aenc(\Tag_1(\langle \proj_1(u),\gfun(r_B)\rangle),\pk(sk_A))).{[}x_B := \ffun(\proj_2(u),r_B)].0\\
    \end{array}
  \right.
$
\hfill\null

\smallskip{}

The first test allows one to check  that $y_B$ is an encryption tagged with $\Tag_1$
and the second one is used to ensure that the content of this
encryption is a pair as expected. Then, the process outputs the encrypted message tagged with $\Tag_1$.
The tagged version (using $\Tag_1$) of $P_A$ can be
  obtained in a similar way, and we obtain the tagged
  version of the whole process by putting the resulting two processes in parallel.
\end{example}

\subsection{Revealing shared keys}

Consider two protocols, one whose security relies on the secrecy of a
shared key whereas the other protocol reveals it. Such a
situation will compromise the security of the whole application. It is
therefore important to ensure that shared keys are not revealed.  To
formalise this hypothesis, and to express the sharing of long-term
keys, we introduce the notion of \emph{composition context}.
 This will help us describe under which long-term keys the
 composition has to be done.

\smallskip{}

A \emph{composition context}~$C$
is defined by  the grammar:\\[1mm]
\null\hfill
$C := \_ \; \mid\; \new \ n.\ C \; \mid \; ! C\;\;\;\;\;\;\;\;\;
\mbox{where~$n$ is a name of base type.}$\hfill\null

\begin{definition}
  \label{def:reveal-main}
  Let $C$ be a composition context, 
  $A$ be an extended process of the form $\triple{\Ec}{C[P]}{\Phi}$,
  $key \in \{n, \pk(n), \vk(n) ~|~ n \mbox{ occurs in } C\}$, and $c$,
  $s$ two fresh names.  We say that
  \emph{$A$ reveals $key$} when\\[1mm]
\null\hfill
$\triple{\Ec \cup \{s\}}{C[P \mid \In(c,x). \,\texttt{if
    } x = key \;\texttt{then}\, \Out(c,s)]}{\Phi}\;\LRstep{\;\tr\;}\;
  \quadruple{\Ec'}{\p'}{\Phi'}{\sigma'}$\hfill\null

\smallskip{}

  \noindent for some $\Ec'$, $\p'$, $\Phi'$, and $\sigma'$ such that
  $\new\; \Ec'. \Phi' \vdash s$.
\end{definition}

\subsection{A first composition result}

Before stating our first result regarding parallel
composition for confidentiality properties, we gather the required
hypotheses in the following definition.

\begin{definition}
  \label{def:composability-main}
  Let $C$ be a composition context and $\Ec_0$ be a finite set of
  names of base type. Let~$P$ and $Q$ be two plain processes together
  with their frames~$\Phi$ and~$\Psi$.
  We say that $P/\Phi$ and $Q/\Psi$ are \emph{composable} under
  $\Ec_0$ and~$C$ when $\fv(P) = \fv(Q) = \emptyset$, $\dom(\Phi) \cap
  \dom(\Psi) =\emptyset$, and
  \begin{enumerate}
  \item $P$ (resp.~$Q$) is built over $\Sigma_\alpha \cup \Sigmazero$
    (resp. $\Sigma_\beta \cup \Sigmazero$), whereas $\Phi$
    (resp. $\Psi$) is built over $\Sigma_\alpha \cup \{\pk,\vk,
    \langle \; \rangle\}$
    (resp. $\Sigma_\beta \cup \{\pk,\vk, \langle \; \rangle\}$), $\Sigma_\alpha \cap
    \Sigma_\beta = \emptyset$, and $P$ (resp.~$Q$) is tagged;
  \item $\Ec_0 \cap (\fn(C[P]) \cup \fn(\Phi)) \cap (\fn(C[Q]) \cup
    \fn(\Psi)) = \emptyset$; and
  \item $\triple{\Ec_0}{C[P]}{{\Phi}}$
    (resp.~$\triple{\Ec_0}{C[Q]}{{\Psi}}$) does not reveal any key in\\[1mm]
\null\hfill
$\{n, \pk(n), \vk(n) ~|~ n \mbox{ occurs in }{\fn(P)
      \mathord{\cap} \fn(Q) \mathord{\cap} \bn(C)}\}.$\hfill\null
  \end{enumerate}
\end{definition}

 Condition~$1$ is about sharing primitives, whereas
Conditions $2$ and $3$ ensure that keys are shared via the composition
context $C$ only (not via $\Ec_0$), and are not revealed by each
protocol individually.

We are now able to state the following theorem which is in the same vein
as those obtained previously in \emph{e.g.} 
~\cite{GuttmanT00,CD-fmsd08}.
However, the setting we  consider here is more general. In particular,
we consider arbitrary primitives, processes with else branches, and
private channels.

\begin{theorem}
  \label{theo:compo-par-reachability-main}
 Let $C$ be a composition context, $\Ec_0$ be a finite set of
  names of base type, and $s$ be a name that occurs in $C$. 
Let ${P}$ and ${Q}$ be two plain processes
  together with their frames 
  $\Phi$ and $\Psi$, and assume that $P/\Phi$ and $Q/\Psi$ are
  composable under $\Ec_0$ and $C$.
  If $\triple{\Ec_0}{C[P]}{{\Phi}}$ and $\triple{\Ec_0}{C[Q]}{{\Psi}}$
  do not reveal $s$ then $\triple{\Ec_0}{C[P \mid Q]}{\Phi \uplus
    \Psi}$ does not reveal $s$.
\end{theorem}

As most of the proofs of similar composition results, we show this
result going
back to the \emph{disjoint case}. 
Indeed, it is well-known that
parallel composition works well
when protocols do not share any data (the so-called
\emph{disjoint case}). We show that all the conditions are satisfied
to apply our generic result (presented in
Appendix~\ref{sec:app-compo-main}) that allows one to go back to the disjoint case. 
Thus, we obtain that the disjoint case $D =
\triple{\Ec_0}{C[P] \mid C[Q]}{\Phi \uplus \Psi}$ exhibits the same traces
as those exhibited by the shared case $S = \triple{\Ec_0}{C[P \mid Q]}{\Phi \uplus
  \Psi}$ (more formally we have that~$D$ and~$S$ are in trace equivalence), and this allows us to conclude.


\section{The case of key-exchange protocols}
\label{sec:confidentiality-seq}

Our goal is to go beyond parallel composition, and to further consider
the particular case of key-exchange protocols.  Assume that $P = \new
\, \tilde n . (P_1 \mid P_2)$ is a protocol that establishes a key
between two parties. The goal of $P$ is to establish a shared session
key between $P_1$ and $P_2$. Assume that $P_1$ stores the key in the
variable $x_1$, while~$P_2$ stores it in the variable~$x_2$, and then
consider a protocol~$Q$ that uses the values stored in $x_1/x_2$ as a
fresh key to secure communications.
%

\subsection{What is a good key exchange protocol?}
\label{subset: good key exchange protocol}

In this setting, sharing between~$P$ and~$Q$ is achieved through the
composition context as well as through assignment variables $x_1$
and~$x_2$.  The idea is to abstract these values with fresh names when
we analyse $Q$ in isolation. However, {in order to abstract them in
  the right way, we need to know their values (or at least whether
  they are equal or not). This is the purpose of the property stated
  below.

\newcommand{\good}{\mathsf{good}}
\begin{definition}
  \label{def:good}
  Let $C$ be a composition context and $\Ec_0$ be a finite set of
  names.  Let $P_1[\_]$ (resp. $P_2[\_]$) be a plain process with an
  hole in the scope of an assignment of the form $[x_1 := t_1]$
  (resp. $[x_2 := t_2]$), and $\Phi$ be a frame.

 We say that
  $P_1/P_2/\Phi$ is a \emph{good} key-exchange protocol under $\Ec_0$
  and~$C$ when $\triple{\Ec_0}{P_{\good}}{\Phi}$ does not reveal
  $\bad$ where $P_\good$ is defined as follows:\\[1mm]
\null\hfill
$\begin{array}{l}
    P_\good =  \new\, \bad. \new\, d. \big(C[\new\, id. ( P_1[\Out(d,\langle x_1, id\rangle)] \mid P_2[\Out(d,\langle x_2,id\rangle)])] \\[0.5mm]
    \;\mid  \In(d,x). \In(d,y). \myIf \  \proj_1(x) = \proj_1(y)  \wedge \proj_2(x) \neq \proj_2(y)\, \myThen \,\Out(c,\bad)\\[0.5mm]
    \;\mid \In(d,x). \In(d,y). \myIf\ \proj_1(x) \neq \proj_1(y)  \wedge \proj_2(x) = \proj_2(y)\, \myThen \,\Out(c,\bad)\\[0.5mm]
    \;\mid \In(d,x). \In(c,z). \myIf\, z \in \{\proj_1(x), \pk(\proj_1(x)), \vk(\proj_1(x))\}\, \myThen\, \Out(c,\bad)\big)\\[0.5mm]
  \end{array}
$\hfill\null

\smallskip{}

  \noindent where $\bad$ is a fresh name of base type, and $c,d$ are
  fresh names of channel type.
\end{definition}

The expressions $u \neq v$ and $u \in \{v_1, \ldots, v_n\}$ used above
are convenient notations that can be rigorously expressed using nested
conditionals. Roughly, the property expresses
that $x_1$ and $x_2$ are assigned to the same value if, and only if,
they are joined together, \emph{i.e.} they share the same $id$. In
particular, two instances of the role $P_1$ (resp. $P_2$) cannot
assign their variable with the same value: a fresh key is established
at each session.  The property also ensures that the data shared
through $x_1 / x_2$ are not revealed.

\begin{example}
We have that $P_A/P_B/\Phi_\Diffie$ described in
  Example~\ref{ex:process}, as well as its
tagged version (see
Example~\ref{ex:tag example Diffe}) 
are \emph{good} key-exchange protocols under $\Ec_0 = \{\sk_A,
\sk_B\}$ and $C = \_$. This corresponds to a scenario where we
consider only a single execution of the protocol (no replication).
\end{example}



Actually, the property mentioned above is quite strong, and never
satisfied when the context~$C$ under study ends with a replication,
\emph{i.e.} when $C$ is of the form $C'[!\_]$.
To cope with this situation, we
consider an other version of this property. 
 When $C$ is of the form
$C'[!\_]$, we define $P_\good$ as
follows (where~$r_1$ and~$r_2$ are two additional fresh names of base type):\\[1mm]
 \null\hfill
$\begin{array}{l}
  \new\, \bad, d, r_1, r_2.\big(C'[\new\, id. !( P_1[\Out(d,\langle x_1, id,r_1\rangle)] \mid P_2[\Out(d,\langle x_2,id,r_2\rangle)])] \\[0.5mm]
    \;\mid  \In(d,x). \In(d,y). \myIf \  \proj_1(x) = \proj_1(y)  \wedge \proj_2(x) \neq \proj_2(y)\, \myThen \,\Out(c,\bad)\\[0.5mm]
    \;\mid \In(d,x). \In(d,y). \myIf\ \proj_1(x) = \proj_1(y)  \wedge \proj_3(x) = \proj_3(y)\, \myThen \,\Out(c,\bad)\\[0.5mm]
    \;\mid \In(d,x). \In(c,z). \myIf\, z \in \{\proj_1(x), \pk(\proj_1(x)), \vk(\proj_1(x))\}\, \myThen\, \Out(c,\bad)\big)\\[0.5mm]
  \end{array}$\hfill\null

\smallskip{}

\noindent Note that the $\id$ is now
generated before the last replication, and thus is not uniquely
associated to an instance of $P_1/P_2$. Instead several instances of
$P_1/P_2$ may now share the same $\id$ as soon as they are
identical. This gives us more flexibility.
The triplet $\langle u_1, u_2,u_3\rangle$ and
  the operator $\proj_3(u)$ used above are  convenient notations that
  can be expressed using pairs. This new version forces distinct
  values in the assignment variables for each instance of~$P_1$
  (resp.~$P_2$) through the 3rd line. However, we do
  not fix in advance which particular instance of~$P_1$ and~$P_2$ should be
  matched, as in the first version.

\begin{example}
We  have that $P_A/P_B/\Phi_\Diffie$ as well as its tagged version
are good key-exchange protocols
under $\Ec_0 = \{\sk_A, \sk_B\}$ and $C = !\, \_$. 
\end{example}


\subsection{Do we need to tag pairs?}

When analysing $Q$ in isolation, the values stored in the assignment
variables $x_1/x_2$ are abstracted by fresh names.  Since $P$ and $Q$
share the common signature~$\Sigmazero$, we need an additional
hypothesis to ensure that in any execution, the values assigned to the
variables $x_1/x_2$ are not of the form $\langle u_1,u_2\rangle$,
$\pk(u)$, or $\vk(u)$.  These symbols are those of the common
signature that are not tagged, thus abstracting them by  fresh names
in $Q$ would not be safe.  This has already been highlighted
in~\cite{CC-csf10}.  They however left as future work the definition
of the needed hypothesis and simply assume that each operator of the
common signature has to be tagged. Here, we formally express the
required hypothesis.

\begin{definition}
  An extended process $A$ satisfies the \emph{abstractability
    property} if for any $\quadruple{\Ec}{\p}{\Phi}{\sigma}$ such that
  $A \LRstep{\;\tr\;} \quadruple{\Ec}{\p}{\Phi}{\sigma}$, for any $x
  \in \dom(\sigma)$ which corresponds to an assignment variable, for
  any $u_1, u_2$, we have that $x\sigma \neq_\E \langle u_1,
  u_2\rangle$, $x\sigma \neq_\E \pk(u_1)$, and $x\sigma \neq_\E
  \vk(u_1)$.
\end{definition}

Note also that, in~\cite{CC-csf10}, the common signature is restricted
to symmetric encryption and pairing only. They do not consider
asymmetric encryption, and signature.  Thus, our composition result
generalizes theirs considering both a richer common signature, and a
lighter tagging scheme (we do not tag pairs).


\subsection{Composition result}

We retrieve the following result which is actually a
generalization of two theorems established in~\cite{CC-csf10} and
stated for specific composition contexts.

\begin{restatable}{theorem}{theoreachseq}
  \label{theo:main-compo-seq-confidentiality}
  Let $C$ be a composition context, $\Ec_0$ be a finite set of
  names of base type, and $s$ be a name that occurs in $C$. 
Let $P_1[\_]$ (resp. $P_2[\_]$) be a plain
  process without replication and with an hole in the scope of an
  assignment of the form $[x_1 := t_1]$ (resp. ${[x_2 := t_2]}$).  Let
  $Q_1$ (resp. $Q_2$) be a plain process such that $\fv(Q_1)
  \subseteq \{x_1\}$ (resp. $\fv(Q_2) \subseteq \{x_2\}$), and $\Phi$
  and $\Psi$ be two frames. Let $P = P_1[0]\mid P_2[0]$ and $Q = \new\
  k. [x_1:= k]. [x_2 := k]. (Q_1 \mid Q_2)$ for some fresh name $k$,
  and assume that:
  \begin{enumerate}
    \setlength{\itemsep}{0.5mm}
  \item $P/\Phi$ and $Q/\Psi$ are composable under $\Ec_0$ and $C$;
  \item $\triple{\Ec_0}{C[Q]}{\Psi}$ does not reveal $k$, $\pk(k)$,
    $\vk(k)$;
  \item $\triple{\Ec_0}{C[P]}{\Phi}$ satisfies the abstractability
    property; and
  \item $P_1/P_2/\Phi$ is a good key-exchange protocol under $\Ec_0$
    and $C$.
  \end{enumerate}
 If $\triple{\Ec_0}{C[P]}{\Phi}$ and
  $\triple{\Ec_0}{C[Q]}{\Psi}$ do not reveal~$s$ then
  ${\triple{\Ec_0}{C[P_1[Q_1] | P_2[Q_2]]}{\Phi \uplus \Psi}}$ does
  not reveal~$s$.
\end{restatable}

Basically, we prove this result relying on our generic composition
result.
In~\cite{CC-csf10}, they do not require $P$ to be good but only ask
for secrecy of the shared key. In particular they do not express any
freshness or agreement property about the established key. Actually,
when considering a simple composition context without replication,
freshness is trivial (since there is only one session). Moreover, in
their setting, agreement is not important since they do not have else
branches. The analysis of $Q$ considering that both parties have
agreed on the key corresponds to the worst scenario. Note that this is
not true anymore in presence of else branches. The following example
shows that as soon as else branches are allowed, as it is the case in
the present work, agreement becomes important.

\begin{example}
Consider a simple situation where:
\begin{itemize}
\item $P_1[0] = \new\, k_1. [x_1 := k_1]. 0$ and $P_2[0] = \new\, k_2.[x_2 :=
k_2].0$;
\item  $Q_1 = \texttt{if } x_1 = x_2 \texttt{ then } \Out(c,\ok) \texttt{ else
} \Out(c,s)$ and $Q_2 = 0$.
\end{itemize}
Let  $\Ec_0 = \emptyset$, and $C =
\new\; s. \_\,$. We consider the processes $P = P_1[0] \mid P_2[0]$,  and $Q = \new\, k. [x_1 :=
k].[x_2:= k].(Q_1 
\mid Q_2)$ and we assume that the frames $\Phi$ and  $\Psi$ are empty.
We clearly have that $\triple{\Ec_0}{C[P]}{\Phi}$ and
$\triple{\Ec_0}{C[Q]}{\Psi}$ do not reveal $s$ whereas
$\triple{\Ec_0}{C[P_1[Q_1] \mid P_2[Q_2]}{\Phi \uplus \Psi}$ does.
The only hypothesis of Theorem~\ref{theo:main-compo-seq-confidentiality} that is violated is the fact
that $P_1/P_2/\Phi$ is not a good key-exchange protocol due to a lack
of agreement on the key which is generated ($\bad$ can be emitted
thanks to the 3rd line of the process $P_{good}$ given in Definition~\ref{def:good}).
\end{example}

Now, regarding their second theorem corresponding to a context of the
form $\new\, s. \, !\_\,$, as before agreement is not mandatory but
freshness of the key established by the protocol $P$ is crucial.
As illustrated by the following example, 
this hypothesis is missing in the theorem stated in~\cite{CC-csf10}
(Theorem~3).

\begin{example}
\label{ex:counter-example-csf10}
Consider
$A = \triple{\{k_P\}}{\new\, s.! ( [x_1:= k_P].0 \mid
  [x_2:=k_P].0)}{\emptyset}$, as well as
$B = \triple{\{k_P\}}{\new\, s.\,  ! Q}{\emptyset}$ where 
$Q = \new\, k.[x_1 :=k].[x_2:=k].(Q_1 \mid Q_2)$ with\\[1mm]
\null\hfill $Q_1 = \Out(c,\senc(\senc(s,k),k))$; and
$Q_2 = \In(c,x).\Out(c,\sdec(x,k))$.\hfill\null

\smallskip{}

\noindent 
Note that neither $A$ nor $B$ reveals $s$. In particular,
the process $Q_1$ emits the secret $s$ encrypted twice with a fresh
key $k$, but 
$Q_2$ only  allows us to remove one level of encryption with
$k$.
Now, if we plug the key-exchange protocol given above with
no guarantee of freshness (the same key  is established at each
session), the
resulting process, \emph{i.e.} $\triple{\Ec_0}{C[P_1[Q_1] \mid
  P_2[Q_2]]}{\emptyset}$ does reveal $s$.
\end{example}

Note that this example is not a counter example of our
Theorem~\ref{theo:main-compo-seq-confidentiality}: $P_1/P_2/\emptyset$ is not a good
key-exchange protocol according to our definition.

\section{Dealing with equivalence-based properties}
\label{sec:difficulties-equiv}

Our ultimate goal is to analyse privacy-type properties in a modular
way. In~\cite{ACD-csf12}, we propose several composition results
w.r.t. privacy-type properties, but 
for parallel composition only. Here, we want to go
beyond parallel composition, and consider the case of key-exchange
protocols.

\subsection{A problematic example}
\label{subsec:problematic-example}
Even in a quite simple setting
(the shared keys are not revealed, protocols do not share any
primitives), such a sequential composition result does not hold.
Let $C = \new\ k . ! \, \new \ k_1 . ! \, \new \ k_2. \ \_$ be a
composition context, $\yes/\no$, $\ok/\ko$ be public constants, $u =
\senc(\langle k_1,k_2 \rangle, k)$, and consider the following
processes:\\[1mm]
  \null\hfill $
\begin{array}{rl}
    Q(z_1,z_2) \,=\, &\Out(c,u). \In(x). \If\, x = u \,\Then\;  0 \, \Else \\
    \hspace{1cm} & \If \, \proj_1(\sdec(x,k)) = k_1 \, \Then \,
    \Out(c,z_1)
    \, \Else \,\Out(c,z_2)\\[2mm]
    P[\_] \,=\, &\Out(c,u). \big(\_ \mid \In(c,x). \If\, x = u \,\Then\;  0 \, \Else \\
    & \If \, \proj_1(\sdec(x,k))
    = k_1 \, \Then \, \Out(c,\ok) \, \Else \,\Out(c,\ko) \big) \\
  \end{array}$
\hfill\null

\smallskip{}

Of course, we have that $C[P[0]] \approx C[P[0]]$. Actually, we have
also that ${C[Q(\yes,\no)] \approx C[Q(\no,\yes)]}$.  This latter
equivalence is non-trivial. Intuitively, when $C[Q(\yes,\no)]$
unfolds its outermost~! and then performs an output, then
$C[Q(\no,\yes)]$ has to mimic this
step by unfolding its innermost~! and by performing the only available
output. This will allow it to react in the
same way as $C[Q(\yes,\no)]$ in case encrypted messages are used to
fill some input actions.
  
Since the two processes $P[0]$ and $Q(\yes,\no)$ (resp. $Q(\no,\yes)$)
are almost ``disjoint'', we could expect the equivalence
$C[P[Q(\yes,\no)]] \approx C[P[Q(\no,\yes)]]$ to hold. Actually,
this equivalence does \emph{not} hold.  The presence of the process
$P$ gives to the attacker some additional distinguishing power. In
particular, through the outputs $\ok/\ko$ outputted by $P$, the
attacker will learn which ! has been unfolded.
  This result holds even if we rename function symbols so that
  protocols $P$ and $Q$ do not share any primitives. 
  In the example above, the problem is that the two equivalences we
  want to compose hold for different reasons, \emph{i.e.} by unfolding
  the replications in a different and incompatible way.  Thus, when
  the composed process $C[P[Q(\yes,\no)]]$ reaches a point where
  $Q(\yes,\no)$ can be executed, on the other side, the process
  $Q(\no,\yes)$ is ready to be executed but the instance that is
  available is not the one that was used when establishing the
  equivalence $C[Q(\yes,\no)] \approx C[Q(\no,\yes)]$.  Therefore, in
  order to establish equivalence-based properties in a modular way, we
  rely on a stronger notion of equivalence, namely
  \emph{diff-equivalence}, that will ensure that the two {``small''
    equivalences} are satisfied in a compatible way.

Note that this
problem does not arise when considering reachability properties and/or
parallel composition.
In particular, we have that:\\[2mm]
\null\hfill
$C[P[0] \mid Q(\yes,\no)] \approx C[P[0] \mid Q(\no,\yes)].$\hfill\null




\subsection{Biprocesses and diff-equivalence}
 
We consider pairs of
processes, called \emph{biprocesses}, that have the same structure and
differ only in the terms and tests that they contain. Following the
approach of~\cite{BlanchetAF08}, we introduce a special symbol $\diff$
of arity 2 in our signature.  The idea being to use this $\diff$
operator to indicate when the terms manipulated by the processes are
different.
Given a biprocess~$B$, we define two processes $\fst(B)$ and $\snd(B)$
as follows: $\fst(B)$ is obtained by replacing each occurrence of
$\diff(M,M')$ (resp. $\diff(\varphi,\varphi')$) with $M$
(resp. $\varphi$), and similarly $\snd(B)$ is obtained by replacing
each occurrence of $\diff(M,M')$ (resp. $\diff(\varphi,\varphi')$)
with $M'$ (resp. $\varphi'$).


 \smallskip{}

The semantics of biprocesses (detailed in
Appendix~\ref{sec:biprocesses}) is defined as expected via a
relation that expresses when and how a biprocess
may evolve.  A biprocess reduces if, and only if, both sides of the
biprocess reduce in the same way: a communication succeeds on both
sides, a conditional has to be evaluated in the same way in both sides
too. When the two sides of the biprocess reduce in different ways, the
biprocess blocks. For instance, the {\Else} rule is as
follows:\\[2mm]
\null\hfill$
\begin{array}{c}
  \quadruple{\Ec}{\{\mbox{\texttt{if} $\diff(\varphi_L, \varphi_R)$
      \texttt{then} $Q_1$ \texttt{else} $Q_2$}\}\uplus\p}{\Phi}{\sigma}\lrstep{\;\tau\;}_\bi  \quadruple{\Ec}{Q_2\uplus\p}{\Phi}{\sigma} \\\mbox{if $u\sigma \neq_\E
    v\sigma$ for some $u = v \in \varphi_L$, and $u'\sigma \neq_\E
    v'\sigma$ for some $u' = v' \in \varphi_R$}\\
\end{array}$\hfill\null

\smallskip{}

The relation $\LRstep{\;\tr\;}_\bi$ on biprocesses is defined as for
processes.  This leads us to the following notion of \emph{diff-equivalence}.


\begin{definition}
  An extended biprocess $B_0$ satisfies \emph{diff-equivalence} if for
  every biprocess
  $B= \quadruple{\Ec}{\p}{\Phi}{\sigma}$ such that $B_0
  \,\LRstep{\;\tr\;}_\bi\, B$ for some trace $\tr$, we have that
   \begin{enumerate}
   \item
  $\new\ \Ec. \fst(\Phi) \statequiv \new\ \Ec. \snd(\Phi)$
  
   \item
  if $\fst(B) \lrstep{\ell} A_L$
  then there exists 
  $B'$ such that $B \lrstep{\ell}_\bi B'$ and $\fst(B') = A_L$ (and
  similarly for~$\snd$).
   \end{enumerate}
\end{definition}


Note that, considering diff-equivalence instead of trace equivalence,
the example given in Section~\ref{subsec:problematic-example} is not a counter-example anymore.  Indeed,
 the biprocess $B = Q_0(\diff(\yes,\no), \diff(\no,\yes))$,
  does not satisfy diff-equivalence.

\smallskip{}

The notions introduced so far on processes are extended as expected on
biprocesses: the property has to hold on both $\fst(B)$ and
$\snd(B)$. Sometimes, we also 
say that the biprocess~$B$ is in trace equivalence instead of writing~${\fst(B)
  \approx \snd(B)}$.

%
%

%
As expected, this notion of diff-equivalence is actually stronger than the usual notion of trace equivalence.

 \begin{lemma}
A biprocess  $B$ that satisfies diff-equivalence  is in
trace equivalence.
 \end{lemma}

\section{Composition results for diff-equivalence}
\label{sec:privacy}

We first consider the case of parallel composition.
This result is in the spirit of the one
established in~\cite{ACD-csf12}.
However, in order to combine this composition result with 
the one in the case of key-exchange protocol
(Theorem~\ref{theo:main-compo-seq-equiv}), 
we also adapt it to diff-equivalence.


\begin{theorem}
  \label{theo:compo-par-diff-equivalence-main}
  Let $C$ be a composition context and $\Ec_0$ be a finite set of
  names of base type. Let ${P}$ and ${Q}$ be two plain biprocesses
  together with their frames 
  $\Phi$ and~$\Psi$, and assume that $P/\Phi$ and $Q/\Psi$ are
  composable under $\Ec_0$ and $C$.

  If $\triple{\Ec_0}{C[P]}{{\Phi}}$ and $\triple{\Ec_0}{C[Q]}{{\Psi}}$
  satisfy diff-equivalence (resp. trace equivalence) then the
  biprocess $\triple{\Ec_0}{C[P \mid Q]}{\Phi \uplus \Psi}$ satisfies
  diff-equivalence (resp. trace equivalence).
\end{theorem}

\begin{proof}\emph{(sketch)}
As for the proof for Theorem~\ref{theo:compo-par-reachability-main}, 
parallel composition works well when processes do not share any
data. Hence, we easily deduce that $D = \triple{\Ec_0}{C[P] \mid
  C[Q]}{\Phi \uplus \Psi}$ 
satisfies the diff-equivalence (resp. trace equivalence). Then, our
generic composition result allows one to compare the behaviours of
the biprocess $D$ to those of 
the biprocess $S = \triple{\Ec_0}{C[P \mid Q]}{\Phi \uplus
  \Psi}$. More precisely, this allows us to establish that $\fst(D)$ and $\fst(S)$ are in
diff-equivalence (as well as $\snd(D)$ and $\snd(S)$), and then we
conclude relying on the transitivity of the equivalence. \qed
 \end{proof}


Now, regarding sequential composition and the particular case of
key-exchange protocols, we obtain the following composition result.

\begin{restatable}{theorem}{theoequivseq}
  \label{theo:main-compo-seq-equiv}
  Let $C$ be a composition context and $\Ec_0$ be a finite set of
  names of base type. Let $P_1[\_]$ (resp. $P_2[\_]$) be a plain
  biprocess without replication and with an hole in the scope of an
  assignment of the form $[x_1 := t_1]$ (resp. ${[x_2 := t_2]}$).  Let
  $Q_1$ (resp. $Q_2$) be a plain biprocess such that $\fv(Q_1)
  \subseteq \{x_1\}$ (resp. $\fv(Q_2) \subseteq \{x_2\}$), and $\Phi$
  and $\Psi$ be two frames. Let $P = P_1[0]\mid P_2[0]$ and $Q = \new\
  k. [x_1:= k]. [x_2 := k]. (Q_1 \mid Q_2)$ for some fresh name $k$,
  and assume that:
  \begin{enumerate}
    \setlength{\itemsep}{0.5mm}
  \item $P/\Phi$ and $Q/\Psi$ are composable under $\Ec_0$ and $C$;
  \item $\triple{\Ec_0}{C[Q]}{\Psi}$ does not reveal $k$, $\pk(k)$,
    $\vk(k)$;
  \item $\triple{\Ec_0}{C[P]}{\Phi}$ satisfies the abstractability
    property; and
  \item $P_1/P_2/\Phi$ is a good key-exchange protocol under $\Ec_0$
    and $C$.
  \end{enumerate}

\noindent Let $P^+ \hspace{-0.1cm} = \hspace{-0.1cm} P_1[\Out(d,x_1)] \mid P_2[\Out(d,x_2)] \mid
\In(d,x).\In(d,y). \myIf\, x = y \,\myThen\; 0 \;\myElse\; 0$. 
 If the biprocesses $\triple{\Ec_0}{\new\; d. C[P^+]}{\Phi}$ and $\triple{\Ec_0}{C[Q]}{\Psi}$ satisfy diff-equivalence then
$\triple{\Ec_0}{C[{P_1[Q_1] \mid P_2[Q_2]}]}{{\Phi} \uplus {\Psi}}$
  satisfies diff-equivalence.
\end{restatable}

Note that we require $\triple{\Ec_0}{\new\; d. C[P^+]}{\Phi}$ to be
in diff-equivalence (and not simply $\triple{\Ec_0}{C[P]}{\Phi}$).
Actually, when the
composition context $C$ under study is not of the form $C'[! \_]$, and
under the hypothesis 
that $P_1/P_2/\Phi$ is a good key-exchange protocol under $\Ec_0$ and $C$, we have
that these two requirements coincide.
However, the stronger hypothesis is important to
conclude when $C$ is of the form $C'[!\_]$. Indeed, in this case,  we do not know in advance what are the instances of~$P_1$
and~$P_2$ that will be ``matched''. This is not a problem but to
conclude about the diff-equivalence of the whole process (\emph{i.e.} $\triple{\Ec_0}{C[{P_1[Q_1] \mid P_2[Q_2]}]}{{\Phi} \uplus {\Psi}}$), we need to
ensure that such a matching is the same on both sides of the
equivalence.
Note that to conclude about trace equivalence only, this additional requirement
is actually not necessary.

 \section{Case studies}
\label{sec:casestudies}


As mentioned in the introduction, many applications 
rely on several protocols running in
composition (parallel, sequential, or nested). In this section, we
show that our results can help in the analysis of this sort of complex
system.  ProVerif models of our case studies
are made available online at:
\begin{center}
\url{http://www.loria.fr/~chevalvi/other/compo/}.
\end{center}


\subsection{3G mobile phones}

We look at confidentiality and privacy guarantees provided by the
\AKA\ protocol and the Submit SMS procedure (\sSMS) when run in
composition as specified by the 3GPP consortium in~\cite{TS33102}.

\paragraph*{Protocols description.}
The \sSMS\xspace protocol allows a mobile station (MS) to send an SMS
to another MS through a serving network (SN).
The confidentiality of the sent SMS relies on a session
key~$\mathit{ck}$ established through the execution of the \AKA\xspace
protocol between the MS and the SN.
The \AKA\xspace protocol achieves mutual authentication between a MS
and a SN, and allows them to establish a shared session key
$\mathit{ck}$.  The \AKA\xspace protocol consists in the exchange of
two messages: the \emph{authentication request} and the
\emph{authentication response}.  The \AKA\xspace protocol as deployed
in real 3G telecommunication systems presents a linkability
attack~\cite{AMRR-CCS12}, and thus we consider here its fixed version
as described in~\cite{AMRR-CCS12}.
At the end of a successful execution of this protocol, both parties
should agree on a fresh ciphering key $\mathit{ck}$.
This situation can be modelled in our calculus as follows: 
\[
\begin{array}{l}
\new\ \KPrSN. \ !\new\ \imsi.\ \new\ \KIMSI.\ ! \new\ sqn.\ \new\
\sms. \\
\hspace{5cm} \ (\AKA^{SN}[\sSMS^{SN}] \mid \AKA^{MS}[\sSMS^{MS}])
\end{array}
\]

\noindent where $\KPrSN$ represents the private key of the network;
while $\imsi$ and $\KIMSI$ represent respectively the long-term
identity and the symmetric key of the MS. The name $sqn$ models the
sequence number on which SN and MS are synchronised. The two
subprocesses $\AKA^{MS}$ and
$\sSMS^{MS}$ 
(\emph{resp.} $\AKA^{SN}$, and $\sSMS^{SN}$) model one session of the
MS's (\emph{resp.} SN's) side of the \AKA, and \sSMS\ protocols
respectively.  Each MS, identified by its identity $\imsi$ and its key
$\KIMSI$, can run multiple times the \AKA\ protocol followed by the
\sSMS\ protocol.

\paragraph*{Security analysis.}
We explain how some confidentiality and privacy properties 
of the \AKA\ protocol and the
\sSMS\ procedure can be derived relying on our composition results.
We do not need to tag the protocols under study to perform our
analysis since they do not share any primitive but the pairing
operator.  Note that the \AKA\ protocol can \emph{not} be modelled in
the calculus given in~\cite{CC-csf10} due to the need of non-trivial
else branches.

\smallskip{}

\noindent \emph{Strong unlinkability} requires that an
observer does not see the difference between the two following
scenarios: {\it (i)} a same mobile phone sends several SMSs; or {\it
  (ii)} multiple mobile phones send at most one SMS each. To model
this requirement, we consider the composition context\footnote{We use
  $\Let\ x=M\ \In\ P$ to denote the process $P\{M/x\}$.}:\\[1mm]
\null\hfill $\begin{array}{rl} C_U[\_]\ \stackrel{\defi}{=}
  & 
  !\new\ \imsi_1.\ \new\ \KIMSI_1.\   ! \new\ \imsi_2.\ \new\ \KIMSI_2.\  \\  
  & \ \Let\ \imsi = \diff[\imsi_1, \imsi_2]\ \In 
  \ \Let\ \KIMSI = \diff[\KIMSI_1,
  \KIMSI_2]\ \In \ \\ 
  & \ \new\ sqn.\ \new\ sms.\ \_ \\
\end{array}
$\hfill\null\\[1mm]
To check if the considered 3G protocols satisfy strong unlinkability, 
one needs to check if the following biprocess satisfies diff-equivalence ($\Phi_0 =\{w_1 \refer \pk(\KPrSN)\}$):\\[1mm]
\null\hfill$
\begin{array}{c}
  \quad\triple{\KPrSN}{C_U[\AKA^{SN}[\sSMS^{SN}] \mid
    \AKA^{MS}[\sSMS^{MS}]]}{\Phi_0} 
\end{array} 
$\hfill\null\\[1mm]
\noindent Hypotheses (1-4) stated in Theorem~\ref{theo:main-compo-seq-equiv}
are satisfied, and thus this equivalence can be derived from the
following two ``smaller'' diff-equivalences:\\[2mm]
\null\hfill$
\quad\triple{\KPrSN}{\new\ d.\ C_U[\AKA^+]}{\Phi_0} \text{\quad and \quad} \triple{\KPrSN}{C'_U[\sSMS]}{\emptyset}
$\hfill\null\\[-1mm]
\noindent where:
\begin{itemize}
\item  $sSMS \stackrel{\defi}{=} \sSMS^{SN} \mid \sSMS^{MS}$, 
\item  $\AKA^+ \stackrel{\defi}{=}
  \AKA^{SN}[\Out(d,xck_{SN})] \mid \AKA^{MS}[\Out(d,xck_{MS})] \mid \\
\null \hspace{1.7cm}  \In(d, x).\ \In(d, y).\ \myIf\; x = y \;\myThen\; 0 \;\myElse\; 0
$
\item $ C'_U[\_]\ \stackrel{\defi}{=} C_U[\new\ ck.  \Let\
xck_{SN}=ck\ \In \ \Let\ xck_{MS}=ck\ \In\ \_\,]$.
\end{itemize}

\noindent\emph{Weak secrecy} requires that the sent/received SMS is not deducible by
an outsider, and can be modelled using the context\\[1mm]
\null\hfill $C_{WS}[\_]\
\stackrel{\defi}{=} !\new\ \imsi.\ \new\ \KIMSI.\ ! \new\ sqn. \new\
sms. \_$.  \hfill\null

\smallskip{}

Note that the composition context $C_{WS}$ is the same as
$\fst(C_U)$ (up to some renaming), thus Hypotheses (1-4) of
Theorem~\ref{theo:main-compo-seq-confidentiality} also hold and we derive the weak
secrecy property by simply analysing this property on $\AKA$ and
$\sSMS$ in isolation.
\smallskip{}

\noindent\emph{Strong secrecy}
means that an outsider should not be able to distinguish
the situation where $\sms_1$ is sent (resp. received), from the situation
where $\sms_2$ is sent (resp. received), although he might know the content
of $\sms_1$ and $\sms_2$. This can be modelled using the following
composition context:\\[1mm]
$
\begin{array}{c}
  \quad C_{SS}[\_]\ \stackrel{\defi}{=}   !\new\ \imsi.\ \new\ \KIMSI.\ ! \new\ sqn. 
  \ \Let\ \sms = \diff[\sms_1, \sms_2]\
  \In\ \_
\end{array}
$\\[1mm] 
where $\sms_1$ and $\sms_2$ are two free names known to the
attacker. Again, our Theorem~\ref{theo:main-compo-seq-equiv} allows us to
reason about this property in a modular way.

\subsection{E-passport application}

We look at privacy guarantees provided by three protocols of the
e-passport application
when run in composition as specified in~\cite{ICAO-Passport}.

\paragraph*{Protocols description.}

The information stored in the chip of the passport is organised in
data groups ($\DG_1$ to $\DG_{19}$). For example, $\DG_5$ contains a
JPEG copy of the displayed picture, and~$\DG_7$ contains the displayed
signature. The verification key~$\vk(\KPrAA)$ of the passport,
together with its certificate $\sign(\vk(\KPrAA), \KPrDS)$ issued by
the Document Signer authority are stored in~$\DG_{15}$. The
corresponding signing key~$\KPrAA$ is stored in a tamper resistant
memory, and cannot 
be read or copied. For authentication purposes, a hash of all the $\DG$s together with a signature on this hash value issued by the Document Signer authority are stored in a separate file, 
the Security Object Document:\\[1mm]
\null\hfill
$\SOD \stackrel{\defi}{=} \langle \sign(\h(\DG_1, \dots, \DG_{19}),
\KPrDS), \; \h(\DG_1, \dots, \DG_{19}) \rangle.$\hfill\null

\smallskip{}
 
The ICAO standard specifies several protocols through which these
information can be accessed~\cite{ICAO-Passport}. First, the \emph{Basic Access Control} (\BAC)
protocol establishes a key seed $\mathit{kseed}$ from which two sessions keys
$kenc$ and $kmac$ are derived.
The purpose of $kenc$ and
$kmac$ is to prevent  skimming and eavesdropping on the subsequent
communication with the e-passport.
The security of the \BAC\xspace protocol relies on two master keys,
$\KE$ and $\KM$, which are 
optically retrieved from the passport by the reader before executing the \BAC\xspace protocol.
Once the \BAC\xspace protocol has been successfully executed, the
reader gains access to the information stored in the RFID tag through
the \emph{Passive Authentication} (\PAuth)
and the \emph{Active Authentication} (\AAuth) protocols  that can be executed
in any order. 
The \PAuth\xspace protocol is an
authentication mechanism that proves that the content of the RFID chip
is authentic whereas
the \AAuth\xspace protocol 
can be  used to prevent
cloning of the passport chip. 
It relies on the fact that the secret key~$\KPrAA$ of the passport
cannot be read or copied. 
%
This situation  can be modelled 
calculus 
as follows:\\[1mm]
\null\hfill
$
\begin{array}{ll}
  P\ \stackrel{\defi}{=} &  \new\ \KPrDS. \ !\new\ \KE.\ \new\ \KM.\
  \new\ \KPrAA.  \new\ \id.\ \new\ sig.\ \new\ pic.\ \dots \\ 
  & \qquad\qquad \quad\quad! (\BAC^R[\PAuth^R \mid \AAuth^R] \mid \BAC^P[\PAuth^P \mid \AAuth^P])
\end{array}
$\hfill\null

\smallskip{}

\noindent where $\id$, $sig$, $pic$, ... represent the name, the
signature, the displayed picture, \emph{etc} of the e-passport's
owner, \emph{i.e.} the data stored in the $\DG$s ($1$-$14$) and
($16$-$19$). The subprocesses~$\BAC^P$, $\PAuth^P$ and~$\AAuth^P$
(\emph{resp.} $\BAC^R$, $\PAuth^R$ and~$\AAuth^R$) model one session
of the passport's (\emph{resp.} reader's) side of the \BAC,
\PAuth\xspace and \AAuth\xspace protocols respectively. The
name~$\KPrDS$ models the signing key of the Document Signing authority
used in all passports. Each passport (identified by its master
keys~$\KE$ and~$\KM$, its signing key~$\KPrAA$, the owner's name,
picture, signature, ...) can run multiple times the \BAC\xspace
protocol followed by the \PAuth\xspace and \AAuth\xspace protocols.

\paragraph*{Security analysis.}
We explain below how \emph{strong anonymity} of these three protocols
executed together can be derived from the analysis performed on each
protocol in isolation. 
%
In~\cite{ACD-csf12}, as sequential composition could not be handled,
the analysis of the e-passports application 
had to exclude the execution of the \BAC\xspace protocol. 
Instead, it was assumed that the key  $\mathit{kenc}$ (resp. $\mathit{kmac}$) is 
``magically'' pre-shared between the passport and the reader. 
Thanks to our Theorem~\ref{theo:main-compo-seq-equiv}, 
we are now able to complete the analysis of the e-passport application.



To express strong anonymity, 
we need on one hand to consider a system in which the particular
e-passport with publicly known $id_1$, $sig_1$, $pic_1$, \emph{etc.}
is being executed, while on the other hand it is a different
e-passport with publicly known $id_2$, $sig_2$, $pic_2$, \emph{etc.}
which is being executed. 
We consider the following
context:\\[1mm]
\null\hfill
$\begin{array}{ll}
  C_A[\_]\ \stackrel{\defi}{=} 
   !\new\ \KE.\ \new\ \KM.\ \new\ \KPrAA.  
 \Let\ id = \diff[id_1, id_2]\ \In \ 
 \ \dots\ ! \ \_
\end{array}$
\hfill\null

\smallskip{}

This composition context differs in the e-passport being executed on
the left-hand process and on the right-hand process. In other words,
the system satisfies anonymity if an observer cannot distinguish the
situation where the e-passport with publicly known $id_1$, $sig_1$,
$pic_1$, \emph{etc.} is being executed, from the situation where it is
another e-passport which is being executed. 
%
To check if the tagged version of the e-passport application (we
assume here that $\BAC$, $\PAuth$, and $\AAuth$ are tagged in
different ways)
preserves strong anonymity, one thus needs to check if the following
biprocess satisfies diff-equivalence {(with $\Phi_0 = \{w_1 \refer \vk(\KPrDS)\}$):}\\[1mm]
\null\hfill
$\triple{\KPrDS}{C_A[{\BAC^R[\PAuth^R\mid \AAuth^R] \mid \BAC^P[\PAuth^P\mid \AAuth^P]}]}{\Phi_0}
$\hfill\null

\smallskip{}

We can instead check whether \BAC, \PAuth\xspace and \AAuth\xspace
satisfy anonymity in isolation, 
\emph{i.e.} if the following three diff-equivalences hold:\\[1mm]
\null\hfill$
\begin{array}{clccl}
\multirow{2}{*}{$\triple{\KPrDS}{\new\ d.\ C_A[\BAC^+]}{\emptyset}$}& \multirow{2}{*}{$(\alpha)$} 
  &\;\;\;\;&\triple{\KPrDS}{C'_A[{\PAuth^R \mid \PAuth^P}]}{\Phi_0}  & (\beta) \\
  &&&\triple{\KPrDS}{C'_A[{\AAuth^R \mid \AAuth^P}]}{\emptyset}  & (\gamma)
\end{array}
$\hfill\null\\
\noindent where 
\begin{itemize}
\item $\BAC^+ \stackrel{\defi}{=}
  \begin{array}[t]{l}
   \phantom{\mid } \BAC^R[\Out(d, (xkenc_R, xkmac_R))] \\
\mid \BAC^P[\Out(d, (xkenc_P, xkmac_P))] \\
    \mid \In(d, x).\ \In(d, y).\ \myIf\, x = y \,\myThen\; 0 \;\myElse\; 0;
  \end{array}
  $
\item  $C'_A[\_]   \stackrel{\defi}{=}  C_A[C''_A[\_]]$; and
\item $C''_A[\_]  \stackrel{\defi}{=}  \New\ kenc. \New\ kmac.\
  \Let\ (xkenc_R, xkmac_R) = (kenc, kmac) \ \In \\ 
\null\hspace{4.6cm} \Let\ (xkenc_P, xkmac_P) = (kenc, kmac) \ \In\ \_$.
\end{itemize}

\noindent Then, applying
Theorem~\ref{theo:compo-par-diff-equivalence-main} to
$(\beta)$ and $(\gamma)$ we derive that the following biprocess satisfies
diff-equivalence:\\[1mm]
\null\hfill
$
\begin{array}{cl}
  \triple{\KPrDS}{C'_A[{\PAuth^R \mid \AAuth^R \mid \PAuth^P\mid \AAuth^P}]}{\Phi_0} & (\delta) \\
\end{array}
$\hfill\null

\smallskip{}

\noindent and applying Theorem~\ref{theo:main-compo-seq-equiv} to $(\alpha)$
and $(\delta)$, we derive the required diff-equivalence:\\[1mm]
\null\hfill
$
\triple{\KPrDS}{C_A[{\BAC^R[\PAuth^R\mid \AAuth^R] \mid \BAC^P[\PAuth^P\mid \AAuth^P]}]}{\Phi_0}
$
\hfill\null
\smallskip{}

Note that we can do so because Hypotheses (1-4) stated in Theorem~\ref{theo:main-compo-seq-equiv} are satisfied, and in particular because $\BAC^R/\BAC^P/\emptyset$ is a good key-exchange protocol under $\{\KPrDS\}$ and $C_A$.

 \section{Conclusion}
\label{sec:conclusion}

We investigate composition results for reachability properties as well as privacy-type properties expressed using a notion of equivalence. 
Relying on a generic composition result
that allows one to strongly relate any trace of the composed protocol to a trace of the so-called disjoint case, we
 derive parallel composition results, as well as a sequential composition results (the case of key-exchange protocols under various composition contexts).

All these results work in a quite general setting, \emph{e.g.} processes may have non trivial else branches, we consider arbitrary primitives expressed using an equational theory, and processes may even share some standard primitives as long as they are tagged in different ways. We illustrate the usefulness of our results through the mobile phone and e-passport applications.

We believe that our generic result could be used to derive further composition results. 
We may want for instance to relax the notion of being a \emph{good protocol} at the price of studying 
a less ideal scenario when analysing the protocol $Q$ in isolation. 
We may also want to consider situations
 where sub-protocols  sharing some data are arbitrarily interleaved.
Moreover, even if we consider arbitrary primitives, sub-protocols can only share some standard primitives provided that they are tagged. It would be nice to relax these conditions. 
This would allow one to compose  protocols (and not their tagged versions) or to compose protocols that both rely on primitives for which no tagging scheme actually exists (\emph{e.g.} exclusive-or).

\noindent \paragraph*{Acknowledgement}
The research leading to these results has received funding from the European 
   Research Council under the European Union's Seventh Framework Programme 
   (FP7/2007-2013) / ERC grant agreement $n^{\circ}$ 258865, project
   ProSecure, as well as the ANR project JCJC VIP n$^o$ 11 JS02 006 01. 

\bibliographystyle{abbrv}
\bibliography{main}

\appendix

\section{Case study: 3G mobile phones}
\label{sec:app-phones}

In this section, we
look at the confidentiality and
 privacy guarantees provided by the
Authentication and Key Agreement protocol (\AKA) and the Submit SMS
procedure (\sSMS), when run in composition as specified by the 3GPP
consortium in~\cite{TS33102}. 


\smallskip{}

\noindent {\bf The \AKA\xspace protocol}  achieves mutual
authentication between a Mobile Station (MS) and the Serving Network (SN), and
allows them to establish shared session keys to be used to secure
subsequent communications. We consider here its fixed version as
described in~\cite{AMRR-CCS12}
which relies on a public key infrastructure.
 In particular, in case of
failure, \emph{i.e.} 
$\varphi_{\mathsf{test}}$ is not satisfied, the answer $\mathit{RES}$ 
is encrypted using the public key of the SN,
\emph{i.e.} $\pk(\sk_{SN})$.

\begin{figure}[ht]
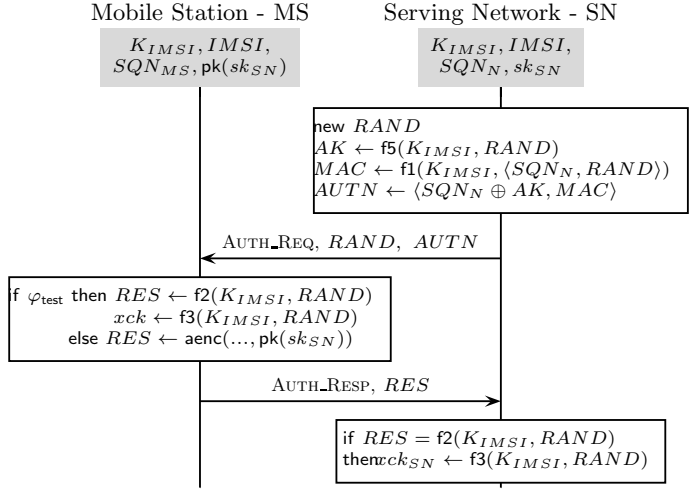

\begin{center}
 \vspace{-0.5cm}
    \setmsckeyword{}
    \drawframe{no}
    \begin{msc}{}
      \setlength{\instfootheight}{0\mscunit}
      \setlength{\instheadheight}{.5\mscunit}
      \setlength{\instwidth}{0\mscunit}
      \setlength{\regionbarwidth}{0\mscunit}
      \setlength{\instdist}{1cm}
      \setlength{\bottomfootdist}{0\mscunit}
      \declinst{MS}{
        \begin{tabular}[c]{c}
          Mobile Station - MS \\ 
          \colorbox{gris}{\scriptsize{$\begin{array}{c}K_{IMSI}, IMSI, \\SQN_{MS}, \pk(\sk_{SN})\end{array}$}} 
        \end{tabular}
      }{} 
      \dummyinst{x1}
      \dummyinst{x2}
      \dummyinst{x3}
      \declinst{SN}{
        \begin{tabular}[c]{c}
          Serving Network - SN \\ 
          \colorbox{gris}{\scriptsize{$\begin{array}{c}K_{IMSI}, IMSI, \\SQN_{N}, \sk_{SN}\end{array}$}} 
        \end{tabular}
      }{}

      \nextlevel[-1]
      \action*{\parbox{4.7cm}{\scriptsize{$
          \begin{array}[c]{l}
           \hspace{-2mm} \mathsf{new}\ RAND \\
           \hspace{-2mm} AK\leftarrow \mathsf{f5}(K_{IMSI},RAND) \\
           \hspace{-2mm} MAC\leftarrow \mathsf{f1}(K_{IMSI},\langle SQN_{N}, RAND\rangle) \\
           \hspace{-2mm} AUTN\leftarrow \langle SQN_{N} \oplus AK, MAC\rangle
          \end{array}
          $}}}{SN}
      \nextlevel[4]

      \mess{\scriptsize{$\textsc{Auth\_Req, } RAND,\ AUTN$}}{SN}{MS}
      \nextlevel[0.5]

      \action*{\parbox{4.8cm}{\scriptsize{$
          \begin{array}[c]{l}
          \hspace{-2mm}  \mathsf{if}\ \varphi_\mathsf{test} \ \mathsf{then}\ RES\leftarrow \mathsf{f2}(K_{IMSI},RAND) \\
          \hspace{12mm} xck\leftarrow \mathsf{f3}({K_{IMSI}},RAND) \\
          \hspace{6mm}  \mathsf{else}\ RES \leftarrow \aenc( ..., \pk(\sk_{SN})) \\
          \end{array}
          $}}}{MS}
    \nextlevel[2.6]
      \nextlevel[0.7]
      
      \mess{\scriptsize{$\textsc{Auth\_Resp, } RES$}}{MS}{SN}
      \nextlevel[0.5]

      \action*{\parbox{4.1cm}{\scriptsize{$
          \begin{array}[c]{l}
          \hspace{-1mm}   \mathsf{if}\  RES = \mathsf{f2}(K_{IMSI},RAND) \\
          \hspace{-1mm}   \mathsf{then}
             \begin{array}[t]{l}
          \hspace{-1mm} xck_{SN}\leftarrow \mathsf{f3}({K_{IMSI}},RAND) \\
     	   \end{array} \\
          \end{array}
          $}}}{SN}
      \nextlevel[1]
      
      \setlength{\instwidth}{0\mscunit}
    \end{msc}

  \end{center}
  \vspace{-0.25cm}
 \caption{The \AKA\xspace protocol (variant proposed in~\cite{AMRR-CCS12})}
  \vspace{-0.25cm}
  \label{fig:aka-fix-simplified}
\end{figure}

The functions $\mathsf{f1}-\mathsf{f5}$, used to compute the authentication parameters, are one-way keyed cryptographic functions, and $\oplus$ denotes the exclusive-or operator. 
%
$\mathit{AUTN}$ contains a MAC of the concatenation of the
random number with a sequence number $\mathit{SQN_N}$ generated by the
network using an individual counter for each subscriber. 
The sequence number $\mathit{SQN_N}$ allows the mobile station to verify the freshness of the authentication request to defend against replay attacks.
%
%
The mobile station computes the ciphering key $\mathit{ck}$ 
and stores it in~$\mathit{xck_{MS}}$. 
It also computes the authentication response $\mathit{RES}$ and sends
it to the network. The network authenticates the mobile station by
verifying whether the received response is equal to the expected
one. If so, the network also computes its version of the key $\mathit{ck}$  
and stores it in~$\mathit{xck_{SN}}$. 


\smallskip{}

\noindent{\bf The \sSMS\xspace protocol} 
allows a MS to send an SMS to another MS through the Network. 
The confidentiality of the sent SMS relies on the session
key~$\mathit{ck}$ established through the execution of the \AKA\xspace 
protocol between the MS and the network. 
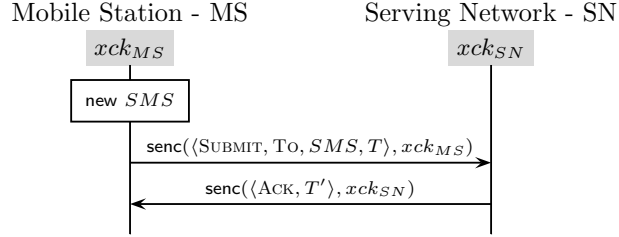
\begin{center}
   \vspace{-0.5cm}
  \setmsckeyword{}
  \drawframe{no}
  \begin{msc}{}
    \setlength{\instfootheight}{0\mscunit}
    \setlength{\instheadheight}{0\mscunit}
    \setlength{\instwidth}{0\mscunit}
    \setlength{\regionbarwidth}{0\mscunit}
    \setlength{\instdist}{1.2cm}
    \setlength{\bottomfootdist}{0\mscunit}
    \declinst{MS}{
      \begin{tabular}[c]{c}
        Mobile Station - MS \\ 
        \colorbox{gris}{\small{$xck_{MS}$}} \\ \\
      \end{tabular}
    }{} 
    \dummyinst{x1}
    \dummyinst{x2}
    \dummyinst{x3}
    \declinst{SN}{
      \begin{tabular}[c]{c}
        Serving Network - SN \\ 
        \colorbox{gris}{\small{$xck_{SN}$}} \\ \\
      \end{tabular}
    }{}
    
    \nextlevel[-1.9]
    \action*{{\scriptsize{$
          \begin{array}[c]{c}
            \mathsf{new}\ SMS\\
          \end{array}
          $}}}{MS}
    \nextlevel[2.2]
    
    \mess{\scriptsize{$\senc(\langle \textsc{Submit}, \textsc{To}, SMS, T\rangle,xck_{MS})$}}{MS}{SN}
    
    \nextlevel[1.1]
    \mess{\scriptsize{$\senc(\langle \textsc{Ack}, T'\rangle, xck_{SN})$}}{SN}{MS}
    
    \setlength{\instwidth}{0\mscunit}
  \end{msc}
  
\end{center}

It is always the MS that initiates the \sSMS\xspace procedure. It does
so by encrypting 
the content of
the SMS it wants to submit, together with the number of the
destination MS and a timestamp $T$, with the session key $\mathit{ck}$ previously
established. The message also contains a constant \textsc{Submit}. The
Network acknowledges the receipt of this message with a 
message that includes a constant $\textsc{Ack}$ and a timestamp $T'$, encrypted with~$\mathit{ck}$.


\paragraph*{Security analysis.}

The \sSMS\xspace procedure uses a ciphering session key $CK$ established through the execution of the \AKA\xspace protocol for the confidentiality of the sent and received SMSs. We can thus use Theorem~\ref{theo:main-compo-seq-confidentiality} and Theorem~\ref{theo:main-compo-seq-equiv}
to reason in a modular way about the  confidentiality and privacy guarantees
provided by these two protocols. 


\medskip{}

\noindent {\bf Strong unlinkability} requires that an outside observer does not
see the difference between the two following scenarios:
{\it (i)} a same mobile phone sends several SMSs; or {\it (ii)}
multiple mobile phones send at most one SMS each. To model this
requirement, we consider the 
composition context\footnote{We use 
$\Let\ x=M\ \In\ P$ to denote the process
  $P\{M/x\}$.}:
\[
\begin{array}{rl}
  C_U[\_]\ \stackrel{\defi}{=} & 
!\new\ \imsi_1.\ \new\ \KIMSI_1.\   ! \new\ \imsi_2.\ \new\ \KIMSI_2.\  \\  
                               & \ \Let\ \imsi = \diff[\imsi_1, \imsi_2]\ \In 
                                \ \Let\ \KIMSI = \diff[\KIMSI_1, \KIMSI_2]\ \In \\ 
                               & \ \new\ sqn.\ \new\ sms.\ \_ \\
\end{array}
\]

In the left-hand process, the identity of the phone in the filling
process is $\imsi_1$ and the long-term key is $\KIMSI_1$, allowing the
same phone to execute multiple times the \AKA\xspace protocol followed
by the \sSMS\xspace protocol. In the right-hand process, the values
that are used in the filling process are $\imsi_2$ and $\KIMSI_2$,
restricting 
the execution of the considered protocols to at most one time.
To check if the considered 3G protocols 
satisfy
strong unlinkability, one needs 
to check if the following biprocess satisfies diff-equivalence:
\[
\begin{array}{c}
\triple{\KPrSN}{C_U[\AKA^{SN}[\sSMS^{SN}] \mid
  \AKA^{MS}[\sSMS^{MS}]]}{\Phi_0} \;\;\;
\mbox{where $\Phi_0 =\{w_1 \refer \pk(\KPrSN)\}$.} 
\end{array} 
\]

\noindent Actually, thanks Theorem~\ref{theo:main-compo-seq-equiv}, this equivalence can be derived from the following two smaller diff-equivalences:
$$
\quad\triple{\KPrSN}{\new\ d.\ C_U[\AKA^+]}{\Phi_0} \text{\quad and \quad} \triple{\KPrSN}{C'_U[\sSMS]}{\emptyset}
$$
where $sSMS \stackrel{\defi}{=} \sSMS^{SN} \mid \sSMS^{MS}$,
$$
\AKA^+ \stackrel{\defi}{=}
\begin{array}[t]{l}
  \AKA^{SN}[\Out(d,xck_{SN})] \mid \AKA^{MS}[\Out(d,xck_{MS})] \mid \\
  \In(d, x).\ \In(d, y).\ \myIf\, x = y \,\myThen\; 0 \;\myElse\; 0
\end{array}
$$
and $ C'_U[\_]\ \stackrel{\defi}{=} C_U[\new\ ck.  \Let\
xck_{SN}=ck\ \In \ \Let\ xck_{MS}=ck\ \In\ \_\,]$.
\smallskip{}

Indeed, let $P = \AKA^{SN}[0] \mid \AKA^{MS}[0]$ and $Q = 
\new\ ck. [xck_{SN} := ck].[xck_{MS} := ck]. (\sSMS^{SN} \mid
\sSMS^{MS})$, and assume that $\Psi$ is the empty frame.
Considering the \AKA\xspace and \sSMS\xspace protocols, we can
check that $P/\Phi_0$ and $Q/\Psi$ are composable under $\Ec_0  =
\{\KPrSN\}$ and $C_U$ (according to
Definition~\ref{def:composability-main}). Note that $\fn(P) \cap
\fn(Q) \cap \bn(C_U) = \emptyset$, and thus the last condition
trivially holds. 
Furthermore, using ProVerif we can show that the remaining properties
are also satisfied, and that the two ``small'' equivalences also
hold. 

\medskip{}

\noindent{\bf Weak secrecy} requires  that the
sent/received SMS is not deducible by an outsider, and
can be modelled using the 
context
$$C_{WS}[\_]\ \stackrel{\defi}{=}   !\new\ \imsi.\ \new\ \KIMSI.\ ! \new\
  sqn. \new\ sms. \_.$$

To check if the considered 3G protocols satisfy weak secrecy of sent/received
SMSs
\emph{w.r.t.} 
some initial intruder knowledge, \emph{e.g.} $\Phi_0 = \{w_1 \refer
\pk(\KPrSN)\}$, 
one needs to check if the following process does not reveal $sms$\\
\null\hfill
$\triple{{\KPrSN}}{C_{WS}[\AKA^{SN}[\sSMS^{SN}] \mid
  \AKA^{MS}[\sSMS^{MS}]]}{\Phi_0}.$\hfill\null

\smallskip{}

However, according to Theorem~\ref{theo:main-compo-seq-confidentiality} we can instead
check whether \AKA\ and \sSMS\ satisfy
confidentiality of SMSs in isolation, \emph{i.e.} whether the
following processes do not reveal $sms$: 
  \[
\begin{array}{ccc}
\triple{\KPrSN}{C_{WS}[\AKA^{SN}[0] \mid \AKA^{MS}[0]]}{\Phi_0}\;
 & (\alpha) \\[2mm]
  \triple{\KPrSN}{C'_{WS}[\sSMS^{SN} \mid \sSMS^{MS}]}{\emptyset} &   \; (\beta)
\end{array}
\]

\noindent where $
  C'_{WS}[\_]\ \stackrel{\defi}{=}  C_{WS}[\new\ ck.
    \Let\ xck_{SN}=ck\ \In 
    \ \Let\ xck_{MS}=ck\ \In\ \_]$.

Note that the composition context $C_{WS}$ is the same as $\fst(C_U)$ (up
to some renaming), thus
Hypotheses (1-4) of Theorem~\ref{theo:main-compo-seq-confidentiality} also hold and
we derive the weak secrecy property by simply analysing this property on 
$\AKA$ and $\sSMS$ in isolation.

We are left with verifying that \AKA\xspace and \sSMS\xspace preserve
 weak secrecy of exchanged SMSs. \AKA\ trivially does, since $sms$ is
 not used in \AKA\xspace.
Using ProVerif we can show that \sSMS\xspace also preserves weak
secrecy of SMSs.

\medskip{}

\noindent{\bf Strong secrecy} 
requires that an outside oberver does not distinguish the situation where  $\sms_1$  is sent
, from the situation where $\sms_2$  is sent
, although he might know the content of $\sms_1$ and $\sms_2$. To
model this requirement we consider the following composition
context. \\
\null\hfill
$  C_{SS}[\_]\ \stackrel{\defi}{=}   !\new\ \imsi.\ \new\ \KIMSI.\ ! \new\ sqn. 
                              \ \Let\ \sms = \diff[\sms_1, \sms_2]\
                              \In\ \_ $\hfill\null

\smallskip{}

\noindent where $\sms_1$ and $\sms_2$ are two free names known to the attacker. This composition context differs on the content of the SMS being sent on the left-hand process and on the right-hand process.
To check if the considered 3G protocols satisfy confidentiality of
sent SMSs \emph{w.r.t.} 
some initial intruder knowledge, \emph{e.g.} $\Phi_0 = \{w_1 \refer
\pk(\KPrSN)\}$, 
one needs to check if the following biprocess satisfies
diff-equivalence \\
\null\hfill
$\triple{{\KPrSN}}{C_{SS}[\AKA^{SN}[\sSMS^{SN}] \mid
  \AKA^{MS}[\sSMS^{MS}]]}{\Phi_0}.$\hfill\null

\smallskip{}

However, according to Theorem~\ref{theo:main-compo-seq-equiv} we can instead
check whether \AKA\ and \sSMS\ satisfy
strong secrecy of SMSs in isolation: \\[1mm]
\null\hfill
$  \triple{\KPrSN}{C_{SS}[\AKA^+]}{\Phi_0}\; (\alpha) \mbox{\;\; and \;\;} \triple{\KPrSN}{C'_{SS}[\sSMS]}{\emptyset} \; (\beta)$\hfill\null

\smallskip{}

\noindent where $\AKA^+$ and $\sSMS$ defined as for unlinkability, and 
$$C'_{SS}[\_]\ \stackrel{\defi}{=}  C_{SS}[\new\ ck. \Let\ xck_{SN}=ck\ \In \ \Let\ xck_{MS}=ck\ \In\ \_].$$

 \smallskip{}
 Indeed, \AKA\xspace/$\Phi_0$ and \sSMS\xspace/$\Psi$ (for
 $\Psi=\emptyset$) are composable under $\Ec_0 = \{\KPrSN\}$ and
 $C_{SS}$. Regarding
 the conditions of Theorem~\ref{theo:main-compo-seq-equiv}:
 \emph{(i)} it is easy to see that \AKA\xspace satisfies the 
 abstractability property: both the MS and the SN compute the key $\mathit{ck}$
 and store it respectively in the assignment variables $xck_{MS}$ and
 $xck_{SN}$ by applying the function $\mathsf{f3}$, which is a one way
 function, to $K_{\imsi}$ and $\mathit{RAND}$; \emph{(ii)} using
 ProVerif we can show that the considered two protocols do not reveal
 $xck_{MS}$ and $xck_{SN}$, and that \AKA\xspace is actually a good
 key-exchange protocol.
 We are left with verifying that \AKA\xspace and \sSMS\xspace preserve
 strong secrecy of exchanged SMSs. \AKA\ trivially does, since the
 left and the right-hand processes are syntactically equal. 
Using ProVerif we can show that \sSMS\xspace also preserves strong
secrecy of SMSs.






\section{Case study: e-passport}
\label{sec:epassports}

As mentioned in the introduction, many applications like electronic passports or mobile phones rely on several protocols running in
composition (parallel, sequential, or nested). In this section, we
show that our results can help in the analysis of this sort of complex
system considering the  e-passport application. 


\subsection{Protocols description}

The information stored in the chip of the passport is organised in data groups ($\DG_1$ to $\DG_{19}$). For example, $\DG_5$ contains a JPEG copy of the displayed picture, and~$\DG_7$ contains the displayed signature. The verification key~$\vk(\KPrAA)$ of the passport, together with its certificate $\sign(\vk(\KPrAA), \KPrDS)$ issued by the Document Signer authority are stored in~$\DG_{15}$. The corresponding signing key~$\KPrAA$ is stored in a tamper resistant memory, and cannot be read or copied. For authentication purposes, a hash of all the $\DG$s together with a signature on this hash value issued by the Document Signer authority are stored in a separate file, 
the Security Object Document:\\[1mm]
\null\hfill
$\SOD \stackrel{\defi}{=} \langle \sign(\h(\DG_1, \dots, \DG_{19}),
\KPrDS), \; \h(\DG_1, \dots, \DG_{19}) \rangle.$\hfill\null

\smallskip{}
 
The ICAO standard specifies several protocols through which these
information can be accessed~\cite{ICAO-Passport}. 
\smallskip{}

\noindent {\bf The Basic Access Control (\BAC) protocol} (see
Figure~\ref{fig:BAC}) establishes a key seed $xkseed$ from which two
sessions keys $\ksenc$ and $\ksmac$ are derived.
The purpose of $ksenc$ and $ksmac$ is to prevent  skimming and eavesdropping on the subsequent communication with the e-passport (see below). The security of the \BAC\xspace protocol relies on two master keys,
$\KE$ and $\KM$, which are optically retrieved from the passport by the reader before executing the \BAC\xspace protocol.

The reader initiates the protocol by sending a challenge to the
passport and the passport replies with a random 64-bit string $n_P$. The reader then creates its own random nonce and some new random key material, both 64-bits. These are encrypted, along with the tag's nonce and sent back to the reader. A MAC is computed using the $\KM$ key and sent along with the message, to ensure the message is received correctly. The tag receives this message, verifies the MAC, decrypts the message and checks that its nonce is correct; this guarantees to the tag that the message from the reader is not a replay of an old message. The tag then generates its own random 64-bits of key material and sends this back to the reader in a similar message, except this time the order of the nonces is reversed, this stops the readers message being replayed directly back to the reader. The reader checks the MAC and its nonce, and both the tag and the reader use the xor of the key material as the seed for a session key, with which to encrypt the rest of the session.

\begin{figure}[ht]
  \vspace{-0.5cm} 
  \hspace{-0.5cm} 
  $$  \setmsckeyword{} 
  \drawframe{no}
  \begin{msc}{}
    \setlength{\instwidth}{0\mscunit}
    \setlength{\instdist}{0.75cm} 
    \dummyinst{d0}
    \declinst{pport}{
      \begin{tabular}[c]{c}
        Passport Tag \\
        \colorbox{gris}{\small{$\KE, \KM$}}
      \end{tabular}}{}
    
    \dummyinst{d1} 
    \dummyinst{d2} 
    \dummyinst{d3} 
    \dummyinst{d4}
    
    \declinst{reader}{
      \begin{tabular}[c]{c}
        Reader \\
        \colorbox{gris}{\small{$\KE, \KM$}}
      \end{tabular}}{}
    
    \nextlevel[-1]
    \mess{\small{$\getC$}}{reader}{pport}
    \nextlevel[0.5]

    \action*{
      \small{$\New\ n_P$}
    }{pport}
    \nextlevel[1.5]

    \mess{\small{$n_P$}}{pport}{reader}
    \nextlevel[0.5]

    \action*{\small{$
        \begin{array}[c]{l}
          \New\ n_R,\ \New\ n_R \\
          xenc \leftarrow \senc(\langle n_R, n_P, k_R\rangle, \KE) \\
          xmac \leftarrow \mac(xmac, \KM) \\
        \end{array}
        $}}{reader} 
    \nextlevel[4]
    
    \mess{\small{$\langle xenc, xmac \rangle$}}{reader}{pport} \nextlevel[0.5]
    
    \action*{\small{$
        \begin{array}[c]{l}
          \mathsf{if}\ \varphi_P \,  \mathsf{then} \; xkseed \leftarrow \ldots \\ 
          \quad \ksenc_P \leftarrow \mathsf{ekg}(xkseed) \\
          \quad \ksmac_P \leftarrow \mathsf{mkg}(xkseed) \\
        \end{array}
        $}}{pport} 
    \nextlevel[4]
    
    \mess{\small{$resp$}}{pport}{reader}
    \nextlevel[1]

    \action*{\small{$
        \begin{array}[c]{l}
          \mathsf{if}\ \varphi_R \, \mathsf{then} \; xkseed \leftarrow\ldots \\
          \quad \ksenc_R \leftarrow \mathsf{ekg}(xkseed) \\
          \quad \ksmac_R \leftarrow \mathsf{mkg}(xkseed) 
        \end{array}
        $}}{reader} 
    \nextlevel[2.5]
  \end{msc}
  $$
\vspace{-1cm}
\caption{The \BAC\xspace  protocol}
\label{fig:BAC}
\end{figure}
Once the \BAC\xspace protocol has been successfully executed, the
reader gains access to the information stored in the RFID tag through
the Passive Authentication (\PAuth)
and the Active Authentication (\AAuth) protocols  that can be executed in any order.
\smallskip{}

\noindent {\bf The \PAuth\xspace protocol} (see Figure~\ref{fig:PA}) is an authentication mechanism that proves that the content of the RFID chip is authentic. Through \PAuth\xspace the reader retrieves the information stored in the $\DG$s and the~$\SOD$. It then verifies that the hash value stored in the~$\SOD$ corresponds to the one signed by the Document Signer authority. It further checks that this hash value is consistent with the received~$\DG$s.

\begin{figure}[ht]
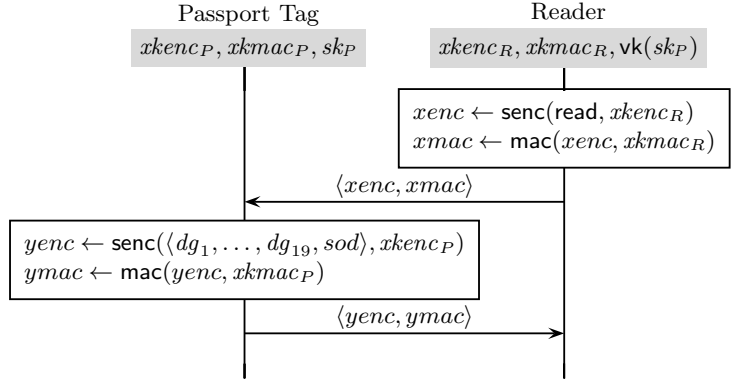

  \[
  \centering
  \setmsckeyword{} 
  \drawframe{no}    
  \begin{msc}{}
    \setlength{\instwidth}{0\mscunit}
    \setlength{\instdist}{0.85cm} 
    \dummyinst{d0}
    \declinst{pport}{
      \begin{tabular}[c]{c}
        Passport Tag \\
        \colorbox{gris}{\small{$\ksenc_P, \ksmac_P, \KPrAA$}}
      \end{tabular}}{}
    
    \dummyinst{d1} 
    \dummyinst{d2} 
    \dummyinst{d3} 
    \dummyinst{d4}
    
    \declinst{reader}{
      \begin{tabular}[c]{c}
        Reader \\
        \colorbox{gris}{\small{$\ksenc_R, \ksmac_R, \vk(\KPrAA)$}}
      \end{tabular}}{}
    
    \nextlevel[-1.25]
    
    
    \action*{\small{$
        \begin{array}[c]{l}
          xenc\leftarrow\senc(\mathsf{read},\ksenc_R) \\
          xmac\leftarrow\mac(xenc,\ksmac_R)
        \end{array}
        $}}{reader} 
    \nextlevel[3]
    
    \mess{\small{$\langle xenc, xmac \rangle$}}{reader}{pport} 
    \nextlevel[0.5]

    \action*{\small{$
        \begin{array}[c]{l}
          yenc\leftarrow \senc(\langle \DG_1, \dots, \DG_{19}, sod\rangle,\ksenc_P) \\
          ymac\leftarrow \mac(yenc,\ksmac_P)
        \end{array}
        $}}{pport} 
    \nextlevel[3]

    \mess{\small{$\langle yenc, ymac \rangle$}}{pport}{reader}
    
  \end{msc}
  \]
  \vspace{-1.4cm}
  \caption{Passive Authentication protocol}
  \label{fig:PA}
\end{figure}

\smallskip{}

\noindent {\bf The \AAuth\xspace protocol} (see Figure~\ref{fig:AA}) is an authentication mechanism that prevents cloning of the passport chip. It relies on the fact that the secret key~$\KPrAA$ of the passport cannot be read or copied. The reader sends a random challenge to the passport, that has to return a signature on this challenge using its private signature key~$\KPrAA$. The reader can then verify using the verification key~$\vk(\KPrAA)$ that the signature was built using the expected passport key.
\begin{figure}[ht]
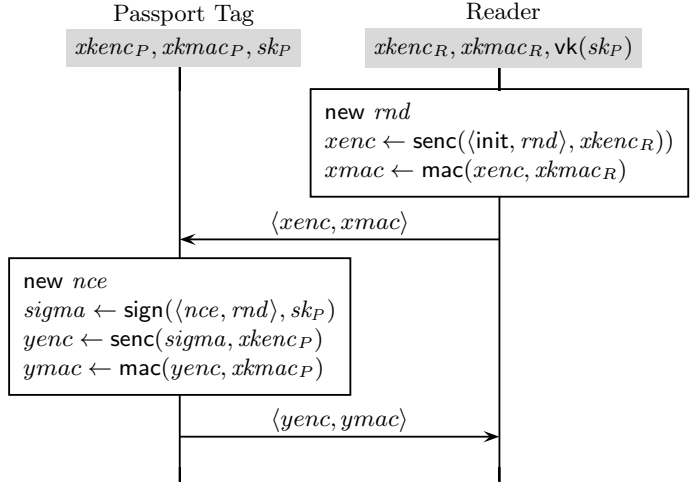

\[
  \centering
  \setmsckeyword{} 
  \drawframe{no}
  \begin{msc}{}
    \setlength{\instwidth}{0\mscunit}
    \setlength{\instdist}{0.85cm}

    \declinst{pport}{
    \begin{tabular}[c]{c}
      Passport Tag \\
      \colorbox{gris}{\small{$\ksenc_P, \ksmac_P, \KPrAA$}}
    \end{tabular}}{}

    \dummyinst{d1}
    \dummyinst{d2}
    \dummyinst{d3}
    \dummyinst{d4}

    \declinst{reader}{ 
    \begin{tabular}[c]{c}
      Reader \\
      \colorbox{gris}{\small{$\ksenc_R, \ksmac_R, \vk(\KPrAA)$}}
    \end{tabular}}{}

    \nextlevel[-1.25]


    \action*{\small{$
    \begin{array}[c]{l}
      \mathsf{new}\ \mathit{rnd} \\
      xenc\leftarrow\senc(\langle \mathsf{init}, \mathit{rnd}\rangle, \ksenc_R)) \\
      xmac\leftarrow\mac(xenc,\ksmac_R)
    \end{array}
    $}}{reader}
    \nextlevel[4]

    \mess{\small{$\langle xenc, xmac\rangle$}}{reader}{pport}
    \nextlevel[0.5]

    \action*{\small{$
    \begin{array}[c]{l}
      \mathsf{new}\ \mathit{nce} \\
      \mathit{sigma}\leftarrow \sign(\langle \mathit{nce}, \mathit{rnd}
      \rangle, \KPrAA) \\
      yenc\leftarrow \senc(\mathit{sigma},\ksenc_P) \\
      ymac\leftarrow \mac(yenc,\ksmac_P)
    \end{array}
    $}}{pport}
    \nextlevel[4.75]

    \mess{\small{$\langle yenc, ymac\rangle$}}{pport}{reader}
  \end{msc}
\]
\vspace{-1.4cm}
  \caption{Active Authentication protocol}
  \label{fig:AA}
\end{figure}

\subsection{Privacy analysis}
All three protocols \BAC, \PAuth\xspace and \AAuth\xspace, rely on symmetric encryption and message authentication codes. 
Note that the only publicly known verification key is $\vk(\KPrDS)$ and is only used by the \PAuth\ protocol. Thus, we can use our composition results, and in particular tour Theorems~\ref{theo:compo-par-diff-equivalence-main} and~\ref{theo:main-compo-seq-equiv}, to reason in a modular way about the privacy guarantees provided by the tagged version of the e-passport application. 

In~\cite{ACD-csf12}, as sequential composition could not be handled, the analysis of the e-passports application had to exclude the execution of the \BAC\xspace protocol. Instead, it was assumed that the keys $kenc$ and $kmac$ were ``magically'' pre-shared. With our sequential composition result (Theorem~\ref{theo:main-compo-seq-equiv}), we avoid this unsafe abstraction, as we can now consider the execution of the \BAC\xspace protocol for the establishment of these two keys. In this way, we are here able to complete the analysis of the e-passport application.

\smallskip{}

According to the ICAO standard, the reader optically retrieves the
passport's master keys $\KE$ and $\KM$ before executing the
\BAC\xspace protocol to establish the key seed for $kenc$ and
$kmac$. The reader can then decide to execute \PAuth\xspace and/or
\AAuth\xspace in any order. Formally, this corresponds to the
sequential composition of the \BAC\xspace protocol and of the
\PAuth\xspace and \AAuth\xspace 
protocols composed in parallel. This system  can be modelled in our calculus as follows:\\[1mm]
\null\hfill
$
\begin{array}{ll}
  P\ \stackrel{\defi}{=} &  \new\ \KPrDS. \ !\new\ \KE.\ \new\ \KM.\ \new\ \KPrAA.\ \new\ \id.\ \new\ sig.\ \new\ pic.\ \dots \\ 
  & \quad\quad! (\BAC^R[\PAuth^R \mid \AAuth^R] \mid \BAC^P[\PAuth^P \mid \AAuth^P])
\end{array}
$\hfill\null

\smallskip{}

\noindent where $\id$, $sig$, $pic$, ... represent the name, the signature, the displayed picture, \emph{etc} of the e-passport's owner, \emph{i.e.} the data stored in the $\DG$s ($1$-$14$) and ($16$-$19$). The subprocesses~$\BAC^P$, $\PAuth^P$ and~$\AAuth^P$ (\emph{resp.} $\BAC^R$, $\PAuth^R$ and~$\AAuth^R$) model one session of the passport's (\emph{resp.} reader's) side of the \BAC, \PAuth\xspace and \AAuth\xspace protocols respectively. The name~$\KPrDS$ models the signing key of the Document Signing authority used in all passports. Each passport (identified by its master keys~$\KE$ and~$\KM$, its signing key~$\KPrAA$, the owner's name, picture, signature, ...) can run multiple times the \BAC\xspace protocol followed by the \PAuth\xspace and \AAuth\xspace protocols in any order.
\smallskip{}

To express strong anonymity, 
we need on one hand to consider a system in which the particular
e-passport with publicly known $id_1$, $sig_1$, $pic_1$, \emph{etc.}
is being executed, while on the other hand it is a different
e-passport with publicly known $id_2$, $sig_2$, $pic_2$, \emph{etc.}
which is being executed. For this we consider the following
composition context:\\[1mm]
\null\hfill
$\begin{array}{ll}
  C_A[\_]\ \stackrel{\defi}{=} 
& \Let\ id = \diff[id_1, id_2]\ \In \ 
 \ \dots\ !\_
\end{array}$
\hfill\null

\smallskip{}

This composition context differs in the e-passport being executed on
the left-hand process and on the right-hand process. In other words,
the systems satisfies anonymity if an observer cannot distinguish whether the e-passport with publicly known $id_1$, $sig_1$, $pic_1$, \emph{etc.} is being executed, or
another e-passport is being executed 
(with publicly known $id_2$, $sig_2$, $pic_2$, \emph{etc.})
\smallskip{}

To check if the tagged version of the e-passport application (we
assume here that $\BAC$, $\PAuth$, and $\AAuth$ are colored using
three distinct colors, and thus will be tagged in
different ways)
preserves strong anonymity, one thus needs to check if the following biprocess satisfies diff-equivalence:\\[1mm]
\null\hfill
$\triple{\KPrDS}{C_A[\TAGG{\BAC^R[\PAuth^R\mid \AAuth^R] \mid \BAC^P[\PAuth^P\mid \AAuth^P]}]}{\Phi_0}
$\hfill\null

\smallskip{}

We can instead check whether \BAC, \PAuth\xspace and \AAuth\xspace
satisfy anonymity in isolation, 
\emph{i.e.} if the following three diff-equivalences hold:
\[
\begin{array}{cl}
  \triple{\KPrDS}{\new\ d.\ C_A[\TAGG{\BAC^+}]}{\emptyset}& (\alpha) \\
  \triple{\KPrDS}{C'_A[\TAGG{\PAuth^R \mid \PAuth^P}]}{\Phi_0}  & (\beta) \\
  \triple{\KPrDS}{C'_A[\TAGG{\AAuth^R \mid \AAuth^P}]}{\emptyset}  & (\gamma)
\end{array}
\]
where $
\begin{array}[t]{lcl}
  \BAC^+    &  \stackrel{\defi}{=} &
  \begin{array}[t]{l}
    \BAC^R[\Out(d, (xkenc_R, xkmac_R))] \mid \BAC^P[\Out(d, (xkenc_P, xkmac_P))] \\
    \In(d, x).\ \In(d, y).\ \myIf\, x = y \,\myThen\, 0 \,\myElse\, 0
  \end{array} \\
  C'_A[\_]  &  \stackrel{\defi}{=} & C_A[C''_A[\_]] \\
  C''_A[\_] & \stackrel{\defi}{=} & \New\ kenc. \ \New\ kmac \\
  &                     & \Let\ (xkenc_R, xkmac_R) = (kenc, kmac)\ \In \\
  &                     & \Let\ (xkenc_P, xkmac_P) = (kenc, kmac)\ \In\ \_ \\
\end{array}
$

\noindent Then, applying Theorem~\ref{theo:compo-par-diff-equivalence-main} to $(\beta)$ and $(\gamma)$ we derive that the following biprocess satisfies diff-equivalence:\\[1mm]
\null\hfill
$
\begin{array}{cl}
  \triple{\KPrDS}{C'_A[\TAGG{\PAuth^R \mid \AAuth^R \mid \PAuth^P\mid \AAuth^P}]}{\Phi_0} & (\delta) \\
\end{array}
$\hfill\null

\smallskip{}

\noindent and applying Theorem~\ref{theo:main-compo-seq-equiv} to $(\alpha)$ and $(\delta)$, we derive the required diff-equivalence:\\[1mm]
\null\hfill
$
\triple{\KPrDS}{C_A[\TAGG{\BAC^R[\PAuth^R\mid \AAuth^R] \mid \BAC^P[\PAuth^P\mid \AAuth^P]}]}{\Phi_0}
$
\hfill\null
\smallskip{}

Indeed, let  $P = \BAC^R[0] | \BAC^P[0]$;
 and $Q = C''_A[\PAuth \mid \AAuth]$, and
assume that $\Psi$ is the empty frame. 
We can check that $P/\Psi$ and $Q/\Phi_0$ are composable under
$\mathcal{E}_0 = \{\KPrDS\}$ and $C_A$ (according to Definition
6). Note that $\fn(P )\cap \fn(Q) \cap \bn(C_A) = \emptyset$, and thus
the last condition trivially holds. Furthermore, using ProVerif we can
show that properties $(\alpha)$ and $\gamma$ are also
satisfied. Unfortunately, ProVerif does not terminate when given the
script corresponding to equivalence $(\beta)$. Note that ProVerif does
not terminate when given the script corresponding to the hole system
either. At this point our only solution would be to rely on a manual
proof. Our composition results have allowed us to reduce a big equivalence that existing tools cannot handle, to a much smaller one.

\section{Sharing primitives via tagging}
\label{sec:app-tagging}



We recall in this section the tagging scheme as presented
in~\cite{ACD-csf12}.  However, since
we would like to be able to iterate our composition results (in order
to compose \emph{e.g.} three protocols), we consider a fixed set of
colors (not only two), and we allow a process to be colored with many
colors.  Actually, a colored process is a process with a color
assigned to each of its action. This gives us enough flexibility to
allow different kinds of compositions, and to iterate our composition
results.

We consider a family of signatures $\Sigma_1,\ldots,\Sigma_p$ disjoint
from each other and disjoint from $\Sigma_0$. In order to tag a
process, we introduce a new family of signatures $\Sigma^{\Tag}_1,
\ldots, \Sigma^{\Tag}_p$. For each $i \in \{1,\ldots, p\}$, we have
that $\Sigma^{\Tag}_i= \{ \Tag_i, \unTag_i\}$ where $\Tag_i$ and
$\unTag_i$ are two function symbols of arity~$1$ that we will use for
tagging.  The role of the $\Tag_i$ function is to tag its argument
with the tag~$i$.  The role of the $\unTag_i$ function is to remove
the tag. To model this interaction between $\Tag_i$ and $\unTag_i$, we
consider the equational theory: $\E_{\Tag_i} = \{ \unTag_i(\Tag_i(x))
= x\}$.

For our composition result, we will assume that the two protocols we
want to compose only share symbols in $\Sigma_0$. Thus, for this, we
split the set $\{1,\ldots,p\}$ into two disjoint sets~$\alpha$
and~$\beta$.  Given a subset $\gamma \subseteq \{1,\ldots,p\}$, we
denote:
\begin{center}
  $\begin{array}{ccccc} \Sigma_\gamma \stackrel{\defi}{=} \bigcup_{i
      \in \gamma} \Sigma_i && \Sigma^{\Tag}_{\gamma}
    \stackrel{\defi}{=} \bigcup_{i\in \gamma} \Sigma^{\Tag}_{i} &&
    \Sigma^+_{\gamma} \stackrel{\defi}{=} \Sigma_{\gamma} \cup
    \Sigma^{\Tag}_{\gamma}\\
    \E_\gamma \stackrel{\defi}{=} \bigcup_{i \in \gamma} \E_i &&
    \E^{\Tag}_\gamma \stackrel{\defi}{=} \bigcup_{i \in \gamma}
    \E^{\Tag}_i &&\E^+_{\gamma} \stackrel{\defi}{=} \E_{\gamma} \cup
    \E^{\Tag}_{\gamma}
  \end{array}
  $
\end{center}


\begin{definition}
  Let $i \in \{1,\ldots,p\}$, and $u$ be a term built over $\Sigma_i
  \cup \Sigmazero$. The $i$-tagged version of~$u$, denoted
  $\TAG{u}{i}$ is defined as follows:

  \smallskip{}

  \noindent$
  \begin{array}{rcl}
    \TAG{\senc(u,v)}{i} &\stackrel{\defi}{=}&
    \senc(\Tag_i(\TAG{u}{i}),\TAG{v}{i})\\
    \TAG{\aenc(u,v)}{i} &\stackrel{\defi}{=}&
    \aenc(\Tag_i(\TAG{u}{i}),\TAG{v}{i})\\
    \TAG{\sign(u,v)}{i} &\stackrel{\defi}{=}&
    \sign(\Tag_i(\TAG{u}{i}),\TAG{v}{i})\\
    \TAG{\h(u)}{i} &\stackrel{\defi}{=} &\h(\Tag_i(\TAG{u}{i}))
  \end{array}
  \begin{array}{rcl}
    \TAG{\sdec(u,v)}{i} &\stackrel{\defi}{=}& \unTag_i(\sdec(\TAG{u}{i},\TAG{v}{i}))\\
    \TAG{\adec(u,v)}{i} &\stackrel{\defi}{=}&
    \unTag_i(\adec(\TAG{u}{i},\TAG{v}{i}))\\
    \TAG{\checksign(u,v)}{i} &\stackrel{\defi}{=}&
    \unTag_i(\checksign(\TAG{u}{i},\TAG{v}{i}))\\
    \TAG{\ffun(u_1, \ldots, u_n)}{i}  &\stackrel{\defi}{=}&
    \ffun(\TAG{u_1}{i}, \ldots, \TAG{u_n}{i}) \; \mbox{otherwise.}
  \end{array}
  $\hfill\null
\end{definition}

Note that we do not tag the pairing function symbol (this is actually
useless), and we do not tag the~$\pk$ and~$\vk$ function symbols. Note
that tagging $\pk$ and $\vk$ would lead us to consider an unrealistic
modelling for asymmetric keys.  This definition is extended as
expected to formulas~$\varphi$ (those involved in conditionals) by
applying the transformation on each term that occurs in~$\varphi$.

\smallskip{}

\begin{example}
  \label{ex:tag}
  Let $\Sigma_1 = \{\ffun, \gfun\}$, and consider the terms $u =
  \senc(\gfun(r),k)$ and $v =\ffun(\sdec(y, k),r)$ built on $\Sigma_1
  \cup \Sigma_0$.  We have that $\TAG{u}{1} = \senc(\Tag_1(\gfun(r)),
  k)$, and $\TAG{v}{1} = \ffun(\unTag_1(\sdec(y, k)),r)$.
\end{example}

\smallskip{}

We also introduce the following notion that allows us to associate a
color to a term that is not necessarily well-tagged.

\smallskip{}

\begin{definition}
  Let $u$ be a term. We define $\racinebis(u)$, namely the tag of the
  root of $u$ as follows:
  \begin{itemize}
  \item $\racinebis(u) = \bot$ when $u \in \N \cup \X$;
  \item $\racinebis(u) = i$ if $u = \ffun(u_1, \ldots, u_n)$ and
    either $\ffun \in \Sigma_{i} \cup \Sigma^\Tag_{i}$, or $\ffun \in
    \{ \senc, \aenc, \sign, \h\}$ and $u_1 = \Tag_i(u'_1)$ for some
    $u'_1$.
  \item $\racinebis(u) = 0$ otherwise.
  \end{itemize}
\end{definition}

\medskip{}

Before extending the notion of tagging to processes, we have to
express the tests that are performed by an agent when he receives a
message that is supposed to be tagged. This is the purpose of
$\TestTag{u}{i}$ that represents the tests which ensure that every
projection and every untagging performed by an agent during the
computation of~$u$ is successful.

\smallskip{}

\begin{definition}
  Let $i \in \{1,\ldots,p\}$, and $u$ be a term built on $\Sigma^+_i
  \cup \Sigmazero$. We define $\TestTag{u}{i}$ as follows:
  \begin{itemize}
  \item $\TestTag{u}{i} \stackrel{\defi}{=}\TestTag{u_1}{i} \wedge
    \TestTag{u_2}{i} \wedge \Tag_i(\unTag_i(u)) = u$ when $u =
    \gfun(u_1,u_2)$ with $\gfun \in \{\sdec,\adec,\checksign\}$
  \item $\TestTag{u}{i} \stackrel{\defi}{=} \TestTag{u_1}{i} \wedge
    u_1 = \langle \proj_1(u_1), \proj_2(u_1) \rangle$ when $u =
    \proj_j(u_1)$ with $j \in \{1,2\}$
  \item $\TestTag{u}{i} \stackrel{\defi}{=} \true$ when $u$ is a name
    or a variable
  \item $\TestTag{u}{i} \stackrel{\defi}{=} \TestTag{u_1}{i} \wedge
    \ldots \wedge \TestTag{u_n}{i}$ otherwise (with $u = f(u_1, \dots, u_n)$).
  \end{itemize}
\end{definition}
This definition is extended as expected to formulas $\varphi$,
\emph{i.e.} $\TestTag{\varphi}{i} \stackrel{\defi}{=} \bigwedge_{u = v
  \in \varphi} \TestTag{u}{i} \wedge \TestTag{v}{i}$.  \medskip{}

\begin{example}
  \label{ex:test}
  Again, consider $u = \senc(\gfun(r),k)$ and $v =\ffun(\sdec(y,
  k),r)$. We have that:
  \begin{itemize}
  \item $\TestTag{\TAG{u}{1}}{1} = \true$
  \item $\TestTag{\TAG{v}{1}}{1} = \Tag_1(\unTag_1(\sdec(y, k))) =
    \sdec(y,k)$
  \end{itemize}
\end{example}

\medskip{}

We consider colored plain processes meaning that initially the actions
of a plain process will be annotated with a color, \emph{i.e.} an
integer in $\{1,\ldots, p\}$. The actions that need to be annotated
are those that involve some composed terms, \emph{i.e.} inputs,
outputs, conditionals, and assignments.
An action colored by $i \in \{1,\ldots,p\}$ can only contain function
symbol from $\Sigma_i$. Given a set $\gamma \subseteq \{ 1, \ldots,
p\}$, we say than an action is colored with $\gamma$ if this action is
colored by $i \in \{1,\ldots,p\}$.  For colored plain processes, the
transformation $\TAGG{P}$ is defined as follows:
\noindent\[
\begin{array}{c}
  \TAGG{0}   \stackrel{\defi}{=} 0 \hfill 
  \TAGG{!P}   \stackrel{\defi}{=}\ !\TAGG{P}  \hfill
  \TAGG{\new\; k. P}   \stackrel{\defi}{=} \new\; k. \TAGG{P} \hfill
  \TAGG{P \mid Q} \stackrel{\defi}{=} \TAGG{P} \mid \TAGG{Q}
  \\
  \TAGG{\In(u,x)^i.P}   \stackrel{\defi}{=} \In(u,x)^i.\TAGG{P} \hfill
  \TAGG{[x := v]^i.P} \stackrel{\defi}{=} \mbox{(\texttt{if} $\TestTag{\TAG{v}{i}}{i}$ \texttt{then} $[x:=\TAG{v}{i}]^i.\TAGG{P}$)$^i$}\\[2mm]
  \TAGG{\Out(u,v)^i.Q} \stackrel{\defi}{=} \mbox{(\texttt{if} $\TestTag{\TAG{v}{i}}{i}$ \texttt{then} $\Out(u,\TAG{v}{i})^i.\TAGG{Q}$)$^i$}\\[2mm]
  \TAGG{\mbox{(\texttt{if} $\varphi$ \texttt{then} $P$ \texttt{else} $Q$)$^i$}}  \stackrel{\defi}{=} 
  (\mbox{\texttt{if} $\varphi_{\mathsf{test}}$ \texttt{then} } (\mbox{\texttt{if} $\TAG{\varphi}{i}$ \texttt{then} $\TAGG{P}$ \texttt{else} $\TAGG{Q}$}))^i \mbox{ \texttt{else} $0$})^i\\
  \hfill\mbox{ where } \varphi_{\mathsf{test}} = \TestTag{\TAG{\varphi}{i}}{i}\\
\end{array}
\]

Roughly, instead of simply outputting a term~$v$, a process will first
perform some tests to check that the term is correctly tagged and he
will output its $i$-tagged version~$\TAG{v}{i}$. For an assignment, we
will also check that the term is correctly tagged. For a conditional,
the process will first check that the terms involved in the test
$\varphi$ are correctly tagged before checking that the test is
satisfied.  The annotations that occur on a plain process do not
affect its semantics.

\begin{definition}
  Consider a set $\gamma \subseteq \{1,\ldots,p\}$. Consider a plain
  process $P$ built over $\Sigma^+_\gamma \cup \Sigma_0$. We say that
  $P$ is tagged if there exists a colored plain process $Q$ built over
  $\Sigma_\gamma$ such that $P = \TAGG{Q}$.
\end{definition}



\section{Biprocesses}
\label{sec:biprocesses}

The semantics of biprocesses is defined via a relation
$\lrstep{\ell}_\bi$ that expresses when and how a biprocess may evolve. Intuitively,
a biprocess reduces if and only if both sides of the
biprocess reduce in the same way: a communication succeeds on both
sides, a conditional has to be evaluated in the same way in both sides
too. When the two sides of the biprocess reduce in different ways, the
biprocess blocks. 
The semantics of biprocesses is formally described in
Figure~\ref{fig:semantics-biprocesses}.

\begin{figure*}[ht]
\[
\begin{array}{lr}
\quadruple{\Ec}{\{\mbox{\texttt{if} $\diff(\varphi_L, \varphi_R)$ \texttt{then} $Q_1$ \texttt{else} $Q_2$}\}\uplus\p}{\Phi}{\sigma} 
\lrstep{\tau}_\bi  \quadruple{\Ec}{Q_1\uplus\p}{\Phi}{\sigma}
&  \mbox{(\sc Then)}\\
\multicolumn{1}{r}{\mbox{if $u\sigma =_\E v\sigma$ for each $u = v \in \varphi_L \cup \varphi_R$}}
&\\[2mm]
\quadruple{\Ec}{\{\mbox{\texttt{if} $\diff(\varphi_L, \varphi_R)$ \texttt{then} $Q_1$ \texttt{else} $Q_2$}\}\uplus\p}{\Phi}{\sigma}
\lrstep{\tau}_\bi 
\quadruple{\Ec}{Q_2\uplus\p}{\Phi}{\sigma} & \mbox{(\sc Else)}\\
\multicolumn{1}{r}{\mbox{if $u_L\sigma \neq_\E v_L\sigma$ for some $u_L = v_L \in \varphi_L$}}
&\\
\multicolumn{1}{r}{\mbox{and $u_R\sigma \neq_\E v_R\sigma$ for some $u_R = v_R \in \varphi_R$}}
&\\[2mm]
\quadruple{\Ec}{\{\Out(c,u) . Q_1; \In(c,x).Q_2\}\uplus\p}{\Phi}{\sigma} 
\lrstep{\tau}_\bi
\quadruple{\Ec}{Q_1 \uplus Q_2 \uplus\p}{\Phi}{\sigma \cup \{x \mapsto u\sigma\}}  
& \mbox{(\sc Comm)} \\[2mm]
\quadruple{\Ec}{\{[x := v] . Q\}\uplus\p}{\Phi}{\sigma}
\lrstep{\tau}_\bi
\quadruple{\Ec}{Q \uplus\p}{\Phi}{\sigma \cup \{x \mapsto v\sigma\}} 
& \mbox{(\sc Assgn)}\\[2mm] 
\quadruple{\Ec}{\{\In(c,z).Q\}\uplus\p}{\Phi}{\sigma} 
\lrstep{\In({c},M)}_{\bi} 
\quadruple{\Ec}{Q \uplus\p}{\Phi}{\sigma \cup\{z\mapsto u\}} & \hfill\mbox{(\sc In)}
\\
\multicolumn{1}{r}{\mbox{if ${c} \not\in \Ec$, $M\Phi =u$, $\fv(M) \subseteq \dom(\Phi)$ and $\fn(M) \cap \Ec = \emptyset$}}&
\\[2mm]
\quadruple{\Ec}{\{\Out(c,u).Q\}\uplus\p}{\Phi}{\sigma}
\lrstep{\nu w_n.\Out({c},w_n)}_{\bi}
\quadruple{\Ec}{Q\uplus\p}{\Phi\cup\{w_n\refer u\sigma\}}{\sigma}&
\mbox{(\sc Out-T)}
\\
\multicolumn{1}{r}{\mbox{if ${c} \not\in \Ec$, $u$ is a term of base type, and $w_n$ is a variable such that $n = | \Phi| + 1$}}&
\\[2mm]
\quadruple{\Ec}{\{\new\; n.Q\} \uplus \p}{\Phi}{\sigma} \lrstep{\tau}_{\bi}
\quadruple{\Ec \cup \{ n \}}{Q\uplus \p }{\Phi}{\sigma}
&\mbox{(\sc New)}\\[2mm]
\quadruple{\Ec}{\{!Q\} \uplus \p}{\Phi}{\sigma} \lrstep{\tau}_{\bi} 
\quadruple{\Ec}{\{!Q; Q\rho\} \uplus \p }{\Phi}{\sigma}& \mbox{(\sc
  Repl)}\\[2mm]
\multicolumn{1}{r}{\mbox{where $\rho$ is used to rename variables in
    $\bv(Q)$}}&\\
\multicolumn{1}{r}{\mbox{(resp. names in $\bn(Q)$) with fresh variables
    (resp. names).}}&\\[2mm]
\quadruple{\Ec}{\{P_1 \mid P_2\} \uplus \p}{\Phi}{\sigma} \lrstep{\tau}_{\bi}
\quadruple{\Ec}{\{P_1,P_2\} \uplus \p}{\Phi}{\sigma} & \hfill \mbox{(\sc Par)}
\end{array}
\]
where $n$ is a name, $c$ is a name of channel type (here we can have $c
= \diff(c_1,c_2)$), $u$, $v$ are
terms that may contain the $\diff$ operator, and $x, z$ are variables. 
The term $M$ used in the {\sc In}  rule is a term that does not
contain any occurrence of the $\diff$ operator. The attacker has to do
the same computation in both sides.

\caption{Semantics for biprocesses}
\label{fig:semantics-biprocesses}
\end{figure*}

\section{The disjoint case for a trace}
\label{sec:app-compo-main}

Composition usually works well in the so-called disjoint case,
\emph{i.e.} when the protocols under study do not share any
secrets. The goal of this section is to show that we can map any trace
corresponding to an execution of a protocol (with some sharing) to
another trace which corresponds to an execution of a ``disjoint case''
(where protocols do not share any secrets) preserving static
equivalence. We need a strong mapping to ensure that processes evolve
simultaneously, and we rely for this on the notion of
\emph{biprocesses}.

We will see in this section that the composition of processes sharing
some secrets (the so-called shared case) behaves as if they did not
share any secret (the so-called disjoint case), provided that the
shared secrets are never revealed and processes are tagged.

\subsection{Material for combination}
\label{subsec:combMaterial}

To handle the different signatures and equational theories, we
consider the notion of \emph{ordered rewriting}.  It has been shown
that by applying the \emph{unfailing completion procedure} to $\E$
where $\E =\E_1 \uplus \E_2 \uplus \ldots \E_p$ is the union of
disjoint equational theories $(\Sigma_i,\E_i)$ (for all $i,j$, we have
that $\Sigma_i \cap \Sigma_j =\emptyset$), we can derive a (possibly
infinite) set of equations~$\Or$ such that on ground terms:
\begin{enumerate}
\item the relations~$=_\Or$ and~$=_\E$ are equal, 
\item the rewriting system~$\to_\Or$ is
  convergent. 
\end{enumerate}
Since the relation~$\to_\Or$ is convergent on ground terms, we
define~$M\mathord{\downarrow}_\E$ (or briefly $M\mathord{\downarrow}$)
as the unique normal form of the ground term~$M$ for~$\to_\Or$.
{These notations are extended as expected to sets of terms.}

\smallskip

We now introduce our notion of \emph{factors} and state some properties on them w.r.t. the different equational theories. A similar notion is also used in~\cite{CRicalp05}.


\begin{definition}[factors]
\label{def:subterm}
Let $M \in \T(\Sigma, \N \cup \X)$.  The \emph{factors} of~$M$, denoted
$\fct(M)$, are the maximal syntactic subterms of $M$ that are alien to~$M$
\end{definition}


\begin{lemma}
\label{lem:app_CRicalp05}
Let $M$ be a ground term such that all its factors are in normal
form and $\racine(M) \in \Sigma_i$. Then
\begin{itemize}
\item  either $M\mydownarrow \in \fct(M) \cup \{n_{min}\}$,
\item  or $\racine(M\mydownarrow) \in \Sigma_i$ and $\fct(M\mydownarrow) \subseteq
  \fct(M) \cup \{n_{min}\}$.
\end{itemize}
\end{lemma}


\begin{lemma}
\label{lem:change alien}
Let $t$ be a ground term with $t = C_1[u_1, \ldots, u_n]$ where
$C_1$ is a context built 
on $\Sigma_i$, $i \in \{1, \ldots, p\}$ and the terms $u_1, \ldots, u_n$ are the factors of $t$ in normal form. 
Let $C_2$ be a context built on $\Sigma_i$ (possibly a hole) such that
$t\mydownarrow = C_2[u_{j_1}, \ldots, u_{j_k}]$ with $j_1, \ldots, j_k
\in \{0 \ldots n\}$ 
and $u_0 = n_{min}$ (the existence is given by Lemma~\ref{lem:app_CRicalp05}). 
We have that for all ground terms $v_1, \ldots, v_n$ in normal form and alien to $t$, if 
\begin{center}
for every $q, q'\in \{1 \ldots n\}$ we have  $u_q = u_{q'} \Leftrightarrow v_q = v_{q'}$
\end{center}
then $C_1[v_1, \dots, v_n]\mydownarrow = C_2[v_{j_1}, \ldots, v_{j_k}]$ with $v_0 = n_{min}$.
\end{lemma}

A proof of these lemmas can be found in~\cite{CD-jar10,cheval-phd2012}.

\subsection{Generic composition result}
\label{subsec:theo-main}

We consider two sets $\alpha, \beta$ such that $\alpha \cup \beta =
\{1, \ldots, p\}$ and $\alpha \cap \beta = \emptyset$.
We consider a plain colored process~$P$ built on $\Sigma_\alpha \cup
\Sigma_\beta {\cup \Sigma_0}$ without replication and such that
$\bn(P) = \fv(P) = \emptyset$. This means that~$P$ is a process with
no free variables, and we assume that it contains no name restrictions
(\emph{i.e.} no $\new$ instructions).

\medskip{}

\begin{example}
  We consider the process $P_{\Diffie}$ as given in
  Example~\ref{ex:process} but we replace
\begin{itemize}
\item the $0$ at the end of $P_A$ with $Q_A = \new\,
  s_A. \Out(c,\sencDiffie(s_A,x_A))$, and
\item the $0$ at the end of $P_B$ with $Q_B = \new\,
  s_B. \Out(c,\sencDiffie(s_B,x_B))$.
\end{itemize}
Intuitively, once the Diffie-Hellman key has been established and
stored in~$x_A$ (resp.~$x_B$), each participant will use it to encrypt
a fresh secret, namely $s_A$ or $s_B$, and then send it to the other
participant.

 Note that when
function symbols of $\Sigmazero$ are used by only one of the protocols
to compose, we can either consider them as part of $\Sigmazero$ and so
they will be tagged, or they can be put into distinct signatures
(using renaming as above) and so they will not be tagged. The
composition theorem can be applied both ways.

To avoid confusion between the encryption schemes that processes can
share, \emph{i.e.} the function symbols in $\Sigmazero$, and the
asymmetric encryption used in $P_{\Diffie}$ but not used in $Q_A$
and~$Q_B$, we will rename them by $\aencDiffie, \adecDiffie,
\pkDiffie$.
Thus, we consider
$p = 2$, $\alpha = \{1\}$, $\beta = \{ 2\}$ with
$(\Sigma_\alpha,\E_\alpha) = (\Sigma_{\Diffie}, \E_{\Diffie})$ with
\begin{itemize}
\item $\Sigma_\Diffie = \{\aencDiffie, \adecDiffie, \pkDiffie, \ffun,
\gfun\}$, and 
\item $\E_\Diffie =
\{\adecDiffie(\aencDiffie(x,\pkDiffie(y)),y) = x, \;\;\ffun(\gfun(x),y) =
\ffun(x,\gfun(y))\}$
\end{itemize}
whereas $\Sigma_\beta = \{\sencDiffie,\sdecDiffie\}$ and $\E_\beta =
\{\sdecDiffie(\sencDiffie(x,y),y) = x\}$. This equational theory is
used to model symmetric encryption/decryption, \emph{i.e.} the
primitives used in the processes~$Q_A$ and~$Q_B$. 
\smallskip{}

\noindent Now, we consider $P = P'_A \mid P'_B$ where:
\begin{itemize}
\item $P'_A = 
  \begin{array}[t]{l}
    \Out(c, \aencDiffie(\langle n_A, \gfun(r_A) \rangle,\pkDiffie(sk_B))). \In(c,y_A).\\ 
    \texttt{if } \proj_1(\adec(y_A, sk_A)) = n_A \\
    \texttt{then }[x_A := \ffun(\proj_2(\adecDiffie(y_A,sk_A)),r_A)]. \Out(c,\sencDiffie(s_A,x_A))
  \end{array}$
\item $P'_B = 
      \begin{array}[t]{l}
      \In(c,y_B). \Out(c, \aencDiffie(\langle \proj_1(\adecDiffie(y_B,sk_B)),\gfun(r_B)\rangle,\pkDiffie(sk_A))). \\
      \lbrack x_B := \ffun(\proj_2(\adecDiffie(y_B,sk_B)),r_B)\rbrack. \Out(c,\sencDiffie(s_B,x_B))
    \end{array}$
\end{itemize}

\smallskip{}

Note that $\bn(P) = \fv(P) = \emptyset$.  We choose to color the three
first actions of $P'_A$ (resp. $P'_B$) with $1 \in \alpha$, and the
remaining ones (\emph{i.e.} those that come from $Q_A$ and $Q_B$) with
$2 \in \beta$.
\end{example}

\medskip{}

We denote $\fn^\gamma(P)$ the set of free names of $P$ that occur in
actions colored with $\gamma$, and $\fv^\gamma(P)$ the set of
variables of $P$ that occur in an action colored with $\gamma$, and
that are not bound by an action colored with $\gamma$.
We consider a set $\Ec_0$ of names such that $\fn^\alpha(P) \cap
\fn^\beta(P) \cap \Ec_0 = \emptyset$. This means that each name in
$\Ec_0$ can only occur in one type of actions (those colored~$\alpha$
or those colored $\beta$).  We denote $z^{\alpha}_1, \ldots,
z^{\alpha}_k$ (resp. $z^{\beta}_1, \ldots, z^{\beta}_l$) the variables
occurring in the left-hand side of an {assignment} colored $\alpha$
(resp. $\beta$), \emph{i.e.} the variable $x$ such that the action $[x
:= v]$ occurs in $P$ and is colored $\alpha$ (resp. $\beta$).  We
assume that $\fv^\alpha(P) \subseteq \{z^\beta_1,\ldots,z^\beta_l\}$
and $\fv^\beta(P) \subseteq \{z^\alpha_1,\ldots,z^\alpha_k\}$.

\smallskip{}

These conditions ensure that sharing between the parts of the process
which are colored in different ways is only possible via the
assignment variables. This is not a real limitation but this allows us
to easily keep track of the shared data.

\smallskip{}

\begin{example}
  Continuing our example, we have $\fn^{\alpha}(P) =
  \{r_A,r_B,n_A,sk_A,sk_B\}$ and $\fn^{\beta}(P) = \{s_A,s_B\}$. Regarding
  variables: $\fv^{\alpha}(P) = \emptyset$, whereas $\fv^{\beta}(P) =
  \{x_A, x_B\}$.

  Let $\Ec_0 = \fn^\alpha(P) \cup \fn^{\beta}(P)$.  To follow the same
  notation as those introduced in this section, we may want to rename
  $x_A$ with $z^\alpha_1$ and $x_B$ with $z^\alpha_2$.  Note that
  $\fv^{\beta}(P) \subseteq \{z^\alpha_1,z^\alpha_2\}$.
\end{example}

\smallskip{}

Let $\Ec_\alpha = \{n^\alpha_1, \ldots, n^\alpha_k\}$ and $\Ec_\beta =
\{n^\beta_1, \ldots, n^\beta_l\}$ be two sets of fresh names of base
type such that $\Ec_\alpha \cap \Ec_\beta = \emptyset$.  We
define~$\rho_\alpha$ and~$\rho_\beta$ as follows:
\begin{itemize}
  \setlength{\itemsep}{0mm}
\item $\dom(\rho_\alpha) = \{z^\beta_1, \ldots, z^\beta_l\}$,
  $\dom(\rho_\beta) = \{z^\alpha_1, \ldots, z^\alpha_k\}$;
\item $\rho_\alpha(z^\beta_i) = n^\beta_i$ for each $i \in \{1,
  \ldots, l\}$; and
\item $\rho_\beta(z^\alpha_i) = n^\alpha_i$ for each $i \in \{1,
  \ldots, k\}$.
\end{itemize}

We do not assume that names in $\Ec_\alpha$ (resp. $\Ec_\beta$) are
distinct. For instance, we may have $n^\alpha_j = n^\alpha_{j'}$ for
some $j \neq j'$.

\smallskip{}

Given a colored plain process~$P$, we denote
by~$\delta_{\rho_\alpha,\rho_\beta}(P)$, the process obtained by
applying $\rho_\alpha$ on actions colored $\alpha$, and~$\rho_\beta$
on actions colored $\beta$.
{This transformation maps the shared case to a particular disjoint
  case.}

\smallskip{}


\begin{example}
  Let $\Ec_\alpha = \{k^\alpha\}$ and $\Ec_\beta = \emptyset$, and
  consider the function $\rho_\beta$ defined as follows:
  $\rho_\beta(z^\alpha_1) = \rho_\beta(z^\alpha_2) = k^\alpha$.
  Applying $\delta_{\rho_\alpha,\rho_\beta}$ on $P$ gives us $D_A \mid
  D_B$ where:
  \begin{itemize}
    \setlength{\itemsep}{0.5mm}
  \item $D_A =
    \begin{array}[t]{l}
      \Out(c, \aencDiffie(\langle n_A, \gfun(r_A) \rangle,\pkDiffie(sk_B))). \In(c,y_A).\\ 
      \texttt{if } \proj_1(\adecDiffie(y_A, sk_A)) = n_A \\
      \texttt{then }[x_A := \ffun(\proj_2(\adecDiffie(y_A,sk_A)),r_A)]. \Out(c,\sencDiffie(s_A,k^\alpha))
    \end{array}
    $
  \item $D_B =
    \begin{array}[t]{l}
      \In(c,y_B). \Out(c, \aencDiffie(\langle \proj_1(\adecDiffie(y_B,sk_B)),\gfun(r_B)\rangle,\pkDiffie(sk_A))). \\
      \lbrack x_B := \ffun(\proj_2(\adecDiffie(y_B,sk_B)),r_B)\rbrack. \Out(c,\sencDiffie(s_B,k^\alpha))
    \end{array}
    $
  \end{itemize}
  Note that there is no sharing anymore between the part of the
  process colored $\alpha$ and the part of the process colored
  $\beta$.
\end{example}


Actually, the disjoint case obtained using the transformation
$\delta_{\rho_\alpha,\rho_\beta}$ behaves as the shared case but only
along executions that are \emph{compatible} with the chosen
abstractions, \emph{i.e.} executions that preserve the equalities and
the inequalities among assignment variables as done by the chosen
abstraction.  This notion is formally defined as follows:

\smallskip{}

Let $A$ be any extended process derived from $\triple{\Ec_\alpha
  \uplus \Ec_\beta\uplus \Ec_0}{\TAGG{P}}{\emptyset}$, \emph{i.e.}
such that $\triple{\Ec_\alpha \uplus \Ec_\beta \uplus
  \Ec_0}{\TAGG{P}}{\emptyset} \,\LRstep{\;\tr\;}\, A$.
For $\gamma \in \{ \alpha,\beta\}$, we say that $\rho_\gamma$ is
\emph{compatible} with $A = \quadruple{\Ec}{\p}{\Phi}{\sigma}$ when:
\begin{enumerate}
  \setlength{\itemsep}{0mm}
\item for all $x,y \in \dom(\sigma) \cap \dom(\rho_\gamma)$, we have
  that $x\sigma =_{\E} y\sigma$ if, and only if, $x\rho_\gamma =
  y\rho_\gamma$; and
\item for all $z \in \dom(\rho_\gamma)$, either
  $\racinebis(z\sigma\mydownarrow) = \bot$ or
  $\racinebis(z\sigma\mydownarrow) \not\in \gamma \cup \{0\}$.
\end{enumerate}

We say that $(\rho_\alpha, \rho_\beta)$ is \emph{compatible} with $A$
when both~$\rho_\alpha$ and~$\rho_\beta$ are compatible with $A$.
For $\gamma \in \{\alpha,\beta\}$, we define the 
\emph{extension} of $\rho_\gamma$, denoted $\rho_\gamma^+$, as follows:
\begin{itemize}
\item $\dom(\rho_\gamma^+) = \dom(\rho_\gamma) \cup \{{x\sigma\mydownarrow} ~|~ x \in \dom(\rho_\gamma)
  \}$, and 
\item for any $x \in \dom(\rho_\gamma)$, $\rho_\gamma^+(x) \stackrel{\defi}{=} \rho(x)$ and  $\rho_\gamma^+(x\sigma) \stackrel{\defi}{=} \rho_\gamma(x)$ .
\end{itemize}

\smallskip{}

Before stating our generic composition result, we have also to
formalize the fact that the shared keys are not revealed. Since
sharing is performed via the assignment variables,
we say that \emph{$A_0$ does not reveal the value of its assignments}
w.r.t. $(\rho_\alpha,\rho_\beta)$ if for any extended process $A =
\quadruple{\Ec}{\p}{\Phi}{\sigma}$ derived from $A_0$ and such that
$(\rho_\alpha,\rho_\beta)$ is compatible with~$A$, we have:\\[1mm]
\null\hfill $\new \, \Ec. \Phi \not\vdash k$ for any $k \in K_\alpha
\cup K_\beta$ \hfill\null

\noindent where for all $\gamma \in \{\alpha,\beta\}$, $K_\gamma =
\{t,\pk(t),\vk(t) \mid z \in \dom(\sigma) \cap \dom(\rho_\gamma)$ and
$(t = z\sigma \mbox{ or } t = z\rho_\gamma) \}$.

\medskip{}



\begin{restatable}{theorem}{theoremmain}
  \label{theo:main-main}
  Let $P$ be a plain colored process as described above, and~$B_0$ be
  an extended colored biprocess such that:
  \begin{itemize}
    \setlength{\itemsep}{0mm}
  \item $S_0 = \quadruple{\Ec_\alpha \uplus \Ec_\beta \uplus
      \Ec_0}{\TAGG{P}}{\emptyset}{\emptyset} \stackrel{\defi}{=}
    \fst(B_0)$,
  \item $D_0 = \quadruple{\Ec_\alpha \uplus \Ec_\beta \uplus
      \Ec_0}{P_D}{\emptyset}{\emptyset} \stackrel{\defi}{=} \snd(B_0)
    $, and
  \item $P_D =\delta_{\rho_\alpha,\rho_\beta}(\TAGG{P})$ for some
    $(\rho_\alpha,\rho_\beta)$ compatible with $D_0$, and
  \item $D_0$ does not reveal its assignments
    w.r.t. $(\rho_\alpha,\rho_\beta)$.
  \end{itemize}
 
  \noindent We have that:
  \begin{enumerate}
  \item For any extended process $S =
    \quadruple{\Ec_S}{\p_S}{\Phi_S}{\sigma_S}$ such that $S_0
    \,\LRstep{\;\tr\;}\, S$ with $(\rho_\alpha,\rho_\beta)$ compatible
    with $S$,
    there exists a biprocess~$B$ and an extended process $D =
    \quadruple{\Ec_D}{\p_D}{\allowbreak \Phi_D}{\sigma_D} $ such that
    $B_0 \,\LRstep{\;\tr\;}_\bi \, B$, $\fst(B) = S$, $\snd(B) = D$,
    and $\new\, \Ec_S. \Phi_S \statequiv \new\, \Ec_D. \Phi_D$.
  \item For any extended process $D =
    \quadruple{\Ec_D}{\p_D}{\Phi_D}{\sigma_D}$ such that $D_0
    \,\LRstep{\;\tr\;}\, D$ with $(\rho_\alpha,\rho_\beta)$ compatible
    with $D$,
    there exists a biprocess $B$ and an extended process $S =
    \quadruple{\Ec_S}{\p_S}{\allowbreak \Phi_S}{\sigma_S}$ such that
    $B_0 \, \LRstep{\;\tr\;}_\bi \, B$, $\fst(B) = S$, $\snd(B) = D$,
    and $\new\, \Ec_S. \Phi_S \statequiv \new\, \Ec_D. \Phi_D$.
  \end{enumerate}
\end{restatable}

This theorem is proved by induction on the length of the
derivation. For this, a strong correspondence between the
process~$S_0$ (shared case) and $D_0$ (disjoint case) has to be
maintained along the derivation, and the transformation
$\delta_{\rho_\alpha,\rho_\beta}$ has to be extended to allow
replacements also in~$\sigma$ and~$\Phi$. The rest of this section is dedicated to the proof of this theorem.

\begin{example}
  Going back to our running example, and forming a biprocess with $S_0
  = \quadruple{\Ec_0 \cup \{k^\alpha\}}{P'_A \mid
    P'_B}{\emptyset}{\emptyset}$ and $D_0 = \quadruple{\Ec_0 \cup
    \{k^\alpha\}}{D_A \mid D_B}{\emptyset}{\emptyset}$,
  Theorem~\ref{theo:main-main} gives us that these two processes
  behave in the same way when considering executions that are
  compatible with the chosen abstraction $\rho_\beta$, \emph{i.e.}
  executions that instantiate $x_A$ and $x_B$ by the same value.
\end{example}

A similar result as the one stated in Theorem~\ref{theo:main-main} was
proved in~\cite{CC-csf10}. Here, we consider in addition else
branches, and we consider a richer common signature. Moreover, relying
on the notion of biprocess, we show a strong link between the shared
case and the disjoint case, and we prove in addition static
equivalence of the resulting frames.

\subsection{Name replacement}
\label{sec : name replacement}

Now that we have fixed some notations, we have to explain how the replacement
will be applied on the shared process to extract the disjoint case.
Actually a same term will be abstracted differently depending on the
context which is just above it.

\smallskip{}


\begin{definition}
\label{def:rho-extension}
Let $(\rho^+_\alpha, \rho^+_\beta)$ be two functions from terms of
base type to names of base type. 
Let $\delta^{\rho^+_\alpha,\rho^+_\beta}_\gamma$, or shortly $\delta_\gamma$, ($\gamma \in\{\alpha,\beta\}$) be
the functions on terms that is defined as follows:
\begin{center}
$\delta_\gamma(u) =
  u\mydownarrow\rho^+_\gamma$ when 
  $\left\{\begin{array}{l}
    u\mydownarrow \in \dom(\rho^+_\gamma)\\
    \text{and }\racinebis(u) \not\in \gamma \cup \{0\}
  \end{array}\right.$
\end{center}

Otherwise, we have that $\delta_\gamma(u) = u$ when $u$ is a name or a
  variable; and $\delta_\gamma(\ffun(t_1, \ldots, t_k))$ is equal to
\begin{itemize}
\item  $\ffun(\delta_\gamma(t_1),
  \ldots, \delta_\gamma(t_k))$ if
  $\racinebis(\ffun(t_1, \ldots, t_n)) = 0$;
\item   $\ffun(\delta_{\alpha}(t_1),\ldots,\delta_{\alpha}(t_k))$
  if $\racinebis(\ffun(t_1, \ldots, t_n)) \in \alpha$. 
\item   $\ffun(\delta_{\beta}(t_1),\ldots,\delta_{\beta}(t_k))$
  if $\racinebis(\ffun(t_1, \ldots, t_n)) \in \beta$.
\end{itemize}
\end{definition}


\begin{definition}[Factor for $\Sigmazero$]
Let $u$ be a term. We define $\fct_{\Sigmazero}(u)$ the factors of a term~$u$ for $\Sigmazero$ as the maximal syntactic subterms $v$ of $u$ such that $\racinebis(v) \neq 0$.
\end{definition}

%
%

Let $\sigma$ be a substitution. 
We consider a pair $(\rho_\alpha,\rho_\beta)$ as defined in Section~\ref{subsec:theo-main} and compatible with~$\sigma$. We denote $(\rho^+_\alpha,\rho^+_\beta)$ the extension of $(\rho_\alpha,\rho_\beta)$ w.r.t. $\sigma$.
Thanks to compatibility,
$\rho^+_\alpha$ (resp. $\rho^+_\beta$) is injective on $\dom(\rho^+_\alpha) \smallsetminus \{z^\beta_1,\ldots, z^\beta_l\}$ (resp. $\dom(\rho^+_\beta) \smallsetminus \{z^\alpha_1,\ldots, z^\alpha_k\}$). Moreover, we also have that for all $u \in \dom(\rho^+_\alpha) \smallsetminus \{z^\beta_1,\ldots, z^\beta_l\}$ (resp. $\dom(\rho^+_\beta) \smallsetminus \{z^\alpha_1,\ldots, z^\alpha_k\}$), either $\racinebis(u) = \bot$ or $\racinebis(u) \not\in \alpha \cup \{0\}$ (resp. $\racinebis(u) \not\in \beta \cup \{0\}$).


\begin{lemma}
  \label{lem:transfoV}
 Let $t_1$ and $t_2$ be ground terms in normal form such that $(\fn(t_1) \cup \fn(t_2))
  \cap (\Ec_\alpha \uplus \Ec_\beta) = \emptyset$. We have that:
\begin{center}
$t_1 = t_2$ if, and only if,
$\delta_\gamma(t_1) =
\delta_\gamma(t_2)$
\end{center}
where $\gamma \in \{\alpha,\beta\}$.
\end{lemma}

\begin{proof}
 The right implication is trivial. We consider the left
 implication, and we prove the result by induction on
 $\mathrm{max}(|t_1|, |t_2|)$ when $\gamma = \alpha$.
The other case $\gamma = \beta$ can be handled in  a similar way.
\medskip

\noindent\emph{Base case $\mathrm{max}(|t_1|, |t_2|) = 1$:} In such a case, we
have that $t_1, t_2 \in \N$. We first assume that $\delta_\alpha(t_1)$
(and thus also $\delta_\alpha(t_2)$) is 
in $\Ec_\alpha\uplus \Ec_\beta$.
By hypothesis, we know that
$t_2$ and $t_1$ do not use names in $\Ec_\alpha \uplus \Ec_\beta$. 
Therefore, by definition of
$\delta_\alpha$, we can deduce that $t_1, t_2 \in
\dom(\rho^+_\alpha)$ and $t_1\rho^+_\alpha = t_2\rho^+_\alpha$, and
thus $t_1 = t_2$ thanks to $\rho^+_\alpha$ being injective on $\dom(\rho^+_\alpha) \smallsetminus \{z^\beta_1,\ldots, z^\beta_l\}$.
Now, we assume that $\delta_\alpha(t_1)$ (and thus also
$\delta_\alpha(t_2)$) is not in $\Ec_\alpha \uplus \Ec_\beta$.
In such a case, by
definition of $\delta_\alpha$, we have that $\delta_\alpha(t_1) = t_1$ and
$\delta_\alpha(t_2) = t_2$, and thus $t_1 = t_2$.

\medskip

\noindent\emph{Inductive step $\mathrm{max}(|t_1|, |t_2|) > 1$:} 
Assume w.l.o.g. that $|t_1| >1$. Thus, there exists a symbol function
$\ffun$ and terms $u_1, \ldots u_n$ 
such that $t_1 = \ffun(u_1, \ldots u_n)$. 
We do a case analysis on $t_1$ which is in normal form.

\medskip

\underline{Case $t_1 \in \dom(\rho^+_\alpha)$:} 
In such a case, $\delta_\alpha(t_1) = \delta_\alpha(t_2) = n$ for some
$n \in \Ec_\alpha$. By hypothesis, we know that
$t_2$ and $t_1$ do not use names in $\Ec_\alpha$, and we have that
$t_1\rho^+_\alpha =t_2\rho^+_\alpha$.
Therefore, we necessarily have that $t_1 = t_2$. 

\medskip

\underline{Case $t_1 \not\in  \dom(\rho^+_\alpha)$:} 
We do a new case analysis on $t_1$.

\smallskip{}

\emph{Case $\ffun \in \Sigma^+_i$ for some $i \in \{1,\ldots,p\}$:}
Let $\gamma \in \{\alpha,\beta\}$ such that $i \in \gamma$.
In such
a case, we have that $\delta_\alpha(t_1) = \ffun(\delta_\gamma(u_1), \ldots,
\delta_\gamma(u_n))$. But $\delta_\alpha(t_2) = \delta_\alpha(t_1)$ and by
definition of $\delta_\alpha$, it implies that there exist $v_1,
\ldots, v_n$ such that $t_2 = \ffun(v_1, \ldots, v_n)$ and $\ffun(\delta_\gamma(v_1), \ldots, \allowbreak\delta_\gamma(v_n)) = \delta_\alpha(t_2)$. Thus we have that
$\delta_\gamma(v_j) = \delta_\gamma(u_j)$ for all $j \in \{1, \ldots,
n\}$. Furthermore, since $t_1$ and $t_2$ are in normal form and not using names in $\Ec_\alpha\uplus \Ec_\beta$, we also know that $u_j$ and $v_j$ are in normal form and not using names in $\Ec_\alpha\uplus \Ec_\beta$, for every $j$. Since, we have that $\mathrm{max}(|t_1|,
|t_2|) > \mathrm{max}(|u_j|, |v_j|)$, for any $j$, by our inductive
hypothesis, we can deduce that $u_j = v_j$, for all $j$ and so $t_1 = \ffun(u_1,
\ldots, u_n) = \ffun(v_1, \ldots, v_n) = t_2$.

\emph{Case $t_1 = \ffun(\Tag_i(w_1), w_2)$ with  $i \in \{1,\ldots,p\}$ and $\ffun \in
 \{\senc, \allowbreak\aenc, \sign\}$:} 
Let $\gamma \in \{\alpha,\beta\}$ such that $i \in \gamma$.
In such a case, we know that $\delta_\alpha(t_1) =
\ffun(\Tag_i(\delta_\gamma(w_1)), \delta_\gamma(w_2))$. But we know that
$\delta_\alpha(t_2) = \delta_\alpha(t_1) = \ffun(\Tag_i(\delta_\gamma(w_1)),
\delta_\gamma(w_2))$. Thus thanks to $t_2$ being in normal form and by
definition of $\delta_\alpha$, it implies that there exists $v_1$ and $v_2$ such
that $t_2 = \ffun(\Tag_i(v_1), v_2)$ and so $\delta_\alpha(t_2) =
\ffun(\Tag_i(\delta_\gamma(v_1)), \delta_\gamma(v_2))$. Thus, we have that
$\delta_\gamma(v_1) = \delta_\gamma(u_1)$ and $\delta_\gamma(v_2) =
\delta_\gamma(u_2)$. Moreover, $t_1$ and $t_2$ being in normal form and not
using names in $\Ec_\alpha \uplus \Ec_\beta$, so are $u_j$ and $v_j$ for $j\in\{1, 2\}$, so we can apply
inductive hypothesis and conclude that $v_1 = u_1$ and $v_2 = u_2$ and so $t_1 =
t_2$.

\emph{Case $t_1 = \h(\Tag_i(w_1))$ with $i \in \{1,\ldots, p\}$:} This case is analogous
to the previous one.

\emph{Case $\ffun \in \Sigmazero$ and $\racine(u_1) \neq \Tag_i$, $i = 1\ldots p$:}
By definition of $\delta_\alpha$,
we can deduce that $\delta_\alpha(t_1) = \ffun(\delta_\alpha(u_1), \ldots,
\delta_\alpha(u_n))$. Since $\delta_\alpha(t_1) = \delta_\alpha(t_2)$, we can
deduce that the top symbol of $t_2$ is also $\ffun$ and so there exists $v_1,
\ldots, v_n$ such that $t_2 = \ffun(v_1, \ldots, v_n)$. In the previous cases,
we showed that if $\ffun \in \{\senc, \aenc, \sign, \h\}$ and the top symbol of
$v_1$ is $\Tag_j$ for some $j \in \{1,\ldots,p\}$ then $\delta_\alpha(t_1) = \delta_\alpha(t_2)$
implies that the top symbol of $u_1$ is also $\Tag_{j}$.
Thus, thanks
to our hypothesis, we can deduce that either $\ffun \not\in \{\senc, \aenc,
\sign, \h\}$ or the top symbol of $v_1$ is different from $\Tag_j$ for
some $j \in \{1,\ldots,p\}$.
Hence by definition of $\delta_\alpha$, we can deduce that $\delta_\alpha(t_2) =
\ffun(\delta_\alpha(v_1), \ldots, \delta_\alpha(v_n))$ and so
$\delta_\alpha(v_j) = \delta_\alpha(u_j)$ for all $j \in \{1, \ldots,
n\}$. Moreover, $t_1$ and $t_2$ being in normal form and not using names in
$\Ec_\alpha \uplus \Ec_\beta$, implies that so are $u_j$ and $v_j$ for all $j\in\{1,\dots,
n\}$. We can thus apply our inductive hypothesis and conclude that $u_j = v_j$
for all $j \in \{1, \ldots, n\}$ and so $t_1 = t_2$. 
\end{proof}


\begin{lemma}
  \label{lem:eqdeltaAetB}
  Let $t_1$ and $t_2$ be ground terms in normal form such that 
$(\fn(t_1) \cup \fn(t_2)) \cap (\Ec_\alpha \uplus \Ec_\beta)
=\emptyset$. We have that:
\begin{center}
$\delta_\alpha(t_1) =
\delta_\beta(t_2)$ implies that $t_1 =
t_2$.
\end{center}
\end{lemma}

\begin{proof}
 We prove the result by induction on $|\delta_\alpha(t_1)|$.

 \smallskip

 \noindent\emph{Base case $|\delta_\alpha(t_1)| = 1$:} 
Since $\delta_\alpha(t_1)
 =\delta_\beta(t_2)$, $\Ec_\alpha \cap \Ec_\beta = \emptyset$, and
 $t_1,t_2$ do not use names in $\Ec_\alpha \uplus \Ec_\beta$, we
 necessarily have that $t_1 \not\in \dom(\rho^+_\alpha)$ and $t_2
 \not\in \dom(\rho^+_\beta)$. Hence, we have that $\delta_\alpha(t_1)
 = t_1$ and $\delta_\beta(t_2) = t_2$. This allows us to conclude.

 \medskip

 \noindent \emph{Inductive step $|\delta_\alpha(t_1)| > 1$:} In that case, we
 have that $\delta_\alpha(t_1) = \ffun(u_1, \ldots, u_n) =
 \delta_\beta(t_2)$. Assume that $\ffun \in \Sigma_i \cup \Sigma_{\Tag_i}$
for some $i \in \{1,\ldots, p\}$. Let $\gamma \in \{\alpha,\beta\}$
such that $i\in \gamma$.  By definition of $\delta_\alpha$ and $\delta_\beta$,
 we can deduce that $\racine(t_1) = \ffun = \racine(t_2)$. Furthermore, if we
 assume that $t_1 = \ffun(v_1, \ldots, v_n)$ and $t_2 = \ffun(w_1, \ldots,
 w_n)$, we would have $\delta_\gamma(v_j) = \delta_\gamma(w_j)$ for all $j \in
 \{1, \ldots, n\}$. By Lemma~\ref{lem:transfoV}, we deduce that $v_j = w_j$ for
 all $j \in \{1, \ldots, n\}$. Hence, we conclude that $t_1 = t_2$.  Assume now
 that $\ffun \in \Sigmazero$. According to the definition of $\delta_\alpha$
 and $\delta_\beta$, there exists $v_1, \ldots, v_n$ and $w_1, \ldots, w_n$
 such that $t_1 = \ffun(v_1, \ldots, v_n)$, $t_2 = \ffun(w_1, \ldots, w_n)$ and
$\delta_{\gamma_1}(v_j) = \delta_{\gamma_2}(w_j)$, for some $\gamma_1,
\gamma_2 \in \{\alpha,
 \beta\}$. Moreover, $t_1$ and $t_2$ being in normal form and not using names in
 $\Ec_\alpha \uplus \Ec_\beta$ implies that so are $v_j$ and $w_j$ for all $j\in\{1, \dots,
 n\}$. Now, either $\gamma_1 = \gamma_2$ and so by Lemma~\ref{lem:transfoV}, we have that
 $v_j = w_j$, else $\gamma_1 \neq \gamma_2$ but then by our inductive hypothesis, we also
 have $v_j = w_j$. Hence we conclude that $t_1 = t_2$.
\end{proof}


\begin{lemma}
  \label{lem:deltaandnormalform}
  Let $u$ be a ground  term in normal form such that $\fn(u) \cap (\Ec_\alpha\
 \uplus \Ec_\beta) = \emptyset$. Let $\gamma \in \{\alpha,\beta\}$.
We have that:
\begin{itemize} 
\item $\delta_\gamma(u)$ is in normal form; and
\item either
  $\racine(\delta_\gamma(u)) = \racine(u)$ or
  $\racine(\delta_\gamma(u)) = \bot$.
\item either
  $\racinebis(\delta_\gamma(u)) = \racinebis(u)$ or
  $\racinebis(\delta_\gamma(u)) = \bot$.
\end{itemize}
\end{lemma}

\begin{proof}
 We prove this result by induction on $|u|$ and we assume
 w.l.o.g. that $\gamma = \alpha$.

\smallskip{}

\noindent\emph{Base case $|u| =1$:} In such a case, we have that $u
\in \N$, and we   also
 have that $\delta_\alpha(u) \in \N$ and so $\delta_\alpha(u)$ is in normal
 form with the same root as $u$, namely $\bot$. Moreover, we have $\racinebis(\delta_\gamma(u)) = \bot$.

 \medskip

\noindent \emph{Inductive $|u| > 1$:}
Assume first that $u\mydownarrow \in \dom(\rho^+_\alpha)$ and $\racinebis(u) \not\in \alpha\cup \{0\}$. Hence
by definition of $\delta_\alpha$, we have that $\delta_\alpha(u) \in
\Ec_\alpha$.
Thus,
we trivially obtain that $\delta_\alpha(u)$ is in normal form, $\racine(\delta_\alpha(u)) = \bot$ and $\racinebis(\delta_\alpha(u)) = \bot$.

\medskip

Otherwise, we distinguish two cases:

\medskip{}

\noindent \underline{\emph{Case 1.}} 
We have that $u = C[u_1, \ldots, u_n]$  where $C$ is built
on $\Sigma_j \cup \Sigma_{\Tag_j}$ with $j \in \{1,\ldots, p\}$, 
$C$ is
different from a hole, $u_k$ are factors in normal form of $u$, $k =
1\ldots n$. Let $\varepsilon \in \{\alpha,\beta\}$ such that $j \in
\varepsilon$.
Hence, since $u \not\in \dom(\rho^+_\alpha)$, then by 
definition of $\delta_\alpha$, we deduce that
 $\delta_\alpha(u) = C[\delta_\varepsilon(u_1), \ldots,
 \delta_\varepsilon(u_n)]$. 
 Since $C$ is not a hole, thanks to our
 inductive hypothesis on $u_1, \ldots, u_n$, we have that $\delta_\varepsilon(u_1),
 \ldots, \delta_\varepsilon(u_n)$ are in normal form and
 $\delta_\varepsilon(u_1),
 \ldots, \delta_\varepsilon(u_n)$ are factors of $\delta_\alpha(u)$. Thus, since $u$
 is in normal form, we have that $C[u_1,\ldots,u_n]\mydownarrow =
 C[u_1,\ldots,u_n]$. By Lemmas~\ref{lem:transfoV} and~\ref{lem:change alien}, we
 deduce that 
\[
C[\delta_\varepsilon(u_1), \ldots,
 \delta_\varepsilon(u_n)]\mydownarrow =  C[\delta_\varepsilon(u_1), \ldots,
 \delta_\varepsilon(u_n)]\]
 \emph{i.e.} $\delta_\alpha(u)\mydownarrow = \delta_\alpha(u)$.

 Furthermore, we also have that
 $\racine(\delta_\alpha(u)) = \racine(u)$ and $\racinebis(\delta_\alpha(u)) = \racinebis(u)$.
 
 \medskip
 
 \noindent \underline{\emph{Case 2.}} We have that $u
 =\ffun(v_1,\ldots,v_n)$ for some $\ffun \in \Sigma_0$.
By definition
 of $\delta_\alpha$ there exists $\varepsilon \in \{\alpha,\beta\}$ such that
 $\delta_\alpha(u) = \ffun(\delta_\varepsilon(v_1), \allowbreak\ldots,
 \delta_\varepsilon(v_m))$. We do a case analysis on $\ffun$:

\smallskip{}

 \emph{Case $\ffun \in \{\senc, \aenc, \pk, \sign, \vk, \h, \langle\;
   \rangle\}$:} In this case, we have that
 $\delta_\alpha(u)\mydownarrow = \ffun(\delta_\varepsilon(v_1)\mydownarrow, \ldots,
 \delta_\varepsilon(v_m)\mydownarrow)$. Since by inductive hypothesis,
 $\delta_\varepsilon(v_k)$ is in normal form, for all $k \in \{1, \ldots, m\}$, we
 can  deduce that $\delta_\alpha(u)$ is also in normal form and
 $\racine(\delta_\alpha(u)) = \ffun = \racine(u)$. If $\ffun \in \{\pk,\vk, \langle\; \rangle\}$ then we trivially have that $\racinebis(\delta_\alpha(u)) = \racinebis(u)$. Let's focus on $\ffun \in \{\senc, \aenc, \sign\}$. If $\racinebis(u) \not\in \{0\}$ then it means that $\racine(v_1) = \Tag_i$ for some $i \in \varepsilon$. But by definition of $\delta_\varepsilon$, we would have that $\racine(\delta_\varepsilon(v_1)) = \Tag_i$. Hence $\racinebis(\delta_\alpha(u)) = \racinebis(u)$. Now if $\racinebis(u) \not\in \{0\}$, it means that $\racine(v_1) \not \{ \Tag_1\ldots \Tag_p\}$. But by inductive hypothesis, $\racine(\delta_\varepsilon(v_1)) = \bot$ or $\racine(\delta_\varepsilon(v_1)) = \racine(v_1)$ and so we can conclude that $\racinebis(\delta_\alpha(u)) \in \{0\}$.

\smallskip{}

 \emph{Case $\ffun = \sdec$:} Then $m=2$, and by definition of
 $\delta_\alpha$, 
we have that $\delta_\alpha(u) =
 \sdec(\delta_\alpha(v_1), \delta_\alpha(v_2))$. Thus, in such a case,
 we have that $\racine(\delta_\alpha(u)) = \ffun = \racine(u)$.  By
 inductive hypothesis, we have that $\delta_\alpha(v_1)$ and
 $\delta_\alpha(v_2)$ are both in normal form.
 Assume that $\sdec$ cannot be reduced,
 \emph{i.e.} $\delta_\alpha(u)\mydownarrow = \sdec(\delta_\alpha(v_1)\mydownarrow,
 \delta_\alpha(v_2)\mydownarrow) = \sdec(\delta_\alpha(v_1),
 \delta_\alpha(v_2))$. Thus the result holds.  Otherwise, if $\sdec$ can be
 reduced, there exist $w_1, w_2$ with  $\delta_\alpha(v_1)
 = \senc(w_1, w_2)$ and $\delta_\alpha(v_2) = w_2$. By definition of
 $\delta_\alpha$, there must exist $\varepsilon' \in \{\alpha, \beta\}$, and $w'_1, w'_2$ such that
 $\delta_\alpha(v_1) = \senc(\delta_{\varepsilon'}(w'_1), \delta_{\varepsilon'}(w'_2))$, $v_1 =
 \senc(w'_1, w'_2)$, $w_1 = \delta_{\varepsilon'}(w'_1)$ and $w_2 =
 \delta_{\varepsilon'}(w'_2)$. Thus, we have that $\delta_\alpha(v_2) =
 \delta{\varepsilon'}(w'_2)$. Thanks to Lemmas~\ref{lem:transfoV}
 and~\ref{lem:eqdeltaAetB}, we have that $v_2 = w'_2$. Hence, $u =
 \sdec(\senc(w'_1,w'_2), w'_2)$. But in such a case, we would have that $u$ is
 not in normal form which contradicts our hypothesis.
 
 At last, since $\racine(\delta_\alpha(u)) = \bot$ or $\racine(\delta_\alpha(u)) = \racine(u) = \sdec$ then we can deduce that $\racinebis(\delta_\alpha(u)) = \bot$ or $\racinebis(\delta_\alpha(u)) = \racinebis(u) = 0$.
\smallskip{}

The cases  where $\ffun = \checksign$ or $\ffun = \adec$ are
  analogous to the previous one. 
\end{proof}
\subsection{$\delta_\alpha$ and $\delta_\beta$ on tagged term}
\label{subsec: delta and tag}

Let $\sigma_0$ be a ground substitution. Similarly to the previous section, we consider a pair $(\rho_\alpha,\rho_\beta)$ as defined in Section~\ref{subsec:theo-main} and compatible with $\sigma_0$. We denote $(\rho^+_\alpha,\rho^+_\beta)$ the extension of $(\rho_\alpha,\rho_\beta)$ w.r.t. this substitution. We also denote by $\Ec_\alpha$ and $\Ec_\beta$ the respective image of $\rho_\beta$ and $\rho_\alpha$, and we assume that $\sigma_0$ does not use any name in $\Ec_\alpha$ and $\Ec_\beta$.

Thanks to compatibility, $\rho^+_\alpha$ (resp. $\rho^+_\beta$) is injective on $\dom(\rho^+_\alpha) \smallsetminus \{z^\beta_1,\ldots, z^\beta_l\}$ (resp. $\dom(\rho^+_\beta) \smallsetminus \{z^\alpha_1,\ldots, z^\alpha_k\}$). Moreover, we also have that for all $z \in \{z^\beta_1,\ldots, z^\beta_l\}$ (resp. $z \in \{z^\alpha_1,\ldots, z^\alpha_k\}$), either $\racinebis(z\sigma_0\mydownarrow) = \bot$ or $\racinebis(z\sigma_0\mydownarrow) \not\in \alpha \cup \{0\}$ (resp. $\racinebis(z\sigma_0\mydownarrow) \not\in \beta \cup \{0\}$).

\medskip

Let $i \in \{1,\ldots,p\}$. Let $u \in \T(\Sigma_i \cup \Sigma_{\Tag_i} \cup \Sigmazero, \N \cup \X)$. 
As defined in Section~\ref{sec:app-tagging}, $\TestTag{u}{i}$ is a
conjunction of elementary formulas (equalities between terms). Given a
substitution $\sigma$ such that $\fv(u) \subseteq \dom(\sigma)$, we say that~$\sigma$ satisfies $t_1 = t_2$, denoted $\sigma \vDash t_1 = t_2$, if $t_1\sigma\mydownarrow = t_2\sigma\mydownarrow$. 

At last, for all substitution~$\sigma$, for all $\gamma \in \{\alpha,\beta\}$, we denote by $\delta_\gamma(\sigma)$ the substitution such that $\dom(\sigma) = \dom(\delta_\gamma(\sigma))$ and for all $x \in \dom(\delta_\gamma(\sigma))$, $x\delta_\gamma(\sigma) = \delta_\gamma(x\sigma)$.

\medskip{}




\begin{lemma}
\label{lem:Testandequality}
Let $u \in \T(\Sigma_i \cup \Sigmazero, \N \cup
\X)$ for some $i \in \{1,\ldots,p\}$. Let $\gamma \in \{\alpha,
\beta\}$ such that $i \in \gamma$ and $\sigma_0$ be a ground
substitution such that $\fv(u)
\subseteq \dom(\sigma_0)$.
Moreover, assume that $u$ does not use
names in
$\Ec_\alpha \uplus \Ec_\beta$. We have that:
\begin{itemize}
\setlength{\itemsep}{0.5mm}
\item $\delta_\gamma(\TAG{u}{i}(\sigma_0\mydownarrow)) =
  \delta_\gamma(\TAG{u}{i})\delta_\gamma(\sigma_0\mydownarrow)$; and
\item If $\sigma_0 \vDash \TestTag{\TAG{u}{i}}{i}$ then
  $\delta_\gamma(\TAG{u}{i}(\sigma_0\mydownarrow))\mydownarrow =
  \delta_\gamma(\TAG{u}{i}\sigma_0\mydownarrow)$.
\end{itemize}
\end{lemma}

\begin{proof}
 Let  $\sigma$ be the substitution $\sigma_0\mydownarrow$. We prove the two results separately. First, we show by
 induction on $|u|$ that $\delta_\gamma(\TAG{u}{i}\sigma) =
 \delta_\gamma(\TAG{u}{i})\delta_\gamma(\sigma)$:

\medskip

\noindent\emph{Base case $|u| = 1$:} In this case, $u \in \N \cup
\X$. If $u \in \N$ then we have that $\TAG{u}{i} = u$ and so $\TAG{u}{i}\sigma =
u$ and $\delta_\gamma(u) \in \N$. Thus, we have that
$\delta_\gamma(\TAG{u}{i}\sigma) = \delta_\gamma(u) =
\delta_\gamma(u)\delta_\gamma(\sigma) =
\delta_\gamma(\TAG{u}{i})\delta_\gamma(\sigma)$. Otherwise, we have
that  $u \in \X$ and  $\TAG{u}{i} = u$. W.l.o.g., we assume that $\gamma = \alpha$. First,
if $u \not\in \{z^\beta_1, \ldots, z^\beta_l\}$, then we
have that $\delta_\alpha(u) = u$. Thus,
$\delta_\alpha(u)\delta_\alpha(\sigma) = u\delta_\alpha(\sigma)$. Since $u \in
\X$ and $\fv(u) \subseteq \dom(\sigma)$, we have that $u\delta_\alpha(\sigma) =
\delta_\alpha(u\sigma)$, thus $\delta_\alpha(\TAG{u}{i}\sigma) =
\delta_\alpha(u\sigma) = u\delta_\alpha(\sigma) =
\delta_\alpha(u)\delta_\alpha(\sigma) =
\delta_\alpha(\TAG{u}{i})\delta_\alpha(\sigma)$.
{Now, it remains the case where $u = z^\beta_j$ for some $j \in
\{1,\ldots, l\}$. In such a case, we have that:}
\begin{itemize}
\item $\delta_\alpha(\TAG{z^\beta_j}{i}\sigma) =
\delta_\alpha(z^\beta_j\sigma) = n^\beta_j$, and
\item $\delta_\alpha(\TAG{z^\beta_j}{i})\delta_\alpha(\sigma) =
\delta_\alpha(z^\beta_j)\delta_\alpha(\sigma)  =
n^\beta_j\delta_\alpha(\sigma) = n^\beta_j$. 
\end{itemize}
\medskip

\noindent\emph{Inductive step $|u| > 1|$, i.e. $u = \ffun(u_1,
\ldots, u_n)$.} We do a case analysis on $\ffun$.

\smallskip{}

\emph{Case $\ffun \in \Sigma_i$:} In such a case, $\TAG{u}{i} =
\ffun(\TAG{u_1}{i}, \ldots, \TAG{u_n}{i})$. By definition of
$\delta_\gamma$, $\delta_\gamma(\TAG{u}{i}\sigma) =
\ffun(\delta_\gamma(\TAG{u_1}{i}\sigma), \ldots,
\delta_\gamma(\TAG{u_n}{i}\sigma))$ and $\delta_\gamma(\TAG{u}{i}) =
\ffun(\delta_\gamma(\TAG{u_1}{i}), \ldots, \delta_\gamma(\TAG{u_n}{i}))$. By our
inductive hypothesis, we can deduce that for all $k \in \{1, \ldots, n\}$, we
have that $\delta_\gamma(\TAG{u_k}{i}\sigma) =
\delta_\gamma(\TAG{u_k}{i})\delta_\gamma(\sigma)$. Thus, we can deduce that
$\delta_\gamma(\TAG{u}{i}\sigma) = \ffun(\delta_\gamma(\TAG{u_1}{i}), \ldots,
\delta_\gamma(\TAG{u_n}{i}))\delta_\gamma(\sigma) =
\delta_\gamma(\TAG{u}{i})\delta_\gamma(\sigma)$.

\smallskip{}

\emph{Case $\ffun \in \{ \aenc, \sign, \senc\}$:} In this case $n=2$, and by
definition of $\TAG{u}{i}$, we have that $\TAG{u}{i} =
\ffun(\Tag_i(\TAG{u_1}{i}), \TAG{u_2}{i})$. Thus,
$\delta_\gamma(\TAG{u}{i}) = \ffun(\Tag_i(\delta_\gamma(\TAG{u_1}{i})),
\delta_\gamma(\TAG{u_2}{i}))$ and $\delta_\gamma(\TAG{u}{i}\sigma) =
\ffun(\Tag_i(\delta_\gamma(\TAG{u_1}{i}\sigma)),
\delta_\gamma(\TAG{u_2}{i}\sigma))$. But by our inductive hypothesis,
we have $\delta_\gamma(\TAG{u_k}{i}\sigma) =
\delta_\gamma(\TAG{u_k}{i})\delta_\gamma(\sigma)$ with $k \in \{1,
2\}$.
We conclude that
\[
\begin{array}{rcl}
\delta_\gamma(\TAG{u}{i}\sigma) &=&
\ffun(\Tag_i(\delta_\gamma(\TAG{u_1}{i})\delta_\gamma(\sigma)),
\delta_\gamma(\TAG{u_2}{i})\delta_\gamma(\sigma))\\& =&
\delta_\gamma(\TAG{u}{i})\delta_\gamma(\sigma).
\end{array}
\]

\smallskip{}

\emph{Case $\ffun = \h$:} This case is analogous to the previous one and can be
handled in a similar way.

\smallskip{}
\emph{Case $\ffun \in \{\sdec, \adec, \checksign\}$:} In this case $n=2$, and by
definition of $\TAG{u}{i}$, we have that $\TAG{u}{i} =
\unTag_i(\ffun(\TAG{u_1}{i},\TAG{u_2}{i}))$. Thus,
$\delta_\gamma(\TAG{u}{i}) = \unTag_i(\ffun(\delta_\gamma(\TAG{u_1}{i}),
\delta_\gamma(\TAG{u_2}{i})))$ and $\delta_\gamma(\TAG{u}{i}\sigma) =
\unTag_i(\ffun(\delta_\gamma(\TAG{u_1}{i}\sigma),
\delta_\gamma(\TAG{u_2}{i}\sigma)))$. Relying on our inductive hypothesis,
we deduce that 
\begin{center}
$\delta_\gamma(\TAG{u_k}{i}\sigma) =
\delta_\gamma(\TAG{u_k}{i})\delta_\gamma(\sigma)$ with $k \in
\{1,2\}$.
\end{center}
We conclude that \[
\begin{array}{rcl}
\delta_\gamma(\TAG{u}{i}\sigma) &=&
\unTag_i(\ffun(\delta_\gamma(\TAG{u_1}{i}),
\delta_\gamma(\TAG{u_2}{i})))\delta_\gamma(\sigma)\\ &=&
\delta_\gamma(\TAG{u}{i})\delta_\gamma(\sigma).
\end{array}
\]

\smallskip{}

Otherwise, by definition of $\TAG{u}{i}$, we have
that:
\begin{itemize}
\item  $\TAG{u}{i} = \ffun(\TAG{u_1}{i},\ldots, \TAG{u_n}{i})$, and
\item 
$\delta_\gamma(\TAG{u}{i}) = \ffun(\delta_\gamma(\TAG{u_1}{i}), \ldots,
\delta_\gamma(\TAG{u_n}{i}))$. 
\end{itemize}
Thus, this case is similar to the case $\ffun \in
\Sigma_i$. Hence the result holds.

\medskip{}

We now prove the second property, \emph{i.e.} if $\sigma \vDash
\TestTag{\TAG{u}{i}}{i}$, then $\delta_\gamma(\TAG{u}{i}\sigma)\mydownarrow =
\delta_\gamma(\TAG{u}{i}\sigma\mydownarrow)$. We prove the result
by induction on $|u|$:

\medskip

\noindent\emph{Base case $|u| = 1$:} In this case, $u \in \N \cup
\X$. In both cases, we have that $\TAG{u}{i} = u$ and $\TestTag{u}{i} =
\true$. If $u \in \N$, we know that $\delta_\gamma(u) \in \N$ and so
$\delta_\gamma(u)\mydownarrow = \delta_\gamma(u)$. We also have that
$u\sigma\mydownarrow = u\sigma = u$. This allows us to conclude that
\[
\delta_\gamma(u\sigma)\mydownarrow = \delta_\gamma(u)\mydownarrow =
\delta_\gamma(u) = \delta_\gamma(u\sigma\mydownarrow).
\]
Otherwise, we have $u \in \X$. Since $\sigma$ is is normal form, we deduce that
$u\sigma\mydownarrow = u\sigma$. By Lemma~\ref{lem:deltaandnormalform}, we also know that
$\delta_\gamma(u\sigma\mydownarrow)\mydownarrow =
\delta_\gamma(u\sigma\mydownarrow)$. Thus, we conclude that
\[\delta_\gamma(u\sigma\mydownarrow) =
\delta_\gamma(u\sigma\mydownarrow)\mydownarrow =
\delta_\gamma(u\sigma)\mydownarrow.
\]

\noindent\emph{Inductive step $|u| > 1$, i.e. $u =
\ffun(u_1, \ldots, u_n)$.} We do a case analysis on $\ffun$.

\smallskip{}

\emph{Case $\ffun \in \Sigma_i$:} We have that $\TAG{u}{i} =
\ffun(\TAG{u_1}{i}, \ldots, \TAG{u_n}{i})$. Hence, we have that
$\delta_\gamma(\TAG{u}{i}\sigma) = \ffun(\delta_\gamma(\TAG{u_1}{i}\sigma),
\ldots, \delta_\gamma(\TAG{u_n}{i}\sigma))$ and so
$\delta_\gamma(\TAG{u}{i}\sigma)\mydownarrow =
\ffun(\delta_\gamma(\TAG{u_1}{i}\sigma)\mydownarrow, \ldots,
\allowbreak \delta_\gamma(\TAG{u_n}{i}\sigma)\mydownarrow)\mydownarrow$. We have
that 
$\TestTag{\TAG{u}{i}}{i} = \bigwedge_{j=1}^n \TestTag{\TAG{u_j}{i}}{i}$ which
means that $\sigma \vDash \TestTag{\TAG{u_j}{i}}{i}$ for each $j \in
\{1,\ldots,n\}$. By applying  our
inductive hypothesis on $u_1, \ldots, u_n$, we deduce that
\[
\begin{array}{rcl}
\delta_\gamma(\TAG{u}{i}\sigma)\mydownarrow &=&
\ffun(\delta_\gamma(\TAG{u_1}{i}\sigma\mydownarrow), \ldots,
\delta_\gamma(\TAG{u_n}{i}\sigma\mydownarrow))\mydownarrow\\ &=&
\delta_\gamma(\ffun(\TAG{u_1}{i}\sigma\mydownarrow, \ldots,
\TAG{u_n}{i}\sigma\mydownarrow))\mydownarrow
\end{array}
\]

Let  $t = \ffun(\TAG{u_1}{i}\sigma\mydownarrow, \ldots,
\TAG{u_n}{i}\sigma\mydownarrow)$. We can assume that there exists a
context $C$ built on $\Sigma_i$ such that $t = C[t_1, \ldots, t_m]$
with $\fct(t) = \{t_1, \ldots, t_m\}$ and $t_1, \ldots, t_m$ are in
normal form. Thus, by Lemma~\ref{lem:app_CRicalp05}, there exists a
context $D$ (possibly a hole) such that $t\mydownarrow = D[t_{j_1},
\ldots, t_{j_k}]$ with $j_1, \ldots, j_k \in \{0, \ldots, m\}$ and
$t_0 = n_{min}$. Since $t_1, \ldots, t_m$ are in normal form
and thanks to Lemma~\ref{lem:deltaandnormalform}, we know that
for all $k \in \{0, \ldots, m\}$, $\delta_\gamma(t_k)$ is also in
normal form and its root is not in $\Sigma_i$. Hence, we can apply
Lemma~\ref{lem:change alien} such that $C[\delta_\gamma(t_1), \ldots,
\delta_\gamma(t_m)]\mydownarrow = D[\delta_\gamma(t_{j_1}), \ldots,
\delta_\gamma(t_{j_k})]$. But since $C$ and $D$ are both built upon
$\Sigma_i$, we have that:
\begin{itemize}
\item  $C[\delta_\gamma(t_1), \ldots,
\delta_\gamma(t_m)]\mydownarrow = \delta_\gamma(C[t_1, \ldots,
t_m])\mydownarrow$, and 
\item $D[\delta_\gamma(t_{j_1}), \ldots,
\delta_\gamma(t_{j_k})] = \delta_\gamma(D[t_{j_1}, \ldots,
t_{j_k}])$.
\end{itemize}
 Hence, we can deduce that $\delta_\gamma(t)\mydownarrow =
\delta_\gamma(t\mydownarrow)$. But we already know that $t\mydownarrow
= \TAG{u}{i}\sigma\mydownarrow$ and $\delta_\gamma(t)\mydownarrow =
\delta_\gamma(\TAG{u}{i}\sigma)\mydownarrow$. Thus, we can conclude
that $\delta_\gamma(\TAG{u}{i}\sigma)\mydownarrow =
\delta_\gamma(\TAG{u}{i}\sigma\mydownarrow)$.

\smallskip{}

\emph{Case $\ffun \in \{\senc, \aenc, \sign\}$:} In such a case, we
have that:
\begin{itemize}
\item $\TAG{u}{i} = \ffun(\Tag_i(\TAG{u_1}{i}), \TAG{u_2}{i})$, and
\item $\TestTag{\TAG{u}{i}}{i} = \TestTag{\TAG{u_1}{i}}{i} \wedge
\TestTag{\TAG{u_2}{i}}{i}$. 
\end{itemize}
Hence, we have that
$\TAG{u}{i}\sigma\mydownarrow =
\ffun(\Tag_i(\TAG{u_1}{i}\sigma\mydownarrow),
\TAG{u_2}{i}\sigma\mydownarrow)$, and also
$\delta_\gamma(\TAG{u}{i}\sigma)\mydownarrow =
\ffun(\Tag_i(\delta_\gamma(\TAG{u_1}{i}\sigma)\mydownarrow),
\delta_\gamma(\TAG{u_2}{i}\sigma)\mydownarrow)$. By our inductive
hypothesis on $u_1$ and $u_2$, we have that:
\begin{center}
$\delta_\gamma(\TAG{u_k}{i}\sigma)\mydownarrow
=\delta_\gamma(\TAG{u_k}{i}\sigma\mydownarrow)$ with $k \in \{1,
2\}$.
\end{center}
 Hence, we can deduce that
\[
\begin{array}{rcl}
\delta_\gamma(\TAG{u}{i}\sigma)\mydownarrow &=&
\ffun(\Tag_i(\delta_\gamma(\TAG{u_1}{i}\sigma\mydownarrow)),
\delta_\gamma(\TAG{u_2}{i}\sigma\mydownarrow)) \\ &=&
\delta_\gamma(\ffun(\Tag_i(\TAG{u_1}{i}\sigma\mydownarrow),
\TAG{u_2}{i}\sigma\mydownarrow)) \\
&=& \delta_\gamma(\TAG{u}{i}\sigma\mydownarrow).
\end{array}
\]

\smallskip{}

\emph{Case $\ffun = \h$:} This case is analogous to the previous one and can be
handled in a similar way.

\smallskip{}

\emph{Case $\ffun \in \{ \pk, \vk, \langle\;\rangle\}$:} In such a case, we have
that:
\begin{itemize}
\item $\TAG{u}{i} = \ffun(\TAG{u_1}{i}, \ldots, \TAG{u_n}{i})$ with
$n\in\{1,2\}$, and 
\item $\TestTag{\TAG{u}{i}}{i} = \wedge_{j=1}^{n}
\TestTag{\TAG{u_j}{i}}{i}$. 
\end{itemize}
We have that
$\TAG{u}{i}\sigma\mydownarrow = \ffun(\TAG{u_1}{i}\sigma\mydownarrow, \ldots,
\TAG{u_n}{i}\sigma\mydownarrow)$. Thus, this case is similar to the $\senc$ case
and can be handled similarly.

\smallskip{}

\emph{Case $\ffun \in \{\sdec, \adec, \checksign\}$:} In such a case, we have
that:
\begin{itemize}
\item  $\TAG{u}{i} = \unTag_i(\ffun(\TAG{u_1}{i}, \TAG{u_2}{i}))$, and
\item $\TestTag{\TAG{u}{i}}{i}$ is the following formula: 
\[
\begin{array}{l}
(\Tag_i(\unTag_i(\ffun(\TAG{u_1}{i}, \TAG{u_2}{i}))) =
\ffun(\TAG{u_1}{i}, \TAG{u_2}{i})) \\
\wedge \; \TestTag{\TAG{u_1}{i}}{i} \; \wedge \;
\TestTag{\TAG{u_2}{i}}{i}\\
\end{array}
\]
\end{itemize}
By hypothesis, we have that $\sigma \vDash
\TestTag{\TAG{u}{i}}{i}$, thus 
$\Tag_i(\unTag_i(\ffun(\TAG{u_1}{i},\TAG{u_2}{i})))\sigma\mydownarrow =
\ffun(\TAG{u_1}{i},\TAG{u_2}{i})\sigma\mydownarrow$. 
Hence, we deduce that the
root function symbol $\ffun$ can be reduced and the root of the plaintext is
$\Tag_i$. More formally, there exist $v_1, v_2$ such that:
\begin{itemize}
\item $\ffun = \sdec$: $\TAG{u_1}{i}\sigma\mydownarrow =
 \senc(\Tag_i(v_1),v_2)$, $\TAG{u_2}{i}\sigma\mydownarrow = v_2$ and
 $\TAG{u}{i}\sigma\mydownarrow = v_1$. This implies that:
 \[
\delta_\gamma(\TAG{u_1}{i}\sigma\mydownarrow) =
 \senc(\Tag_i(\delta_\gamma(v_1)), \delta_\gamma(v_2)).
\] 
Thus, we can
 deduce that:
 \[
\begin{array}{r@{\,}c@{\,}l}
\unTag_i(\sdec(\delta_\gamma(\TAG{u_1}{i}\sigma\mydownarrow),
 \delta_\gamma(\TAG{u_2}{i}\sigma\mydownarrow)))\mydownarrow &=&
 \delta_\gamma(v_1) \\ &=& \delta_\gamma(\TAG{u}{i}\sigma\mydownarrow)
\end{array}
\]
\item $\ffun = \adec$: $\TAG{u_1}{i}\sigma\mydownarrow =
 \aenc(\Tag_i(v_1),\pk(v_2))$, $\TAG{u_2}{i}\sigma\mydownarrow =
 v_2$, and  $\TAG{u}{i}\sigma\mydownarrow = v_1$.
\item $\ffun = \checksign$: $\TAG{u_1}{i}\sigma\mydownarrow =
 \sign(\Tag_i(v_1),v_2)$, $\TAG{u_2}{i}\sigma\mydownarrow =
 \vk(v_2)$,  and $\TAG{u}{i}\sigma\mydownarrow = v_1$.
\end{itemize}
In each case, we have that:
\[
\unTag_i(\ffun(\delta_\gamma(\TAG{u_1}{i}\sigma\mydownarrow),
\delta_\gamma(\TAG{u_2}{i}\sigma\mydownarrow)))\mydownarrow =
\delta_\gamma(\TAG{u}{i}\sigma\mydownarrow).
\]
 By inductive
hypothesis, we have $\delta_\gamma(\TAG{u_k}{i}\sigma\mydownarrow)
= \delta_\gamma(\TAG{u_k}{i}\sigma)\mydownarrow$ with $k \in
\{1,2\}$. We also have that:
\[
\delta_\gamma(\TAG{u}{i}\sigma)\mydownarrow =
\unTag_i(\ffun(\delta_\gamma(\TAG{u_1}{i}\sigma)\mydownarrow,
\delta_\gamma(\TAG{u_2}{i}\sigma)\mydownarrow))\mydownarrow.
\] 
This allows us to conclude that
\[
\begin{array}{rcl}
\delta_\gamma(\TAG{u}{i}\sigma)\mydownarrow &=&
\unTag_i(\ffun(\delta_\gamma(\TAG{u_1}{i}\sigma\mydownarrow),
\delta_\gamma(\TAG{u_2}{i}\sigma\mydownarrow)))\mydownarrow \\ &=&
\delta_\gamma(\TAG{u}{i}\sigma\mydownarrow).
\end{array}
\]

\smallskip{}

\emph{Case $\ffun = \proj_j$, $j = 1,2$:} In such a case, we have that
$n=1$, and $\TAG{u}{i} =
\ffun(\TAG{u_1}{i})$. Since $\sigma \vDash \TestTag{\TAG{u}{i}}{i}$, we have
that there exist $v_1, v_2$ such that $\TAG{u_1}{i}\sigma\mydownarrow = \langle
v_1, v_2 \rangle$ and so $\delta_\gamma(\TAG{u}{i}\sigma\mydownarrow) =
\delta_\gamma(v_j)$. But by inductive hypothesis, we have that
$\delta_\gamma(\TAG{u_1}{i}\sigma)\mydownarrow =
\delta_\gamma(\TAG{u_1}{i}\sigma\mydownarrow) = \langle \delta_\gamma(v_1),
\delta_\gamma(v_2) \rangle$. Hence, $\delta_\gamma(\TAG{u}{i}\sigma)\mydownarrow
= \ffun(\delta_\gamma(\TAG{u_1}{i}\sigma))\mydownarrow =
\ffun(\delta_\gamma(\TAG{u_1}{i}\sigma)\mydownarrow)\mydownarrow =
\delta_\gamma(v_j)\mydownarrow$. We have shown that $\delta_\gamma(v_j) =
\delta_\gamma(\TAG{u}{i}\sigma\mydownarrow)$, thus by
Lemma~\ref{lem:deltaandnormalform}, 
$\delta_\gamma(v_j)$ is in normal form and which allows us to
conclude.
\end{proof}

\medskip{}


\begin{corollary}
  \label{cor:equivalenceeq}
  Let $u,v \in \T(\Sigma_i\cup \Sigmazero, \N \cup \X)$ for some $i \in \{1,\ldots,p\}$. 
  Let $\gamma \in \{\alpha,\beta\}$ such that $i \in \gamma$.
  Assume that $\fv(u) \cup \fv(v) \subseteq \dom(\sigma_0\mydownarrow)$, and $\sigma_0 \vDash \TestTag{\TAG{u}{i}}{i} \wedge \TestTag{\TAG{v}{i}}{i}$.
  Moreover, assume that $u$, $v$ do not use names in $\Ec_\alpha \uplus \Ec_\beta$. 
  \[
  \TAG{u}{i}\sigma_0 \mydownarrow = \TAG{v}{i}\sigma_0\mydownarrow
  \Leftrightarrow
  \delta_\gamma(\TAG{u}{i})\delta_\gamma(\sigma_0\mydownarrow)\mydownarrow =
  \delta_\gamma(\TAG{v}{i})\delta_\gamma(\sigma_0\mydownarrow)\mydownarrow.
  \]
\end{corollary}


\begin{proof}
 Thanks to Lemma~\ref{lem:transfoV}, we have that
\[
\TAG{u}{i}\sigma_0\mydownarrow =
 \TAG{v}{i}\sigma_0\mydownarrow \Leftrightarrow
\delta_\gamma(\TAG{u}{i}\sigma_0\mydownarrow) =
 \delta_\gamma(\TAG{v}{i}\sigma_0\mydownarrow).
\]
 Thanks to
 Lemma~\ref{lem:Testandequality}, we have that:
\begin{itemize}
\item $\delta_\gamma(\TAG{u}{i}\sigma_0\mydownarrow) =
 \delta_\gamma(\TAG{u}{i}(\sigma_0\mydownarrow))\mydownarrow =
 \delta_\gamma(\TAG{u}{i})\delta_\gamma(\sigma_0\mydownarrow)\mydownarrow$, and
\item  $\delta_\gamma(\TAG{v}{i}\sigma_0\mydownarrow) =
 \delta_\gamma(\TAG{v}{i}(\sigma_0\mydownarrow))\mydownarrow =
 \delta_\gamma(\TAG{v}{i})\delta_\gamma(\sigma_0\mydownarrow)\mydownarrow$. 
\end{itemize}
This allows us to conclude.
\end{proof}

\medskip{}


\begin{lemma}
  \label{lem:TestTagequivalentdelta}
Let $u\in \T(\Sigma_i \cup \Sigmazero, \N
  \cup \X)$ for some $i\in\{1, \ldots, p\}$. Let $\gamma \in
  \{\alpha,\beta\}$ such that $i \in \gamma$. Assume that $\fv(u)
  \subseteq \dom(\sigma_0)$.
Moreover, assume that $u$ does not use names in $\Ec_\alpha
\cup \Ec_\beta$. We have that :
  \[
  \sigma_0\mydownarrow \vDash \TestTag{\TAG{u}{i}}{i} \; \Leftrightarrow \; 
  \delta_\gamma(\sigma_0\mydownarrow) \vDash \TestTag{\delta_\gamma(\TAG{u}{i})}{i}
  \]
\end{lemma}
\begin{proof}
To simplify the proof, we denote by $\sigma$ the substitution $\sigma_0\mydownarrow$. We prove this result by induction on $|u|$ :

 \medskip

 \noindent \emph{Base case $|u| = 1$:} In this case, we have that
 $u\in\N\cup\X$, and thus $\TAG{u}{i}, \delta_\gamma(\TAG{u}{i})\in
 \N\cup\X$. In such a case, we have that
$\TestTag{\TAG{u}{i}}{i} = \true$
 and $\TestTag{\delta_\gamma(\TAG{u}{i})}{i} = \true$. Hence, the
 result trivially holds.

\medskip

\noindent\emph{Inductive step $|u| > 1$, i.e. $u =
\ffun(u_1, \ldots, u_n)$.} We do a case analysis on $\ffun$:

\smallskip{}

\emph{Case $\ffun \in \Sigma_i \cup \{ \pk, \vk, \langle\;\rangle\}$:}
In this case, we have that $\TAG{u}{i} = \ffun(\TAG{u_1}{i}, \ldots,
\TAG{u_n}{i})$ and $\delta_\gamma(\TAG{u}{i}) =
\ffun(\delta_\gamma(\TAG{u_1}{i}), \ldots,
\delta_\gamma(\TAG{u_n}{i}))$. Thus, we deduce that
$\TestTag{\TAG{u}{i}}{i} = \bigwedge_{j=1}^n
\TestTag{\TAG{u_j}{i}}{i}$ and $\TestTag{\delta_\gamma(\TAG{u}{i})}{i}
= \bigwedge_{j=1}^n \TestTag{\delta_\gamma(\TAG{u_j}{i})}{i}$. By
inductive hypothesis on $u_1, \ldots, u_n$, the result holds.

\smallskip{}

\emph{Case $\ffun \in \{ \senc, \aenc, \sign \}$:} In this case, we
have that:
\begin{itemize}
\item $\TAG{u}{i} = \ffun(\Tag_i(\TAG{u_1}{i}), \TAG{u_2}{i})$, and
\item $\delta_\gamma(\TAG{u}{i}) =
\ffun(\Tag_i(\delta_\gamma(\TAG{u_1}{i})),
\delta_\gamma(\TAG{u_2}{i}))$. 
\end{itemize}
Thus, we deduce that
$\TestTag{\TAG{u}{i}}{i} = \TestTag{\TAG{u_1}{i}}{i} \wedge
\TestTag{\TAG{u_2}{i}}{i}$ and $\TestTag{\delta_\gamma(\TAG{u}{i})}{i}
= \TestTag{\delta_\gamma(\TAG{u_1}{i})}{i} \wedge
\TestTag{\delta_\gamma(\TAG{u_2}{i})}{i}$. By inductive hypothesis on
$u_1, u_2$, the result holds.

\smallskip{}

\emph{Case $\ffun = \h$:} This case is analogous to de previous one and can be
handled in a similar way.

\smallskip{}

\emph{Case $\ffun \in \{ \sdec, \adec, \checksign\}$:} In this case,
we have that:
\begin{itemize}
\item $\TAG{u}{i} = \unTag_i(\ffun(\TAG{u_1}{i},
\TAG{u_2}{i}))$, and 
\item $\delta_\gamma(\TAG{u}{i}) =
\unTag_i(\ffun(\delta_\gamma(\TAG{u_1}{i}),
\delta_\gamma(\TAG{u_2}{i})))$. 
\end{itemize}
Thus, we deduce that $\TestTag{\TAG{u}{i}}{i}$ is the following formula:
\[
\begin{array}{l}
\TestTag{\TAG{u_1}{i}}{i} \wedge \TestTag{\TAG{u_2}{i}}{i} \\
\wedge  \Tag_i(\unTag_i(\ffun(\TAG{u_1}{i}, \TAG{u_2}{i}))) = \ffun(\TAG{u_1}{i}, \TAG{u_2}{i}) \\
\end{array}
\]
and $\TestTag{\delta_\gamma(\TAG{u}{i})}{i}$ is the following formula:
\[
\begin{array}{@{}l@{}}
 \TestTag{\delta_\gamma(\TAG{u_1}{i})}{i} \wedge \TestTag{\delta_\gamma(\TAG{u_2}{i})}{i}\, \wedge \\ 
\Tag_i(\unTag_i(\ffun(\delta_\gamma(\TAG{u_1}{i}), \delta_\gamma(\TAG{u_2}{i})))) =  \ffun(\delta_\gamma(\TAG{u_1}{i}), \delta_\gamma(\TAG{u_2}{i}))\\
\end{array}
\]
Whether we assume that $\sigma \vDash \TestTag{\TAG{u}{i}}{i}$ or
$\delta_\gamma(\sigma) \vDash \TestTag{\delta_\gamma(\TAG{u}{i})}{i}$,
we have by inductive hypothesis that $\sigma \vDash
\TestTag{\TAG{u_k}{i}}{i}$ with $k \in \{1, 2\}$. Thus by
Lemma~\ref{lem:Testandequality}, it implies that
$\delta_\gamma(\TAG{u_k}{i}\sigma\mydownarrow) =
\delta_\gamma(\TAG{u_k}{i})\delta_\gamma(\sigma)\mydownarrow$ with $k
\in \{1, 2\}$. We do a case analysis on $\ffun$. We detail below the
case where $\ffun = \sdec$. The cases where $\ffun = \adec$, and
$\ffun = \checksign$ can be done in a similar way.

\smallskip{}
In such a case ($\ffun = \sdec$), we have that 
\[\sigma \vDash
 \Tag_i(\unTag_i(\ffun(\TAG{u_1}{i}, \TAG{u_2}{i}))) =
 \ffun(\TAG{u_1}{i}, \TAG{u_2}{i})
\]
is equivalent to there exists
 $v_1, v_2$ s.t. $\TAG{u_2}{i}\sigma\mydownarrow =
 v_2$ and $\TAG{u_1}{i}\sigma\mydownarrow =
 \senc(\Tag_i(v_1),v_2)$. But by Lemma~\ref{lem:transfoV}, it is equivalent to
 $\delta_\gamma(\TAG{u_1}{i}\sigma\mydownarrow) =
 \senc(\Tag_i(\delta_\gamma(v_1)), \delta_\gamma(v_2))$ and
 $\delta_\gamma(\TAG{u_2}{i}\sigma\mydownarrow) =
 \delta_\gamma(v_2)$. Thus, it is equivalent to:
\begin{itemize}
\item $\delta_\gamma(\TAG{u_1}{i})\delta_\gamma(\sigma)\mydownarrow =
 \senc(\Tag_i(\delta_\gamma(v_1)), \delta_\gamma(v_2))$, and
\item $\delta_\gamma(\TAG{u_2}{i})\delta_\gamma(\sigma)\mydownarrow =
 \delta_\gamma(v_2)$.
\end{itemize}
 Hence it is equivalent to
 \[
\delta_\gamma(\sigma) \vDash
\left(
\begin{array}{l}
\Tag_i(\unTag_i(\ffun(\delta_\gamma(\TAG{u_1}{i}),\delta_\gamma(\TAG{u_2}{i})))) \\
=  \ffun(\delta_\gamma(\TAG{u_1}{i}),\delta_\gamma(\TAG{u_2}{i}))
\end{array}
\right)
\]

\smallskip{}

\emph{Case $\ffun \in \{ \proj_1, \proj_2\}$:} In such a case, we have
that $\TAG{u}{i} = \ffun(\TAG{u_1}{i})$ and $\delta_\gamma(\TAG{u}{i}) =
\ffun(\delta_\gamma(\TAG{u_1}{i}))$. Thus, we deduce that $\TestTag{\TAG{u}{i}}{i}$ is the following formula:
\[
\begin{array}{l}
\TestTag{\TAG{u_1}{i}}{i} \wedge \langle\proj_1(\TAG{u_1}{i}), \proj_2(\TAG{u_1}{i})\rangle =
\TAG{u_1}{i}\\
\end{array}
\]
and $\TestTag{\delta_\gamma(\TAG{u}{i})}{i}$ is the following formula:
\[
\begin{array}{l}
 \TestTag{\delta_\gamma(\TAG{u_1}{i})}{i}\, \wedge \\
 \langle
 \proj_1(\delta_\gamma(\TAG{u_1}{i})),
 \proj_2(\delta_\gamma(\TAG{u_1}{i}))\rangle =
 \delta_\gamma(\TAG{u_1}{i})
\end{array}
\]
Whether we assume that $\sigma \vDash \TestTag{\TAG{u}{i}}{i}$ or
$\delta_\gamma(\sigma) \vDash \TestTag{\delta_\gamma(\TAG{u}{i})}{i}$,
we have by inductive hypothesis that $\sigma \vDash
\TestTag{\TAG{u_1}{i}}{i}$. Thus by Lemma~\ref{lem:Testandequality},
it implies that $\delta_\gamma(\TAG{u_1}{i}\sigma\mydownarrow) =
\delta_\gamma(\TAG{u_1}{i})\delta_\gamma(\sigma)\mydownarrow$.

Actually $\sigma \vDash \langle\proj_1(\TAG{u_1}{i}),
\proj_2(\TAG{u_1}{i})\rangle = \TAG{u_1}{i}$ is equivalent to there
exist
$v_1, v_2$ such that $\TAG{u_1}{i}\sigma\mydownarrow = \langle v_1, v_2
\rangle$, which is, thanks to Lemma~\ref{lem:transfoV}, equivalent to
$\delta_\gamma(\TAG{u_1}{i}\sigma\mydownarrow) = \langle
\delta_\gamma(v_1), \delta_\gamma(v_2) \rangle$.

We have shown that this is equivalent to
\[
\delta_\gamma(\TAG{u_1}{i})\delta_\gamma(\sigma)\mydownarrow = \langle
\delta_\gamma(v_1), \delta_\gamma(v_2) \rangle
\]
Thus, we 
conclude that $\sigma \vDash \langle\proj_1(\TAG{u_1}{i}),
\proj_2(\TAG{u_1}{i})\rangle = \TAG{u_1}{i}$ is equivalent
$\delta_\gamma(\sigma) \vDash \langle
\proj_1(\delta_\gamma(\TAG{u_1}{i})),
\proj_2(\delta_\gamma(\TAG{u_1}{i}))\rangle =
\delta_\gamma(\TAG{u_1}{i})$.
\end{proof}
\medskip

For a term $u$ that does not contain any tag, we defined a way to construct a term that is
properly tagged (\emph{i.e.} $\TAG{u}{i}$). Hence, for a term properly
tagged, we would never have $\senc(n, k)$ where $n$ and $k$ are both
nonces, for example. Instead, we would have
$\senc(\Tag_i(n),k)$. However, even if we can force the processes to
properly tag their terms, we do not have any control on what the
intruder can build. Typically, if the intruder is able to deduce $n$
and $k$, he is allowed to send to a process the term
$\senc(n,k)$. Thus, we want to define the notion of \emph{flawed
  tagged term}.

\medskip


\begin{definition}
Let $u$ be a ground term in normal form. Consider $\gamma $ and $\gamma'$ such that $\{ \gamma, \gamma'\} = \{\alpha,\beta\}$. We define the flawed subterms of $u$ w.r.t. $\gamma$, denoted $\FlawedColor{\gamma}{u}$, as follows:
\[
\FlawedColor{\gamma}{u} \stackrel{\defi}{=} \left\{v\in st(u) \left| \begin{array}{l}\racinebis(v) \in \{0\} \cup \gamma' \text{ and } \\\racine(v) \not\in \{\pk, \vk, \langle\, \rangle\} \end{array}\right.\right\}
\]
We define the flawed subterms of $u$, denoted $\Flawed{u}$, as the set $\Flawed{u} = \FlawedColor{\alpha}{u} \cap \FlawedColor{\beta}{u}$
\end{definition}

\medskip


\begin{lemma}
\label{lem: flawed color and tag term}
Let $u \in \T(\Sigma_i \cup \Sigmazero, \N \cup
\X)$ for some $i \in \{1,\ldots,p\}$. Let $\gamma \in
\{\alpha,\beta\}$ such that $i \in \gamma$. Let $\gamma'$ such that $\gamma' \in \{\alpha,\beta\} \smallsetminus \gamma$. Let $\sigma$ be a ground
substitution in normal form such that
$\fv(u) \subseteq \dom(\sigma)$.

If $\sigma \vDash
\TestTag{\TAG{u}{i}}{i}$ then for all $t \in
\FlawedColor{\gamma}{\TAG{u}{i}\sigma\mydownarrow}$, there exists $x \in
\fv(\TAG{u}{i})$ such that $t \in \FlawedColor{\gamma}{x\sigma}$.
\end{lemma}

\smallskip{}

\begin{proof}
We prove the result by induction on $|u|$.

\medskip

\noindent\emph{Base case $|u| = 1$:} In this case, we have that $u
\in \X \cup \N$ and so $\TAG{u}{i} = u$. If $u \in \N$, then $u\sigma$
and $\TAG{u}{i}\sigma\mydownarrow$ are both in~$\N$, which means that $\FlawedColor{\gamma}{\TAG{u}{i}\sigma\mydownarrow} = \emptyset$. Thus, the result
holds. Otherwise, we have that $u \in \X$ and so $\TAG{u}{i} = u \in \dom(\sigma)$ which
means that the result trivially holds.

\medskip

\noindent\emph{Inductive step $|u| > 1$, i.e. $u = \ffun(u_1,
\ldots, u_n)$.} We do a case analysis on $\ffun$.

\smallskip{}

\emph{Case $\ffun \in \Sigma_i$:} In this case, $\TAG{u}{i} =
\ffun(\TAG{u_1}{i}, \ldots, \TAG{u_n}{i})$ and $\TAG{u}{i}\sigma\mydownarrow =
\ffun(\TAG{u_1}{i}\sigma\mydownarrow, \ldots,\allowbreak
\TAG{u_n}{i}\sigma\mydownarrow)\mydownarrow$. By definition, we know that for
all $t \in \FlawedColor{\gamma}{\TAG{u}{i}\sigma\mydownarrow}$, $\racine(t) \not\in \Sigma_\gamma$. Thus, thanks to Lemma~\ref{lem:app_CRicalp05}, 
for all $t \in \FlawedColor{\gamma}{\TAG{u}{i}\sigma\mydownarrow}$, there exists $k \in
\{1, \ldots, n\}$ such that $t \in \st(\TAG{u_k}{i}\sigma\mydownarrow)$. By
hypothesis, $\sigma \vDash \TestTag{\TAG{u}{i}}{i}$ and so $\sigma \vDash
\TestTag{\TAG{u_k}{i}}{i}$. Thus, by inductive hypothesis, we know that there
exists $x \in \fv(\TAG{u_k}{i})$ such that $t \in
\st(x\sigma)$. Since 
$\fv(\TAG{u_k}{i}) \subseteq \fv(\TAG{u}{i})$, we can conclude.

\smallskip{}

\emph{Case $\ffun \in \{\senc, \aenc, \sign\}$:} In such a case,  $\TAG{u}{i} =
\ffun(\Tag_i(\TAG{u_1}{i}), \TAG{u_2}{i})$ and
$\TAG{u}{i}\sigma\mydownarrow = \ffun(\Tag_i(\TAG{u_1}{i}\sigma\mydownarrow),
\TAG{u_2}{i}\sigma\mydownarrow)$.  Moreover, $\sigma \vDash
\TestTag{\TAG{u}{i}}{i}$ implies that $\sigma \vDash \TestTag{\TAG{u_k}{i}}{i}$,
with $k \in \{1,2\}$. Since $\racinebis(\TAG{u}{i}\sigma\mydownarrow) = i$, then we deduce that :
\[
\FlawedColor{\gamma}{\TAG{u}{i}\sigma\mydownarrow} = \FlawedColor{\gamma}{\TAG{u_1}{i}\sigma\mydownarrow} \cup \FlawedColor{\gamma}{\TAG{u_2}{i}\sigma\mydownarrow}
\]
Thanks to our inductive hypothesis on
$u_1$ and $u_2$, the result holds.

\smallskip{}

\emph{Case $\ffun = \h$:} This case is analogous to the previous one and can be
handled in a similar way.

\smallskip{}

\emph{Case $\ffun = \langle\;\rangle$:} In this case, we have that $\TAG{u}{i} =
\ffun(\TAG{u_1}{i}, \TAG{u_2}{i})$, and $\TAG{u}{i}\sigma\mydownarrow =
\ffun(\TAG{u_1}{i}\sigma\mydownarrow, \TAG{u_2}{i}\sigma\mydownarrow)$. Moreover,
$\sigma \vDash \TestTag{\TAG{u}{i}}{i}$ implies that $\sigma \vDash
\TestTag{\TAG{u_k}{i}}{i}$ with  $k \in \{1,2\}$. By definition, since $\racine(\TAG{u}{i}\sigma\mydownarrow) = \langle\ \rangle$, we
have that $\FlawedColor{\gamma}{\TAG{u}{i}\sigma\mydownarrow} = \FlawedColor{\gamma}{\TAG{u_1}{i}\sigma\mydownarrow}
\cup \FlawedColor{\gamma}{\TAG{u_2}{i}\sigma\mydownarrow}.$
Applying our inductive hypothesis
on~$u_1$ and~$u_2$, we conclude.

\smallskip{}

\emph{Case $\ffun = \{\vk, \pk\}$:} In this case, we have $u = \ffun(v)$ with $v \in \N \cup \X$. Thus $\TAG{u}{i} = u$ and so by definition, $\FlawedColor{\gamma}{u\sigma\mydownarrow} = \emptyset$. Thus, the result trivially holds.

\smallskip{}

\emph{Case $\ffun \in \{ \sdec, \adec, \checksign\}$:} In this case,
we have that $\TAG{u}{i} = \unTag_i(\ffun(\TAG{u_1}{i},
\TAG{u_2}{i}))$ and 

\[
\begin{array}{rcl}
\TestTag{\TAG{u}{i}}{i} &=&
\TestTag{\TAG{u_1}{i}}{i} \wedge \TestTag{\TAG{u_2}{i}}{i} \wedge \\
&& \Tag_i(\TAG{u}{i}) = \ffun(\TAG{u_1}{i}, \TAG{u_2}{i}).
\end{array}
\] 

By
hypothesis, we know that $\sigma \vDash \TestTag{\TAG{u}{i}}{i}$ and more specifically $\Tag_i(\TAG{u}{i})\sigma\mydownarrow = \ffun(\TAG{u_1}{i}, \TAG{u_2}{i})\sigma\mydownarrow$. It implies that there exist $v_1$, $v_2$ such that $\TAG{u_1}{i}\sigma\mydownarrow = \gfun(\Tag_i(v_1), v_2)$ and $\TAG{u}{i}\sigma\mydownarrow = v_1$,
with $\gfun \in \{\senc, \aenc, \sign\}$. Thus, for all $t \in \FlawedColor{\gamma}{\TAG{u}{i}\sigma\mydownarrow}$, $t \in
\FlawedColor{\gamma}{\TAG{u_1}{i}\sigma\mydownarrow}$. Since $\sigma \vDash
\TestTag{\TAG{u_1}{i}}{i}$, the result holds by inductive hypothesis.

\smallskip{}

\emph{Case $\ffun = \proj_j$, $j \in \{1,2\}$:} We have that
$\TAG{u}{i} = \ffun(\TAG{u_1}{i})$ and $\TestTag{\TAG{u}{i}}{i} =
\TestTag{\TAG{u_1}{i}}{i} \wedge \langle \proj_1(\TAG{u_1}{i}),
\proj_2(\TAG{u_1}{i}) \rangle = \TAG{u_1}{i}$. Hence, $\sigma \vDash
\TestTag{\TAG{u}{i}}{i}$ implies that there exist $v_1, v_2$ such that
$\TAG{u_1}{i}\sigma\mydownarrow = \langle v_1, v_2 \rangle$ and
$\TAG{u}{i}\sigma\mydownarrow = v_j$. Thus, for all $t \in
\FlawedColor{\gamma}{\TAG{u}{i}\sigma\mydownarrow}$, $t \in
\FlawedColor{\gamma}{\TAG{u_1}{i}\sigma\mydownarrow}$. 
Since $\sigma
\vDash \TestTag{\TAG{u_1}{i}}{i}$, our inductive hypothesis allows
us to conclude.
\end{proof}

\medskip


\begin{corollary}
\label{lem:flawedandtagterm}
Let $u \in \T(\Sigma_i \cup \Sigmazero, \N \cup
\X)$ for some $i \in \{1,\ldots,p\}$. Let $\gamma \in
\{\alpha,\beta\}$ such that $i \in \gamma$. Let $\sigma$ be a ground
substitution in normal form such that
$\fv(u) \subseteq \dom(\sigma)$.

If $\sigma \vDash
\TestTag{\TAG{u}{i}}{i}$ then for all $t \in
\Flawed{\TAG{u}{i}\sigma\mydownarrow}$, there exists $x \in
\fv(\TAG{u}{i})$ such that $t \in \Flawed{x\sigma}$.
\end{corollary}

\medskip


\begin{corollary}
Let $u \in \T(\Sigma_i \cup \Sigmazero, \N
  \cup \X)$ for some $i \in \{1,\ldots, p\}$. Let $\gamma \in
  \{\alpha,\beta\}$ such that $i \in \gamma$. Assume that $\fv(u) \subseteq
  \dom(\sigma_0)$.
Moreover, assume that $u$ does not use names in
$\Ec_\alpha \cup \Ec_\beta$.
If  $\delta_\gamma(\sigma_0\mydownarrow) \vDash
  \TestTag{\delta_\gamma(\TAG{u}{i})}{i}$, then for all $t \in
  \FlawedColor{\gamma}{\delta_\gamma(\TAG{u}{i})\delta_\gamma(\sigma_0\mydownarrow)\mydownarrow}$,
  there exists $x \in \fv(\delta_\gamma(\TAG{u}{i}))$ such that $t \in
  \FlawedColor{\gamma}{x\delta_\gamma(\sigma_0\mydownarrow)}$.
\end{corollary}

\begin{proof}
By Lemma~\ref{lem:deltaandnormalform}, we deduce
  that $\delta_\gamma(\sigma_0\mydownarrow)$ is a substitution in normal form. 
Moreover, 
since $u \in  \T(\Sigma_i \cup \Sigmazero, \N
 \cup \X)$ and by definition of $\delta_\gamma$,  and $\TAG{\ }{i}$,
 we deduce that there exists $v \in  \T(\Sigma_i \cup \Sigmazero, \N
 \cup \X)$ such that $\TAG{v}{i} = \delta_\gamma(\TAG{u}{i})$. By
 application of Lemma~\ref{lem: flawed color and tag term}, we deduce that for
 all $t \in \FlawedColor{\gamma}{\TAG{v}{i}\delta_\gamma(\sigma_0\mydownarrow)\mydownarrow}$,
 there exists $x \in \fv(\TAG{v}{i})$ such that $t \in
 \st(x\delta_\gamma(\sigma_0\mydownarrow))$. Hence, we conclude 
that for all $t \in
\FlawedColor{\gamma}{\delta_\gamma(\TAG{u}{i})\delta_\gamma(\sigma_0\mydownarrow)\mydownarrow}$, 
there exists $x \in \fv(\delta_\gamma(\TAG{u}{i}))$ such that $t \in \st(x\delta_\gamma(\sigma_0\mydownarrow))$.
\end{proof}

\medskip

\begin{definition}
Let $u \in \T(\Sigma, \N \cup \X)$. The $\alpha$-factors (resp. $\beta$-factors) of $u$, denoted $\fct_\alpha(u)$, are the maximal syntactic subterms of $u$ that are also in $\FlawedColor{\alpha}{u}$ (resp. $\FlawedColor{\beta}{u}$).
\end{definition}

\medskip


\begin{lemma}
\label{lem: flawed color and tag term 2}
Let $u \in \T(\Sigma_i \cup \Sigmazero, \N \cup
\X)$ for some $i \in \{1,\ldots,p\}$. Let $\gamma \in
\{\alpha,\beta\}$ such that $i \in \gamma$. Let $\sigma$ be a ground
substitution in normal form such that
$\fv(u) \subseteq \dom(\sigma)$. 

If $\sigma \vDash \TestTag{\TAG{u}{i}}{i}$ then 
\begin{itemize}
\item either $\TAG{u}{i}\sigma\mydownarrow \in \fct_\gamma(\TAG{u}{i}\sigma)$,
\item otherwise $\fct_\gamma(\TAG{u}{i}\sigma\mydownarrow) \subseteq  \fct_\gamma(\TAG{u}{i}\sigma)$
\end{itemize}
\end{lemma}

\begin{proof}
We prove the result by induction on $|u|$.

\medskip

\noindent\emph{Base case $|u| = 1$:} In this case, we have that $u \in \X \cup \N$ and so $\TAG{u}{i} = u$. If $u \in \N$, then $u\sigma$
and $\TAG{u}{i}\sigma\mydownarrow$ are both in~$\N$, which means that $\fct_\gamma(\TAG{u}{i}\sigma) = \emptyset$ and $\fct_\gamma(\TAG{u}{i}\sigma\mydownarrow) = \emptyset$. Thus, the result
holds. Otherwise, we have that $u \in \X$ and so $\TAG{u}{i} = u$. But $\sigma$ is in normal form hence $\TAG{u}{i}\sigma\mydownarrow = \TAG{u}{i}\sigma$. Thus, $\fct_\alpha(\TAG{u}{i}\sigma\mydownarrow) =  \fct_\alpha(\TAG{u}{i}\sigma\mydownarrow)$ and so the result holds.

\medskip

\noindent\emph{Inductive step $|u| > 1$, i.e. $u = \ffun(u_1,
\ldots, u_n)$.} We do a case analysis on $\ffun$.

\smallskip{}

\emph{Case $\ffun \in \Sigma_i$:} In this case, $\TAG{u}{i} =
\ffun(\TAG{u_1}{i}, \ldots, \TAG{u_n}{i})$ and $\TAG{u}{i}\sigma\mydownarrow =
\ffun(\TAG{u_1}{i}\sigma\mydownarrow, \ldots,\allowbreak
\TAG{u_n}{i}\sigma\mydownarrow)\mydownarrow$.

By definition, we know that for all $t \in \fct_\gamma{\TAG{u}{i}\sigma\mydownarrow}$, $\racine(t) \not\in \Sigma_\gamma$. Thus, thanks to Lemma~\ref{lem:app_CRicalp05}, for all $t \in \fct_\gamma{\TAG{u}{i}\sigma\mydownarrow}$, there exists $k \in
\{1, \ldots, n\}$ such that $t \in \fct_\gamma(\TAG{u_k}{i}\sigma\mydownarrow)$. 
By hypothesis, $\sigma \vDash \TestTag{\TAG{u}{i}}{i}$ and so $\sigma \vDash
\TestTag{\TAG{u_k}{i}}{i}$. Thus, by inductive hypothesis, we know that 
\begin{itemize}
\item either $\TAG{u_k}{i}\sigma\mydownarrow \in \fct_\gamma(\TAG{u_k}{i}\sigma)$,
\item otherwise $\fct_\gamma(\TAG{u_k}{i}\sigma\mydownarrow) \subseteq  \fct_\gamma(\TAG{u_k}{i}\sigma)$
\end{itemize}
Thus, if $\TAG{u_k}{i}\sigma\mydownarrow \in \fct_\gamma(\TAG{u_k}{i}\sigma)$ then it means that $t = \TAG{u_k}{i}\sigma\mydownarrow$ and so $t \in \fct_\gamma(\TAG{u_k}{i}\sigma)$ (otherwise it contradicts the notion of maximal subterm). Thus in both cases, we obtain that $t \in \fct_\gamma(\TAG{u_k}{i}\sigma)$. Since $\fct_\gamma(\TAG{u_k}{i}\sigma) \subseteq \fct_\gamma(\TAG{u}{i}\gamma)$ then we deduce that $t \in \fct_\gamma(\TAG{u}{i}\gamma)$ hence the result holds.

\smallskip{}

\emph{Case $\ffun \in \{\senc, \aenc, \sign\}$:} In such a case,  $\TAG{u}{i} =
\ffun(\Tag_i(\TAG{u_1}{i}), \TAG{u_2}{i})$ and
$\TAG{u}{i}\sigma\mydownarrow = \ffun(\Tag_i(\TAG{u_1}{i}\sigma\mydownarrow),
\TAG{u_2}{i}\sigma\mydownarrow)$.  Moreover, $\sigma \vDash
\TestTag{\TAG{u}{i}}{i}$ implies that $\sigma \vDash \TestTag{\TAG{u_k}{i}}{i}$,
with $k \in \{1,2\}$. Since $\racinebis(\TAG{u}{i}\sigma\mydownarrow) = i$, then we deduce that :
\[
\fct_\gamma(\TAG{u}{i}\sigma\mydownarrow) = \fct_\gamma(\TAG{u_1}{i}\sigma\mydownarrow) \cup \fct_\gamma(\TAG{u_2}{i}\sigma\mydownarrow)
\]
Thanks to our inductive hypothesis on
$u_1$ and $u_2$, the result holds.

\smallskip{}

\emph{Case $\ffun = \h$:} This case is analogous to the previous one and can be
handled in a similar way.

\smallskip{}

\emph{Case $\ffun = \langle\;\rangle$:} In this case, we have that $\TAG{u}{i} =
\ffun(\TAG{u_1}{i}, \TAG{u_2}{i})$, and $\TAG{u}{i}\sigma\mydownarrow =
\ffun(\TAG{u_1}{i}\sigma\mydownarrow, \TAG{u_2}{i}\sigma\mydownarrow)$. Moreover,
$\sigma \vDash \TestTag{\TAG{u}{i}}{i}$ implies that $\sigma \vDash
\TestTag{\TAG{u_k}{i}}{i}$ with  $k \in \{1,2\}$. By definition, since $\racine(\TAG{u}{i}\sigma\mydownarrow) = \langle\ \rangle$, we
have that $\fct_\gamma(\TAG{u}{i}\sigma\mydownarrow) = \fct_\gamma(\TAG{u_1}{i}\sigma\mydownarrow)
\cup \fct_\gamma(\TAG{u_2}{i}\sigma\mydownarrow).$
Applying our inductive hypothesis
on~$u_1$ and~$u_2$, we conclude.

\smallskip{}

\emph{Case $\ffun = \{\vk, \pk\}$:} In this case, we have $u = \ffun(v)$ with $v \in \N \cup \X$. Thus $\TAG{u}{i} = u$ and so by definition, $\fct_\gamma(u\sigma\mydownarrow) = \emptyset$. Thus, the result trivially holds.

\smallskip{}

\emph{Case $\ffun \in \{ \sdec, \adec, \checksign\}$:} In this case,
we have that $\TAG{u}{i} = \unTag_i(\ffun(\TAG{u_1}{i},
\TAG{u_2}{i}))$ and 

\[
\begin{array}{rcl}
\TestTag{\TAG{u}{i}}{i} &=&
\TestTag{\TAG{u_1}{i}}{i} \wedge \TestTag{\TAG{u_2}{i}}{i} \wedge \\
&& \Tag_i(\TAG{u}{i}) = \ffun(\TAG{u_1}{i}, \TAG{u_2}{i}).
\end{array}
\] 

By
hypothesis, we know that $\sigma \vDash \TestTag{\TAG{u}{i}}{i}$ and more specifically $\Tag_i(\TAG{u}{i})\sigma\mydownarrow = \ffun(\TAG{u_1}{i}, \TAG{u_2}{i})\sigma\mydownarrow$. It implies that there exists $v_1$, $v_2$ such that $\TAG{u_1}{i}\sigma\mydownarrow = \gfun(\Tag_i(v_1), v_2)$ and $\TAG{u}{i}\sigma\mydownarrow = v_1$, with $\gfun \in \{\senc, \aenc, \sign\}$. 
Thus, for all $t \in \fct_\gamma(\TAG{u}{i}\sigma\mydownarrow)$, $t \in
\fct_\gamma(\TAG{u_1}{i}\sigma\mydownarrow)$. Since $\sigma \vDash
\TestTag{\TAG{u_1}{i}}{i}$, the result holds by inductive hypothesis.

\smallskip{}

\emph{Case $\ffun = \proj_j$, $j \in \{1,2\}$:} We have that
$\TAG{u}{i} = \ffun(\TAG{u_1}{i})$ and $\TestTag{\TAG{u}{i}}{i} =
\TestTag{\TAG{u_1}{i}}{i} \wedge \langle \proj_1(\TAG{u_1}{i}),
\proj_2(\TAG{u_1}{i}) \rangle = \TAG{u_1}{i}$. Hence, $\sigma \vDash
\TestTag{\TAG{u}{i}}{i}$ implies that there exist $v_1, v_2$ such that
$\TAG{u_1}{i}\sigma\mydownarrow = \langle v_1, v_2 \rangle$ and
$\TAG{u}{i}\sigma\mydownarrow = v_j$. Thus, for all $t \in
\fct_\gamma(\TAG{u}{i}\sigma\mydownarrow)$, $t \in
\fct_\gamma(\TAG{u_1}{i}\sigma\mydownarrow)$. 
Since $\sigma
\vDash \TestTag{\TAG{u_1}{i}}{i}$, our inductive hypothesis allows
us to conclude. 
\end{proof}

\subsection{Frame of a tagged process}
\label{subset: frame flagged process}

In this subsection, we will state and prove the lemmas regarding
frames and static equivalence.
Let $\nu \Ec. \Phi$ be a frame such that:
\[
\Phi = \{ w_1
\refer u_1, \ldots, w_n \refer u_n\}.
\]
Let $M$ be a recipe, \emph{i.e.} a term such that $\fv(M) \subseteq \dom(\Phi)$ and $\fn(M) \cap
\Ec = \emptyset$, we define the measure $\M$ as follows:
 \[
\M(M) = (i_\mathsf{max}, |M|)
\]
where $i_\mathsf{max} \in \{1,\ldots,n\}$ is the maximal indice $i$ such that
$w_ i \in \fv(M)$, and $|M|$ denotes the size of the term~$M$, \emph{i.e.} the number of
symbols that occur in~$M$.

We have that $\M(M_1) \stackrel{\defi}{=} (i_1, s_1) < \M(M_2)
\stackrel{\defi}{=} (i_2,s_2)$ when
either $i_1 < i_2$;
or $i_1 = i_2$ and $s_1 < s_2$.

Once again, we denote by $z^\alpha_1, \ldots, z^\alpha_k$ and $z^\beta_1, \ldots, z^\beta_\ell$ the assignment variables of the extended processes that we are considering. 

\begin{definition}
 Let $\quadruple{\Ec}{\p}{\Phi}{\sigma}$ be an extended process, $\prec$ be a total order on $\dom(\Phi) \cup \dom(\sigma)$ and $\vcol$ be a mapping from $\dom(\Phi) \cup \dom(\sigma)$ to $\{ 1, \ldots, p\}$. We say that  $\quadruple{\Ec}{\p}{\Phi}{\sigma}$ is a \emph{derived well-tagged extended process} w.r.t.~$\prec$ and $\vcol$ if for every $x \in \dom(\Phi)$ (resp. $x \in \dom(\sigma)$), there exists $\{\gamma,\gamma'\} = \{\alpha,\beta\}$ such that one of the following condition is
satisfied:
\begin{enumerate}
\item there exist $v$ and $i = \vcol(x) \in \gamma$ such that
  $u = \TAG{v}{i}\sigma$, $\sigma  \vDash
  \TestTag{\TAG{v}{i}}{i}$, and for all $z \in \fv(v)$, $z \prec x$ and either $\vcol(z) \in \gamma$ or there exists $j$ such that $z = z^{\gamma'}_j$; or
\item there exists $M$ such that $\fv(M)
    \subseteq \dom(\Phi) \cap \{z ~|~ z \prec x\}$, $\fn(M) \cap \Ec = \emptyset$ and
    $M\Phi = u$.
\end{enumerate}
where $u = x\Phi$ (resp. $u = x\sigma$). 
\end{definition}

{In the case of variables instantiated through an output, and or an internal communication, it will be the first item that needs to hold; while in the case of variables intantiated through inputs on public channels it is the second item that needs to hold.} Intuitively, the order $\prec$ on $\dom(\Phi) \cup \dom(\sigma)$ corresponds to the order in which the variables in $\dom(\Phi) \cup \dom(\sigma)$ have been introduced along the execution. In particular, we have that $w_1 \prec w_2 \prec \ldots \prec w_n$ where $\dom(\Phi) = \{w_1,\ldots, w_n\}$. In the following, we sometimes simply say that $\quadruple{\Ec}{\p}{\Phi}{\sigma}$ is a derived well-tagged extended process.




\begin{lemma}
  \label{lem:Flawedandframeelement}
  Let $\quadruple{\Ec}{\p}{\Phi}{\sigma}$ be a derived well-tagged extended process w.r.t $\prec$ and $\vcol$. Let $x \in \dom(\Phi)$ (resp. $x \in \dom(\sigma)$) and $t \in   \Flawed{x\Phi\mydownarrow}$ (resp. $t \in \Flawed{x\sigma\mydownarrow}$). We have that there exists $M$ such that $\fv(M) \subseteq \dom(\Phi) \cap \{z ~|~ z \prec x\}$, $\fn(M) \cap \Ec = \emptyset$ and $t \in \Flawed{M\Phi\mydownarrow}$.
\end{lemma}

\begin{proof}
We prove this result by induction on $\dom(\Phi) \cup \dom(\sigma)$
with the order $\prec$.

\smallskip{}

\noindent \emph{Base case $u = x\sigma$ or $u = x\Phi$ with $x \prec
  z$ for any $z \in \dom(\Phi) \cup \dom(\sigma)$.} Assume $t \in \Flawed{u\mydownarrow}$.
By definition of a derived well-tagged extended process w.r.t $\prec$ and $\vcol$, one of the
following condition is satisfied:
\begin{enumerate}
\item There exist $v$ and $i = \vcol(x)$ such that $u
=\TAG{v}{i}\sigma$, $\sigma \vDash \TestTag{\TAG{v}{i}}{i}$, and $z
\prec x$ for any $z \in \fv(v)$.  Since $u = \TAG{v}{i}\sigma$ and $\sigma \vDash
 \TestTag{\TAG{v}{i}}{i}$, we can apply Lemma~\ref{lem: flawed color and tag term}
 to $v$ and~$\sigma\mydownarrow$. Thus, we have that there exists $z \in
 \fv(\TAG{v}{i})$ such that $t \in
 \Flawed{z\sigma\mydownarrow}$. However, since $x$ is mimimal w.r.t. $\prec$,
we know that $\fv(v) = \emptyset$.   Hence, we obtain a
contradiction. This case is impossible.
\item There exists $M$ such that $\fv(M) \subseteq \dom(\Phi) \cap \{z
  ~|~ z \prec x\}$, $\fn(M) \cap \Ec = \emptyset$, and $M\Phi =
  u$. Thus, we have that $M\Phi\mydownarrow = u\mydownarrow$, and we
  have that $t \in \Flawed{M\Phi\mydownarrow}$.
\end{enumerate}

\medskip{}

\noindent \emph{Inductive case $u = x\sigma$ or $u = x\Phi$.} Assume $t \in \Flawed{u\mydownarrow}$.
By definition of a derived well-tagged extended process w.r.t $\prec$ and $\vcol$, one of the
following condition is satisfied:
\begin{enumerate}
\item There exist $v$ and $i = \vcol(x)$ such that $u
=\TAG{v}{i}\sigma$, $\sigma \vDash \TestTag{\TAG{v}{i}}{i}$, and $z
\prec x$ for any $z \in \fv(v)$.  Since $u = \TAG{v}{i}\sigma$ and $\sigma \vDash
 \TestTag{\TAG{v}{i}}{i}$, we can apply Lemma~\ref{lem: flawed color and tag term}
 to $v$ and~$\sigma\mydownarrow$. Thus, we have that there exists $z \in
 \fv(\TAG{v}{i})$ such that $t \in
 \Flawed{z\sigma\mydownarrow}$, and we have that $z \prec x$.
Hence, we conclude by applying our induction hypothesis.
\item There exists $M$ such that $\fv(M) \subseteq \dom(\Phi) \cap \{z
  ~|~ z \prec x\}$, $\fn(M) \cap \Ec = \emptyset$, and $M\Phi =
  u$. Thus, we have that $M\Phi\mydownarrow = u\mydownarrow$, and we
  have that $t \in \Flawed{M\Phi\mydownarrow}$.
\end{enumerate}
This allows us to conclude.
\end{proof}


\begin{lemma}
  \label{lem:FlawedColor and frame element direct element}
  Let $\quadruple{\Ec}{\p}{\Phi}{\sigma}$ be a derived well-tagged extended process w.r.t $\prec$ and $\vcol$. Let $\{\gamma,\gamma'\} = \{\alpha,\beta\}$. Let $x \in \dom(\Phi)$ (resp. $x \in \dom(\sigma)$) such that $\vcol(x) \in \gamma$. Let $u = x\Phi$ (resp. $u = x\sigma$). Let $t \in \fct_\gamma(u\mydownarrow)$. We have that 
  \begin{itemize}
  \item either there exists $M$ such that $\fv(M) \subseteq \dom(\Phi) \cap \{z ~|~ z \prec x\}$, $\fn(M) \cap \Ec = \emptyset$ and $t \in \fct_\gamma(M\Phi\mydownarrow)$; 
  \item otherwise there exists $j$ such that $z^{\gamma'}_j \prec x$ and $z^{\gamma'}_j\sigma\mydownarrow = t$.
  \end{itemize}
\end{lemma}

\begin{proof}
We prove this result by induction on $\dom(\Phi) \cup \dom(\sigma)$
with the order $\prec$.

\smallskip{}

\noindent \emph{Base case $u = x\sigma$ or $u = x\Phi$ with $x \prec
  z$ for any $z \in \dom(\Phi) \cup \dom(\sigma)$.} Let $t \in \fct_\gamma(u\mydownarrow)$ and $\vcol(x) \in \gamma$ with $\gamma \in \{\alpha,\beta\}$.
By definition of a derived well-tagged extended process w.r.t $\prec$ and $\vcol$, one of the
following condition is satisfied:
\begin{enumerate}
\item There exist $v$ and $i = \vcol(x)$ such that $u
=\TAG{v}{i}\sigma$, $\sigma \vDash \TestTag{\TAG{v}{i}}{i}$, and $z
\prec x$ for any $z \in \fv(v)$. Since $x$ is minimal by $\prec$ then $\fv(v)  = \emptyset$. Hence $u = \TAG{v}{i}$. Thus we deduce that $\fct_\gamma(u\mydownarrow) = \emptyset$. Hence there is a contradiction with $t \in \fct_\gamma(u\mydownarrow)$ and so this condition cannot be satisfied.
\item There exists $M$ such that $\fv(M) \subseteq \dom(\Phi) \cap \{z
  ~|~ z \prec x\}$, $\fn(M) \cap \Ec = \emptyset$, and $M\Phi =
  u$. Thus, we have that $M\Phi\mydownarrow = u\mydownarrow$ and so the result holds.
  \end{enumerate}
  
  \medskip{}
  
\noindent \emph{Inductive case $u = x\sigma$ or $u = x\Phi$.} Assume $t \in \fct_\gamma(u\mydownarrow)$ and $\vcol(x) \in \gamma$.
By definition of a derived well-tagged extended process w.r.t $\prec$ and $\vcol$, one of the
following condition is satisfied:  
\begin{enumerate}
\item There exist $v$ and $i = \vcol(x) \in \gamma$ such that
  $u = \TAG{v}{i}\sigma$, $\sigma  \vDash
  \TestTag{\TAG{v}{i}}{i}$, and for all $z \in \fv(v)$, $z \prec x$ and either $\vcol(z) \in \gamma$ or there exists $j$ such that $z = z^{\gamma'}_j$.
  Since $u = \TAG{v}{i}\sigma$ and $\sigma \vDash
 \TestTag{\TAG{v}{i}}{i}$, we can apply Lemma~\ref{lem: flawed color and tag term 2}
 to $v$ and~$\sigma\mydownarrow$. Thus we have that $t \in \fct_\gamma(\TAG{v}{i}(\sigma\mydownarrow)$. In such a case, it means that there exists $z \in \fv(v)$ with $z \prec x$ such that $t \in \fct_\gamma(z\sigma\mydownarrow)$ and one of the two conditions is satisfied:
 \begin{itemize}
 \item $\vcol(z) \in \gamma$: In such a case, we can apply our inductive hypothesis on $t$ and $z$ and so the result holds.
 \item there exists $j$ such that $z = z^{\gamma'}_j$: Otherwise, we know by hypothesis that $z^{\gamma'}\sigma\mydownarrow \in \N$ or $\fct_\gamma(z^{\gamma'}\sigma\mydownarrow) = \{ z^{\gamma'}\sigma\mydownarrow\}$. Since $t \in \fct_\gamma(z\sigma\mydownarrow)$, we deduce that $z^{\gamma'}\sigma\mydownarrow \not\in \N$ and so $\fct_\gamma(z^{\gamma'}\sigma\mydownarrow) = \{ z^{\gamma'}\sigma\mydownarrow\}$. But this implies that $t = z\sigma\mydownarrow$. Hence the result holds.
 \end{itemize}
\item There exists $M$ such that $\fv(M) \subseteq \dom(\Phi) \cap \{z
  ~|~ z \prec x\}$, $\fn(M) \cap \Ec = \emptyset$, and $M\Phi =
  u$. Thus, we have that $M\Phi\mydownarrow = u\mydownarrow$, and we
  have that $t \in \fct_\gamma(M\Phi\mydownarrow)$.
\end{enumerate}
This allows us to conclude.
\end{proof}


\begin{lemma}
  \label{lem:flawed,smallerrecipe}
  Let $\quadruple{\Ec}{\p}{\Phi}{\sigma}$ be a derived well-tagged extended process. Let $M$ be a term such that $\fn(M) \cap \Ec = \emptyset$ and $\fv(M) \subseteq \dom(\Phi)$. Let $\ffun(t_1, \ldots, t_m) \in \Flawed{M\Phi\mydownarrow}$.  There exists $M_1, \ldots, M_m$ such that $\fv(M_k) \subseteq \dom(\Phi)$, $\fn(M_k) \cap \Ec = \emptyset$, $M_k\Phi\mydownarrow = t_k$, and $\M(M_k) < \M(M)$, for all $k \in \{1, \ldots, m\}$.
\end{lemma}

\begin{proof}
 We prove this result by induction on $\M(M)$.

\smallskip{}

 \noindent\emph{Base case $\M(M) = (j, 1)$:} In this case, either we
 have that $M \in \N$ or $M = w_j$. If $M \in \N$, then we have
 $M\Phi\mydownarrow = M \in \N$ and $\Flawed{M\Phi\mydownarrow} =
 \emptyset$. Thus the result holds. If $M = w_j$
then, by Lemma~\ref{lem:Flawedandframeelement},
$\ffun(t_1, \ldots, t_m) \in
 \Flawed{w_j\Phi\mydownarrow}$ implies that 
there exists $M'$ such that:
\begin{itemize}
\item  $\fv(M') \subseteq
 \{w_1, \ldots, w_{j-1}\}$, 
\item $\fn(M) \cap \Ec = \emptyset$, and 
\item $\ffun(t_1,
 \ldots, t_m) \in \Flawed{M'\Phi\mydownarrow}$.
\end{itemize} 
Since $\M(M') < \M(M)$, thanks to
 our inductive hypothesis, we  deduce that there exist $M_1, \ldots, M_m$
 such that for each $k \in \{1,
 \ldots, m\}$, we have that:
$\fv(M_k) \subseteq \dom(\Phi)$, $\fn(M_k) \cap \Ec =
  \emptyset$,
 $M_k\Phi\mydownarrow = t_k$, and 
$\M(M_k) < \M(M') < \M(M)$.

 \medskip

 \noindent \emph{Inductive step $\M(M) > (j,1)$:} In such a case, we have that $M = \ffun(M_1,
 \ldots, M_n)$. Let $t = \gfun(t_1, \ldots, t_m) \in
 \Flawed{M\Phi\mydownarrow}$. We do a case analysis on $\ffun$.

\smallskip{}

 \emph{Case $\ffun \in \Sigma_i \cup \Sigma_{\Tag_i}$ for some $i \in
   \{1,\ldots,p\}$:} In such a  case, $M\Phi\mydownarrow = \ffun(M_1\Phi\mydownarrow, \ldots,
 M_n\Phi\mydownarrow)\mydownarrow$. By definition, we know that for all $t \in
 \Flawed{M\Phi\mydownarrow}$, we have that $\racine(t) \not\in \Sigma_i \cup
 \Sigma_{\Tag_i}$. Thus, thanks to Lemma~\ref{lem:app_CRicalp05}, we 
 deduce that 
\[
\Flawed{M\Phi\mydownarrow} \subseteq  \Flawed{M_1\Phi\mydownarrow}
\cup \ldots \cup \Flawed{M_n\Phi\mydownarrow}.
\] 
Since $\M(M_k)
 < \M(M)$ for any $k \in \{1,\ldots, n\}$, thanks to our inductive hypothesis, we know that there exists $M'_1,
 \ldots, M'_m$ such that $\fv(M'_j) \subseteq \dom(\Phi)$, $\fn(M'_j) \cap \Ec =
 \emptyset$, $M'_j\Phi\mydownarrow = t_i$ and $\M(M'_j) < \M(M_k) < \M(M)$, for
 $j \in \{1, \ldots, m\}$. Hence the result holds.

\smallskip{}

 \emph{Case $\ffun = \langle\;\rangle$:} In such a case, $M\Phi\mydownarrow =
 \ffun(M_1\Phi\mydownarrow, M_2\Phi\mydownarrow)$. Moreover, we have
 that 
$\Flawed{M\Phi\mydownarrow} = \Flawed{M_1\Phi\mydownarrow} \cup
 \Flawed{M_2\Phi\mydownarrow}$. Since $\M(M_1) < \M(M)$, $\M(M_2) < \M(M)$ and
 $t \in \Flawed{M_1\Phi\mydownarrow} \cup \Flawed{M_2\Phi\mydownarrow}$, we
 conclude by applying our inductive hypothesis on $M_1$ (or $M_2$).

\smallskip{}

 \emph{Case $\ffun \in \{\pk, \vk\}$:} In this case, $M\Phi\mydownarrow =
 \ffun(M_1\Phi\mydownarrow)$ and we have
 that $\Flawed{M\Phi\mydownarrow} = \emptyset$. Hence the result trivially holds.

\smallskip{}
 
 \emph{Case $\ffun \in \{\senc, \aenc, \sign\}$:} In such a case, we
 have that
 $M\Phi\mydownarrow = \ffun(M_1\Phi\mydownarrow, M_2\Phi\mydownarrow)$. We need
 to distinguish whether $\racine(M_1\Phi\mydownarrow) = \Tag_i$ for
 some $i \in \{1,\ldots,p\}$ or not.

 If $\racine(M_1\Phi\mydownarrow) = \Tag_i$ for some $i \in
 \{1,\ldots, p\}$,  then there exists $u_1$ such that $M_1\Phi\mydownarrow =
 \Tag_i(u_1)$. Hence, we have that $\Flawed{M_1\Phi\mydownarrow} =
 \Flawed{u_1}$. We have also that:
\[
 \Flawed{M\Phi\mydownarrow} = \Flawed{u_1} \cup
 \Flawed{M_2\Phi\mydownarrow}.
\]
 We deduce that $t \in \Flawed{M_1\Phi\mydownarrow}$ or $t \in
 \Flawed{M_2\Phi\mydownarrow}$. Since $\M(M_1) < \M(M)$ and $\M(M_2) < \M(M)$,
 we conclude by applying our inductive hypothesis on $M_1$ or $M_2$.

 Otherwise $\racine(M_1\Phi\mydownarrow) \not\in \{\Tag_1, \ldots, \Tag_p\}$. In such a
 case, $\Flawed{M\Phi\mydownarrow} = \Flawed{M_1\Phi\mydownarrow} \cup
 \Flawed{M_2\Phi\mydownarrow} \cup \{M\Phi\mydownarrow\}$. If $t =
 M\Phi\mydownarrow$, we have that $t_1 = M_1\Phi\mydownarrow$, $t_2 =
 M_2\Phi\mydownarrow$ and $\M(M_1) < \M(M)$, $\M(M_2) < \M(M)$. Thus the result
 holds. If $t \in \Flawed{M_1\Phi\mydownarrow} \cup
 \Flawed{M_2\Phi\mydownarrow}$, we conclude by applying our inductive
 hypothesis on $M_1$ or $M_2$.

\smallskip{}

 \emph{Case $\ffun = \h$:} This case is analogous to the previous one and can
 be handled similarly.

\smallskip{}

 \emph{Case $\ffun \in \{\sdec, \adec, \checksign\}$:} In such a case,
 we have to distinguish two cases depending on whether $\ffun$ is reduced in
 $M\Phi\mydownarrow$, or not.

 If $\ffun$ is not reduced, \emph{i.e.} $M\Phi\mydownarrow =
 \ffun(M_1\Phi\mydownarrow, M_2\Phi\mydownarrow)$, then we have that
 \[
\Flawed{M\Phi\mydownarrow} = \{M\Phi\mydownarrow\} \cup
 \Flawed{M_1\Phi\mydownarrow} \cup \Flawed{M_2\Phi\mydownarrow}.
\] 
Thus if $t =
 M\Phi\mydownarrow$, we have that $t_1 = M_1\Phi\mydownarrow$, $t_2 =
 M_2\Phi\mydownarrow$ and $\M(M_1) < \M(M)$, $\M(M_2) < \M(M)$. Thus the result
 holds. Otherwise, we have that $t \in \Flawed{M_1\Phi\mydownarrow}$ or $t \in
 \Flawed{M_2\Phi\mydownarrow}$. Since $\M(M_1) < \M(M)$, $\M(M_2) < \M(M)$, we
 can conclude by applying our inductive hypothesis on~$M_1$ or~$M_2$.

 If $\ffun$ is reduced, then we have that $M_1\Phi\mydownarrow =
 \ffun'(u_1,u_2)$ with $M\Phi\mydownarrow = u_1$ and $\ffun' \in \{\senc,
 \aenc, \sign\}$. If $\racine(u_1) = \Tag_i$ for some $i \in \{1,\ldots,p\}$, then
 we have that there exists $u'_1$ such that $u_1 = \Tag_i(u'_1)$,
 $\Flawed{M\Phi\mydownarrow} = \Flawed{u'_1}$ and $\Flawed{M_1\Phi\mydownarrow}
 = \Flawed{u'_1} \cup \Flawed{u_2}$. Thus, we have that 
$\Flawed{M\Phi\mydownarrow} \subseteq
 \Flawed{M_1\Phi\mydownarrow}$.
Otherwise, if $\racine(u_1) \not\in \{\Tag_1,\ldots,\Tag_p\}$, then we
have that 
\[
\Flawed{M_1\Phi\mydownarrow} =
 \{M_1\Phi\mydownarrow\} \cup \Flawed{u_1} \cup \Flawed{u_2}
\]
 and
 $\Flawed{M\Phi\mydownarrow} = \Flawed{u_1}$. Thus, 
 $\Flawed{M\Phi\mydownarrow} \subseteq \Flawed{M_1\Phi\mydownarrow}$. In
 both cases, we have that $\Flawed{M\Phi\mydownarrow} \subseteq
 \Flawed{M_1\Phi\mydownarrow}$ and since $\M(M_1) < \M(M)$, we can conclude by
 applying our inductive hypothesis on~$M_1$.
\end{proof}

\newcommand{\fctpair}{\fct_{\langle\ \rangle}}

In the following lemma, we will use the factors of the signature only composed of $\langle\ \rangle$, denoted $\fctpair$. Typically, for all terms $u$, for all context built only on $\langle\ \rangle$, for all terms $u_1, \ldots, u_n$, if $u = C[u_1, \ldots, u_n]$ and for all $k \in \{1, \ldots, n\}$, $\racine(u_i) \neq \langle\ \rangle$ then $\fctpair(u) = \{ u_1, \ldots, u_n\}$.


\begin{lemma}
  \label{lem:FlawedColor and frame element direct element 2}
  Let $\quadruple{\Ec}{\p}{\Phi}{\sigma}$ be a derived well-tagged extended process w.r.t $\prec$ and $\vcol$. Assume that for all assignment variables $z$, $\new \Ec. \Phi \not\vdash z\sigma\mydownarrow$. Let $M$ such that $\fv(M) \subseteq \dom(\Phi)$, $\fn(M) \cap \Ec = \emptyset$. For all $\{\gamma,\gamma'\} = \{\alpha,\beta\}$, for all $t \in \fct_\gamma(M\Phi\mydownarrow)$, if $t \not\in \fctpair(M\Phi\mydownarrow)$ and for all assignment variable $z$, for all $w \in \dom(\Phi)$, $z \prec w$ and $\M(w) \leq \M(M)$ implies $z\sigma\mydownarrow \neq t$ then there exists $M'$ such that $\M(M') < \M(M)$, $\fn(M) \cap \Ec = \emptyset$ and $t \in \fctpair(M'\Phi\mydownarrow)$.
\end{lemma}

\begin{proof}
We do a proof by induction on $\M(M)$:

\medskip

\noindent\emph{Base case $\M(M) = (0, 1)$:} In this case, we have that $M \in \N$ which means that $M\Phi\mydownarrow = M \in \N$ and $\fct_\gamma(M\Phi\mydownarrow) = \emptyset$. Thus the result holds. 

\medskip

\noindent\emph{Base case $\M(M) = (j,1)$:} In this case, we have $M = w_j$. Let $\{\gamma,\gamma'\} = \{\alpha,\beta\}$. Let $t \in \fct_\gamma(M\Phi\mydownarrow)$ such that $t \not\in \fctpair(M\Phi\mydownarrow)$. We do a case analysis on $\vcol(w_j)$:

\emph{Case $\vcol(w_j) \in \gamma$:} In this case, since for all assignment variable $z$, for all $w \in \dom(\Phi)$, $z \prec w$ and $\M(w) \leq \M(M)$ implies $z\sigma\mydownarrow \neq t$, than we can deduce that for all assignment variables $z \prec w_j$, $z\sigma\mydownarrow \neq t$. Thus by Lemma~\ref{lem:FlawedColor and frame element direct element}, we obtain that there exists $M'$ such that $\fv(M') \subseteq \dom(\Phi) \cap \{z ~|~ z \prec x\}$, $\fn(M') \cap \Ec = \emptyset$ and $t \in \fct_\gamma(M'\Phi\mydownarrow)$. $\fv(M') \subseteq \dom(\Phi) \cap \{z ~|~ z \prec x\}$ implies that $\M(M') = (k,k')$ with $k < j$ and so $\M(M') < \M(M)$. If $t \in \fctpair(M'\Phi\mydownarrow)$ then the result holds. Otherwise, we can apply our inductive hypothesis on $t$ and $M'$ and so the result holds.
 
 \emph{Case $\vcol(w_j) \in \gamma'$ :} Since $t \not\in \fctpair(M\Phi\mydownarrow)$, we deduce that there exists $u \in \fctpair(M\Phi\mydownarrow)$ s.t. $\racinebis(u) = \gamma$ and $t \in \fct_\gamma(u)$. Note that $\racinebis(u) \not\in \gamma' \cup \{0\}$ otherwise it would contradict the fact that $t \in \fct_\gamma(M\Phi\mydownarrow)$. But $u \in \fct_{\gamma'}(M\Phi\mydownarrow)$. Moreover, $u \in \fctpair(M\Phi\mydownarrow)$ implies that $u$ is deducible in $\new\ \Ec. \Phi$. Thus we deduce that for all assignment variables $z$, $z\sigma\mydownarrow \neq u$. By applying the same proof as case $\vcol(w_j) \in \gamma$, we deduce that there exists $M'$ such that $\fn(M') \cap \Ec = \emptyset$, $\M(M') < \M(M)$ and $u \in \fctpair(M'\Phi\mydownarrow)$. But $t \in \fct_\gamma(u)$, $\racinebis(u) = \gamma$ and $u \in \fctpair(M'\Phi\mydownarrow)$ implies that $t \in \fct_\gamma(M'\Phi\mydownarrow)$ and $t \not\in \fctpair(M'\Phi\mydownarrow)$. Hence we can apply our inductive hypothesis on $M'$ and $t$ which allows us to conclude.

\medskip

 \noindent \emph{Inductive step $\M(M) > (j,1)$:} In such a case, we have that $M = \ffun(M_1,\ldots, M_n)$. Let $t \in \fct_\gamma(M\Phi\mydownarrow)$ such that $t \not\in \fctpair(M\Phi\mydownarrow)$. We do a case analysis on $\ffun$.

\smallskip{}

 \emph{Case $\ffun \in \Sigma_i \cup \Sigma_{\Tag_i}$ for some $i \in \gamma$:} In such a  case, $M\Phi\mydownarrow = \ffun(M_1\Phi\mydownarrow, \ldots,
 M_n\Phi\mydownarrow)\mydownarrow$. By definition, we know that for all $t \in
 \fct_\gamma(M\Phi\mydownarrow)$, we have that $\racine(t) \not\in \Sigma_i \cup
 \Sigma_{\Tag_i}$. Thus, thanks to Lemma~\ref{lem: flawed color and tag term 2}, we 
 deduce that there exists 
\[
\fct_\gamma(M\Phi\mydownarrow) \subseteq  \fct_\gamma(M_1\Phi\mydownarrow)
\cup \ldots \cup \fct_\gamma(M_n\Phi\mydownarrow).
\] 
Thus there exists $k \in \{1, \ldots, n\}$ such that $t \in \fct_\gamma(M_k\Phi\mydownarrow)$. If $t \in \fctpair(M_k\Phi\mydownarrow)$ then the result holds, else we apply our inductive hypothesis on $t$ and $M_k$ and so the result also holds.

\smallskip{}

 \emph{Case $\ffun \in \Sigma_i \cup \Sigma_{\Tag_i}$ for some $i \not\in \gamma$:} In such a  case, $M\Phi\mydownarrow = \ffun(M_1\Phi\mydownarrow, \ldots, M_n\Phi\mydownarrow)\mydownarrow$. We assumed that $t \not\in \fctpair(M\Phi\mydownarrow)$ hence there exists $u \in \fctpair(M\Phi\mydownarrow)$ s.t. $\racinebis(u) = \gamma$ and $t \in \fct_\gamma(u)$. But it also implies that $\racinebis(M\Phi\mydownarrow) \in \gamma \cup \{ 0 \}$. Hence, by applying Lemma~\ref{lem:app_CRicalp05}, we deduce that there exists $k \in \{1, \ldots, n\}$ such that $M\Phi\mydownarrow \in \st(M_k\Phi\mydownarrow)$. Moreover, it also implies that $u \in \fct_\gamma'(M_k\Phi\mydownarrow)$. 
 
 If $u \in \fctpair(M_k\Phi\mydownarrow)$ then we deduce that $\racine(M_k\Phi\mydownarrow) \not\in \gamma'$ and so, by Lemma~\ref{lem:app_CRicalp05}, $M_k\Phi\mydownarrow = M\Phi\mydownarrow$. Since we had $t \not\in \fctpair(M\Phi\mydownarrow)$, then we also have $t \not\in \fctpair(M_k\Phi\mydownarrow)$ and so we conclude by applying our inductive hypothesis on $t$ and $M_k$.
 
 if $u \not\in \fctpair(M_k\Phi\mydownarrow)$ then we can apply our inductive hypothesis on $u, \gamma'$ and $M_k$. Indeed, since $u \in \fctpair(M\Phi\mydownarrow)$, then $u$ is deducible in $\new\ \Ec. \Phi$ and so we deduce that for all assignment variable $z$, $z\sigma\mydownarrow \neq u$. Hence we obtain that there exists $M'$ such that $\M(M') < \M(M_k)$, $\fn(M) \cap \Ec = \emptyset$ and $u \in \fctpair(M'\Phi\mydownarrow)$. But $t \in \fct_\gamma(u)$ and $u \in \fctpair(M'\Phi\mydownarrow)$. Hence we deduce that $t \in \fct_\gamma(M'\Phi\mydownarrow)$ and $t \not\in \fctpair(M'\Phi\mydownarrow)$. We conclude by applying once again our inductive hypothesis but on $t, \gamma$ and $M'$.
 
 \smallskip{}

 \emph{Case $\ffun = \langle\;\rangle$:} In such a case, $M\Phi\mydownarrow =
 \ffun(M_1\Phi\mydownarrow, M_2\Phi\mydownarrow)$. Moreover, we have
 that 
$\fct_\gamma(M\Phi\mydownarrow) = \fct_\gamma(M_1\Phi\mydownarrow) \cup
 \fct_\gamma(M_2\Phi\mydownarrow)$. Since $\M(M_1) < \M(M)$, $\M(M_2) < \M(M)$ and
 $t \in \fct_\gamma(M_1\Phi\mydownarrow) \cup \fct_\gamma(M_2\Phi\mydownarrow)$, we
 conclude by applying our inductive hypothesis on $t$ and $M_1$ (or $M_2$).

\smallskip{}

 \emph{Case $\ffun \in \{\pk, \vk\}$:} In this case, $M\Phi\mydownarrow =
 \ffun(M_1\Phi\mydownarrow)$ and we have
 that $\fct_\gamma{M\Phi\mydownarrow} = \emptyset$. Hence the result trivially holds.

\smallskip{}
 
 \emph{Case $\ffun \in \{\senc, \aenc, \sign\}$:} In such a case, we
 have that
 $M\Phi\mydownarrow = \ffun(M_1\Phi\mydownarrow, M_2\Phi\mydownarrow)$. We need
 to distinguish whether $\racine(M_1\Phi\mydownarrow) = \Tag_i$ for
 some $i \in \{1,\ldots,p\}$ or not.

 If $\racine(M_1\Phi\mydownarrow) = \Tag_i$ for some $i \in
 \{1,\ldots, p\}$,  then there exists $u_1$ such that $M_1\Phi\mydownarrow =
 \Tag_i(u_1)$. Assume first that $i \in \gamma'$. In such a case $\fct_\gamma(M\Phi\mydownarrow) = \{ \fct_\gamma(M\Phi\mydownarrow) \}$ and $\fctpair(M\Phi\mydownarrow) = \{ \fctpair(M\Phi\mydownarrow) \}$. Hence it contradicts the fact that $t \not\in \fctpair(M\Phi\mydownarrow)$. We can thus deduce that $i \in \gamma$. But in such a case, we have that $\fct_\gamma(M_1\Phi\mydownarrow) = \fct_\gamma(u_1)$ and:
\[
 \fct_\gamma(M\Phi\mydownarrow) = \fct_\gamma(u_1) \cup
 \fct_\gamma(M_2\Phi\mydownarrow).
\]
 We deduce that $t \in \fct_\gamma(M_1\Phi\mydownarrow)$ or $t \in
 \fct_\gamma(M_2\Phi\mydownarrow)$. Since $\M(M_1) < \M(M)$ and $\M(M_2) < \M(M)$,
 we conclude by applying our inductive hypothesis on $M_1$ or $M_2$.

 Otherwise $\racine(M_1\Phi\mydownarrow) \not\in \{\Tag_1, \ldots, \Tag_p\}$. In such a
 case, $\fct_\gamma(M\Phi\mydownarrow) = \{M\Phi\mydownarrow\}$ and $\fctpair(M\Phi\mydownarrow) = \{ M\Phi\mydownarrow\}$. But we assume that $t \not\in \fctpair(M\Phi\mydownarrow)$ hence this case is impossible.

\smallskip{}

 \emph{Case $\ffun = \h$:} This case is analogous to the previous one and can
 be handled similarly.

\smallskip{}

 \emph{Case $\ffun \in \{\sdec, \adec, \checksign\}$:} In such a case,
 we have to distinguish two cases depending on whether $\ffun$ is reduced in
 $M\Phi\mydownarrow$, or not.

 If $\ffun$ is not reduced, \emph{i.e.} $M\Phi\mydownarrow =
 \ffun(M_1\Phi\mydownarrow, M_2\Phi\mydownarrow)$, then we have that
 \[
\fct_\gamma(M\Phi\mydownarrow) = \{M\Phi\mydownarrow\}.
\] 
Once again this is in contradiction with our hypothesis that $t \not\in \fctpair(M\Phi\mydownarrow)$.

We now focus on the case where $\ffun$ is reduced: we have that $M_1\Phi\mydownarrow =
 \ffun'(u_1,u_2)$ with $M\Phi\mydownarrow = u_1$ and $\ffun' \in \{\senc,
 \aenc, \sign\}$. We have to do a case analysis on $\racine(u_1)$:
 \begin{itemize}
 \item if $\racine(u_1) = \Tag_i$ for some $i \in \gamma$. In such a case, there exists $u'_1$ such that $u_1 = \Tag_i(u'_1)$, $\fct_\gamma(M\Phi\mydownarrow) = \fct_\gamma(u'_1)$ and $\fct_\gamma(M_1\Phi\mydownarrow) = \fct_\gamma(u'_1) \cup \fct_\gamma(u_2)$. Thus we deduce that $\fct_\gamma(M\Phi\mydownarrow) \subseteq \fct_\gamma(M_1\Phi\mydownarrow)$. We can conclude thanks to our inductive hypothesis on $t$ and $M_1$. 
\item if $\racine(u_1) = \Tag_i$ for some $i \not\in \gamma$. In such a case, $\fct_\gamma(M\Phi\mydownarrow) = \{ M\Phi\mydownarrow)$ which contradicts the hypothesis $t \not\in \fctpair(M\Phi\mydownarrow)$.
\item otherwise, $\racine(u_1) \not\in \{\Tag_1,\ldots,\Tag_p\}$, then we
have that $\ffun'(u_1,u_2) \in \Flawed{M_1\Phi\mydownarrow}$. By Lemma~\ref{lem:flawed,smallerrecipe}, we deduce that there exists $M'$ such that $\M(M') < \M(M_1)$, $\fn(M') \cap \Ec = \emptyset$ and $M'\Phi\mydownarrow = u_1$. Since $u_1 = M\Phi\mydownarrow$ and $\M(M') < \M(M)$ then we can apply our inductive hypothesis on $t, \alpha$ and $M'$ and so the result holds.
\end{itemize}
\end{proof}


\begin{lemma}
  \label{lem : deductibily of fct_C}
 {Let $A = \quadruple{\Ec}{\p}{\Phi}{\sigma}$ be a derived well-tagged process, and let $(\rho_\alpha, \rho_\beta)$ be compatible with $A$.} Let $u$ be a ground term in normal form that do not use names in $\Ec_\alpha \uplus \Ec_\beta$. We have that there exists a context $C$ (possibly a hole) built only using $\langle\; \rangle$, and terms $u_1, \ldots, u_m$ such that $u = C[u_1, \ldots, u_m]$, and for all $i \in \{1, \ldots, m\}$,
  \begin{itemize}
  \item either $u_i \in \Flawed{u}$;
  \item or $u_i \in \fct_{\Sigmazero}(u)$ and $\delta_\alpha(u_i) =
   \delta_\beta(u_i)$,
  \item or $u_i = \ffun(n)$ for some $\ffun \in \{\pk, \vk\}$ and $n \in \N$,
  \item or $u_i \in \dom(\rho^+_\alpha) \cup \dom(\rho^+_\beta)$.
  \end{itemize}
\end{lemma}

\begin{proof}
 Let $u$ a ground term in normal form and let $\{v_1, \ldots, v_n\} =
 \fct_{\Sigmazero}(u)$. Thus there exists a context $D$ (possibly a hole) built
 on $\Sigmazero$ such that $u = D[v_1, \ldots, v_n]$. We now prove the result by
 induction on $|D|$.

 \medskip

 \noindent \emph{Base case $|D| = 0$:} 
We show that the result holds and in such a case the context $C$ is
reduced to a hole.
Since $|D| = 0$, we know that $\fct_{\Sigmazero}(u) = u$ and so
 either $\racinebis(u) = i$ with $i \in \{1,\ldots,p\}$ or
 $\racinebis(u) = \bot$. 
If $u \in \dom(\rho^+_\alpha) \cup \dom(\rho^+_\beta)$, then the
result trivially holds. Otherwise, we have 
that $\delta_\alpha(u) = \delta_\beta(u)$ by definition of
$\delta_\alpha$ and $\delta_\beta$. Hence the result holds.

 \medskip

 \noindent \emph{Inductive step $|D| > 0$:} There exists $\ffun \in
 \Sigmazero$, and $v_1, \ldots, v_k$ such that $u = \ffun(u_1, \ldots,
 u_k)$. We do a case analysis on $\ffun$.

\smallskip{}

 \emph{Case $\ffun = \langle\;\rangle$:} In such a case, there exist
 two contexts $D_1,
 D_2$  (possibly holes) built on $\Sigmazero$ such that:
\begin{itemize}
\item  $D =
 \pair{D_1}{D_2}$ with $|D_1|, |D_2| < |D|$, 
\item $u_1 = D_1[v^1_1, \ldots, v^1_{n_1}]$ and $\{v^1_1, \ldots,
 v^1_{n_1}\} = \fct_{\Sigmazero}(u_1)$,
\item $u_2 = D_1[v^2_1, \ldots, v^2_{n_1}]$ and $\{v^2_1, \ldots,
 v^2_{n_2}\} = \fct_{\Sigmazero}(u_2)$
\end{itemize}
By applying our inductive hypothesis on $u_1$ and $u_2$, we know that
there exist two contexts $C_1$ and $C_2$.
Since 
\begin{itemize}
\item $\Flawed{u} =
   \Flawed{u_1} \cup \Flawed{u_2}$, and 
\item  $\fct_{\Sigmazero}(u) =
   \fct_{\Sigmazero}(u_1) \uplus \fct_{\Sigmazero}(u_2)$, 
\end{itemize}
we conclude
   that $C = \langle C_1,C_2 \rangle$ satisfies all the conditions
   stated in the lemma.

\smallskip{}

 \emph{Case $\ffun \in \{\pk, \vk\}$ and $u = \ffun(n)$ for some $n \in N$:}
 The result trivially hold by choosing the context $C$ to be a hole.

\smallskip{}

Otherwise, we have that 
\[
\Flawed{u} = \{u\} \cup \Flawed{u_1} \cup \ldots \cup 
\Flawed{u_k}.
\] 
Since $u \in \Flawed{u}$, we can choose $C$ to be the context reduced
to a hole. 
The result trivially holds.
\end{proof}


\begin{lemma}
\label{lem:samerecipesymmetric}
Let $A = \quadruple{\Ec}{\p}{\Phi}{\sigma}$ be a derived well-tagged extended process, and let $(\rho_\alpha, \rho_\beta)$ be compatible with $A$.
Let  $M$ be a term such that $\fv(M) \subseteq \dom(\Phi)$ and $\fn(M) \cap \Ec =
\emptyset$. We assume that 
$\Ec =  \Ec_0 \uplus \Ec_\alpha
\uplus \Ec_\beta$, 
 $\fn(\Phi) \cap (\Ec_\alpha \uplus \Ec_\beta) = \emptyset$, and
one of the two following conditions
is satisfied: 
\begin{enumerate}
\item  $\new\ \Ec. \Phi \not\vdash k$ for any $k \in K_S$; or
\item  $\new\ \Ec. \delta(\Phi\mydownarrow) \not\vdash k$ for any {$k \in \delta_\alpha(K_S) \cup \delta_\beta(K_S)$.} 
\end{enumerate}
with $K_S = \{t, \pk(t),\vk(t) ~|~  \mbox{$t$ ground}, t \in \dom(\rho^+_\alpha) \cup
\dom(\rho^+_\beta)\}$. 
We have that $\delta_\gamma(M\Phi\mydownarrow) =
M\delta(\Phi\mydownarrow)\mydownarrow$ with $\gamma \in \{\alpha,\beta\}$.
\end{lemma}

\begin{proof}
Let $\Phi\mydownarrow =\{w_1 \refer u_1,\ldots, w_n \refer u_n\}$.
We prove this result by induction on $\M(M)$:

\medskip

\noindent\emph{Base case $\M(M) = (0,0)$:} There exists no term $M$ such that $|M| = 0$, thus the result holds.

\medskip

\noindent\emph{Inductive step $\M(M) > (0,0)$:} We first prove there exists $\gamma
\in \{\alpha, \beta\}$ such that $\delta_{\gamma}(M\Phi\mydownarrow) =
M\delta(\Phi\mydownarrow)\mydownarrow$ and then we show that
$\delta_\alpha(M\Phi\mydownarrow) = \delta_\beta(M\Phi\mydownarrow)$.

\medskip{}

Assume first that $|M| = 1$, \emph{i.e.} either $M\in\N$ or there exists $j \in \{1,
\ldots, n\}$ such that $M = w_j$. 

\noindent \emph{Case $M\in \N$.} In such a case, we have that
$M\Phi\mydownarrow = M$, and $M \not\in \Ec$. Hence, we have that 
$\new \ \Ec. \Phi \vdash M$ and also that $\new\
\Ec. \delta(\Phi\mydownarrow) \vdash M$. {In case condition $1$ is
satisfied, we easily deduce that $M \not\in K_S$. Otherwise, we know
that the condition $2$ is satisfied, and thus $M \not\in
\delta_\alpha(K_S) \cup \delta_\beta(K_S)$. Again, we want to conclude
that $M \not\in K_S$. Assume that this is not the case, \emph{i.e.} $M \in
K_S$. This means that $M$ is a name in $\dom(\rho^+_\alpha)$ (or
$\dom(\rho^+_\beta)$). Hence, we  have that $\delta_\beta(M) \in
\delta_\beta(K_S)$, and {$\delta_\beta(M) = M$}. Hence, we
deduce that $M \in \delta_\beta(K_S)$, and this leads to a
contradiction, since in such a case, by hypothesis $M$ can not be deducible from $\new\
\Ec. \delta(\Phi\mydownarrow)$.} 
Thus, in any case, we have that  $M \not\in K_S$, and
thus $M \not\in \dom(\rho^+_\alpha) \cup \dom(\rho^+_\beta)$.
Hence, we have that 
$\delta_\gamma(M\Phi\mydownarrow) =
\delta_\gamma(M) = M = M\delta(\Phi\mydownarrow)\mydownarrow$ for any $\gamma\in\{\alpha, \beta\}$.

\noindent \emph{Case $M = w_j$ for some $j \in \{1,\ldots,n\}$.} 
We know that $w_j$ is colored with $\gamma \in \{\alpha,\beta\}$. Hence, we have
that  $w_j\delta(\Phi\mydownarrow) =
\delta_\gamma(w_j\Phi\mydownarrow)$. Since $u_j$ is in normal form, then by
Lemma~\ref{lem:deltaandnormalform}, 
we know that $\delta_\gamma(w_j\Phi)$ is also in
normal form. Thus, we have that $\delta_\gamma(M\Phi\mydownarrow) =
M\delta(\Phi\mydownarrow)\mydownarrow$.

\medskip{}

Otherwise, if $|M| > 1$, then there exists a symbol $\ffun$ and $M_1, \ldots,
M_n$ such that $M = \ffun(M_1, \ldots, M_n)$. We do a case analysis on $\ffun$.

\smallskip{}

\emph{Case $\ffun \in \Sigma_i\cup \Sigma_{\Tag_i}$ with $i \in \{ 1,
  \ldots, p\}$.} Consider $\gamma \in \{\alpha,\beta\}$ such that $i \in \gamma$. In such a case, let $t = \ffun(M_1\Phi\mydownarrow, \ldots,
M_n\Phi\mydownarrow)$. Since $\ffun \in \Sigma_i$ (resp. $\Sigma_{\Tag_i}$), then there exists a
context $C$ built upon $\Sigma_i$ (resp. $\Sigma_{\Tag_i}$) such that $t = C[u_1, \ldots, u_m]$ and
$u_1, \ldots, u_m$ are factor of $t$ in normal form. By
Lemma~\ref{lem:app_CRicalp05}, we know that there exists a context $D$ (possibly
a hole) over $\Sigma_i$ (resp. $\Sigma_{\Tag_i}$) such that $t\mydownarrow = D[u_{i_1}, \ldots,
u_{i_k}]$ with $i_1, \ldots, i_k \in \{0, \ldots, m\}$ and $u_0 = n_{min}$. But
thanks to Lemma~\ref{lem:change alien}, ~\ref{lem:transfoV} and~\ref{lem:deltaandnormalform}, we also that $C[\delta_\gamma(u_1), \ldots,
\delta_\gamma(u_m)]\mydownarrow = D[\delta_\gamma(u_{i_1}), \ldots,
\delta_\gamma(u_{i_k})]$. But $C$ and $D$ are both built on $\Sigma_i$ (resp. $\Sigma_{\Tag_i}$), thus by definition of $\delta_\gamma$, we have that
$\delta_\gamma(t)\mydownarrow = C[\delta_\gamma(u_1), \ldots,
\delta_\gamma(u_m)]\mydownarrow$ and $\delta_\gamma(t\mydownarrow) =
D[\delta_\gamma(u_{i_1}), \ldots, \delta_\gamma(u_{i_k})]$. Hence, the
equality, $\delta_\gamma(t\mydownarrow) = \delta_\gamma(t)\mydownarrow$,
holds. But $t\mydownarrow = M\Phi\mydownarrow$ which means that
$\delta_\gamma(M\Phi\mydownarrow) = \delta_\gamma(t)\mydownarrow$.
We have that:
 \[
\begin{array}{rcl}
\delta_\gamma(t)\mydownarrow &=&
\delta_\gamma(\ffun(M_1\Phi\mydownarrow, \ldots,
M_n\Phi\mydownarrow))\mydownarrow \\ &=&
\ffun(\delta_\gamma(M_1\Phi\mydownarrow), \ldots,
\delta_\gamma(M_n\Phi\mydownarrow))\mydownarrow
\end{array}
\]
Since $\M(M_1) <
\M(M)$, \ldots, $\M(M_n) < \M(M)$, we can apply our inductive hypothesis
on $M_1, \ldots, M_n$. This gives us $\delta_\gamma(t)\mydownarrow =
\ffun(M_1\delta(\Phi\mydownarrow)\mydownarrow, \ldots,
M_n\delta(\Phi\mydownarrow)\mydownarrow)\mydownarrow =  \ffun(M_1,\ldots,
M_n)\delta(\Phi\mydownarrow)\mydownarrow$. Thus we can conclude that
$\delta_\gamma(M\Phi\mydownarrow) = \delta_\gamma(t)\mydownarrow = M\delta(\Phi\mydownarrow)\mydownarrow$.

\smallskip{}

\emph{Case $\ffun \in \Sigmazero \smallsetminus \{\sdec, \adec, \checksign\}$:} In
this case, we have that $M\Phi\mydownarrow = \ffun(M_1\Phi\mydownarrow,\allowbreak \ldots,
M_n\Phi\mydownarrow)$. By applying our inductive hypothesis on $M_1, \ldots,
M_n$, we have that 
\begin{center}
$\delta_\alpha(M_k\Phi\mydownarrow) =
\delta_\beta(M_k\Phi\mydownarrow)$, for all $k \in \{1, \ldots,
n\}$. 
\end{center}
Thus we
have that $\delta_\gamma(M\Phi\mydownarrow) =
\ffun(\delta_{\gamma'}(M_1\Phi\mydownarrow), \ldots,
\delta_{\gamma'}(M_n\Phi\mydownarrow))$ with $\gamma,
\gamma'\in\{\alpha, \beta\}$.
Applying our inductive hypothesis on $M_1, \ldots, M_n$, we
deduce that 
\[\delta_\gamma(M\Phi\mydownarrow) =
\ffun(M_1\delta(\Phi\mydownarrow)\mydownarrow, \ldots,
M_n\delta(\Phi\mydownarrow)\mydownarrow) = M\delta(\Phi\mydownarrow)\mydownarrow.
\]

\smallskip{}

\emph{Case $\ffun \in \{\sdec, \adec, \checksign\}$:} If we first assume that
the root occurence $\ffun$ is not reduced in $M\Phi\mydownarrow$ then the proof
is similar to the previous case. Thus, we focus on the case where the root
occurence of $\ffun$ is reduced, and we consider the case where $\ffun
= \sdec$. The other cases can be done in a similar way.
In such a situation, we know that there exist $v_1, v_2$ such that
$M_1\Phi\mydownarrow = \senc(v_1,v_2)$, $M_2\Phi\mydownarrow = v_2$ and
 $M\Phi\mydownarrow = v_1$. According to the definition of $\delta_\gamma$, we
 know that there exists $\gamma \in \{\alpha,\beta\}$ such that
 $\delta_{\gamma}(\senc(v_1,v_2)) = \senc(\delta_{\gamma}(v_1),
 \delta_{\gamma}(v_2))$. For such $\gamma$, we have that
 $\sdec(\delta_{\gamma}(M_1\Phi\mydownarrow),
 \delta_\gamma(M_2\Phi\mydownarrow))\mydownarrow =
 \delta_\gamma(M\Phi\mydownarrow)$. But by applying our inductive hypothesis on
 $M_1$ and $M_2$, we obtain $\delta_\gamma(M\Phi\mydownarrow) =
 \sdec(M_1\delta(\Phi\mydownarrow)\mydownarrow,
 M_2\delta(\Phi\mydownarrow)\mydownarrow)\mydownarrow =
 M\delta(\Phi\mydownarrow)\mydownarrow$.

\bigskip{}

It remains to prove that $\delta_\alpha(M\Phi\mydownarrow) =
\delta_\beta(M\Phi\mydownarrow)$. We have shown that there exists $\gamma_0 \in \{\alpha,\beta\}$
such that $\delta_{\gamma_0}(M\Phi\mydownarrow) =
M\delta(\Phi\mydownarrow)\mydownarrow$. Thanks to Lemma~\ref{lem : deductibily of fct_C},
we know that there exists a context $C$ built over $\{ \langle\rangle\}$, and
$v_1, \ldots, v_m$ terms such that $M\Phi\mydownarrow = C[v_1, \ldots, v_m]$ and
for all $i \in \{1, \ldots, m\}$:
\begin{itemize}
\item either $v_i \in \Flawed{M\Phi\mydownarrow}$
\item or $v_i \in \fct_{\Sigmazero}(M\Phi\mydownarrow)$ and $\delta_\alpha(v_i) = \delta_\beta(v_i)$.
\item or $v_i = \ffun(n)$ for some $\ffun \in \{\pk, \vk\}$ and $n \in \N$,
\item or $v_i \in \dom(\rho^+_\alpha) \cup \dom(\rho^+_\beta)$.
\end{itemize}

Note that $C$ being built upon $\{ \langle \rangle \}$ means that
$v_i$ is deducible in $\new\ \Ec. \Phi$ for all $i \in \{1, \ldots, m\}$. Furthermore,
since $C[v_1, \ldots, v_m]$ is in normal form,
\[
\delta_{\gamma_0}(M\Phi\mydownarrow) = C[\delta_{\gamma_0}(v_1),\ldots,
\delta_{\gamma_0}(v_m)].
\]
But we have shown that
$\delta_{\gamma_0}(M\Phi\mydownarrow) = M\delta(\Phi\mydownarrow)\mydownarrow$, thus
$\delta_{\gamma_0}(v_i)$ is deducible from $\delta(\Phi\mydownarrow)$, for all $i \in
\{1, \ldots, m\}$. Now, we distinguish several cases depending on
which condition is fullfilled by $v_i$.
\smallskip{}

\emph{Case $v_i \in \Flawed{M\Phi\mydownarrow}$:} There exists $w_1, \ldots,
w_\ell$ terms and a function symbol $\ffun$ such that $v_i = \ffun(w_1, \ldots,
w_\ell)$. By Lemma~\ref{lem:flawed,smallerrecipe}, there exists $N_1,
\ldots, N_\ell$ such that for all $k \in \{1, \ldots, \ell\}$, $\M(N_k) < \M(M)$
and $N_k\Phi\mydownarrow = w_k$. Hence, by applying inductive hypothesis on
$N_1, \ldots, N_\ell$, we obtain that $\delta_\alpha(N_k\Phi\mydownarrow) =
\delta_\beta(N_k\Phi\mydownarrow)$, for all $k \in \{1, \ldots, \ell\}$. Thus,
thanks to $v_i$ being in normal form, we can conclude that $\delta_\alpha(v_i) =
\delta_\beta(v_i)$.

\smallskip{}

\emph{Case $v_i \in \fct_{\Sigmazero}(M\Phi\mydownarrow)$:} In
  such a case, we have that 
$\delta_\alpha(v_i) = \delta_\beta(v_i)$. Hence, we easily conclude.

\smallskip{}

\emph{Case $v_i = \ffun(n)$ for some $\ffun \in \{\pk, \vk\}$ and $n \in \N$:} {By
hypothesis, we know that either $\new\ \Ec. \Phi \not\vdash k$,
for all $k \in K_S$; or $\new\ \Ec. \delta(\Phi\mydownarrow) \not\vdash k$, for
all $k \in \delta_\alpha(K_S) \cup \delta_\beta(K_S)$.
Since we have shown that $v_i$ is deducible from $\new\ \Ec. \Phi$
and $\delta_{\gamma_0}(v_i)$ is deducible from $\new\ \Ec. \delta(\Phi\mydownarrow)$, both
hypotheses imply that $n \not\in \dom(\rho^+_\alpha) \cup
\dom(\rho^+_\beta)$, and so $\delta_\alpha(v_i) =
\delta_\beta(v_i)$.}

\smallskip{}

\emph{Case $v_i \in \dom(\rho^+_\alpha) \cup \dom(\rho^+_\beta)$:} 
{By hypothesis, we know that either $\new\ \Ec. \Phi \not\vdash k$,
for all $k \in K_S$; or $\new\ \Ec. \delta(\Phi\mydownarrow) \not\vdash k$, for
all $k \in \delta_\alpha(K_S) \cup \delta_\beta(K_S)$.
Since we have shown that $v_i$ is deducible from $\new\ \Ec. \Phi$
and $\delta_{\gamma_0}(v_i)$ is deducible from $\new\ \Ec. \delta(\Phi\mydownarrow)$, both
hypotheses imply that $v_i \not\in \dom(\rho^+_\alpha) \cup
\dom(\rho^+_\beta)$ and lead us to a contradiction.}
\end{proof}


\begin{corollary}
\label{cor:deltaandkeyhidden}
Let $A = \quadruple{\Ec}{\p}{\Phi}{\sigma}$ be a derived well-tagged extended
process and let $(\rho_\alpha, \rho_\beta)$ be compatible with $A$, such that $\Ec = \Ec_0 \uplus \Ec_\alpha \uplus \Ec_\beta$, and
$\fn(\Phi) \cap (\Ec_\alpha \uplus \Ec_\beta) = \emptyset$.  The two following
conditions are equivalent:
\begin{enumerate}
\item  $\new\ \Ec. \Phi \not\vdash k$ for any $k \in K_S$; or
\item  $\new\ \Ec. \delta(\Phi\mydownarrow) \not\vdash k$ for any {$k \in \delta_\alpha(K_S) \cup \delta_\beta(K_S)$.} 
\end{enumerate}
with 
$K_S = \{t, \pk(t),\vk(t) ~|~ t \in \dom(\rho^+_\alpha) \cup
\dom(\rho^+_\beta),  \mbox{ $t$ ground}\}$. 
\end{corollary}

\smallskip{}

\begin{proof}
We prove the two implications separately.

\noindent $(2) \Rightarrow (1)$: Let $k \in K_S$ such that $\new\
\Ec. \Phi \vdash k$. In such a case, there exists $M$ such that
$\fv(M) \subseteq \dom(\Phi)$, $\fn(M) \cap \Ec = \emptyset$, and
$M\Phi\mydownarrow = k\mydownarrow$. We assume w.l.o.g. that $k \in
\{t, \pk(t),\vk(t) ~|~ t \in \dom(\rho^+_\alpha) \mbox{ and $t$
  ground}\}$. Let $\gamma \in \{\alpha,\beta\}$. 
By Lemma~\ref{lem:transfoV}, we
have that {$\delta_\gamma(M\Phi\mydownarrow) =
\delta_\gamma(k\mydownarrow)$.}
Thanks to Lemma~\ref{lem:samerecipesymmetric}, we have that
{$\delta_\gamma(M\Phi\mydownarrow) = M\delta(\Phi\mydownarrow)\mydownarrow$}, and by
Definition of {$\delta_\gamma$}, we have that
{$\delta_\gamma(k\mydownarrow) \in \delta_\gamma(K_S)$}. Thus, we deduce that there
exists {$k' \in \delta_\gamma(K_S)$} such that $\new\ \Ec. \delta(\Phi\mydownarrow) \vdash
k'$.

\smallskip{}

\noindent $(1) \Rightarrow (2)$:{Let $k \in \delta_\gamma(K_S)$ with $\gamma \in \{\alpha,\beta\}$}, and $M$ be a term such that $\fv(M) \subseteq \dom(\Phi)$, $\fn(M) \cap \Ec = \emptyset$, and $M\delta(\Phi\mydownarrow)\mydownarrow = k\mydownarrow$. {$k \in \delta_\gamma(K_S)$ implies the existence of $k' \in K_S$ such that $k =\delta_\gamma(k')$, and thus such that $M\delta(\Phi\mydownarrow)\mydownarrow = \delta_\gamma(k')\mydownarrow$. Thanks to Lemma~\ref{lem:samerecipesymmetric}, we have that $\delta_\gamma(M\Phi\mydownarrow)\mydownarrow = \delta_\gamma(k')\mydownarrow$. Now, if $k'\in K_S$ there must exist $k''\in dom(\rho_{\gamma'})$ such that either $k' = k''$, or $k' = \pk(k'')$, or $k'=\vk(k'')$. In any case, because $\rho_{\gamma'}$ is in normal form, we know that $k''\mydownarrow = k''$ and thus that $k'\mydownarrow = k'$. Hence $M\delta(\Phi\mydownarrow)\mydownarrow = \delta_\gamma(k'\mydownarrow)\mydownarrow$. But, then according to Lemma~\ref{lem:deltaandnormalform}, $\delta_\gamma(M\Phi\mydownarrow) = \delta_\gamma(M\Phi\mydownarrow)\mydownarrow = \delta_\gamma(k'\mydownarrow)\mydownarrow = \delta_\gamma(k'\mydownarrow)$. Finally, thanks to Lemma~\ref{lem:transfoV} we can derive that $M\Phi\mydownarrow = k'\mydownarrow = k'$.} This implies that $M\Phi\mydownarrow \in K_S$, and thus there is a term in $K_S$ that is deducible from $\new\, \Ec. \Phi$.
\end{proof}


\begin{corollary}
\label{cor:framestatequiv}
Let $A = \quadruple{\Ec}{\p}{\Phi}{\sigma}$ be a derived extended process and let $(\rho_\alpha, \rho_\beta)$ be compatible with $A$
such
that $\Ec =  \Ec_0 \uplus \Ec_\alpha
\uplus \Ec_\beta$, 
 $\fn(\Phi) \cap (\Ec_\alpha \uplus \Ec_\beta) = \emptyset$, 
and
$\new\ \Ec. \Phi \not\vdash k$ for 
any $k \in K_S$.
We have that
$\new\; \Ec. \Phi \sim
  \new\; \Ec. \delta(\Phi\mydownarrow)$.
\end{corollary}

\begin{proof}
 The proof directly follows from Lemmas~\ref{lem:transfoV}
 and~\ref{lem:samerecipesymmetric}. 
Indeed, $M\Phi\mydownarrow = N\Phi\mydownarrow$ is
 equivalent to $\delta_\gamma(M\Phi\mydownarrow) =
 \delta_\gamma(N\Phi\mydownarrow)$ (thanks to Lemma~\ref{lem:transfoV}), which
 is equivalent to $M\delta(\Phi\mydownarrow)\mydownarrow =
 N\delta(\Phi\mydownarrow)\mydownarrow$ (thanks to Lemma~\ref{lem:samerecipesymmetric}).
\end{proof}

\subsection{Proof of Theorem~\ref{theo:main-main}}
\label{subsec: proof main theorem}

The goal of this section is to prove Theorem~\ref{theo:main-main}. 
We first state and prove two propositions.

Let $S = \quadruple{\Ec_S}{\p_S}{\Phi_S}{\sigma_S}$ 
and $D =  \quadruple{\Ec_D}{\p_D}{\Phi_D}{\sigma_D}$. 
We say that $D = \delta(S)$ if $\Ec_S = \Ec_D$, $\p_D = \delta(\p_S)$,
$\Phi_D\mydownarrow = \delta(\Phi_S\mydownarrow)$, and
$\sigma_D\mydownarrow = \delta(\sigma_S\mydownarrow)$. 


\begin{proposition}
  \label{pro:shared-to-disjoint}
  Let $P_0$ be a plain coloured process without replication and such that $\bn(P_0) = \fv(P_0) = \emptyset$.
  Let $B_0$ be an extended coloured biprocess such that:
  \begin{itemize}
  \item  $S_0 = \quadruple{\Ec_\alpha \uplus \Ec_\beta \uplus \Ec_0}{\TAGG{P_0}}{\emptyset}{\emptyset} \stackrel{\defi}{=} \fst(B_0)$, 
  \item  $D_0 = \quadruple{\Ec_\alpha \uplus \Ec_\beta \uplus \Ec_0}{P'_0}{\emptyset}{\emptyset}
    \stackrel{\defi}{=} \snd(B_0) $, and  
  \item $D_0 =\delta^{\rho^+_\alpha, \rho^+_\beta}(S_0)$ for some $(\rho_\alpha,\rho_\beta)$, and
  \item $D_0$ does not reveal the value of its assignments w.r.t. $(\rho_\alpha,\rho_\beta)$.
  \end{itemize}
 
  For any extended process $S = \quadruple{\Ec_S}{\p_S}{\Phi_S}{\sigma_S}$ 
  such that $S_0 \LRstep{\tr} S$ with $(\rho_\alpha,\rho_\beta)$ compatible with $S$,
  there exists a biprocess~$B$ and an extended process $D = \quadruple{\Ec_D}{\p_D}{\Phi_D}{\sigma_D} $ 
  such that $B_0 \LRstep{\tr}_\bi B$, ${\fst(B) = S}$, $\snd(B) = D$, $D = \delta(S)$ and with $(\rho_\alpha,\rho_\beta)$ compatible with $D$.
\end{proposition}

\begin{proof} 
  Let $\Ec = \Ec_0 \uplus \Ec_\alpha \uplus \Ec_\beta$. 
  We show the result by induction on the length of the derivation. 
  The base case when $S = S_0$ is trivial. We simply conclude by considering $B = B_0$, and $D = D_0$. 
  Now, we assume that $S_0 \LRstep{\tr'} S'$ such that $(\rho_\alpha,\rho_\beta)$ is compatible with $S'$. This means that there exists $S'$, $\tr$ and $\ell$ such that:
  \[
  S_0 \LRstep{\tr} S \lrstep{\ell} S' \mbox{ with $\tr'  = \tr \cdot \ell$}
  \]
  Moreover, we have that $(\rho_\alpha,\rho_\beta)$ is compatible with $S$.

  By induction hypothesis, we have that there exists an extended biprocess $B$ and an extended process $D$ such that $\fst(B) = S$, $\snd(B) = D$, $B_0 \LRstep{\tr}_\bi B$, and $D = \delta(S)$.
  We will show by case analysis on the rule involved in $S \lrstep{\ell} S'$ that exists a biprocess~$B$ and an extended process $D' = \quadruple{\Ec_D'}{\p_D'}{\Phi_D'}{\sigma_D'}$ such that $B_0 \LRstep{\tr}_\bi B'$, ${\fst(B') = S'}$, $\snd(B') = D'$, $D' = \delta(S')$. 
  Then it will remain to prove that $(\rho_\alpha,\rho_\beta)$ compatible with $D'$. To do so, we rely on the fact that $\delta(\sigma_{S'}\mydownarrow) = \sigma_{D'}\mydownarrow$. In particular, by Lemma~\ref{lem:transfoV}, we have that for all assignment variables $z,z'$,  $z\sigma_{S'}\mydownarrow = z'\sigma_{S'}\mydownarrow$ is equivalent to $\delta(z\sigma_{S'}\mydownarrow) = \delta(z'\sigma_{S'}\mydownarrow)$ which is also equivalent to $z\delta(\sigma_{S'}\mydownarrow) = z'\delta(\sigma_{S'}\mydownarrow)$. 
Moreover, since $(\rho_\alpha,\rho_\beta)$ is compatible with $S'$, then for all assignment variable $z \in \dom(\rho_\gamma)$, either $\racinebis(z\sigma_{S'}\mydownarrow) =  \bot$ or $\racinebis(z\sigma_{S'}\mydownarrow)  \not\in \gamma \cup \{ 0\}$. Thus, by Lemma~\ref{lem:deltaandnormalform}, we deduce that either $\racinebis(\delta_\gamma'(z\sigma_{S'}\mydownarrow)) = \bot$ or $\racinebis(\delta_\gamma'(z\sigma_{S'}\mydownarrow)) \not\in \gamma \cup \{ 0\}$. This allows us to conclude that either $\racinebis(z\sigma_{D'}) = \bot$ or $\racinebis(z\sigma_{D'}\mydownarrow) \not\in \gamma \cup \{ 0\}$, and so that $(\rho_\alpha,\rho_\beta)$ is compatible with $D'$.
  
\medskip

Let's now prove the core part of the result. Let $S = \quadruple{\Ec_S}{\p_S}{\Phi_S}{\sigma_S}$ and $S'  = \quadruple{\Ec'_S}{\p'_S}{\Phi_S}{\sigma'_S}$.

  \medskip{}
  
  \noindent\emph{Case of the rule {\sc Out-T}.}
  In such a case, we have that $\Ec'_S = \Ec_S = \Ec$, $\sigma'_S = \sigma_S$, $\p_S = \{\Out(c,\TAG{u}{i})^i.Q\} \uplus \q_S$, $\p'_S = \{Q\} \uplus \q_S$, and $\Phi'_S = \Phi_S \cup \{w_n \refer \TAG{u}{i}\sigma_S\}$.
  Furthermore, we have that $\ell = \new\ w_n. \Out(c,w_n)$, $c \not\in \Ec$, and $n = |\Phi_S| + 1$. 
  Lastly, since $S$ is issued from $\quadruple{\Ec}{\TAGG{P_0}}{\emptyset}{\emptyset}$, we have that $\sigma_S \vDash \TestTag{\TAG{u}{i}}{i}$.

  By hypothesis, we have that $D =\delta(S)$. Hence, we have that:
  \[
  D = \quadruple{\Ec}{\{\Out(c, \delta_\gamma(\TAG{u}{i})).\delta(Q)\} \uplus \delta(\q_S)}{\Phi_D}{\sigma_D}
  \]
  with $\Phi_D\mydownarrow = \delta(\Phi_S\mydownarrow)$, $\sigma_D\mydownarrow = \delta(\sigma_S\mydownarrow)$, and $\gamma \in \{\alpha,\beta\}$ such that $i \in \gamma$.

  Hence, we have that $D \lrstep{\new\, w_n. \Out(c,w_n)} D'$ where
  \[
  D' = \quadruple{\Ec}{\delta(Q) \uplus \delta(\q_S)}{\Phi_D \cup \{w_n \refer \delta_\gamma(\TAG{u}{i})\sigma_D\}}{\sigma_D}.
  \]
  Hence, we have that $B \lrstep{\new\, w_n. \Out(c, w_n)}_\bi B'$ with $\fst(B') = S'$ and ${\snd(B') = D'}$. It remains to show that $D' = \delta(S')$, i.e.
  \[
  (\delta_\gamma(\TAG{u}{i})\sigma_D)\mydownarrow = \delta_\gamma(\TAG{u}{i}\sigma_S\mydownarrow).
  \]
  Since $\sigma_D\mydownarrow = \delta(\sigma_S\mydownarrow)$, we have that:
  \[
  (\delta_\gamma(\TAG{u}{i})\sigma_D)\mydownarrow  = 
  (\delta_\gamma(\TAG{u}{i})\delta(\sigma_S\mydownarrow))\mydownarrow 
  \]
  Let $\gamma'$ be equal to $\alpha$ if $\gamma = \beta$, and equal to $\beta$ if $\gamma = \alpha$. 
  Each variable that occurs in $\TAG{u}{i}$ also occurs in $\dom(\sigma_S)$ and such a variable is either colored with a color in $\gamma$, or an assignation variable $z^{\gamma'}_j$. 
  Thus, we have that $\delta_\gamma(\TAG{u}{i})$ only contains variables that are colored with a color in $\gamma$.  
  Hence, we have that
  \[
  (\delta_\gamma(\TAG{u}{i})\delta(\sigma_S\mydownarrow))\mydownarrow = 
  (\delta_\gamma(\TAG{u}{i})\delta_\gamma(\sigma_S\mydownarrow))\mydownarrow
  \]
  Relying on Lemma~\ref{lem:Testandequality} (note that $\sigma_S\mydownarrow \vDash \TestTag{\TAG{u}{i}}{i}$), we have that:
  \[
  \begin{array}{rclcl}
    (\delta_\gamma(\TAG{u}{i})\sigma_D)\mydownarrow & = &
    (\delta_\gamma(\TAG{u}{i})\delta_\gamma(\sigma_S\mydownarrow))\mydownarrow \\ &= & \delta_\gamma(\TAG{u}{i}(\sigma_S\mydownarrow))\mydownarrow \\
    &=& \delta_\gamma(\TAG{u}{i}(\sigma_S\mydownarrow)\mydownarrow)\\
    &=& \delta_\gamma(\TAG{u}{i}\sigma_S\mydownarrow)
  \end{array}
  \]
  
  \medskip{}

  \noindent \emph{Case of the rule {\sc In}.}
  In such a case, we have that $\Ec'_S = \Ec_S$, $\Phi'_S = \Phi_S$, $\p_S =\{\In(c,x)^i.Q\} \uplus \q_S$, $\p'_S = \{Q\} \uplus \q_S$, $\sigma'_S = \sigma_S \cup \{x \mapsto M\Phi_S\}$, and $\ell = \In(c,M)$ with $c \not\in \Ec_S$, $\fv(M) \subseteq \dom(\Phi_S)$ and $\fn(M) \cap \Ec_S = \emptyset$.
  
  By hypothesis, we have that $D = \delta(S)$. Hence, we have that:
  \[
  D = \quadruple{\Ec}{\{\In(c, x). \delta(Q)\} \uplus \delta(\q_S)}{\Phi_D}{\sigma_D}
  \]
  with $\Phi_D\mydownarrow = \delta(\Phi_S\mydownarrow)$, $\sigma_D\mydownarrow = \delta(\sigma_S\mydownarrow)$.
  Let $\gamma \in \{\alpha,\beta\}$ such that $i \in \gamma$.

  Hence, we have that $D \lrstep{\In(c,M)} D'$ where
  \[
  D' = \quadruple{\Ec}{\delta(Q) \uplus \delta(\q_S)}{\Phi_D}{\sigma_D \cup \{x \mapsto M\Phi_D\}}.
  \]
  Hence, we have that $B \lrstep{\In(c,M)}_\bi B'$ with $\fst(B') = S'$ and $\snd(B') = D'$. It remains to show that $D' = \delta(S')$, i.e. 
  \[
  (M\Phi_D)\mydownarrow = \delta_\gamma(M\Phi_S\mydownarrow).
  \]
By hypothesis, we know that $D_0$ does not reveal the values of its assignment variables w.r.t. $(\rho_\alpha,\rho_\beta)$. 
Hence, for all assignment variable $x$ of color $\alpha$
(resp. $\beta$) in $\dom(\sigma_D)$, for all $k \in \{ k, \pk(k),
\vk(k) \mid k = x\sigma_D \vee k = x\rho_\alpha \text{
  (resp. $x\rho_\beta$)}\}$, $k$ is not deducible in $\new\
\Ec. \Phi_D$. We denote $K$ this set.

\noindent {Let $K_S = \{t,\pk(t), \vk(t)~|~t \in \dom(\rho^+_\alpha)\cup \dom(\rho^+_\beta), \mbox{ $t$ ground}\}$ We know that  $\sigma_D\mydownarrow = \delta(\sigma_S\mydownarrow)$, and by definition of $\rho_\alpha^+$ and $\rho_\beta^+$, we have that $K = \delta_\alpha(K_S) \cup \delta_\beta(K_S)$.} We have also that $\Phi_D\mydownarrow =
\delta(\Phi_S\mydownarrow)$. Hence, we deduce that $\new\
\Ec. \delta(\Phi_S\mydownarrow) \not\vdash k$ for any $k \in
\delta_\alpha(K_S) \cup \delta_\beta(K_S)$ 
This allow us to apply Lemma~\ref{lem:samerecipesymmetric} and thus to obtain that:
  \[
  (M\Phi_D)\mydownarrow =  (M (\Phi_D\mydownarrow))\mydownarrow 
  = (M\delta(\Phi_S\mydownarrow))\mydownarrow
  = \delta_\gamma(M\Phi_S\mydownarrow).
  \]
  
  \medskip{}

  \noindent\emph{Case of the rule {\sc Then}.}
  In such a case, we have that $\Ec'_S = \Ec_S$, $\Phi'_S = \Phi_S$, $\sigma'_S = \sigma_S$, $\p_S = \{P_S\} \uplus \q_S$, and $\p'_S = \{P'_S\} \uplus \q_S$ where $P_S$ and $P'_S$ are as follows:

  \begin{itemize}
  \item \emph{Case a: a test before an output.}
    \[
    \begin{array}{l}
      P_S = \myIf\, \TestTag{\TAG{v}{i}}{i} \, \myThen \, \Out(c,\TAG{v}{i})^i. Q_S\\
      P'_S = \Out(u,\TAG{v}{i})^i. Q_S\\
      \sigma_S \vDash  \TestTag{\TAG{v}{i}}{i}
    \end{array}
    \]
    for some $i \in \{1,\ldots,p\}$.
  \item \emph{Case b: a test before an assignation.}
    \[
    \begin{array}{l}
      P_S = \myIf\, \TestTag{\TAG{v}{i}}{i} \, \myThen \, [z := \TAG{v}{i}]^i. Q_S\\
      P'_S = \{[z := \TAG{v}{i}]^i. Q_S\\
      \sigma_S \vDash  \TestTag{\TAG{v}{i}}{i}
    \end{array}
    \]
    for some $i \in \{1,\ldots,p\}$.
  \item \emph{Case c: a test before a conditional.}
    \[
    \begin{array}{l}
      P_S = \myIf\, \TestTag{\TAG{\varphi}{i}}{i} \, \myThen\, (\myIf\, \TAG{\varphi}{i} \, \myThen\, Q^1_S \, \myElse \, Q^2_S)\\
      P'_S = \myIf\, \TAG{\varphi}{i} \, \myThen\, Q^1_S \, \myElse \, Q^2_S \\
      \sigma_S \vDash  \TestTag{\TAG{\varphi}{i}}{i}
    \end{array}
    \]
    for some $i \in \{1,\ldots, p\}$. 
  \item \emph{Case d: a test of a conditional.} 
    \[
    \begin{array}{l}
      P_S = \myIf\, \TAG{\varphi}{i} \, \myThen \, Q^1_S \, \myElse\, Q^2_S\\
      P'_S = Q^1_S\\
      \sigma_S \vDash  \TAG{\varphi}{i} \mbox{ and } \sigma_S \vDash \TestTag{\TAG{\varphi}{i}}{i}
    \end{array}
    \]
    for some $i \in \{1,\ldots, p\}$. 
  \end{itemize}
  Each case can be handled in a similar way. 
  Note that we rely on Corollary~\ref{cor:equivalenceeq} instead of Lemma~\ref{lem:TestTagequivalentdelta} to establish the result in \emph{Case d.}
  We assume that we are in the first case. 
  Let $\gamma \in \{\alpha,\beta\}$ such that $i \in \gamma$. 
  By hypothesis, we have that $D = \delta(S)$. 
  Hence, we have that $D$ is equal to
  \[
  \begin{array}{@{}l@{}}
  \quadruple{\Ec}{\{\myIf\ \TestTag{\delta_\gamma(\TAG{v}{i})}{i} \, \myThen\\
  \quad\quad\quad\Out(c,\delta_\gamma(\TAG{v}{i})). \delta(Q_S)\} \uplus \q_S}{\Phi_D}{\sigma_D}\\
  \end{array}
  \]
  with $\Phi_D\mydownarrow = \delta(\Phi_S\mydownarrow)$, and $\sigma_D\mydownarrow = \delta(\sigma_S\mydownarrow)$.

  Since $\sigma_S \vDash  \TestTag{\TAG{v}{i}}{i}$, we have also that $(\sigma_S\mydownarrow) \vDash \TestTag{\TAG{v}{i}}{i}$.
  Thanks to Lemma~\ref{lem:TestTagequivalentdelta}, we deduce that $\delta_\gamma(\sigma_S\mydownarrow)  \vDash \TestTag{\delta_\gamma(\TAG{v}{i})}{i}$.
  Actually, each variable that occurs in $\TestTag{\delta_\gamma(\TAG{v}{i})}{i}$ is a variable that occurs in $\dom(\sigma_S)$ and such a variable is necessarily colored with a color in $\gamma$. Hence, we have also that:
  \[
  \delta(\sigma_S\mydownarrow)  \vDash \TestTag{\delta_\gamma(\TAG{v}{i})}{i}.
  \]

  Hence, we have that $D \lrstep{\tau} D'$ where 
  \[
  D' = \quadruple{\Ec}{\{\Out(u,\delta_\gamma(\TAG{v}{i})). \delta(Q_S)\} \uplus \delta(\q_S)}{\Phi_D}{\sigma_D}.
  \]
  Hence, we have that $B \lrstep{\tau}_\bi B'$ with $\fst(B') = S'$ and $\snd(B') = D'$.
  We also have that $D' = \delta(S')$.
  
  \medskip{}

  \noindent\emph{Case of the rule {\sc Else}.} This case is similar to the previous one.

  \medskip{}

  \noindent\emph{Case of the rule {\sc Assgn}.}
  In such a case, we have that $\Ec'_S = \Ec_S$, $\Phi'_S = \Phi_S$, $\p_S = \{[x := \TAG{v}{i}].Q\} \uplus \q_S$, $\p'_S = \{Q\} \uplus \q_S$, $\sigma'_S = \sigma_S \cup \{x \mapsto \TAG{v}{i}\sigma_S\}$, and $\ell = \tau$. Lastly, since $S$ is issued from $\quadruple{\Ec}{\TAGG{P_0}}{\emptyset}{\emptyset}$, we have that $\sigma_S \vDash \TestTag{\TAG{v}{i}}{i}$.

  By hypothesis, we have that $D = \delta(S)$. Hence, we have that:
  \[
  D = \quadruple{\Ec}{[x:=\delta(\TAG{v}{i})].\delta(Q) \uplus \delta(\q_S)}{\Phi_D}{\sigma_D}
  \] 
  with $\Phi_D\mydownarrow = \delta(\Phi_S\mydownarrow)$, $\sigma_D\mydownarrow = \delta(\sigma_S\mydownarrow)$, and $\gamma \in \{\alpha,\beta\}$ such that $i \in \gamma$.

  Hence, we have that $D \lrstep{\tau} D'$ where
  \[
  D' = \quadruple{\Ec}{\delta(Q) \uplus \delta(\q_S)}{\Phi_D}{\sigma_D \cup \{x \mapsto \delta_\gamma(\TAG{v}{i})\sigma_D\}}.
  \]
  Hence, we have that $B \lrstep{\tau} B'$ with $\fst(B')= S'$ and $\snd(B') = D'$. It remains to show that $D' = \delta(S')$, i.e.
  \[
  (\delta_\gamma(\TAG{v}{i})\sigma_D)\mydownarrow = \delta_\gamma(\TAG{v}{i}\sigma_S\mydownarrow).
  \]
  This can be done as in the case of the rule {\sc Out-T}.

  \medskip{}

  \noindent\emph{Case of the rule {\sc Comm}.}
  In such a case, we have that $\Ec'_S = \Ec_S$, $\Phi'_S = \Phi_S$, $\p_S = \{\Out(c,\TAG{u}{i})^{i}.Q_1; \In(c,x)^{i'}.Q_2\} \uplus \q_S$, $\sigma'_S = \sigma_S \cup \{x \mapsto \TAG{u}{i}\sigma_S\}$, and $\ell = \tau$. 
  Lastly, since $S$ is issued from $\quadruple{\Ec}{\TAGG{P_0}}{\emptyset}{\emptyset}$, we have that $\sigma_S \vDash \TestTag{\TAG{u}{i}}{i}$.

  By hypothesis, we have that $D = \delta(S)$. Hence, we have that $D$ is equal to
  \[
  \quadruple{\Ec}{\{\Out(c,\delta_\gamma(\TAG{u}{i})).\delta(Q_1); \In(c,x).\delta(Q_2)\}\uplus \delta(\q_S)}{\Phi_D}{\sigma_D}
  \]
  with $\Phi_D\mydownarrow = \delta(\Phi_S\mydownarrow)$, $\sigma_D\mydownarrow = \delta(\sigma_S\mydownarrow)$.

  Let $\gamma,\gamma' \in \{\alpha,\beta\}$ such that $i \in \gamma$, and $i' \in \gamma'$.
  Hence, we have that $D \lrstep{\tau} D'$ where $D'$ is equal to:
  \[
  \quadruple{\Ec}{\{\delta(Q_1); \delta(Q_2)\} \uplus \delta(\q_S)}{\Phi_D}{\sigma_D \cup \{(x \mapsto \delta_\gamma(\TAG{u}{i})\sigma_D)^{i'}\}}.
  \]

  Hence, we have that $B \lrstep{\tau} B'$ for some biprocess $B'$ such that $\fst(B') = S'$ and $\snd(B')=D'$.
  It remains to show that $D' = \delta(S')$, i.e.
  \[
  (\delta_\gamma(\TAG{u}{i})\sigma_D)\mydownarrow =
  \delta_{\gamma'}((\TAG{u}{i}\sigma_S)\mydownarrow)
  \]
  If $\gamma = \gamma'$, then this can be done as in the previous cases. 

  Otherwise, since names can only be shared through assignments, and assignments only concern variables/terms of base type, we necessarily have that $c \not\in \Ec$. 
  Hence, we have that $S \lrstep{\nu w_n. \Out(c,w_n)} S_\Out$ where:
  \[
  S_\Out = \quadruple{\Ec}{\{Q_1; \In(c,x).Q_2\} \uplus \q_S}{\Phi_S \cup \{w_n \refer \TAG{u}{i}\sigma_S\}}{\sigma_S}
  \]
  Note that $(\rho_\alpha,\rho_\beta)$ is still compatible with $S_\Out$. We would like to apply Lemma~\ref{lem:samerecipesymmetric} with $M = w_n$ on the frame of $S_\Out$ which requires an hypothesis of non deductibility of the shared key. For these, we will rely on our hypothesis that $D_0$ does not reveal the values of its assignments w.r.t. $(\rho_\alpha,\rho_\beta)$:
  
Let $\Phi'_S = \Phi_S \cup \{w_n \refer \TAG{u}{i}\sigma_S\}$. We
already proved our induction result for the rule {\sc Out-T}. 
Hence, we deduce that there exists $D_\Out$ such that $D \lrstep{\nu
  w_n. \Out(c,w_n)} D_\Out$ where $D_\Out = \quadruple{\Ec}{\p_\Out}{\Phi'_D}{\sigma_D}$, $\Phi'_D = \Phi_D \cup \{w_n \refer \delta_\gamma(\TAG{u}{i})\sigma_D\}$. Moreover, it implies that $\Phi'_D\mydownarrow = \delta(\Phi'_S\mydownarrow\})$ and $\sigma_D\mydownarrow = \delta(\sigma_S\mydownarrow)$.
As mentioned, by hypothesis, we know that $D_0$ does not reveal the values of its assignments w.r.t. $(\rho_\alpha,\rho_\beta)$. 
Hence, for all assignment variable $x$ of color $\alpha$
(resp. $\beta$) in $\dom(\sigma_D)$, for all $k \in \{ k, \pk(k),
\vk(k) \mid k = x\sigma_D \vee k = x\rho_\alpha \text{
  (resp. $x\rho_\beta$)}\}$, $k$ is not deducible in $\new\
\Ec. \Phi'_D$. We denote by $K$ such a set.

\noindent {Let $K_S = \{t, \pk(t), \vk(t)~|~t \in \dom(\rho^+_\alpha)
\cup \dom(\rho^+_\beta), \mbox{ $t$ ground}\}$.
Since  $\sigma_D\mydownarrow = \delta(\sigma_S\mydownarrow)$, and by
definition of $\rho_\alpha^+$ and $\rho_\beta^+$, we deduce that $K =
\delta_\alpha(K_S) \cup \delta_\beta(K_S)$. Moreover, we have that
$\Phi'_D\mydownarrow = \delta(\Phi'_S\mydownarrow)$. Hence, we deduce
that $\new\ \Ec. \delta(\Phi'_S\mydownarrow) \not\vdash k$ for any $k
\in \delta_\alpha(K_S) \cup \delta_\beta(K_S)$.}
This allow us to apply Lemma~\ref{lem:samerecipesymmetric} with $M = w_n$ and so we deduce that $\delta_\gamma(\TAG{u}{i}\sigma_S\mydownarrow) = \delta_{\gamma'}(\TAG{u}{i}\sigma_S\mydownarrow)$.
Hence, we can conclude as in the previsous case.

  \medskip{}

  \noindent\emph{Case of the rule {\sc Par}.} It is easy to see that the result holds for this case.

  \medskip{}
  Note that the rules {\sc New} and {\sc Repl} can not be triggered since the processes under study do not contain bounded names and replication.
\end{proof}



\begin{proposition}
  \label{pro:disjoint-to-shared}
  Let $P_0$ be a plain colored process without replication and such that $\bn(P_0) = \fv(P_0) = \emptyset$.
  Let $B_0$ be an extended colored biprocess such that:
  \begin{itemize}
  \item  $S_0 = \quadruple{\Ec_\alpha \uplus \Ec_\beta \uplus \Ec_0}{\TAGG{P_0}}{\emptyset}{\emptyset} \stackrel{\defi}{=} \fst(B_0)$, 
  \item  $D_0 = \quadruple{\Ec_\alpha \uplus \Ec_\beta \uplus \Ec_0}{P'_0}{\emptyset}{\emptyset}
    \stackrel{\defi}{=} \snd(B_0) $, and  
  \item $D_0 =\delta^{\rho^+_\alpha, \rho^+_\beta}(S_0)$ for some $(\rho_\alpha,\rho_\beta)$.
   \item $D_0$ does not reveal the value of its assignments w.r.t. $(\rho_\alpha,\rho_\beta)$.
  \end{itemize}
  
  For any extended process $D = \quadruple{\Ec_D}{\p_D}{\Phi_D}{\sigma_D}$ 
  such that  $D_0 \LRstep{\tr} D$ with $(\rho_\alpha,\rho_\beta)$  compatible with $D$, 
  there exists a  biprocess $B$ and an extended process $S = \quadruple{\Ec_S}{\p_S}{\Phi_S}{\sigma_S}$ 
  such that $B_0 \LRstep{\tr}_\bi B$, ,${\fst(B) = S}$, $\snd(B) = D$, and $D = \delta(S)$.
\end{proposition}

\begin{proof}
  We show the result by induction on the length of the derivation. 
  The base case when $D_0 = D$ is trivial. We simply conclude by considering $B  = B_0$, and $S = S_0$. 
  Now, we assume that $D_0 \LRstep{\tr'} D'$ such that $(\rho_\alpha,\rho_\beta)$ is compatible with $D'$. This means that there exist $D$, $\tr$, and $\ell$ such that:
  \[
  D_0 \LRstep{\tr} D \lrstep{\ell} D' \mbox{ with } \tr' = \tr\cdot \ell
  \]
  Note that we necessarily have that $(\rho_\alpha,\rho_\beta)$ is compatible with~$D$.

  By induction hypothesis, we have that there exists an extended biprocess $B$ and an extended process $S$ such that $\fst(B) = S$, $\snd(B) =D$, $B_0 \LRstep{\tr}_\bi B$, and $D = \delta(S)$.
  We show the result by case analysis on the rule involved in $D \lrstep{\ell} D'$. 
  Let $D = \quadruple{\Ec_D}{\p_D}{\Phi_D}{\sigma_D}$ and $D' = \quadruple{\Ec'_D}{\p'_D}{\Phi'_D}{\sigma'_D}$.
  First, note that since $D$ is issued from $D_0 = \delta(S_0)$ and $S_0 = \quadruple{\Ec}{\TAGG{P_0}}{\emptyset}{\emptyset}$, we know that terms invovled in $D$ are tagged and obtained through the $\delta$ transformation.

  \medskip{}
  
  \noindent\emph{Case of the rule {\sc Out-T}.} 
  In such a case, we have that $\Ec'_D = \Ec_D$, $\sigma'_D = \sigma_D$, $\p_D = \{\Out(c,\delta(\TAG{v}{i})).\delta(Q_S)\}\uplus \delta(\q_S)$, $\p'_D = \{\delta(Q_D)\} \uplus \delta(\q_D)$, and $\Phi'_D = \Phi_D \cup \{w_n \refer \delta_\gamma(\TAG{v}{i})\sigma_D\}$ with $\gamma \in \{\alpha,\beta\}$ such that $i \in \gamma$.
  Furthermore, we have that $\ell = \new\, w_n.\Out(c,w_n)$, $c \not\in \Ec_D$, and $n = |\Phi_D|+1$. 
  We have also 
  \[
  S = \quadruple{\Ec}{\{\Out(c,\TAG{v}{i}).Q_S\} \uplus \q_S)}{\Phi_S}{\sigma_S}
  \]
  with $\Phi_D\mydownarrow = \delta(\Phi_S\mydownarrow)$, and $\sigma_D\mydownarrow = \delta(\sigma_S\mydownarrow)$.
  Hence, we have that $S \lrstep{\new\, w_n. \Out(c,w_n)} S'$ where
  \[
  S' = \quadruple{\Ec}{Q_S \uplus \q_S}{\Phi_S \cup \{w_n \refer \TAG{v}{i}\sigma_S\}}{\sigma_S}.
  \]

  Hence, we have that $B \lrstep{\new \, w_n. \Out(c, w_n)}_\bi B'$ with $\fst(B') = S'$ and $\snd(B') = D'$. 
  It remains to show that $D' = \delta(S')$, i.e.
  \[
  (\delta_\gamma(\TAG{v}{i})\sigma_D)\mydownarrow = \delta_\gamma(\TAG{v}{i}\sigma_S\mydownarrow)
  \]
  Since $D$ is issued from $\quadruple{\Ec}{\delta(\TAGG{P_0})}{\emptyset}{\emptyset}$ and $B_0 \LRstep{\tr}_\bi B$, we have that  $\sigma_D \vDash \TestTag{\delta_\gamma(\TAG{v}{i})}{i}$ and $\sigma_S\vDash \TestTag{\TAG{v}{i}}{i}$.

  Since $\sigma_D\mydownarrow = \delta(\sigma_S\mydownarrow)$, we have that:
  \[
  (\delta_\gamma(\TAG{v}{i})\sigma_D)\mydownarrow  = 
  (\delta_\gamma(\TAG{v}{i})\delta(\sigma_S\mydownarrow))\mydownarrow 
  \]
  Let $\gamma'$ be equal to $\alpha$ if $\gamma = \beta$, and equal to $\beta$ if $\gamma = \alpha$. 
  Each variable that occurs in $\TAG{v}{i}$ also occurs in $\dom(\sigma_S)$ and such a variable is either  colored with a color in $\gamma$, or an assignation variable $z^{\gamma'}_j$. 
  Thus, we have that $\delta_\gamma(\TAG{v}{i})$ only contains variables that are colored with a color in $\gamma$.
  Hence, we have that
  \[
  (\delta_\gamma(\TAG{v}{i})\delta(\sigma_S\mydownarrow))\mydownarrow =
  (\delta_\gamma(\TAG{v}{i})\delta_\gamma(\sigma_S\mydownarrow))\mydownarrow
  \]
  Relying on Lemma~\ref{lem:Testandequality} (note that $\sigma_S\mydownarrow \vDash \TestTag{\TAG{v}{i}}{i}$), we have that:
  \[
  \begin{array}{rclcl}
    (\delta_\gamma(\TAG{v}{i})\sigma_D)\mydownarrow & = &
    (\delta_\gamma(\TAG{v}{i})\delta_\gamma(\sigma_S\mydownarrow))\mydownarrow \\ &= & \delta_\gamma(\TAG{v}{i}(\sigma_S\mydownarrow))\mydownarrow \\
    &=& \delta_\gamma(\TAG{v}{i}(\sigma_S\mydownarrow)\mydownarrow)\\
    &=& \delta_\gamma(\TAG{v}{i}\sigma_S\mydownarrow)
  \end{array}
  \]
  
  \medskip{}

  \noindent\emph{Case of the rule {\sc In}.} 
  In such a case, we have that $\Ec'_D = \Ec_D$, $\Phi'_D = \Phi_D$, $\p_D = \{\in(c,x)^i.\delta(Q_S\} \uplus \delta(\q_S)$, $\p'_D = \{\delta(Q_S)\} \uplus \delta(\q_S)$, $\sigma'_D = \sigma_D \cup \{x \mapsto M\Phi_D\}$, and $\ell = \In(c,M)$ with $c \not\in \Ec_D$, $\fv(M) \subseteq \dom(\Phi_D)$, and $\fn(M) \cap \Ec_D = \emptyset$. 
  Moreover, we have that:
  \[
  S = \quadruple{\Ec}{\{\In(c,x).Q_S\}\uplus \q_S}{\Phi_S}{\sigma_S}
  \]
  with $\Phi_D\mydownarrow = \delta(\Phi_S\mydownarrow)$, and $\sigma_D\mydownarrow = \delta(\sigma_S\mydownarrow)$. Let $\gamma \in \{\alpha,\beta\}$ such that $i \in \gamma$.
  
  Hence, we have that $S \lrstep{\In(c,M)} S'$ where
  \[
  S' = \quadruple{\Ec}{\{Q_S\}\uplus\q_S}{\Phi_S}{\sigma_S \cup \{x \mapsto M\Phi_S}.
  \]
  Hence, we have that $B \lrstep{\In(c,M)}_\bi B'$ with $\fst(B') = S'$, and $\snd(B') = D'$. It remains to show that $D' = \delta(S')$, i.e.
  \[
  (M\Phi_D)\mydownarrow = \delta_\gamma(M\Phi_S\mydownarrow).
  \]
  By hypothesis, we know that $D_0$ does not reveal the value of its assignments w.r.t. $(\rho_\alpha,\rho_\beta)$. Since $\Phi_D\mydownarrow = \delta(\Phi_S\mydownarrow)$ and $\sigma_D\mydownarrow = \delta(\sigma_S\mydownarrow)$. Hence, by following the definition of $\rho_\alpha^+$ and $\rho_\beta^+$, we deduce that the hypothesis of Lemma~\ref{lem:samerecipesymmetric} are satisfied. Hence, by relying on it, we have that:
  \[
  (M\Phi_D)\mydownarrow = (M(\Phi_D\mydownarrow))\mydownarrow =
  (M\delta(\Phi_S\mydownarrow))\mydownarrow = \delta_\gamma(M\Phi_S\mydownarrow).
  \] 

  \medskip{}

  \noindent \emph{Case of the rule {\sc Then}.} 
  In such a case, we have that $\Ec'_D = \Ec_D$, $\Phi'_D = \Phi_D$, $\sigma'_D = \sigma_D$, $\p_D = \{P_D\} \uplus \q_D$, and $\p'_D = \{P'_D\}\uplus \q_D$ where $P_D$ and $P'_D$ are as follows:
  
  \begin{itemize}
  \item \emph{Case a: a test before an output.}
    \[
    \begin{array}{l}
      P_D = \myIf\, \TestTag{\delta_\gamma(\TAG{v}{i})}{i} \, \myThen \, \Out(c,\delta_\gamma(\TAG{v}{i}))^i. \delta(Q_S)\\
      P'_D = \Out(u,\delta_\gamma(\TAG{v}{i}))^i. \delta(Q_S)\\
      \sigma_D \vDash  \TestTag{\delta_\gamma(\TAG{v}{i})}{i}
    \end{array}
    \]
    for some $i \in \{1,\ldots,p\}$ and $\gamma \in \{\alpha,\beta\}$ such that $i \in \gamma$.
  \item \emph{Case b: a test before an assignation.}
    \[
    \begin{array}{l}
      P_D = \myIf\, \TestTag{\delta_\gamma(\TAG{v}{i})}{i} \, \myThen \, [z := \delta_\gamma(\TAG{v}{i})]^i. \delta(Q_S)\\
      P'_D = \{[z := \delta_\gamma(\TAG{v}{i})]^i. \delta(Q_S)\\
      \sigma_D \vDash  \TestTag{\delta_\gamma(\TAG{v}{i})}{i}
    \end{array}
    \]
    for some $i \in \{1,\ldots,p\}$ and $\gamma \in \{\alpha,\beta\}$ such that $i \in \gamma$.
  \item \emph{Case c: a test before a conditional.}
    \[
    \begin{array}{l}
      P_D = \myIf\, \TestTag{\delta_\gamma(\TAG{\varphi}{i})}{i} \, \myThen\, \\
      \quad\quad\quad\quad(\myIf\, \TAG{\varphi}{i} \, \myThen\, \delta(Q^1_S) \, \myElse \, \delta(Q^2_S))\\
      P'_D = \myIf\, \TAG{\varphi}{i} \, \myThen\, \delta(Q^1_S) \, \myElse \, \delta(Q^2_S) \\
      \sigma_D \vDash  \TestTag{\delta_\gamma(\TAG{\varphi}){i}}{i}
    \end{array}
    \]
    for some $i \in \{1,\ldots, p\}$ and $\gamma \in \{\alpha,\beta\}$ such that $i \in \gamma$.
  \item \emph{Case d: a test of a conditional.} 
    \[
    \begin{array}{l}
      P_D = \myIf\, \delta_\gamma(\TAG{\varphi}{i}) \, \myThen \, \delta(Q^1_S) \, \myElse\, \delta(Q^2_S)\\
      P'_D = \delta(Q^1_S)\\
      \sigma_D \vDash  \delta_\gamma(\TAG{\varphi}{i}) \mbox{ and } \sigma_D \vDash \TestTag{\delta_\gamma(\TAG{\varphi}{i})}{i}
    \end{array}
    \]
    for some $i \in \{1,\ldots, p\}$ and $\gamma \in \{\alpha,\beta\}$ such that $i \in \gamma$.
  \end{itemize}
  Each case can be handled in a similar way. 
  Note that we rely in addition on Corollary~\ref{cor:equivalenceeq} instead of Lemma~\ref{lem:TestTagequivalentdelta} to establish the result in \emph{case d}. 
  We assume that we are in the first case. 
  Let $\gamma \in \{\alpha,\beta\}$ such that $i \in \gamma$. 
  We have that $S$ is equal to
  \[
  \quadruple{\Ec}{\{\myIf\, \TestTag{\TAG{v}{i}}{i} \, \myThen \, \Out(c,\TAG{v}{i})^i. Q_S\}\uplus\q_S}{\Phi_S}{\sigma_S}
  \]
  with $\Phi_D\mydownarrow = \delta(\Phi_S\mydownarrow)$, 
  and $\sigma_D\mydownarrow = \delta(\sigma_S\mydownarrow)$.

  Since $\sigma_D \vDash \TestTag{\delta_\gamma(\TAG{v}{i})}{i}$, 
  we have $(\sigma_D\mydownarrow) \vDash  \TestTag{\delta_\gamma(\TAG{v}{i})}{i}$, 
  and thus $\delta(\sigma_S\mydownarrow)  \vDash \TestTag{\delta_\gamma(\TAG{v}{i})}{i}$. 
  As in the previous cases, we deduce that  $\delta_\gamma(\sigma_S\mydownarrow)  \vDash \TestTag{\delta_\gamma(\TAG{v}{i})}{i}$. 
  Thanks to Lemma~\ref{lem:TestTagequivalentdelta}, we deduce that $\sigma_S\mydownarrow \vDash \TestTag{\TAG{v}{i}}{i}$. 
  Hence, we have that $S \lrstep{\tau} S'$ where
  \[
  S' = \quadruple{\Ec}{\{\Out(u,\TAG{v}{i}).Q_S\} \uplus \q_S}{\Phi_S}{\sigma_S}.
  \]
  Hence, we have that $B \lrstep{\tau}_\bi B'$ with $\fst(B') = S'$, and $\snd(B') = D'$. We also have that $D' = \delta(S')$.

  \medskip{}

  \noindent\emph{Case of the rule {\sc Else}.} This case is similar to the previous one.
  
  \medskip{}
  
  \noindent\emph{Case of the rule {\sc Assgn}.} 
  In such a case, we have that $\Ec'_D = \Ec_D$, $\Phi'_D = \Phi_D$, $\p_D = \{[x := \delta_\gamma(\TAG{v}{i})].\delta(Q)\} \uplus \delta(\q_S)$, $\p'_D = \{\delta(Q)\} \uplus \delta(\q_S)$, $\sigma'_D = \sigma_D \cup \{x \mapsto \delta_\gamma(\TAG{v}{i})\sigma_D\}$, and $\ell = \tau$ where $\gamma \in  \{\alpha,\beta\}$ with $i\in \gamma$. 
  We have also that $\sigma_D \vDash \TestTag{\delta(\TAG{v}{i})}{i}$ and $\sigma_S \vDash \TestTag{\TAG{v}{i}}{i}$. 
  Hence, we have that:
  \[
  S = \quadruple{\Ec}{\{[x := \TAG{v}{i}].Q\} \uplus \q_S}{\Phi_S}{\sigma_S}
  \]
  with $\Phi_D\mydownarrow = \delta(\Phi_S\mydownarrow)$, and $\sigma_D\mydownarrow = \delta(\sigma_S\mydownarrow)$.

  Hence, we have that $S \lrstep{\tau} S'$ where:
  \[
  S' = \quadruple{\Ec}{\{Q\} \uplus \q_S}{\Phi_S}{\sigma_S \cup \{x \mapsto \TAG{v}{i}\sigma_S\}}.
  \]
  Hence, we have that $B \lrstep{\tau}_\bi B'$ with $\fst(B') = S'$ and $\snd(B') = D'$. It remains to show that $D ' = \delta(S')$, i.e. 
  \[
  (\delta_\gamma(\TAG{v}{i})\sigma_D)\mydownarrow = \delta_\gamma(\TAG{v}{i}\sigma_S\mydownarrow).
  \]
  This can be done as in the case of the rule {\sc Out-T}.
  
  \medskip{}

  \noindent\emph{Case of the rule {\sc Comm}.}
  In such a case, $\p_D = \{\Out(c,\delta_\gamma(\TAG{u}{i}))^i.\delta(Q_1); \In(c,x)^{i'}.\allowbreak\delta(Q_2)\} \uplus \delta(\q_S)$, $\Ec'_D = \Ec_D$, $\Phi'_D = \Phi_D$, $\sigma'_D = \sigma_D \cup \{x \mapsto \delta_\gamma(\TAG{u}{i})\sigma_D\}$, and $\ell = \tau$. 
  Moreover, we have that $\sigma_D \vDash \delta_\gamma(\TestTag{\TAG{u}{i}}{i})$ and $\sigma_S \vDash \TestTag{\TAG{v}{i}}{i}$ where $\gamma \in \{\alpha,\beta\}$ such that $i \in \gamma$.
  Hence, we have that $S$ is equal to
  \[
  \quadruple{\Ec}{\{\Out(c,\TAG{u}{i}).Q_1; \In(c,x).Q_2\}\uplus \q_S}{\Phi_S}{\sigma_S}
  \]
  with $\Phi_D\mydownarrow = \delta(\Phi_S\mydownarrow)$, and $\sigma_D\mydownarrow = \delta(\sigma_S\mydownarrow)$.

  Let $\gamma' \in \{\alpha,\beta\}$ such that $i' \in \gamma'$.
  Hence, we have that $S \lrstep{\tau} S'$ where $S'$ is equal to:
  \[
  \quadruple{\Ec}{\{Q_1;Q_2\} \uplus \q_S}{\Phi_S}{\sigma_S \cup \{(x \mapsto \TAG{u}{i}\sigma_S)^{i'}\}}.
  \] 
  Hence, we have that $B \lrstep{\tau} B'$ for some biprocess $B'$ such that $\fst(B') = S'$ and $\snd(B') = D'$. It remains to show that $D' =\delta(S')$, i.e.
  \[
  (\delta_\gamma(\TAG{u}{i})\sigma_D)\mydownarrow = \delta_{\gamma'}(\TAG{u}{i}\sigma_S\mydownarrow)
  \]

  If $\gamma = \gamma'$, then this can be done as in the previous cases. 
  Otherwise, since names can only be shared through assignations, and assignations only concern variables/terms of base type, we necessarily have that $c \not\in \Ec$. Hence, we have that $D \lrstep{\nu w_n. \Out(c,w_n)} D_\Out$ where $D_\Out$ is equal to:
  \[
  \quadruple{\Ec_D}{\{\delta(Q_1); \In(c,x).\delta(Q_2)\} \uplus \q_S}{\Phi_D \cup \{w_n \refer \delta_\gamma(\TAG{u}{i})\sigma_D\}}{\sigma_D}
  \]
  Note that $(\rho_\alpha,\rho_\beta)$ is still compatible with $D_\Out$. We would like to apply Lemma~\ref{lem:samerecipesymmetric} with $M = w_n$ on the frame of $D_\Out$ which requires an hypothesis of non deductibility of the shared key. For these, we will rely on our hypothesis that $D_0$ does not reveal the values of his assignment variables w.r.t. $(\rho_\alpha,\rho_\beta)$:

Let  $\Phi'_D = \Phi_D \cup \{w_n \refer
\delta_\gamma(\TAG{u}{i})\sigma_D\}$. 
We already proved our induction result for the rule {\sc
  Out-T}. Hence, we deduce that there exists $S_\Out$ such that $S
\lrstep{\nu w_n. \Out(c,w_n)} S_\Out$ where $S_\Out =
\quadruple{\Ec}{\p'_S}{\Phi'_S}{\sigma_S}$, $\Phi'_S = \Phi_S \cup
\{w_n \refer \TAG{u}{i})\sigma_S\}$. 
Moreover, it implies that $\Phi'_D\mydownarrow = \delta(\Phi'_S\mydownarrow\})$ and $\sigma_D\mydownarrow = \delta(\sigma_S\mydownarrow)$.
As mentioned, by hypothesis, we know that $D_0$ does not reveal the values of its assignments w.r.t. $(\rho_\alpha,\rho_\beta)$. 
Hence, for all assignment variable $x$ of color $\alpha$
(resp. $\beta$) in $\dom(\sigma_D)$, for all $\key \in \{ k, \pk(k),
\vk(k) \mid k = x\sigma_D \vee k = x\rho_\alpha \text{
  (resp. $x\rho_\beta$)}\}$, $\key$ is not deducible from $\new\
\Ec. \Phi'_D$. We denote $K$ this set.

\noindent {Let $K_S = \{t, \pk(t),\vk(t)~|~ t \in \dom(\rho^+_\alpha) \cup
\dom(\rho^+_\beta), \mbox{ $t$ ground}\}$. 
Since 
$\sigma_D\mydownarrow = \delta(\sigma_S\mydownarrow)$, and by
definition 
of $\rho_\alpha^+$ and $\rho_\beta^+$, we deduce that
$K = \delta_\alpha(K_S) \cup \delta_\beta(K_S)$.}
We have also that $\Phi'_D\mydownarrow =
\delta(\Phi'_S\mydownarrow)$. Hence, 
we can now apply Lemma~\ref{lem:samerecipesymmetric} with $M = w_n$ and so we deduce that $\delta_\gamma(\TAG{u}{i}\sigma_S\mydownarrow) = \delta_{\gamma'}(\TAG{u}{i}\sigma_S\mydownarrow)$. Hence, we can conclude as in the previous case.
  \medskip{}
  
  \noindent\emph{Case of the rule {\sc Par}.} It is easy to see that the result holds for this case.

  \medskip{}
  
  Note that the rules {\sc New} and {\sc Repl} can not be triggered since the processes under study do not contain bounded names and replication.
\end{proof}


\theoremmain*

\begin{proof}
We prove the two items separately. 
\begin{enumerate}
\item The first item is actually a direct consequence of
  Proposition~\ref{pro:shared-to-disjoint}. We  rely on Corollary~\ref{cor:framestatequiv} and the fact that $D = \delta(S)$ to
  establish that:
\[
\new \, \Ec_S. \Phi_S \statequiv \new\, \Ec_D. \Phi_D.
\]
\item The second item is actually a direct consequence of
  Proposition~\ref{pro:disjoint-to-shared}. We rely on
  Corollary~\ref{cor:framestatequiv},  Corollary~\ref{cor:deltaandkeyhidden} and the fact
  that $D = \delta(S)$ to establish that:
\[
\new \, \Ec_S. \Phi_S \statequiv \new\, \Ec_D. \Phi_D.
\]
\end{enumerate}
This concludes the proof of the theorem.
\end{proof}

\section{Parallel composition}
\label{sec:app-parallel composition}

The goal of this section is to prove the results that relate to the parallel composition, that are Theorem~\ref{theo:compo-par-reachability-main} and~\ref{theo:compo-par-diff-equivalence-main}.
We prove a slightly improved version of
Theorem~\ref{theo:compo-par-diff-equivalence-main} 
assuming that composition contexts may contain several holes.
To prove these composition results, we will rely on
Theorem~\ref{theo:main-main}, 
and for this we have to explain how to get rid of the replications, and the $\new$ instructions (see Section~\ref{subsec:par-unfold}). 
We have also to rewrite the process to ensure that names are shared via assignment variables only (see Section~\ref{subsec:par-assignation}).

\subsection{Unfolding a biprocess}
\label{subsec:par-unfold}

Given an extended process $A = \triple{\Ec}{\p}{\Phi}$ where $\p$ may contain name restrictions and replications, the idea 
is to unfold the replications and to gather together all the restricted names in the set $\Ec$. Of course, it is not possible to apply such a transformation and to 
preserve the set of possible traces. However, given a specific trace issued from $A$, it is possible to compute an unfolding of $A$ that will exhibit  this specific trace.
The converse is also true, any trace issued from an unfolding of $A$ will correspond to a trace of $A$. Thus, the process $A$ and all its possible unfoldings will exhibit exactly  the same set of traces.
We define this notion directly on biprocesses.

\medskip

\begin{definition}
Let $A = \triple{\Ec}{\p}{\Phi}$ be an extended biprocess. We define the $n\textsuperscript{th}$ unfolding of $A$, denote by $\unfolding{A}{n}$, the biprocess $\triple{\Ec \uplus \Ec_n}{\p_n}{\Phi}$ 
obtained from $A$ by replacing in $\p$ 
each instance of $!Q$ with $n$ instances of $Q$ (applying $\alpha$-renaming to ensure name and variable distinctness), and then removing the 
 $\new$ operations  from the resulting process. These names are then put in the set $\Ec_n$ and added in the first component of the extended process.
\end{definition}

The link between an extended biprocess and its unfoldings is stated in Lemma~\ref{lem:unfolding-diff}.
 
 \medskip
 
\begin{lemma}
\label{lem:unfolding-diff}
Let $A = \triple{\Ec}{\p}{\Phi}$ be an extended biprocess. 
The biprocess $A$ is in diff-equivalence if, and only if, $\unfolding{A}{n}$ is in diff-equivalence for any $n \in \mathbb{N}$.
\end{lemma}


\subsection{Sharing names via assignments}
\label{subsec:par-assignation}

In Theorem~\ref{theo:main-main}, one can note that processes may only share data through assignment variables. 
This is not a real limitation since a name that is shared via the composition context can be assigned to an assignment variable by 
one process and used  by the other through the assignment variables. Below, we describe this transformation that actually preserves diff-equivalence of a biprocess.

\smallskip{}

Let $A = \triple{\Ec}{\p}{\Phi}$ be an extended colored (with colors in $\{1, \ldots,p\} = \alpha \uplus \beta$)
biprocess that does not contain any name restriction nor replication in $\p$. Let $K = k_1, \ldots, k_\ell$ be a sequence of names (of base type) in 
$\Ec$ that contains at least all the names occurring in both type of
actions -- in actions colored $\alpha$ as well as in actions colored
$\beta$ (intuitively $k_1,\ldots,k_\ell$ are the names shared by the
two processes we want to compose). Since we work with a biprocess, we
do this transformation simulatenously on both sides. We do this each
time the transformation is required by
one side of the biprocess. Actually, when we will apply this
transformation, the right-hand side will correspond to the 
disjoint case, whereas the left-hand side will correspond to the shared case,  and all the transformations will arise because of the left-hand side.

Let $Z = z^\alpha_1, \ldots, z^\alpha_\ell$ be a sequence of fresh variables, and $i \in \alpha$. 
We denote by $\ndisjunction{A}{\mathcal{K}}{Z}{i}$ the extended biprocess 
$\triple{\Ec}{P_\ass}{\Phi}$ where $P_\ass$ is defined as follows:
\[
P_\ass =  [z^\alpha_1 := k_1]^i.\ldots.[z^\alpha_\ell := k_\ell]^i. ( \mid_{P \in \p} P\rho^{}_{\beta})
\]
where $\rho^{}_\beta$ replaces each occurrence of the name $k_j$ ($1 \leq j \leq \ell$) that occurs in an action 
$\beta$-colored by its associated assignment variable $z^\alpha_j$ ($1\leq j\leq \ell$).
Note that the replacement $\rho_{\beta}$ will not affect the process corresponding to the disjoint case.
\smallskip{}

Note that in the definition above, the $\alpha$-colored process will
assign the shared names into assignment variables whereas the
$\beta$-colored process will simply use those variables instead of the
corresponding names. 
This choice is arbitrary and the roles played by $\alpha$ and $\beta$ can be swapped. Again, this transformation preserves equivalence.
This result is stated below in Lemma~\ref{lem:assignation-diff}.

\smallskip{}

\begin{lemma}
\label{lem:assignation-diff}
Let $A = \triple{\Ec}{\p}{\Phi}$ and $\ndisjunction{A}{\mathcal{K}}{Z}{i}$ be two extended biprocesses as described above.  We have that
$A$ is in diff-equivalence if, and only if, $\ndisjunction{A}{\mathcal{K}}{Z}{i}$ is in diff-equivalence.
\end{lemma}


\subsection{Composing trace equivalence}
 \label{subsec:par-equiv}

The theorem we want to prove is stated below.
Note that, this theorem differs from the one 
stated in the main body of the paper since we work in a slightly more general setting. 

We denote by $\Sigma_0^c = \{\senc,\aenc,\sign, \pk, \vk, \langle \; \rangle\}$, 
\emph{i.e.} the constructors of the common signature $\Sigma_0$. We consider composition contexts that may contain several holes. They are formally defined as follows:

\smallskip{}

\begin{definition}
\label{def:context} 
A \emph{composition context} $C$ is defined by
the following grammar where~$n$ is a name of base type.
\[
C, C_1, C_2  :=  \_  \; \mid\; \new \ n.\ C \; \mid \; ! C   \; \mid \; C_1|C_2
\]
\end{definition}

We only allow names of base type (typically keys) to be shared between processes
through the composition context.  In particular, they are not allowed to share a
private channel even if each process can used its own private channels to
communicate internally.  We also suppose w.l.o.g. that names occurring in~$C$
are distinct. A composition context may contain several holes. We can index them to avoid
confusion.  We write $C[P_1, \ldots, P_\ell]$ (or shortly
$C[\seq{P}]$) the process obtained by filling the
$i\textsuperscript{th}$ hole with the process~$P_i$ (or the
$i\textsuperscript{th}$ process of the sequence $\seq{P}$).

\smallskip{}

We use the notation $\seq{P} \mid \seq{Q}$ to represent the sequence of processes
obtained by putting in parallel the processes of the sequences
$\seq{P}$ and $\seq{Q}$ componentwise.

\smallskip{}

Parallel composition between tagged processes can only 
be achieved assuming that the shared keys are not revealed. Indeed, if the security of $P$ is ensure through the secrecy 
of the shared key $k$, there is no way to guarantee that $P$ is still secure in an environment where another process $Q$ running in parallel will reveal  this key.

Since, we consider a common signature $\Sigmazero$ and composition
contexts with several holes, we have to
generalize a bit the notion of revealing a shared key stated in the
body of the paper. We have to take into account public keys and
verification keys.

\medskip{}

\begin{definition}
  \label{def:reveal}
  Let $C$ be a composition context, 
  $A$ be an extended process of the form $\quadruple{\Ec}{C[P_1,\ldots,
    P_\ell]}{\Phi}{\sigma}$, and $\key \in \{n, \pk(n), \vk(n) ~|~ 
  \mbox{$n$ occurs in $C$}\}$.  We say that \emph{the extended
    process $A$ reveals the key $\key$} when:
  \begin{itemize}
  \item $\quadruple{\Ec \cup \{s\}}{C[P^+_1, \ldots, P^+_\ell]}{\Phi}{\sigma}
    \LRstep{w} \quadruple{\Ec'}{\p'}{\Phi'}{\sigma'}$ with $P^+_{i_0}
    \stackrel{\defi}{=} P_{i_0} \mid \In(c,x). \,\texttt{if } x = \key
    \;\texttt{then}\, \Out(c,s)$ and $P^+_i \stackrel{\defi}{=} P_i$
    if $i\not=i_0$; and
  \item $M\Phi' =_\E s$ for some $M$ such that $\fv(M) \subseteq
    \dom(\Phi')$ and $\fn(M) \cap \Ec' = \emptyset$
  \end{itemize}
  {where $c$ is a fresh public channel name,  $s$ is a fresh
    name of base type, and the $i_0\textsuperscript{th}$ hole of $C$ is in the scope of ``$\new \ \fn(\key)$''.}
\end{definition}

\medskip

\begin{definition}
  \label{def:composability-app}
  Let $C$ be a composition context and $\Ec_0$ be a finite set of
  names of base type. Let $\seq{P}$ and $\seq{Q}$ be two sequences of
plain processes together
with their frames $\Phi$ and $\Psi$.
  We say that $\seq{P}/\Phi$ and $\seq{Q}/\Psi$ are \emph{composable} under
  $\Ec_0$ and~$C$ when:
  \begin{enumerate}
\setlength{\itemsep}{1mm}
\item $\seq{P}$ (resp. $\seq{Q}$)  are built over $\Sigma_\alpha
  \cup \Sigmazero$ (resp. $\Sigma_\beta \cup \Sigmazero$), whereas
  $\Phi$ (resp. $\Psi$) are built over $\Sigma_\alpha \cup
  \{\pk,\vk\}$ (resp. $\Sigma_\beta \cup \{\pk,\vk\}$),
  $\Sigma_\alpha \cap \Sigma_\beta = \emptyset$, and $\seq{P}$ (resp. $\seq{Q}$) is tagged;
  \item $\fv(\seq{P}) = \fv(\seq{Q}) = \emptyset$, and $\dom(\Phi) \cap \dom(\Psi) =\emptyset$.
  \item $\Ec_0 \cap (\fn(C[\seq{P}]) \cup \fn(\Phi)) \cap (\fn(C[\seq{Q}]) \cup
    \fn(\Psi)) = \emptyset$;
  \item $\triple{\Ec_0}{C[\seq{P}]}{{\Phi}}$
    (resp.~$\triple{\Ec_0}{C[\seq{Q}]}{{\Psi}}$) 
does not reveal any key in: 
\[{\{n, \pk(n), \vk(n) ~|~ n \mbox{ occurs in
  }{\fn(\seq{P}) \mathord{\cap} \fn(\seq{Q}) \mathord{\cap} \bn(C)}\}.}
\]
  \end{enumerate}
This notion is extended as expected to biprocesses requiring that
$\fst(\seq{P})/\fst(\Phi)$ and $\fst(\seq{Q})/\fst(\Psi)$, as well as
$\snd(\seq{P})/\snd(\Phi)$ and $\snd(\seq{Q})/\snd(\Psi)$,  are composable.
\end{definition}

\medskip

\begin{theorem}
  \label{theo:compo-par-diff-equivalence-app}
Let $C$ be a composition context, and $\Ec_0$ be a finite set of names
of base type. Let $\seq{P}$ (resp.~$\seq{Q}$) be a sequence of
plain biprocesses together with its frame $\Phi$ (resp. $\Psi$),
and assume that $\seq{P}/\Phi$ and $\seq{Q}/\Psi$ are composable under~$\Ec_0$ and~$C$.

If $\triple{\Ec_0}{C[\seq{P}]}{\Phi}$ and 
$\triple{\Ec_0}{C[\seq{Q}]}{\Psi}$
satisfy diff-equivalence (resp. trace equivalence), then
$\triple{\Ec_0}{C[\seq{P} \mid \seq{Q}]}{\Phi \uplus
  \Psi}$ 
satisfies diff-equivalence (resp. trace equivalence).
\end{theorem}

\smallskip{}

\begin{proof}
According to our hypothesis, $\seq{P}$ and $\seq{Q}$ are both tagged hence there exists two sequences of colored plain processes $\seq{P_t}$ and $\seq{Q_t}$ such that $\TAGG{\seq{P_t}} = \seq{P}$ and $\TAGG{\seq{Q_t}} = \seq{Q}$. Moreover, we can split the set of names~$\Ec_0$
into two disjoint sets $\Ec_P$ and $\Ec_Q$ depending on whether the name occurs 
in $\seq{P}/\Phi$ or $\seq{Q}/\Psi$.

Let $S = \triple{\Ec_0}{C[\seq{P} \mid
  \seq{Q}]}{\Phi}$. Our goal is to show that $S$ satisfies
diff-equivalence (resp trace equivalence).
By hypothesis, we actually have that
$\triple{\Ec_P}{C[\seq{P}]}{\Phi}$, and 
$\triple{\Ec_Q}{C[\seq{Q}]}{\Psi}$ 
satisfy diff-equivalence (resp trace equivalence).
Let $D = \triple{\Ec_P \uplus \Ec_Q}{C[\seq{P}] \mid C[\seq{Q}]}{\Phi \uplus \Psi}$ (modulo some $\alpha$-renaming to ensure name and variable distinctness of the resulting process). 
Since the two processes that are composed in parallel do not share any
data,  we have that $D$ satisfies diff-equivalence (resp trace equivalence).
In order to conclude that $S$ satisfies diff-equivalence (resp trace equivalence), we will show
that $\fst(S) \approx_\diff \fst(D)$ and $\snd(S) \approx_\diff
\snd(D)$
relying on Theorem~\ref{theo:main-main}.

\smallskip{}
Let $B_1$ be the biprocess obtained by forming a biprocess with $\fst(S)$ and $\fst(D)$. Even if the two processes do not have exactly 
the same structure, this can be achieved by introducing some $\new$ instructions that will not be used in $\fst(S)$.
%
Relying on Lemma~\ref{lem:unfolding-diff}, we have that $B_1$ is in diff-equivalence if and only if $\unfolding{B_1}{n}$ is in diff-equivalence for any $n \in \mathbb{N}$.
Let $n_0 \in \mathbb{N}$. We transform the biprocess
$\unfolding{B_1}{n_0}$ to introduce assignment variables (and
  we may assume w.l.o.g. that the processes under study do not rely on
any assignment variables, thus the resulting process will only contain
the assignment variables introduced by our transformation), namely
$z^\alpha_1,\ldots,z^\alpha_\ell$. This leads us to another biprocess 
and this transformation still preserves diff-equivalence as stated in Lemma~\ref{lem:assignation-diff}. 
Note that, on the right-hand side of the biprocess (the disjoint case), the assignments variables are assigned to names that do not occur in any action colored~$\beta$.
In order to apply our  Theorem~\ref{theo:main-main}, we perform a last
transformation on this biprocess that consists in 
replacing the elements that occur inside the frame by output actions
(colored with $\alpha$ or $\beta$ depending on its origin)   in front
of the biprocess. 
This last transformation preserves also diff-equivalence. We finally
consider $\Ec_\alpha = \emptyset$, and {$\Ec_\beta =
  \{k^\alpha_1,\ldots,k^\alpha_\ell\}$} a set of fresh names, 
and we add these two sets of names to the set of $\Ec_0$ (first argument of the biprocess).
Now, it remains to show that this resulting biprocess $B'_1$ is in diff-equivalence. For this, we rely on Theorem~\ref{theo:main-main}.
Let $\rho_\alpha$ be such that $\dom(\rho_\alpha) = \emptyset$, and $\rho_\beta$ be such that $\dom(\rho_\beta)  = \{z^\alpha_1,\ldots,z^\alpha_\ell\}$, 
and {$z^\alpha_j\rho_\beta = k^\alpha_j$} for $j \in \{1,\ldots, \ell\}$. 
Actually, we have that {$D'_1 = \delta(S'_1)$} where $S'_1 =
\fst(B'_1)$ and $D'_1 = \snd(B'_1)$, and for all possible executions of $S'_1$ or $D'_1$, compatibility will be satisfied.
Indeed, by construction, we know that all the assignment variables
(remember that all the assignments occurring in the process
  have been introduced by our transformation) will be assigned to distinct names. 
Now, to satisfy all the requirements needed to apply Theorem~\ref{theo:main-main}, it remains to establish the non-deducibility of the keys.

\medskip

By hypothesis, $\triple{\Ec_0}{C[\seq{P}]}{\Phi}$ and $\triple{\Ec_0}{C[\seq{Q}]}{\Psi}$ do not reveal  $k$, $\pk(k)$, or $\vk(k)$ for any $k \in \fn(\seq{P}) \cap \fn(\seq{Q}) \cap
 \bn(C)$. Hence, we deduce that $\fst(D)$ (parallel composition -
 disjoint case) does not 
reveal $k$, $\pk(k)$, or $\vk(k)$ for any $k \in \fn(\seq{P}) \cap \fn(\seq{Q}) \cap \bn(C)$. 
 
Note that we want to apply Theorem~\ref{theo:main-main} on $S'_1$ and
$D'_1$ and not on $S_1$ and $D_1$. However, we built $D'_1$ by
unfolding $D_1$ and introducing assignment variables. First, note that these
transformations preserve deducibility.
Moreover, secrecy of $k$, $\pk(k)$, or $\vk(k)$ for any $k \in \fn(\seq{P}) \cap \fn(\seq{Q}) \cap
 \bn(C)$ actually implies that $D'_1$ does not reveal its
 assignments w.r.t. $(\rho_\alpha, \rho_\beta)$.
 This allows us to apply Theorem~\ref{theo:main-main} and so to conclude.
\end{proof}

\subsection{Composing reachability}
\label{subsec:par-diff}

We now prove a variant of
Theorem~\ref{theo:compo-par-reachability-main} considering our  slightly more general setting.

\begin{theorem}
\label{cor:compo-par-reachability-app}
Under the same hypotheses as Theorem~\ref{theo:compo-par-diff-equivalence-app} with processes instead of bioprocesses, and considering a name $s$ that occurs in $C$. If 
$\triple{\Ec_0}{C[\seq{P}]}{\Phi}$ and $\triple{\Ec_0}{C[\seq{Q}]}{\Psi}$ do not reveal
$s$, then $\triple{\Ec_0}{C[\seq{P} \mid \seq{Q}]}{\Phi \cap \Psi}$ does not reveal $s$.
\end{theorem}

\begin{proof}
The proof follows the same lines as the one for dealing with diff-equivalence and trace equivalence. In order to show that the process $S = 
\triple{\Ec_0}{C[\seq{P} \mid \seq{Q}]}{\Phi \uplus \Psi}$ does not reveal $s$, we rely on the fact that the secrecy is preserved by parallel composition of ``disjoint'' processes. Thanks to our hypotheses, we have that $D = \triple{\Ec_0}{C[\seq{P}] \mid C[\seq{Q}]}{\Phi \uplus \Psi}$ does not reveal $s$. Then, by applying
Theorem~\ref{theo:main-main} and more specifically the first bullet point
of this theorem, we can deduce that for all $(\tr, \new\
\Ec_S. \Phi_S) \in \trace(S)$, there exists a trace $(\tr, \new\
\Ec_D. \Phi_D) \in \trace(D)$ such that $\new\ \Ec_S. \Phi_S
\statequiv \new\ \Ec_D. \Phi_D$. 
Since $D$ does not reveal~$s$, we conclude that $S$ does not reveal
$s$ too.
\end{proof}

\section{Sequential composition}
\newcommand{\unf}{\mathsf{u}}
\newcommand{\vari}{\mathsf{v}}
\newcommand{\rename}{\mathsf{r}}
\newcommand{\ru}{\mathsf{ru}}
\newcommand{\vsequ}{\mathsf{vseq}}

In this section we prove Theorems~\ref{theo:main-compo-seq-equiv} and~\ref{theo:main-compo-seq-confidentiality}. 
As for establishing parallel composition results, we will rely on  Theorem~\ref{theo:main-main}. 
This will require to unfold  the processes under study, and to use assignment variables to share data.
However, as already discussed in Section~\ref{sec:difficulties-equiv}, we also
have to tackle some additional difficulties. In particular, to ensure
the compatibility of the executions as required by Theorem~\ref{theo:main-main}.


\subsection{Unfolding biprocesses and sharing names via assignments}

Unfolding the biprocesses for sequential composition follows the same
principles as unfolding the biprocesses for parallel composition.
However, we need to be more specific.  In particular, we need to be
able to easily talk about the replicated instances of a nonce after
unfolding.  We explain in this section how the unfolded biprocesses
are built, and we introduce some notation that we will use throughout
the entire section.

\begin{example}
\label{ex:renaming-one}
Let  $P = ! \new\ k.! \new\ n.\allowbreak\Out(c,\senc(n,k))$. The plain process 
\[
\begin{array}{l}
\ \phantom{\mid }\Out(c,\senc(n[1,1],k[1])) \mid \Out(c,\senc(n[1,2],k[1]))\\
\mid \Out(c,\senc(n[2,1],k[2])) \mid \Out(c,\senc(n[2,2],k[2]))
\end{array}
\]
together with the set 
\[\mathcal{K} = \{k[1], k[2], n[1,1], n[1,2], n[2,1], n[2,2]\}
\]
will correspond to the 2-unfolding of $P$, denoted $\unfolding{P}{2}$. 
In this example, $k[1], k[2], n[1,1], \ldots, n[2,2]$ are considered as distinct names.
\end{example}

\medskip

More generally, in such formalism, two names $n_1[i_1, \ldots, i_p]$ and $n_2[j_1, \ldots, j_q]$ are equal if, and only if,  they are syntactically equal, \emph{i.e.} $n_1 = n_2$, $p = q$ and $i_k = j_k$ for each $k \in\{1\ldots p\}$.  We will use the same convention to represent the variables occurring in the processes. We will also extend this notation to processes. Thus 
$P[i_1, \ldots, i_n]$ will represent the instance of $P$ that
correspond to the $i_1^\textsuperscript{th}$ instance of the $1^\textsuperscript{st}$ replication, 
$i_2^\textsuperscript{th}$ instance of the $2^\textsuperscript{nd}$ replication, \emph{etc}.

\medskip

\begin{example}
\label{ex:renaming-two}
Going back to our previous example, we have that $\unfolding{P}{2} = (Q[1,1] \mid Q[1,2] \mid Q[2,1] \mid Q[2,2], \mathcal{K})$ where $Q[i,j] = \Out(c,\senc(n[i,j], k[i]))$.
\end{example}

\medskip

With such notation, we can now be much more precise on how our disjoint and shared processes are unfolded.

\medskip

Following notation given in Theorem~\ref{theo:main-compo-seq-confidentiality}, we will consider the biprocesses:
\begin{enumerate}
\item $S = \triple{\Ec_0}{C[P_1[Q_1] \mid P_2[Q_2]]}{\Phi \uplus
    \Psi}$, the so-called shared case;
\item $D^\para = \triple{\Ec_0}{C[P] \mid C[Q]}{\Phi
    \uplus \Psi}$, the so-called parallel disjoint case;
\item $D^\sequ = \triple{\Ec_0}{\tilde{C}[P_1[\tilde{Q_1}]\mid
      P_2[\tilde{Q_2}]]}{\Phi \uplus \Psi}$ where $\tilde{C}$ is as
  $C$ but each name $n$ is duplicated $n$/ $n^Q$ in order to ensure
  disjointness. The processes $\tilde{Q_1}$ and $\tilde{Q_2}$ are
  obtained from $Q_1$ and $Q_2$ by replacing each name $n$ occurring
  in $C$ by its copy $n^Q$. This represents the so-called sequential
  disjoint case.
\end{enumerate}

Then, given a biprocess $B$ (typically one given above), we denote by
$B_n$ its $n$\textsuperscript{th} unfolding relying on the naming
convention introduced in Example~\ref{ex:renaming-one} and Example~\ref{ex:renaming-two}.

Using the notation introduced above, it should be clear that for each
unfolding $n$ (with $n \in \mathbb{N}$),  the
biprocess that represents the
parallel disjoint case, \emph{i.e.} $D^\para_n$ exhibits more behaviours
than the biprocess that represents the sequential disjoint case,
\emph{i.e.}  $D^\sequ_n$.

\medskip{}

\begin{lemma}
\label{lem:par to seq}
If $D^{\para}_n$ satisfies diff-equivalence then $D^{\sequ}_n$ satisfies  diff-equivalence
\end{lemma}

\medskip{}

As for parallel composition, once unfolding has been done, we get rid
of names that are shared through the composition context using
assignment variables. We denote these names $r_1, \ldots r_p$ and
their associated assignment variables $z_1, \ldots, z_p$. We also get
rid of the content of the frame by adding some outputs in front of the
resulting process. Note that, we can assume w.l.o.g. that the only
assignment instructions are those that occur in $P_1$ and $P_2$ to
give a value to $x_1$ and $x_2$. Indeed, an assignment of the form $[x
:= t]$ that is ``local'' to $P_1/P_2$ (or $Q_1/Q_2$) has the same
effect as applying the substitution ${x \mapsto t}$ directly on the
process.  This additional hypothesis will help us ensure compatibility
of all executions when applying Theorem~\ref{theo:main-main}.

Given a biprocess $B$, we will denote  $B^\vari$ the biprocess
  resulting from the transformation described above. In particular, we
  will consider $S^\vari_n$ the biprocess obtained by applying the
  transformation above on $S_n$ (the $n$\textsuperscript{th} unfolding
  of the shared case), and also $D^\vsequ_n$ the biprocess obtained by
  applying the transformation on $D^\sequ_n$.

Again, it should be clear that these transformations preserve
diff-equivalence. 

\medskip{}

\begin{lemma}
We have that:
\begin{enumerate}
\item $D^{\vsequ}_n$ satisfies diff-equivalence if, and only if, $D^{\sequ}_n$ satisfies  diff-equivalence
\item $S^{\vari}_n$ satisfies diff-equivalence if, and only if, $S_n$ satisfies  diff-equivalence
\end{enumerate}
\end{lemma}

\medskip{}

\label{subsec:rho}
Relying on this transformation, by colouring actions of~$P$ with
$\alpha$, and actions of~$Q$ with $\beta$, given an integer $n$
corresponding to the unfolding under study, and assuming that the hole
of $C$ is under $m$ replications, we consider $\rho_\alpha$ such that
$\dom(\rho_\alpha) = \emptyset$, and $\rho_\beta$ with
\[
\begin{array}{l}
\dom(\rho_\beta) =   \{ z_1, \ldots,
  z_p\} 
\;\; \cup \{ x_1[i_1,\ldots, i_m], x_2[i_1,\ldots,
i_m] \mid  1 \leq i_1, \ldots, i_m\leq n\}
\end{array}
\]

\begin{itemize}
\item $z_i\rho_\beta = r_i$ for $1 \leq i \leq p$;
\item $\rho_\beta(x_1[i_1,\ldots,i_m]) =k[i_1,\ldots, i_m]$
\item   $\rho_\beta(x_2[i_1,\ldots,i_m]) = k[i_1,\ldots, i_m]$.
\end{itemize}

In other words, we abstract each name shared via the composition
context by a fresh one, \emph{i.e.} $r_i$, and each term shared
through the variables $x_1$ and $x_2$ are abstracted by a fresh name,
a new one for each instance.

\subsection{Secrecy of the shared keys}
We now focus on the fourth condition of Theorem~\ref{theo:main-main}, 
\emph{i.e.} we ensure that $D^{\vsequ}_n$ does not reveal the values
of its assignments w.r.t. $(\rho_\alpha,\rho_\beta)$ as defined in Section~\ref{subsec:rho}.


\begin{lemma}
\label{lem:assignment variable non deduce}
Assume that $P_1/P_2/\Phi$ is a good key-exchange protocol under $\Ec_0$ and $C$.
Assume also that $\quadruple{\Ec_0}{C[Q]}{\Psi}{\emptyset}$ does not reveal any $k, \pk(k)$ and $\vk(k)$.

In such a case, we have that $D^{\vsequ}_n$ 
does not reveal the value of its assignment variables w.r.t. $(\rho_\alpha,\rho_\beta)$.
\end{lemma}

\begin{proof}
By Definition~\ref{def:good},  $P_1/P_2/\Phi$ being a good key-exchange protocol under $\Ec_0$ and $C$ implies that $\triple{\Ec_0}{P_{\good}}{\Phi}$ 
does not reveal $\bad$ where $P_\good$ is defined as follows:
\[
\begin{array}{l}
    P_\good =  \new\, \bad. \new\, d. \big(C[\new\, id. ( P_1[\Out(d,\langle x_1, id\rangle)] \mid P_2[\Out(d,\langle x_2,id\rangle)])] \\[0.5mm]
    \;\mid  \In(d,x). \In(d,y). \myIf \  \proj_1(x) = \proj_1(y)  \wedge \proj_2(x) \neq \proj_2(y)\, \myThen \,\Out(c,\bad)\\[0.5mm]
    \;\mid \In(d,x). \In(d,y). \myIf\ \proj_1(x) \neq \proj_1(y)  \wedge \proj_2(x) = \proj_2(y)\, \myThen \,\Out(c,\bad)\\[0.5mm]
    \;\mid \In(d,x). \In(c,z). \myIf\, z \in \{\proj_1(x), \pk(\proj_1(x)), \vk(\proj_1(x))\}\, \myThen\, \Out(c,\bad)\big)\\[0.5mm]
  \end{array}
\]
In the case were $C$ is of the form $C'[!\_]$, $P_\good$ is defined as follows:
\[
\begin{array}{l}
  \new\, \bad, d, r_1, r_2.\big(C'[\new\, id. !( P_1[\Out(d,\langle x_1, id,r_1\rangle)] \mid P_2[\Out(d,\langle x_2,id,r_2\rangle)])] \\[0.5mm]
    \;\mid  \In(d,x). \In(d,y). \myIf \  \proj_1(x) = \proj_1(y)  \wedge \proj_2(x) \neq \proj_2(y)\, \myThen \,\Out(c,\bad)\\[0.5mm]
    \;\mid \In(d,x). \In(d,y). \myIf\ \proj_1(x) = \proj_1(y)  \wedge \proj_3(x) = \proj_3(y)\, \myThen \,\Out(c,\bad)\\[0.5mm]
    \;\mid \In(d,x). \In(c,z). \myIf\, z \in \{\proj_1(x), \pk(\proj_1(x)), \vk(\proj_1(x))\}\, \myThen\, \Out(c,\bad)\big)\\[0.5mm]
  \end{array}
\]
In both cases, it indicates that the secrecy of $x$, $\pk(x)$ and $\vk(x)$ is preserved, where $x$ is the value of any assignment variable. 
Then, the result is actually a direct consequence of the fact that secrecy is
preserved through disjoint composition, and the transformations that
are performed on the process (\emph{e.g.} unfolding, adding of some
assignments operations) also preserve secrecy.
\end{proof}


\subsection{Compatibility}
To use Theorem~\ref{theo:main-main}, a compatibility condition is
required. As in the case of parallel composition, this property will
be trivially satisfied for assignments that have been added by our
transformation. However, more work is needed to deal with assignments
present in the original processes, that is in our situation,
assignments of the form $[x_1[...] = \_]$ and $[x_2[...] = \_]$ that
come from the unfolding of the process $P_1/P_2$. The idea is that the
 abstractability property and the fact that $P_1/P_2/\Phi$ is a good key-exchange protocol will give us the required conditions to apply Theorem~\ref{theo:main-main}.

\medskip{}

The following lemma focuses on $P_1/P_2/\Phi$ being a good key-exchange protocol under $\Ec_0$ and $C$. However, the definition of a good key-exchange protocol depends on the shape of the composition context, and the properties satisfied by our processes will depends on the distinction. Hence, to avoid any confusion, unless the composition context is explicitely mentioned being of the form $C'[!\_]$, the definition of good key-exchange protocol always follows Definition~\ref{def:good}.


\begin{lemma}
\label{lem:strong-agreement}
Let $\quadruple{\Ec_0}{C[P_1[0] \mid 
    P_2[0]]}{\Phi}{\emptyset}$ be a process such that $P_1/P_2/\Phi$ is a good key-exchange protocol under $\Ec_0$ and $C$ . Let $n$ be an integer, and
 $\quadruple{\Ec}{\p}{\Phi'}{\sigma}$  a process such that
 $\fst(D^{\sequ}_n) \LRstep{\tr} \quadruple{\Ec}{\p}{\Phi'}{\sigma}$. 
Let $i_1, j_1, \ldots, i_m,j_m \in \mathbb{N}$, and $q_1,q_2 \in
\{1,2\}$ such that $x_{q_1}[i_1,\ldots,i_m]$ and $x_{q_2}[j_1, \ldots,
\allowbreak j_m]$ are in $\dom(\sigma)$. 
We have that:
\begin{center} 
$x_{q_1}[i_1,\ldots,i_m]\sigma\mydownarrow = x_{q_2}[j_1, \ldots,
\allowbreak j_m]\sigma\mydownarrow$ \\
if, and only if, \\
$i_p= j_p$ for all $1 \leq p \leq m$.
\end{center}
A similar property holds for $\snd(D^{\sequ}_n)$.
\end{lemma}

\begin{proof}
By definition of $P_1/P_2/\Phi$ being a good key-exchange protocol under $\Ec_0$ and $C$ and since secrecy is preserved
when considering disjoint composition, we have that 
$\quadruple{\Ec_0}{P}{\Phi \uplus \Psi}{\emptyset}$ preserves the secrecy of $bad$ where:
\[
\begin{array}{@{}ll@{}}
P =  \new\, \bad.\, \new\, d. (\\
\quad \tilde{C}[\new\, k. \new\, id. (P_1[\Out(d,\langle x_1,
id\rangle).\tilde{Q_1}\{^k/_{x_1}\}] \\
\quad \quad \hspace{2cm}  \mid
P_2[\Out(d,\langle x_2,id\rangle). \tilde{Q_2}\{^k/_{x_2}\}])] \\
\quad\mid  \In(d,x). \In(d,y). \\
\quad\quad \myIf \  \proj_1(x) = \proj_1(y)  \wedge \proj_2(x) \neq \proj_2(y)\, \\ 
\quad\quad \myThen \,\Out(c,\bad)\\
\quad\quad \myElse\, \myIf\ \proj_1(x) \neq \proj_1(y)  \wedge \proj_2(x) = \proj_2(y)\\
\quad\quad \myThen \,\Out(c,\bad)\\
\quad)\\
\end{array}
\]
Here, the notation $\tilde{C}$, $\tilde{Q_1}$, and $\tilde{Q_2}$ refer
to the same renaming as the one used to define $D^\sequ$.

Let $n$ be an integer. Consider the $n$\textsuperscript{th} unfolding
of $D^\sequ$ as well as the $n$\textsuperscript{th} unfolding of the
process $P$ defined above. First, note that an output on channel $d$ is always of the
form 
\begin{center}
$\Out(d,\langle x_j[i_1,\ldots i_m],id[i_1,\ldots i_m]\rangle)$ with $j
\in\{1,2\}$.
\end{center}
 
Let $\quadruple{\Ec}{\p}{\Phi'}{\sigma}$ be a process such that
$\fst(D^{\sequ}_n) \LRstep{\tr} \quadruple{\Ec}{\p}{\Phi'}{\sigma}$
with  $x_{q_1}[i_1,\ldots,i_m]$ and $x_{q_2}[j_1,\ldots,j_m]$ both in $\dom(\sigma)$.
Moreover, assume that $x_{q_1}[i_1,\ldots,i_m]\sigma\mydownarrow = x_{q_2}[j_1,
\ldots, j_m]\sigma\mydownarrow$.
In such a case, it is easy to build a trace of $\quadruple{\Ec_0}{P_n}{\Phi
  \uplus \Psi}{\emptyset}$ 
such that the pairs 
\begin{itemize}
\item 
$\langle x_{q_1}[i_1,\ldots,\allowbreak
i_m],id[i_1,\ldots,\allowbreak i_m]\rangle$, and 
\item $\langle x_{q_2}[j_1,\ldots, j_m],id[j_1,\ldots, j_m]\rangle$
\end{itemize}
 are
outputted on channel~$d$. 
Since the hole in $P_{q_1}$ (resp. $P_{q_2}$) is not in the scope of a replication, we deduce that these pairs can only be outputted once.
We have seen that such a process preserves the secrecy of $\bad$, and
thus we deduce that $(i_1, \ldots, i_p) = (j_1, \ldots, j_p)$.

\smallskip{} Now, relying on the fact that
$\quadruple{\Ec_0}{P_n}{\Phi \uplus \Psi}{\emptyset}$ preserves the
secrecy of $\bad$, and more precisely on the fact that the following
instructions are part of the process:
\[
\begin{array}{l@{\quad\quad}l}
&\mid  \In(d,x). \In(d,y). \\
&\quad \ldots\\
&\quad \myElse\, \myIf\ \proj_1(x) \neq \proj_1(y)  \wedge \proj_2(x) = \proj_2(y)\\
&\quad \myThen \,\Out(c,\bad)\\\end{array}
\]
we deduce that $(i_1, \ldots, i_p) = (j_1, \ldots, j_p)$ implies that 
$id[i_1,\ldots,\allowbreak i_m] = id[j_1,\ldots, j_m]$ and so we
deduce that  $x_{q_1}[i_1,\ldots,i_m]\sigma\mydownarrow = x_{q_2}[j_1, \ldots, j_m]\sigma\mydownarrow$.
\end{proof}

Note that the property above we established for $D^\sequ_n$ also holds
on $D^\vsequ_n$. We have also a similar result in case the composition context is of the form $C'[!\_]$ that is stated below and can be proved in a similar way.


\begin{lemma}
\label{lem:weak-agreement}
Let $\quadruple{\Ec_0}{C[P_1[0] \mid 
    P_2[0]]}{\Phi}{\emptyset}$ be a process such that $C = C'[!\_]$ for some $C'$ and $P_1/P_2/\Phi$ is a good key-exchange protocol under $\Ec_0$ and $C$. Let $n$ be an
integer and 
 $\quadruple{\Ec}{\p}{\Phi'}{\sigma}$ be a process such that
 $\fst(D^{\sequ}_n) \LRstep{\tr} \quadruple{\Ec}{\p}{\Phi'}{\sigma}$.

 Let $i_1, j_1, \ldots, i_m,j_m \in \mathbb{N}$. We have that:
\begin{itemize}
\item $x_1[i_1,\ldots,i_m]\sigma\mydownarrow = x_2[j_1, \ldots,
  \allowbreak j_m]\sigma\mydownarrow$ implies that $(i_1, \ldots,
  i_{m-1}) = (j_1, \ldots, j_{m-1})$; and
\item  for $q \in \{1,2\}$, $x_q[i_1, \ldots, i_m]\sigma\mydownarrow = x_q[j_1, \ldots,
  j_m]\sigma\mydownarrow$ implies $(i_1,\ldots, i_m) =
  (j_1,\ldots,j_m)$.
\end{itemize}
A similar property holds for $\snd(D^{\sequ}_n)$.
\end{lemma}

Now, regarding assingment variables, and in particular the different
instances of $x_1$ and $x_2$, it remains to show that the values
assigned to these variables will be rooted in the right signature. We
proceed in two steps. First, we discard terms rooted with 
a symbol in $\{ \pk, \vk,
\langle\rangle\}$ (Lemma~\ref{lem:rootNotpk}), and then we show that it is actually
rooted in the right signature (Lemma~\ref{lem:abstractability}).


\begin{definition}
We say that a process $P$ satisfies the \emph{abstractability property} if for all $P \LRstep{\tr} \quadruple{\Ec}{\p}{\Phi}{\sigma}$, for all assignment variable $x \in \dom(\sigma)$,  $\racine(x\sigma\mydownarrow) \not\in \{ \pk, \vk, \langle\rangle\}$.
\end{definition}

\smallskip{}

This property is important for our composition to hold.

\smallskip{}

\begin{example}
Let $P_i = [x_i := \langle k_1,k_2\rangle]$, and $Q_i = \myIf\, x_i = \langle \proj_1(x_i), \proj_2(x_i) \rangle \,\allowbreak \myThen\, \Out(c,\id_i)$. 
Let $C = \new\, k_1. \new\, k_2. \_$. We can see that in the shared case, the branch {\sc Then} of the process $Q_i$ 
will be executed whereas when considering in isolation the process $C[\new\, k. (Q_1\{x_1\mapsto k\}\mid Q_2\{x_2 \mapsto k\})]$ will not exhibit a similar behaviour.
\end{example}

Intuitively, we say that a value of an assignment variable can be
abstracted if it is not a pair, a public key or verification key. This
is due to the fact that those three primitives are not tagged and so
can be used by processes of any colour.


\begin{lemma}
\label{lem:rootNotpk}
Let $\quadruple{\Ec_0}{C[{P_1[0] \mid P_2[0]}]}{\Phi}{\emptyset}$ be a
process satisfying the abstractability property. We have that $D^{\vsequ}$ satisfies the abstractability property.
\end{lemma}

\begin{proof}
First of all, unfolding the process 
$\quadruple{\Ec_0}{C[{P_1[0] \mid P_2[0]}]}{\Phi}{\emptyset}$
preserves the abstractability property. 
Moreover, the transformation that transforms a process $B$ into a
process $B^\vari$ preserves the abstractability property. Thus, to
show that $D^{\vsequ}$ satisfies the abstractability property, it only
remains to show that this property is preserved by disjoint
composition assuming that the process we want to compose does not
introduce new assignments (note that this is the case of $Q_1/Q_2$).

In fact, part of the process brought by $Q_1/Q_2$ can be viewed as a
process executed by the attacker. Thus, for all 
$D^\sequ \LRstep{\tr'} \quadruple{\Ec'}{\p'}{\Phi'}{\sigma'}$, there
exists a correspondig execution $\quadruple{\Ec_0}{C[{P_1[0] \mid
    P_2[0]}]}{\Phi}{\emptyset} \LRstep{\tr''}
\quadruple{\Ec''}{\p''}{\Phi''}{\sigma''}$ 
such that $\sigma''$ and $\sigma'$ coincide on $\dom(\sigma'')$, and
in particular on the values assigned to $x_1[\ldots]$ and
$x_2[\ldots]$. This allows us to deduce that 
$\racine(x\sigma'') \not\in \{ \pk,\vk, \langle \rangle\}$, and thus
$D^\vsequ$ satisfies the abstractability property.
\end{proof}

The next lemma will allow us to conclude that we obtain traces
compatible with $(\rho_\alpha,\rho_\beta)$.

\begin{lemma}
\label{lem:abstractability}
Assume that $D^{\vsequ}_n$ does not reveal the value of its assignment variables w.r.t. $(\rho_\alpha,\rho_\beta)$ and satisfies the abstractability property. We have that for all $D^{\vsequ}_n \LRstep{\tr} \quadruple{\Ec}{\p}{\Phi}{\sigma}$, for all $\gamma \in \{\alpha,\beta\}$, for all $z \in \dom(\rho_\gamma)$, we have that either $\racinebis(z\sigma\mydownarrow) = \bot$ or $\racinebis(z\sigma\mydownarrow) \not\in \gamma \cup \{0\}$.
\end{lemma}

\begin{proof}
Since $D^{\vsequ}_n \LRstep{\tr} \quadruple{\Ec}{\p}{\Phi}{\sigma}$, we know that $\quadruple{\Ec}{\p}{\Phi}{\sigma}$ is a derived well-tagged extended process w.r.t. $\prec$ and $\vcol$, for some $\prec$ and $\vcol$. Moreover, by construction of $D^{\vsequ}_n$, we also know that $\dom(\rho_\alpha) = \emptyset$. We prove the result by induction of the $\dom(\rho_\beta)$ with the order $\prec$.

\medskip

\noindent\emph{Base case $z \prec z'$ for all assignment variables $z'$ different from $z$:} Assume that $\racinebis(z\sigma\mydownarrow) \neq \bot$ and $\racinebis(z\sigma\mydownarrow) \in \beta \cup \{0\}$. We now show that $z\sigma\mydownarrow \in \fct_\alpha(z\sigma)$. Since $\racinebis(z\sigma\mydownarrow) \in \beta \cup \{0\}$, we have that if $\racine(z\sigma\mydownarrow) \not\in \{ \vk,\pk,\langle\ \rangle\}$ then $z\sigma \in \fct_\alpha(z\sigma)$. Thus it remains to show that $\racine(z\sigma\mydownarrow) \not\in \{ \vk,\pk,\langle\ \rangle\}$. But $D^{\vsequ}_n$ satisfies the abstractability property hence we deduce that $\racine(z\sigma\mydownarrow) \not\in \{ \vk,\pk,\langle\ \rangle\}$.

Since $z\sigma\mydownarrow \in \fct_\alpha(z\sigma)$, we can apply Lemma~\ref{lem:FlawedColor and frame element direct element} and so we deduce that:
\begin{enumerate}
\item either there exists $M$ such that $\fv(M) \subseteq \dom(\Phi) \cap \{z' ~|~ z' \prec z\}$, $\fn(M) \cap \Ec = \emptyset$ and $z\sigma\mydownarrow \in \fct_\gamma(M\Phi\mydownarrow)$
\item otherwise there exists $j$ such that $z^{\beta}_j \prec z$ and $z^{\beta}_j\sigma\mydownarrow = z\sigma\mydownarrow$
\end{enumerate}
The second case is trivially impossible since $\dom(\rho_\alpha) = \emptyset$ and so $z^{\beta}_j$ does not exists. We focus on the first case: We know that $z\sigma\mydownarrow \in \fct_\gamma(M\Phi\mydownarrow)$. Since $z\sigma\mydownarrow$ is not deducible in $\new\ \Ec. \Phi$, then $z\sigma\mydownarrow \not\in \fctpair(M\Phi\mydownarrow)$. Moreover, we know that for all assignment variables $z'$ different from $z$, $z \prec z'$. Thus we can apply Lemma~\ref{lem:FlawedColor and frame element direct element 2} and obtain that there exists $M''$ such that $\fv(M) \subseteq \dom(\Phi)$, $\fn(M) \cap \Ec = \emptyset$ and $z\sigma\mydownarrow \in \fctpair(M'\Phi\mydownarrow)$. But this contradicts the fact that $z\sigma\mydownarrow$ is not deducible in $\new\ \Ec. \Phi$.

Since we always reach a contradiction, we can conclude that $\racinebis(z\sigma\mydownarrow) = \bot$ or $\racinebis(z\sigma\mydownarrow) \not\in \beta \cup \{0\}$.

\medskip

\noindent\emph{Inductive case:} Assume once again that $\racinebis(z\sigma\mydownarrow) \neq \bot$ and $\racinebis(z\sigma\mydownarrow) \in \beta \cup \{0\}$. As in the previous case, we can show that $z\sigma\mydownarrow \in \fct_\alpha(z\sigma)$ and so we can apply Lemma~\ref{lem:FlawedColor and frame element direct element} to obtain:
\begin{enumerate}
\item either there exists $M$ such that $\fv(M) \subseteq \dom(\Phi) \cap \{z' ~|~ z' \prec z\}$, $\fn(M) \cap \Ec = \emptyset$ and $z\sigma\mydownarrow \in \fct_\gamma(M\Phi\mydownarrow)$
\item otherwise there exists $j$ such that $z^{\beta}_j \prec z$ and $z^{\beta}_j\sigma\mydownarrow = z\sigma\mydownarrow$
\end{enumerate}
Once again the first case is trivially impossible since $\dom(\rho_\alpha) = \emptyset$. Thus it remain to focus on the second case. As in the previous, we can deduce that $z\sigma\mydownarrow \not\in \fctpair(M\Phi\mydownarrow)$. Moreover, by our inductive hypothesis, we know that for all assignment variable $z' \prec z$, $\racinebis(z'\sigma\mydownarrow) = \bot$ or $\racinebis(z'\sigma\mydownarrow) \not\in \beta \cup \{0\}$. Thus, we can deduce that $z'\sigma\mydownarrow \neq z\sigma\mydownarrow$. Thanks to this, we can apply Lemma~\ref{lem:FlawedColor and frame element direct element 2} and obtain that there exists $M''$ such that $\fv(M) \subseteq \dom(\Phi)$, $\fn(M) \cap \Ec = \emptyset$ and $z\sigma\mydownarrow \in \fctpair(M'\Phi\mydownarrow)$. But this contradicts the fact that $z\sigma\mydownarrow$ is not deducible in $\new\ \Ec. \Phi$.

Since we always reach a contradiction, we can conclude that $\racinebis(z\sigma\mydownarrow) = \bot$ or $\racinebis(z\sigma\mydownarrow) \not\in \beta \cup \{0\}$.
\end{proof}


\medskip We now establish that when the processes are a good  key exchanged protocol, all possible executions are actually compatible
w.r.t. $(\rho_\alpha,\rho_\beta)$.

\begin{lemma}
\label{lem:compatibility agree}
Let $(\rho_\alpha,\rho_\beta)$ be the two abstraction functions as
defined in Section~\ref{subsec:rho}.
 If $\quadruple{\Ec_0}{C[{P_1[0] \mid P_2[0]}]}{\Phi}{\emptyset}$ satisfies the abstractability property and $P_1/P_2/\Phi$ is a good key-exchanged protocol under $\Ec_0$ and $C$ then for any $P$ such that:
\begin{itemize}
\item $\fst(D^{\vsequ}_n) \LRstep{\tr} P$ (resp. $\snd(D^{\vsequ}_n)
  \LRstep{\tr} P$), we have that $P$ is compatible w.r.t. $(\rho_\alpha,\rho_\beta)$.
\item $\fst(S^{\vari}_n) \LRstep{\tr} P$ (resp. $\snd(S^{\vari}_n)
  \LRstep{\tr} P$), we have that $P$ is compatible w.r.t. $(\rho_\alpha,\rho_\beta)$.
\end{itemize}
\end{lemma}

\begin{proof}
Let $P$ be a process such that 
 $\fst(D^{\vsequ}_n) \LRstep{\tr}
 \quadruple{\Ec}{\p}{\Phi}{\allowbreak\sigma}$. Let $x,y \in
 \dom(\sigma) \cap \dom(\rho_\beta)$ and assume that $x\sigma =_\E
 y\sigma$. 
Let us denote $x = x_i[i_1,\ldots, i_m]$ and $y = x_j[j_1,\ldots, j_m]$ where $j_k,i_k \in \{1,\ldots, n\}$, $k \in \{1, \ldots, m\}$ and $i,j \in \{1,2\}$.

By hypothesis, $P_1/P_2/\Phi$ is a good key-exchanged protocol under $\Ec_0$ and $C$ . Hence thanks to Lemma~\ref{lem:strong-agreement}, $x\sigma =_\E y\sigma$ implies that $i_k = j_k$ for all $k \in \{1\ldots m\}$. On the other hand, Lemma~\ref{lem:strong-agreement} also indicates that $x_1[i_1,\ldots, i_m]\sigma = x_2[i_1,\ldots, i_m]\sigma$, for all $i_1,\ldots,i_m$.

Since by definition of $\rho_\beta$, $x_1[i_1,\ldots, i_m]\rho_\beta = x_2[i_1,\ldots, i_m]\rho_\beta \allowbreak = k[i_1,\ldots, i_m]$, we can deduce that $x\sigma = y\sigma$ if and only if $x\rho_\beta = y\rho_\beta$. At last, relying on Lemma~\ref{lem:abstractability}, we can conclude that $\quadruple{\Ec}{\p}{\Phi}{\sigma}$ is compatible with $(\rho_\alpha,\rho_\beta)$.

We now prove the property for $S^\vari_n$: Let $\fst(S^\vari_n) \LRstep{\tr} \quadruple{\Ec}{\p}{\Phi}{\sigma}$. We prove the result by induction on the size of $\tr$. Consider a transition $\quadruple{\Ec}{\p}{\Phi}{\sigma} \lrstep{\ell} A$. By inductive hypothesis, we know that $\quadruple{\Ec}{\p}{\Phi}{\sigma}$ is compatible with $(\rho_\alpha,\rho_\beta)$. But, the only transition that could render $A$ not compatible is the internal transition {(\sc Assgn)}. Hence assume that $\p = \{ [ x := t ]^i. P \} \uplus \q$ where $i \in \gamma$ and $\ell = \tau$.

Since $\quadruple{\Ec}{\p}{\Phi}{\sigma}$ is compatible, then by
Theorem~\ref{theo:main-main} and in particular 
Proposition~\ref{pro:shared-to-disjoint}, we deduce that $\fst(D^{\vsequ}_n) \LRstep{\tr} \quadruple{\Ec'}{\p'}{\Phi'}{\sigma'}$ where $\delta(\sigma\mydownarrow) = \sigma'\mydownarrow$ and $\delta(\p) = \p'$. It implies that $\p' = \{ [x := \delta_\gamma(t)]^i. \delta(P)\} \uplus \delta(\q)$.
Thus, by Lemma~\ref{lem:Testandequality}, we have that $\delta_\gamma(t\sigma\mydownarrow) = \delta_\gamma(t)\sigma'\mydownarrow$. 

On the other hand, if $A$ is not compatible, it means that there exists $y \in \dom(\sigma)$ such that $t\sigma\mydownarrow = y\sigma\mydownarrow$ is not equivalent to $x\rho_\beta = y\rho_\beta$. But $\delta(\sigma\mydownarrow) = \sigma'\mydownarrow$ and $\delta_\gamma(t\sigma\mydownarrow) = \delta_\gamma(t)\sigma'\mydownarrow$.  Hence $t\sigma\mydownarrow = y\sigma\mydownarrow$ is equivalent to $\delta_\gamma(t)\sigma' = y\sigma'$, and so we can deduce that $\delta_\gamma(t)\sigma' = y\sigma'$ is not equivalent to $x\rho_\beta = y\rho_\beta$. 
However, $\quadruple{\Ec'}{\p'}{\Phi'}{\sigma'}$ can also apply the internal transition {(\sc Assgn)} on $[ x := \delta(t) ]^i$ and so we obtain  $\quadruple{\Ec'}{\p'}{\Phi'}{\sigma'} \lrstep{\tau} A'$ with $A'$ not compatible with $(\rho_\alpha,\rho_\beta)$. This is in contradiction with our result on $D^{\vsequ}_n$.
\end{proof}

When the composition context is of the form $C'[!\_]$, the previous lemma does not hold. However, we will show that we can modify any trace to become a compatible trace by applying some permutation on the indices of the names. 
Intuitively, when considering a trace of $S^\vari_n$, if
$x_1[i_1,\ldots, i_m]$ is equal to $x_2[i_1,\ldots, i_{m-1}, i'_m]$
after instantiation with $i_m \neq i'_m$, we want to permute all names
of the form $t[i_1,\ldots, i_{m-1}, i'_m]$ by $t[i_1,\ldots, i_{m-1},
i_m]$. Such permutation is possible since we only consider composition
context of the form $C'[!\_]$. We
will call this an \emph{index permutation}. 
To ensure that such a permutation is always possible when needed, we
simply ensure that we have enough processes that have not started
their execution by requiring that $2n' \le n$ (\emph{i.e.} the length
$n'$ of the derivation under study is two times smaller than the
number of the unfolding we consider).

\begin{lemma}
\label{lem:existence of compatible}
Let $(\rho_\alpha,\rho_\beta)$ be the two abstraction functions of
$S^\vari_n$. For all $S^\vari_n \lrstep{\ell_1} A_1 \lrstep{\ell_2} \ldots
\lrstep{\ell_{n'}} A_{n'}$ with ${2n'} < n$, there exists an index
permutation 
such that $S^\vari_n
\lrstep{\ell_1} A'_1 \lrstep{\ell_2} \ldots \lrstep{\ell_{n'}}
A'_{n'}$ with $A'_k$ being the application of the index permutation on
$A_k$ 
for all $k = 1\ldots n'$, and $A'_{n'}$ is compatible with $(\rho_\alpha,\rho_\beta)$.
\end{lemma}

\begin{proof}
We prove the result by induction on $n'$. The initial step $n' = 0$
being trivial, we focus on the inductive step $n' > 0$. By hypothesis,
we know that there exists an index permutation 
such that $S^\vari_n \lrstep{\ell_1} A'_1 \lrstep{\ell_2} \ldots
\lrstep{\ell_{n'}} A'_{n'-1}$ where $A'_k$ being the application of
the index permutation on $A_k$ 
for all $k = 1\ldots n'-1$, and $A'_{n'-1}$ is compatible with
$(\rho_\alpha,\rho_\beta)$. However, we know that $A_{n'-1}
\lrstep{\ell_{n'}} A_{n'}$. Since $A'_{n'-1}$ is obtained from
$A_{n'-1}$ by an index permutation, 
then $A'_{n'-1} \lrstep{\ell_{n'}} A'_{n'}$ where $A'_{n'}$ is the application of the index permutation on $A_{n'}$. 

Assume first that the transition $A_{n'-1} \lrstep{\ell_{n'}} A_{n'}$ is different from the internal transition {(\sc Assgn)}, then the compatibility of $A'_{n'-1}$ implies the compatibility of $A'_{n'}$. Hence the result holds.

Assume now that the transition $A_{n'-1} \lrstep{\ell_{n'}} A_{n'}$ is the internal transition {(\sc Assgn)}. Consider that $A'_{n'-1} = \quadruple{\Ec}{\p}{\Phi}{\sigma}$ with $\p = \{ [ x := t ]^i. P \} \uplus \q$. Since $A'_{n'-1}$ is compatible with $(\rho_\alpha,\rho_\beta)$, we can apply Proposition~\ref{pro:shared-to-disjoint}. Using similar reasoning as in proof of Lemma~\ref{lem:compatibility agree}, we obtain that $t\sigma = y\sigma$ for some assignment variable $y$ implies w.l.o.g. that $x = x_1[i_1,\ldots, i_{m-1},i_m]$ and $y = x_2[i_1,\ldots, i_{m-1},i'_m]$. Thus, by applying the index permutation between $i_m$ and $i'_m$ on each $A'_k$, we obtain that $S^v_n \lrstep{\ell_1} A''_1 \lrstep{\ell_2} \ldots \lrstep{\ell_{n'}} A''_{n'}$ with $A''_{n'}$ compatible with $(\rho_\alpha,\rho_\beta)$, and $A''_k$ being the application of the index permutation on $A'_k$, for all $k \in \{1, \ldots, m\}$.
\end{proof}


\subsection{Composing diff-equivalence}

We are now able to prove our composition results.

\smallskip{}

\theoequivseq*

\begin{proof}
Let  $S = \quadruple{\Ec_0}{C[P_1[Q_1] \mid P_2[Q_2]]}{\Phi
  \uplus \Psi}{\emptyset}$.  Thanks to Lemma~\ref{lem:unfolding-diff},
we know that $S$ is in diff-equivalence  if, and only if, $S_n$ is in
diff-equivalence for all $n \in \mathbb{N}$.

By hypothesis, we know that:
\begin{itemize}
\item  $\quadruple{\Ec_0}{C[P_1[0] \mid
      P_2[0]]}{\Phi}{\emptyset}$, and 
\item $\quadruple{\Ec_0}{C[\new\ k. (Q_1\{^k/_{x_1}\} \mid
      Q_2\{^k/_{x_2}\})]}{\Psi}{\emptyset}$
\end{itemize} are both in diff-equivalence and $P_1,P_2,Q_1,Q_2$ are tagged.
Hence, since diff-equivalence is preserved by disjoint parallel
composition, we deduce that $D^{\para}$ is in diff-equivalence, and
thus, thanks to Lemma~\ref{lem:unfolding-diff}, we obtain that
$D^{\para}_n$ is in diff-equivalence for all $n \in \mathbb{N}$. 
Applying Lemma~\ref{lem:par to seq}, we deduce that $D^{\sequ}_n$ is
also in diff-equivalence. Note that 
diff-equivalence still holds on the biprocess
$D^\vsequ_n$ obtained from $D^\sequ_n$ by
adding some assignment variables to ``explicit the sharing''.

Given $n \in \mathbb{N}$, 
in order to conclude, we have to show that
$S^\vari_n$ obtained from $S_n$ by adding some assignments variables
to explicit the sharing satisfies diff-equivalence.
We form two new biprocesses $SD_L$ and $SD_R$ as follows:
\begin{itemize}
\item $\fst(SD_L) = \fst(S^\vari_n)$ and $\snd(SD_L) =
  \fst(D^\vsequ_n)$;
\item  $\fst(SD_R) = \snd(S^\vari_n)$ and $\snd(SD_R) =
  \snd(D^\vsequ_n)$;
\end{itemize}
We will apply Theorem~\ref{theo:main-main} on biprocesses $SD_L$ and
$SD_R$ to establish the strong relationship between the two components
of each biprocess,
and together with the fact
$D^\vsequ_n$ satisfies diff-equivalence, this will allow us to
conclude that $S^\vari_n$ satisfies diff-equivalence too.

\smallskip{}

Considering the two abstraction functions $(\rho_\alpha,\rho_\beta)$
as defined in Section~\ref{subsec:rho}, in order to apply
Theorem~\ref{theo:main-main} on $SD_L$ (resp. $SD_R$), it remains to show
that $\fst(D^{\vsequ}_n)$ and $\snd(D^\vsequ_n)$ do not reveal the
value of their assignment variables w.r.t. $(\rho_\alpha,\rho_\beta)$.
This is actually achieved by application of Lemma~\ref{lem:assignment
  variable non deduce} with the facts that
\begin{itemize}
\item $\quadruple{\Ec_0}{C[\new\ k. (Q_1\{^k/_{x_1}\} \mid Q_2\{^k/_{x_2}\})]}{\Psi}{\emptyset}$ and $\quadruple{\Ec_0}{C[P_1[0] \mid P_2[0]]}{\Phi}{\emptyset}$ do not reveal key in $\{ n, \pk(n), \vk(n) \mid n \in\fn(P_1,P_2) \cap \fn(Q_1,Q_2) \cap \bn(C)\}$, and 
\item $\quadruple{\Ec_0}{C[\new\ k. (Q_1\{^k/_{x_1}\} \mid Q_2\{^k/_{x_2}\})]}{\Psi}{\emptyset}$ do not reveal $k,\pk(k), \vk(k)$, and
\item $P_1/P_2/\Phi$ is a good key-exchange protocol under $\Ec_0$ and $C$, that implies in particular that $\triple{\Ec_0}{P_{\good}}{\Phi}$ does not reveal $\bad$ where $P_\good$ is defined as follows:
\[
\begin{array}{l}
P_\good = \new\, \bad, d. \big(\\
\;\;\;C[\new\, id. ( P_1[\Out(d,\langle x_1, id\rangle)] \mid P_2[\Out(d,\langle x_2,id\rangle)])] \\
\;\;\;\mid \In(d,x). \In(c,z). \\
\;\;\;\quad\myIf\, z \in \{\proj_1(x), 
\pk(\proj_1(x)), \vk(\proj_1(x))\}\\
\;\;\;\quad\myThen\, \Out(c,\bad)\big)\\
\end{array}
\]
\end{itemize}

\smallskip{}

Now, let $B_S$ be a biprocess such that  
\[S^\vari_n \LRstep{\tr}_\bi
B_S \stackrel{\defi}{=} \quadruple{\Ec_S}{\p_S}{\Phi_S}{\sigma_S}
\]
for
some $\tr$. By definition of diff-equivalene, we have to show that:
\begin{enumerate}
\item $\new\ \Ec_S. \fst(\Phi_S) \statequiv \new\
  \Ec_S. \snd(\Phi_S)$;
\item if $\fst(B_S) \lrstep{\ell} A_L$ then there exists $B'$ such
  that ${B_S \lrstep{\ell}_\bi B'}$ and $\fst(B') = A_L$ (and similarly
  for $\snd$).
\end{enumerate}

Let us now focus on the case where the composition context is not of the form $C[!\_]$.

We have $\fst(S^\vari_n) \LRstep{\tr} \fst(B_S)$ as well as $\snd(S^\vari_n)
\LRstep{\tr} \snd(B_S)$.
By Lemma~\ref{lem:compatibility agree}, we obtain that $\fst(B_S)$
as well as $\snd(B_S)$ is compatible with $(\rho_\alpha,\rho_\beta)$. 
Hence, relying on Theorem~\ref{theo:main-main} (first item), we deduce that
there exist biprocesses $SD'_L$  and $SD'_R$ such that:
\begin{itemize}
\item $SD_L \LRstep{\tr}_\bi SD'_L$, $\fst(SD'_L) = \fst(B_S)$, and
  static equivalence holds between the two frames issued from the
  biprocess $SD'_L$;
\item $SD_R \LRstep{\tr}_\bi SD'_R$, $\fst(SD'_R) = \snd(B_S)$, and
  static equivalence holds between the two frames issued from the
  biprocess $SD'_R$.
\end{itemize}

Since, we know that $D^\vsequ_n$ satisfies diff-equivalence, we have
that $D^\vsequ_n \LRstep{\tr}_\bi D'^\vsequ_n$ with $\fst(D'^\vsequ_n)
= \snd(SD'_L)$ and $\snd(D'^\vsequ_n) = \snd(SD'_R)$. Then,  by
transitivity of static equivalence, we deduce that 
\[
\new\ \Ec_S. \fst(\Phi_S) \statequiv \new\
  \Ec_S. \snd(\Phi_S).
\]

Now, assume that $\fst(B_S) \lrstep{\ell} A_L$. In such a case, we
have that $\fst(S_n) \LRstep{\tr} \fst(B_S) \lrstep{\ell} A_L$. 
By Lemma~\ref{lem:compatibility agree}, we obtain that $A_L$
is compatible with $(\rho_\alpha,\rho_\beta)$, and relying on
Theorem~\ref{theo:main-main} (first item), we deduce that there exists a biprocess
$SD''_L$ such that:
 $SD_L \LRstep{\tr}_\bi \,\_ \lrstep{\ell}_\bi SD''_L$ with
 $\fst(SD''_L) = A_L$.
Since $D^\vsequ_n$ satisfies diff-equivalence, we have that
$D^\vsequ_n \LRstep{\tr}_\bi \;\lrstep{\ell}_\bi D''^\vsequ_n$ for
some biprocess $D''^\vsequ_n$ with $\fst(D''^\vsequ_n) =
\snd(SD''_L)$.
Now, applying Theorem~\ref{theo:main-main} (second item)
 on biprocess $SD_R$, we deduce that $SD_R  \LRstep{\tr}_\bi \, \_
 \lrstep{\ell}_\bi SD''_R$ with $\snd(SD''_R) = \snd(D''^\vsequ_n)$.
This allows us to ensure the existence of the biprocess~$B'$ required
to show diff-equivalence of $S^\vari_n$. We will have $\fst(B') =
\fst(SD''_L) = A_L$ and $\snd(B') = \fst(SD''_R)$.

\medskip

In the case where the composition context is of the form $C'[!\_]$, all the traces issued from $S^\vari_n$ are not compatible anymore w.r.t. the
abstraction functions $\rho_\alpha$ and $\rho_\beta$. 
Nevertheless, thanks to Lemma~\ref{lem:existence of compatible}, we
can always find a similar trace that is compatible, then using
Theorem~\ref{theo:main-main}, we will ensure that these traces also exist
in the disjoint case, and we also ensure their compatibility (see Proposition~\ref{pro:shared-to-disjoint}).

Then, relying on the diff-equvialence of the biprocess $\triple{\Ec_0}{\new\, d. C[P^+]}{\Phi}$, 
we deduce that for any trace $D^{\vsequ}_n \LRstep{\tr}_\bi D'$,
$\fst(D')$ is compatible w.r.t. 
$(\rho_\alpha,\rho_\beta)$ if and only if  $\snd(D')$ is compatible
w.r.t. $(\rho_\alpha,\rho_\beta)$. This allows us to ensure that
$D^\sequ_n$ is also in diff-equivalence when considering compatible
traces only. Thanks to this, we are able to conclude as in we did in the case where the composition context were not of the form $C'[!\_]$.
\end{proof}

\subsection{Composing reachability}

\theoreachseq*

\begin{proof}
Let  $S = \quadruple{\Ec_0}{C[P_1[Q_1] \mid P_2[Q_2]]}{\Phi
  \uplus \Psi}{\emptyset}$.  
By hypothesis, we know that:
\begin{itemize}
\item $\quadruple{\Ec_0}{C[P_1[0] \mid
      P_2[0]]}{\Phi}{\emptyset}$, and 
\item $\quadruple{\Ec_0}{C[\new\ k. (Q_1\{^k/_{x_1}\} \mid
      Q_2\{^k/_{x_2}\})]}{\Psi}{\emptyset}$
\end{itemize}
 does not reveal $s$. Since, secrecy is preserved by disjoint
 composition, and the transformations introduced at the beginning of
 the section (\emph{e.g.} unfolding, adding assignment variables,
 ...), we easily deduce that $D^\vsequ_n$ do not reveal $s$.

We show the result by contradiction. Assume that $S^\vari_n$
reveals the secrecy $s$. We consider a trace witnessing this fact, \emph{i.e.}
a process $S'^\vari_n$ such that 
\[
S^\vari_n \LRstep{\tr}
S'^\vari_n  \stackrel{\defi}{=}
\quadruple{\Ec_S}{\p_S}{\Phi_S}{\sigma_S}
\]
and for which $\new\ \Ec_S. \Phi_S \vdash s$.

We form a
biprocess $SD$ by grouping together $S^\vari_n$ and $D^\vsequ_n$ in
order to apply Theorem~\ref{theo:main-main}.

In order to apply Theorem~\ref{theo:main-main}, we first must prove that
$D^{\vsequ}_n$ does not reveal the value of its assignment variables
w.r.t. $(\rho_\alpha,\rho_\beta)$ as defined in
Section~\ref{subsec:rho}. This is achieved by application of
Lemma~\ref{lem:assignment variable non deduce} with the facts that
\begin{itemize}
\item $\quadruple{\Ec_0}{C[\new\ k. (Q_1\{^k/_{x_1}\} \mid Q_2\{^k/_{x_2}\})]}{\Psi}{\emptyset}$ and $\quadruple{\Ec_0}{C[P_1[0] \mid P_2[0]]}{\Phi}{\emptyset}$ do not reveal key in $\{ n, \pk(n), \vk(n) \mid n \in\fn(P_1,P_2) \cap \fn(Q_1,Q_2) \cap \bn(C)\}$, and 
\item $\quadruple{\Ec_0}{C[\new\ k. (Q_1\{^k/_{x_1}\} \mid Q_2\{^k/_{x_2}\})]}{\Psi}{\emptyset}$ do not reveal $k,\pk(k), \vk(k)$, and
\item $P_1/P_2/\Phi$ is a good key-exchange protocol under $\Ec_0$ and $C$, that implies in particular that $\triple{\Ec_0}{P_{\good}}{\Phi}$ does not reveal $\bad$ where $P_\good$ is defined as follows:
\[
\begin{array}{l}
P_\good = \new\, \bad, d. \big(\\
\;\;\;C[\new\, id. ( P_1[\Out(d,\langle x_1, id\rangle)] \mid P_2[\Out(d,\langle x_2,id\rangle)])] \\
\;\;\;\mid \In(d,x). \In(c,z). \\
\;\;\;\quad\myIf\, z \in \{\proj_1(x), 
\pk(\proj_1(x)), \vk(\proj_1(x))\}\\
\;\;\;\quad\myThen\, \Out(c,\bad)\big)\\
\end{array}
\]
\end{itemize}

As done previously, relying on Lemma~\ref{lem:compatibility
    agree} (or Lemma~\ref{lem:existence of
    compatible} in case $C$ is of the form $C'[!\_]$), we may assume that the
  trace under study is compatible.  Applying
Theorem~\ref{theo:main-main}, we deduce that there exists a biprocess
$SD'$ such that $SD \LRstep{\tr}_\bi SD'$ with $\fst(SD') =
S'^\vari_n$, and static equivalence holds between the two frames
issued from the biprocess $SD'$. Moreover, if we denote by $\Phi_S$ and $\Phi_D$ the respective frame of $\fst(SD')$ and $\snd(SD')$, we ensure that $\delta(\Phi_S\mydownarrow) = \Phi_D\mydownarrow$ (see Proposition~\ref{pro:shared-to-disjoint}).

Therefore, since $D^{\vsequ}_n$ does not reveal the secret $s$, and we already proved that $D^{\vsequ}_n$ does not reveal his assignment variables, then by Lemma~\ref{lem:samerecipesymmetric}, we can deduce that $S^\vari_n$ does not reveal $s$, and so $S$ does not reveal~$s$ either.
\end{proof}

\end{document}